\documentclass[onesided]{illcdiss}

\usepackage{latexsym}
\usepackage{amsmath}
\usepackage{amssymb}
\usepackage{amsfonts}
\usepackage{ifthen}
\usepackage{revsymb}
\usepackage{graphicx}
\usepackage{backref}
\usepackage{longtable}
\usepackage{mathrsfs}
\def\nn{\nonumber}

\def\locHam{{\sc local Hamiltonian}}
\def\gappedlocHam{{\sc Gapped local Hamiltonian}}
\def\partition{{\sc Partition Function}}
\def\spec{{\sc Spectral Density}}
\def\BQP{{\sf{BQP}}}
\def\COLORING{{\sf{COLORING}}}
\def\FACTORING{{\sf{FACTORING}}}
\def\BPP{{\sf{BPP}}}
\def\QCMA{{\sf{QCMA}}}
\def\QMA{{\sf{QMA}}}
\def\UQCMA{{\sf{UQCMA}}}
\def\BellQMA{{\sf{BellQMA}}}
\def\LOCCQMA{{\sf{LOCCQMA}}}
\def\DQC{{\sf{DQC1}}}
\def\NP{{\sf{NP}}}
\def\UNP{{\sf{UNP}}}
\def\UcoNP{{\sf{Uco-NP}}}
\def\MA{{\sf{MA}}}
\def\A{{\sf{A}}}
\def\PSPACE{{\sf{PSPACE}}}
\def\P{{\sf{P}}}
\def\PP{{\sf{PP}}}
\def\poly{{\sf{poly}}}
\def\SAT{{\sf{SAT}}}
\def\YES{{\sf{YES}}}
\def\NO{{\sf{NO}}}
\def\eff{{\sf{eff}}}

\newcommand{\bra}[1]{\langle #1|}
\newcommand{\ket}[1]{|#1\rangle}
\newcommand{\braket}[2]{\langle #1|#2\rangle}
\newcommand{\tr}{\text{tr}}

\def\01{\{0,1\}}

\newcommand{\id}{\mathbb{I}}

\newenvironment{proof}
{\noindent {\bf Proof. }}
{{\hfill $\Box$}\\
 \smallskip}

\newcommand{\rank}{\mbox{rank}}

\newcounter{protoCount}
\newcounter{protoList}
\newsavebox{\tmpbox}
\newlength{\protobox}

\usepackage{makeidx}     % only with Latex2e
\setlongtables %XXX add for printing

\makeindex
\usepackage{eurosans} %\usepackage{jmacro} \figuresarein{fig/}

\usepackage{afterpage}

\begin{document}

\cleardoublepage
\pagestyle{headings}
\pagenumbering{arabic}

%\noindent

%\cleardoublepage

%\include{titlepage}
%auto-ignore
\pagestyle{empty}

\begin{center}
  {\sc{University of London}}\\[1.0cm]
  \centerline{\includegraphics[height=5cm]{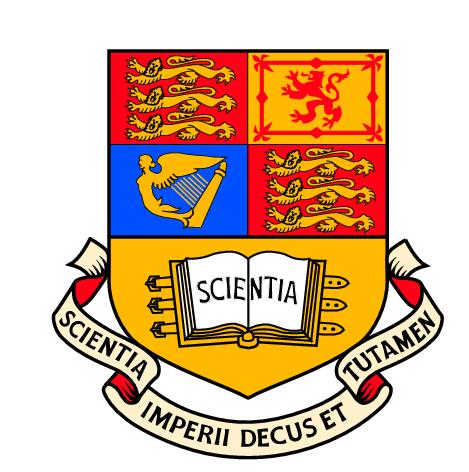}%\hspace{1.5cm}
}
  \vspace{1.0cm}
  Imperial College of
  Science, Technology and Medicine\\ The
  Blackett Laboratory\\ Quantum Optics \& Laser Science Group\\[1.5cm]
  
  {\Huge Entanglement Theory and the Quantum}\\[0.5cm]
  {\Huge Simulation of Many-Body Physics}\\[0.5cm]
  %{\Huge Quantum Optical Systems}\\[2cm]
 
  \vspace{2 cm}

  {\Large{Fernando Guadalupe dos Santos Lins Brand\~ao}}\\[1.5cm]

  Thesis submitted in partial fulfillment of the \\
  requirements for the degree of\\
  Doctor of Philosophy\\
  of the University of London\\
  and the Diploma of Membership of Imperial College.\\[.8cm]
  \vfill
August 2008
\end{center}

% Local Variables:
% TeX-master: "./thesis"
% TeX-master: "C:\\Documents and Settings\\dan\\My Documents\\Thesis\\Chapters\\thesis"
% End:

%\include{newfrontmatter}

%auto-ignore
%\chapter{Abstract}
\clearpage \pagestyle{empty}

%\clearpage \pagestyle{empty}
%\begin{center}
%  {\Huge Acknowledgements}\\[1cm]
%\end{center}
%To people

\begin{center}
  {\Huge Abstract}\\[1cm]
\end{center}

Quantum mechanics led us to reconsider the scope of physics and its building principles, such as the notions of realism and locality. More recently, quantum theory has changed in an equally dramatic manner our understanding of information processing and computation. On one hand, the fundamental properties of quantum systems can be harnessed to transmit, store, and manipulate information in a more efficient and secure way than possible in the realm of classical physics. On the other hand, the development of systematic procedures to manipulate systems of a large number of particles in the quantum regime, crucial to the implementation of quantum-based information processing, has triggered new possibilities in the exploration of quantum many-body physics and related areas. 

In this thesis, we present new results relevant to two important problems in quantum information science: the development of a theory of entanglement, intrinsically quantum correlations, and the exploration of the use of controlled quantum systems to the computation and simulation of quantum many-body phenomena. 

In the first part we introduce a new approach to the study of entanglement by considering its manipulation under operations not capable of generating entanglement. In this setting we show how the landscape of entanglement conversion is reduced to the simplest situation possible: one unique measure completely specifying which transformations are achievable. This framework has intriguing connections to the foundations of the second law of thermodynamics, which we present and explore. On the way to establish our main result, we develop new techniques that are of interest in their own. First, we extend quantum Stein's Lemma, characterizing optimal rates in state discrimination, to the case where the alternative hypothesis might vary over particular sets of possibly correlated states. Second, we employ recent advances in quantum de Finetti type theorems to decide the distillability of the entanglement contained in correlated sequences of states, and to find a new indication that bound entangled states with a non-positive partial transpose exist. 

In the second part we study the usefulness of a quantum computer for calculating properties of many-body systems. Our first result is an application of the phase estimation algorithm to calculate additive approximations to partition functions and spectral densities of quantum local Hamiltonians. We give convincing evidence that quantum computation is superior to classical in solving both problems by showing that they are hard for the class of problems efficiently solved by the one-clean-qubit model of quantum computation, which is believed to contain classically intractable problems. We then present a negative result on the usefulness of quantum computers in determining the ground-state energy of local Hamiltonians: Even under the promise that the spectral gap of the Hamiltonian is larger than an inverse polynomial in the number of sites, already for one-dimensional Hamiltonians the problem is hard for the class Quantum-Classical-Merlin-Arthur, which is believed to contain intractable problems for quantum computation. We also present an application of ideas from entanglement theory into the analysis of quantum verification procedures.

In the third and last part, we approach the problem of quantum simulating many-body systems from a more pragmatic point of view. Based on recent experimental developments on cavity quantum electrodynamics, we propose and analyze the realization of paradigmatic condensed matter Hamiltonians, such as the Bose-Hubbard and the anisotropic Heisenberg models, in arrays of coupled microcavities. We outline distinctive properties of such systems as simulators of quantum many-body physics, such as the full addressability of individual sites and the access to inhomogeneous models, and discuss the feasibility of an experimental realization with state-of-the-art current technology.

%{\singlespacing \tableofcontents}

\tableofcontents

%\listoffigures
%\listoftables

% Local Variables:
% TeX-master: "./thesis"
% TeX-master: "C:\\Documents and Settings\\dan\\My Documents\\Thesis\\Chapters\\thesis"
% End:

%\include{publications} 
%% Format for most of  thesis

%\setlength{\parindent}{0.0in}
%\setlength{\parskip}{0.1in}

\clearpage

\pagestyle{plain}

\vspace{4 cm}

\noindent
Parts of this thesis is based on material published in the following papers:
\begin{itemize}

\vspace{1 cm}

\item 
\cite{BP08b} \textbf{Entanglement Theory and the Second Law of Thermodynamics}\\
F.G.S.L. Brand\~ao and M.B. Plenio\\
Nature Physics, in press\\
(Chapter~\ref{reversible})

\item 
\cite{BP08c} \textbf{A Generalization of Quantum Stein's Lemma}\\
F.G.S.L. Brand\~ao and M.B. Plenio\\
In preparation \\
(Chapter~\ref{QHT})

\item 
\cite{HBP08} \textbf{Quantum Many-Body Phenomena in Coupled Cavity Arrays }\\
M.J. Hartmann, F.G.S.L. Brand\~ao, M.B. Plenio \\
Submitted - arXiv:0808.2557 \\
(Chapters~\ref{AMO}, \ref{BH}, \ref{kerrnonlinearity}, \ref{heisenberg})

\item 
\cite{BP08} \textbf{A Reversible Theory of Entanglement and its Connection to the Second Law}\\
F.G.S.L. Brand\~ao and M.B. Plenio\\
Submitted - arXiv:0710.5827 \\
(Chapter~\ref{reversible})

\item
\cite{HBP07b} \textbf{A polaritonic two-component Bose-Hubbard model }\\
M.J. Hartmann, F.G.S.L. Brand\~ao, M.B. Plenio \\
New J. Phys. \textbf{10}, 033011 (2008)\\
(Chapter~\ref{BH})

\item
\cite{BE08} \textbf{Correlated entanglement distillation and the structure of the set of undistillable states}\\
F.G.S.L. Brand\~ao and J. Eisert \\
J. Math. Phys. \textbf{49}, 042102 (2008)  \\
(Chapter~\ref{NPPT})

\item
\cite{BHP07} \textbf{Light-shift-induced photonic nonlinearities}\\
F.G.S.L. Brand\~ao, M.J. Hartmann, M.B. Plenio\\
New J. Phys. \textbf{10}, 043010 (2008) \\
(Chapters~\ref{kerrnonlinearity}) 

\item
\cite{HBP07a} \textbf{Effective spin systems in coupled micro-cavities}\\
M.J. Hartmann, F.G.S.L. Brand\~ao, M.B. Plenio\\
Phys. Rev. Lett. \textbf{99}, 160501 (2007)\\
(Chapter~\ref{heisenberg})

\item
\cite{HBP06} \textbf{Strongly Interacting Polaritons in Coupled Arrays of Cavities}\\
M.J. Hartmann, F.G.S.L. Brand\~ao, M.B. Plenio \\
Nature Physics \textbf{2}, 849 (2006) \\
(Chapter~\ref{BH})

\end{itemize}

\vspace{1 cm}

\noindent
Other papers not covered in this thesis to which the author contributed:

\vspace{1 cm}

\begin{itemize}

\item 
\cite{CSC+08} \textbf{Geometrically induced singular behavior of entanglement}\\
D. Cavalcanti, P.L. Saldanha, O. Cosme, F.G.S.L. Brand\~ao, C.H. Monken, S. Padua, M. Franca Santos, M.O. Terra Cunha \\
Phys. Rev. A \textbf{78}, 012318 (2008) 

\item
\cite{BHPV07} \textbf{Remarks on the equivalence of full additivity and monotonicity for the entanglement cost}\\
F.G.S.L. Brand\~ao, M. Horodecki, M.B. Plenio, S. Virmani \\
Open Sys. Inf. Dyn. \textbf{14}, 333 (2007) 

\item
\cite{Bra07} \textbf{Entanglement activation and the robustness of quantum correlations }\\
F.G.S.L. Brand\~ao\\
Phys. Rev. A \textbf{76}, 030301(R) (2007) 

\item
\cite{EBA07} \textbf{Quantitative entanglement witnesses}\\
J. Eisert, F.G.S.L. Brand\~ao, K.M.R. Audenaert \\
New J. Phys. \textbf{9}, 46 (2007) 

\item
\cite{CBC06} \textbf{Entanglement quantifiers and phase transitions}\\
D. Cavalcanti, F.G.S.L. Brand\~ao, M.O. Terra Cunha \\
New J. Phys. \textbf{8}, 260 (2006) 

\item
\cite{Bra05b} \textbf{Entanglement as a quantum order parameter}\\
F.G.S.L. Brand\~ao\\
New J. Phys. \textbf{7}, 254 (2005) 

\item
\cite{CBC05} \textbf{Are all maximally entangled states pure?}\\
D. Cavalcanti, F.G.S.L. Brand\~ao, M.O. Terra Cunha \\
Phys. Rev. A \textbf{72}, 040303(R) (2005)

\item
\cite{Bra05} \textbf{Quantifying Entanglement with Witness Operators}\\
F.G.S.L. Brand\~ao \\
Phys. Rev. A \textbf{72}, 022310 (2005)

\item
\cite{BV06} \textbf{Witnessed Entanglement}\\
F.G.S.L. Brand\~ao and R.O. Vianna \\
Int. J. Quant. Inf. \textbf{4}, 331 (2006) 

\item
\cite{BV04b} \textbf{Separable Multipartite Mixed States - Operational Asymptotically Necessary and Sufficient Conditions}\\
F.G.S.L. Brand\~ao and R.O. Vianna \\
Phys. Rev. Lett. \textbf{93}, 220503 (2004)

\item
\cite{BV04a} \textbf{A Robust Semidefinite Programming Approach to the Separability Problem}\\
F.G.S.L. Brand\~ao and R.O. Vianna \\
Phys. Rev. A \textbf{70}, 062309 (2004)

\end{itemize}

\cleardoublepage

\renewcommand{\chaptername}{Chapter} 

\acknowledgments

This thesis would have never been written without the help, incentive, and collaboration of countless people. 

First I would like to thank my supervisor Martin Plenio for his support, guidance, and friendship during the last three years. I have learned from him not only a lot about science, but also - and perhaps more importantly - about how to do research in an efficient, yet relaxed and enjoyable way. I also thank Michael Hartmann for our very fruitful collaboration and for teaching me a lot of new and interesting physics.  

The work I have done throughout the past five years - first as a master student and then as a PhD student - benefited from collaborations with several colleagues, who have unmistakably contributed to my formation. For this I thank Koenraad Audenaert, Daniel Cavalcanti, Olavo Cosme, Nilanjana Datta, Jens Eisert, Marcelo Franca, Michael Hartmann, Michal Horodecki, Carlos Henrique Monken, Martin Plenio, Sebasti\~ao Padua, Pablo Saldanha,  Marcelo Terra Cunha, Reinaldo Vianna, and Shash Virmani.  

The stimulating atmosphere of the Institute of Mathematical Sciences and the Quantum Information Group had an important impact in my research during the past three years. Many thanks to Koenraad Audenaert, Etienne Brion, Adolfo del Campo, Marcus Cramer, Aminash Datta, Chris Dawson, Jens Eisert, Michael Hartmann, Miguel Navascues, Masaki Owari, Martin Plenio, Kenneth Pregnell, Alex Retzker, Terry Rudolph, Alessio Serafini, Shashank Virmani, Harald Wunderlich, my floormates Oscar Dahlsten, Alvaro Feito, David Gross, Konrad Kieling, Doug Plato, and Moritz Reuter, and the affiliated members from QOLS. For making my life easier and my time at work more alive, I would like to thank Eileen Boyce.           

My studies on quantum information started already in my undergraduate and master studies at Universidade Federal of Minas Gerais. I am grateful to Reinaldo Vianna for getting me started in this exciting field and several other people in the department of physics for contributing to my formation. 

I also thank Peter Knight and Andreas Winter for being the examiners of the thesis and for their useful comments. 

I had the privilege to visit several research institutes during my time as a PhD student. For the hospitality and many discussions, I thank the Horodecki family and the quantum information group in Gdansk; Reinhard Werner and the group of Braunschweig; Andreas Winter and the group of Bristol;  Nilanjana Datta, Berry Groisman, Jonathan Oppenheim and the Cambridge group; and Giovanna Morigi and the group of Universitat Autonoma de Barcelona. I have benefited from discussions, correspondence, and guidance from many other colleagues, too many to individually thank each one of them without unwarranted omissions. 

I thank the Brazilian agency Conselho Nacional de Desenvolvimento Cient\'ifico e Tecnol\'ogico for financial support. 

Finally, I would like to warmly thank everyone of my family, in particular my brothers Bernardo and Pedro and my parents Magda and Jacyntho, who teached me from an early age the great fun of learning new things and always supported my studies and everything else I did in my life. The last seven years and specially my time in London would have not been the same without my girlfriend - and soon to be wife - Roberta. Her love and constant support have been crucial to the development of this thesis, which is dedicated to her.

\chapter{Introduction} \label{introduction}

Quantum mechanics has a profound impact on the way we perceive Nature and its fundamental principles \cite{Per93}. In particular, two or more quantum systems can be correlated in a way that defies any explanation in terms of classical shared randomness - correlations that can be created solely by the communication of bits. This type of quantum correlations, termed \textit{entanglement}, is responsible for the impossibility of giving a local realistic interpretation to quantum theory \cite{Bel87} and has been analysed on a foundational level since the beginning of the theory. In the last twenty years, it has emerged that entanglement also has a distinguished role in information processing. It turns out that entangled quantum systems can be harnessed to transmit, store, and manipulate information in a more efficient and secure way than possible in the realm of classical physics (see e.g. \cite{NC00}). In some sense, entanglement can be seen as a \textit{resource} to process information stored in quantum systems in ways that are different, and sometimes superior, to classical information processing. Part I of the thesis deals with entanglement theory and consists of chapters \ref{entanglement}, \ref{NPPT}, \ref{QHT}, and \ref{reversible}. 

\begin{description}

\item \textbf{Chapter \ref{entanglement}} reviews some results of entanglement theory that we use in the subsequent chapters. Although the discussion is not intended to be a complete survey of the subject, it covers some of the key aspects of quantum correlations. These include entanglement detection methods, with emphasis on entanglement witnesses and the Peres-Horodecki (positive partial transposition) criterion; entanglement distillation, bound (undistillable) entanglement, and entanglement activation; and entanglement measures, including the entanglement cost, the distillable entanglement, the relative entropy of entanglement, and the robustness of entanglement. Finally, we review recent extensions of the seminal de Finetti theorem for probability distributions to quantum states, a tool that will prove extremely useful in chapters \ref{NPPT} and \ref{QHT}. 

\item \textbf{Chapter \ref{NPPT}} is concerned with entanglement distillation of correlated states. The view of entanglement as a resource naturally leads to the idea of concentrating - or distilling - entanglement from a noisy form to a pure one, by employing quantum local operations and classical communication only. Determining the circumstances under which entanglement can be distilled is an important problem, one which to date has not been completely solved, even for independent and identically distributed (i.i.d.) sequences of bipartite states. It is known that the positiveness of the partial transpose implies the undistillability of the state. However, is that the complete history? One of the most notorious open problems in entanglement theory asks if the positiveness of the partial transpose is also necessary for undistillability. The conjecture is that in fact bound entanglement with a non-positive partial transpose (NPPT) exists, although the question still remains open. In order to make progress in this problem we propose to study entanglement distillation in a more general context. Instead of i.i.d. sequences, we allow for arbitrary correlations in the state, as long as all single copy reduced density matrices remain the same. The main result of this chapter is that the distillability of such sequences is intricately related to the distillability properties of the individual state, had we an i.i.d. sequence of copies of it. Although the generalization of a problem that we already cannot solve into a more complicated one might appear as a bad strategy to make progress in the former, we show that we do gain further insight into the conjecture about NPPT bound entanglement. The idea is that in this more general setting, some activation results - which show that in some cases bound entanglement can be pumped into another state and made useful - can be strengthened from the single copy case to an asymptotic scenario, which gives a new possible avenue to solve the conjecture. On a technical side, we employ the quantum de Finetti theorems to develop tools for studying entanglement in non-i.i.d. sequences of states, which turns out to be useful also in chapters \ref{QHT} and \ref{reversible}.                                                                      
                                                                                       
\item \textbf{Chapter \ref{QHT}} presents developments on an important subfield of quantum information theory: quantum hypothesis testing. Given several copies of a quantum system which is known to be described either by the state $\rho$ or $\sigma$, which we call the null and alternative hypotheses, what is the best measurement we can perform to learn the identity of the state at hand? This fundamental problem appears in various contexts and has been studied for the past twenty years under different further assumptions. A possible setting - called asymmetric hypothesis testing - considers the case in which we want to minimize to the extreme the probability of mistakenly identifying $\rho$ when $\sigma$ is the state of the system, while only requiring that the probability that $\sigma$ is identified in the place of $\rho$ goes to zero, at any rate, when the number of copies goes to infinity. It turns out that the optimal sequence of measurements gives rise to an exponentially decreasing probability of error with a rate given by the relative entropy of entanglement of $\rho$ and $\sigma$, $S(\rho || \sigma)$. This is the content of quantum Stein's Lemma, an extension of the same result for two probability distributions. There are some natural generalizations of the setting described. One is to consider arbitrary sequences of states $\{ \rho_n \}_{k \in \mathbb{N}}$ and $\{ \sigma_n \}_{k \in \mathbb{N}}$ instead of i.i.d. ones. The second is to let the two hypotheses to be composed not only of a single state each, but actually of a family of states. Generalizations of Stein's Lemma to the case in which the null hypothesis is an ergodic state and in which it is a family of i.i.d. states have been found. Extending the range of possibilities of the alternative hypothesis, however, remained as an open problem. The difficulty is that already in the case of probability distributions, there are non-i.i.d. distributions with very nice mixing properties for which the rate of decay is not even defined. Therefore one seems to need a different set of assumptions, beyond ergodicity and related concepts. The main result of this chapter is one possible generalization to the situation in which the alternative hypothesis is composed of a family of states, which can moreover be non-i.i.d.. We consider sets of states which satisfy some natural properties, the most important being the closedness under permutations of the copies. Then, employing once more the recently established quantum de Finetti theorems, we determine the exponent of the exponential decay of the error in a very similar fashion to quantum Stein's Lemma, in terms of the relative entropy. Although this result is not directly concerned with entanglement theory, it has interesting applications to it, two of which are discussed in the chapter. The results presented are also the key technical element for establishing the theory explored in chapter \ref{reversible}.  
                                                                                     
\item \textbf{Chapter \ref{reversible}} focuses on the analysis of a new paradigm for entanglement theory. The standard way to define entanglement is to start with local operations and classical communication (LOCC) and consider the class of states that cannot be created by such operations. More generally, this is a route commonly taken in resource theories: if only a restricted set of operations is available, there is usually a distinguished set of states which cannot be created by those and which can be used to lift the constraints on the operations available. These states are then seen as a resource in the theory. For the case of entanglement, one can construct a beautiful, but rather complex, 
theory of state transformations by LOCC, which is reviewed in chapter \ref{entanglement}. We might ask whether there are different, yet still meaningful, settings for which a simpler entanglement conversion theory emerges. This question was raised several years ago and is the motivation of this chapter. We propose a new paradigm for entanglement theory, and actually for general resource theories, which goes the opposite way from the standard situation: Given a certain resource, say entanglement, we consider its manipulation under \textit{any} operation not capable of generating it. Hence, instead of going from a restricted set of operations to a resource, we define the restricted set of operations which we might employ \textit{from} the resource under consideration. Although we might loose the operational motivation for the resource theory in question, the point is that a much simpler theory can be derived in this setting, which at the same time still gives relevant information about the original one. We prove that multipartite entanglement manipulations are \textit{reversible} in the asymptotic limit under non-entangling operations, being fully determined by a single quantity. The structure of the proof - largely based on the findings of chapter \ref{QHT} - is rather general and can be applied to other resource theories, such as secret-correlations in tripartite probability distributions, non-gaussianity, non-classicality, and non-locality of quantum states, and even superselection rules violations\footnote{Although in the thesis we present everything in terms of entanglement, the extension to some of those is completely elementary and will be discussed in a future work.}. From a foundational perspective, the paradigm we propose has remarkable connections to axiomatic approaches to the second law of thermodynamics, such as the one of Giles and Lieb and Yngvason. Indeed, although the two theories deals with completely different resources (entanglement and order) and have also a distinct range of applicability, we find that the structural form of the two theories, in the case of the second law explored most recently by Lieb and Yngvason, is actually the same. This gives new insight into previously considered analogies of entanglement theory and thermodynamics.  
\end{description}

The existence of entanglement in quantum theory also has dramatic consequences to our ability to simulate quantum systems in a classical computer: It is widely believed that it is impossible to simulate efficiently on a classical computer - with resources scaling polynomially in the number of particles - the dynamics of quantum many-body systems. While this is a major problem when studying such systems, it naturally leads to the idea of a quantum computer, in which controlled quantum systems are employed to perform computation in a more efficient way than possible by classical means. Such type of computer can efficiently simulate the dynamics of any local quantum many-body system. But it can actually do much more. There are several computational problems, some of them with no connections to physics, for which quantum computation appears to offer an exponential speed-up over classical computation (see e.g. \cite{NC00}), the most well known example being Shor's polynomial quantum algorithm for factoring \cite{Sho97}. Part II of the thesis is concerned with the use and limitations of a quantum computer in determining properties of many-body physics and consists of chapters \ref{complexity}, \ref{partition}, and \ref{QCMA}.

\begin{description}

\item \textbf{Chapter \ref{complexity}} has two strands. The first is an overview of some of the definitions and results from quantum complexity theory that we employ in the subsequent discussion. These include the definitions of the classes of problems efficiently solved, with high probability, by classical (\BPP) and quantum computation (\BQP), as well as by the more restricted model of quantum computation with only one clean qubit (\DQC); the quantum algorithm for phase estimation; the definition and basic properties of the complexity classes $\NP$ and $\MA$, and their quantum analogues $\QMA$ and $\QCMA$; and an overview of the seminal result by Kitaev that the determination, to polynomial accuracy, of the ground state energy of local quantum Hamiltonians is $\QMA$-complete, together with its most recent extension to one dimensional Hamiltonians. The second strand is an application of entanglement theory to quantum complexity theory. In the same way that we believe we can solve problems more efficiently with a quantum computer than with a classical one, we also believe that we can efficiently check the solution of a larger class of problems if we allow for quantum states as proofs and quantum computation to verify their correctness. This idea can be taken further to define the class $\QMA$. An interesting question is whether we can check the solutions of an even larger class of problems if instead of one quantum state as a proof, we are given several with the promise that they are \textit{not} entangled. Note that this is a question that only makes sense in the quantum setting, as classical proofs are obviously never entangled. We now know arguments that indicate that indeed in some cases many unentangled proofs are more powerful than a single one. But what if we have several unentangled proofs, but are only allowed to perform separate measurements on them? This question, recently posed by Aaronson \textit{et al}, is the focus of the second half of this chapter. Based on the understanding of entanglement shareability developed in recent years, we prove that, in fact, any fixed number of unentangled witnesses are not more useful than a single one, for the case that only separate measurements can be realized. The argument reveals a curious feature of entanglement regarding its shareability properties, which seems to depend strongly on whether we consider global entanglement or only entanglement that is \textit{locally} accessible.

\item \textbf{Chapter \ref{partition}} presents a result concerning the usefulness of quantum computation for calculating properties of many-body systems. It is well known that quantum computers are capable of simulating efficiently the evolution of general local Hamiltonians. There are, however, other properties of many-body systems that are also of interest, and whose calculation do not follow from the ability to simulate the dynamics of the system. Two important examples are the ground-state energy and the partition function of local Hamiltonians. There is very little hope that a quantum computer can calculate either of them for general models. However, can it provide meaningful approximations? Of course what we mean by meaningful needs further clarification. Do we consider an approximation meaningful if it tells us something useful about the physics of the problem, or, if we can give evidence that to achieve such a degree of approximation in a classical computer is an intractable problem? Although the first definition is the most adequate for a possible practical use of the approximation, the latter - which we call computationally non-trivial - is more relevant as an indication of the power of quantum computation. For certain classical Hamiltonians, recent work have found quantum algorithms that give computationally non-trivial approximations for their partition functions. They do so by showing that the approximations derived are $\BQP$-hard (can solve any problem which can be solved in a quantum computer). However, to this aim, non-physical complex valued temperatures have to be employed. A natural question is whether there are computationally non-trivial approximations of partition functions for physical instances of the problem. Based on the phase estimation algorithm, we present an elementary quantum algorithm for approximating partition functions which allows us to answer this question in the affirmative. There are crucial aspects of this approach that differs from previous ones. First, for classical partition functions the approximations can be obtained in polynomial time in a classical computer, so the non-trivial instances must be quantum. Second, we prove that the approximation is meaningful only in a weaker sense: The problem is hard for $\DQC$, and not for $\BQP$; moreover, we must employ $O(\log(n))$-local Hamiltonians in the hardness result, instead of strictly local Hamiltonians. For establishing the result we adapt Kitaev's construction from $\QMA$ to $\DQC$. We believe that such an approach might lead to further insight into the interesting class $\DQC$. Using the same techniques, we also find a $\DQC$-hard quantum algorithm for estimating, to polynomial accuracy, the spectral density of local Hamiltonians. Our result can then be seen as an example of the usefulness of the class $\DQC$ as a tool to attest the intractability of certain problems in a classical computer. 

\item \textbf{Chapter \ref{QCMA}} presents a further result on the limitations of a quantum computer for the calculation of ground-state energies of local Hamiltonians. As mentioned before, this problem is already $\QMA$-complete for one dimensional quantum Hamiltonians. Therefore, unless $\BQP = \QMA$, which is extremely unlikely, even a quantum computer cannot solve it in polynomial time. An interesting question in this respect is to determine what types of Hamiltonians are actually hard. A property that seems to have a direct impact is the \textit{spectral gap}, given by the difference of the first excited state energy to the ground energy. While the influence of the gap on the physical properties of many-body models has already been studied for a long time, only recently it has become clear that the gap also has an important role in the complexity of calculating local properties of many-body Hamiltonians. For example, Hastings recently established that one dimensional Hamiltonians with a constant gap can be efficiently stored and processed classically, which makes it very unlikely, baring a $\QMA = \NP$ surprise, that it is $\QMA$-complete to calculate their ground-state energy. In this chapter we consider the intermediate regime of Hamiltonians with an inverse polynomial (in the number of sites of the model) spectral gap. Our main result is that even for poly-gapped Hamiltonians, the estimation of the ground-state energy of local Hamiltonians is an intractable problem for quantum computation. More concretely, we show that under probabilistic reductions, the problem is $\QCMA$-hard ($\QCMA$ being the class of problems which have polynomial \textit{classical} proofs that can be checked on a \textit{quantum} computer in polynomial time). Crucial in our approach is the celebrated Valiant-Vazirani Theorem on the hardness of $\NP$ with unique witnesses, which we review and extend to probabilistic complexity classes. The result established also has implications to methods for storing and efficiently manipulating ground-states of general one dimensional local Hamiltonians.   
\end{description}

The new possibilities for information processing offered by quantum systems motivates the development of experimental techniques for the manipulation of individual quantum systems with a high degree of control and in a large number. Although several physical systems are being actively explored for the realization of quantum computation in a large scale, it is clear that with current and near future technology, it is out of reach to build a full-working quantum computer, operating below the fault-tolerance threshold for which error correction and hence scalability is possible \cite{NC00}. A natural approach is then to consider quantum systems that can be controlled with a good accuracy, but much above the fault-tolerance threshold, to simulate the dynamics of particular quantum many-body systems of interest, employing the former quantum system as a \textit{quantum simulator} to the latter. Part III of the thesis concerns the application of arrays of coupled microcavities for the quantum simulation of many-body physics and consists of chapters \ref{AMO}, \ref{BH}, \ref{kerrnonlinearity}, and \ref{heisenberg}.

\begin{description}

\item \textbf{Chapter \ref{AMO}} contains an overview of the quantum regime in arrays of coupled microcavities. The goal of this and the following chapters is to describe a new physical platform for the quantum simulation of many-body physics, consisting of arrays of coupled microcavities, interacting in the quantum regime with each other and with atomic-like structures. We show that this system is rich enough to allow for the realization of many interesting many-body models, and we propose schemes to create a few particular Hamiltonians. These proposals, if realized, would put photons or combined photonic-atomic excitations in new states of matter, which do not normally appear in Nature. The set-up we consider offer advantages over other physical platforms for the realization of strongly interacting many-body models, most notably the addressability to individual sites of the model, while it also presents new experimental challenges towards a realization of such strong interacting many-body regime. After a summary of some of the main achievements in the construction of quantum simulators in atomic and optical systems, we present a few of the basic aspects of cavity quantum electrodynamics (cQED), which we employ in the following three chapters. Then we turn to an outline of three promising physical realizations of microcavities for cQED, which are used as a testbed for the feasibility of the proposals we discuss in the following chapters. Finally, we derive the quantum interaction of an array of microcavities, which are coupled due to the hopping of photons between neighboring cavities, and discuss its applicability in real cQED realizations. 

\item \textbf{Chapter \ref{BH}} discusses the realization of Bose-Hubbard models in arrays of coupled cavities. The Hubbard model has an important role in solid state physics in the description of conducting-to-insulating transitions. Its extension to bosonic particles also has key importance as a paradigmatic example of interacting bosons on a lattice. The system exhibits a quantum phase transition from a superfluid to a Mott insulator phase, and describes a wide range of physical systems, such as Josephson junctions arrays and cold atoms in an optical lattice. In this chapter we propose the realization of one- and two-component Bose-Hubbard models of polaritons - joint atomic-photonic excitations - in arrays of interacting microcavities. We use the strong light-matter interaction in cavity QED, operating in the strong coupling regime, and techniques for producing large nonlinearites, based on electromagnetic-induced-transparency (EIT), to create an effective repulsion term for polaritons in the same cavity. This is then combined with the tunneling of photons from neighboring cavities to form the Bose-Hubbard model. We also comment on the possibility of forming a Mott insulator phase for photons in such a set-up. Finally, we discuss detection methods for local properties of the system and argue that a realization of such a model is within reach with present state-of-the-art technology. 

\item \textbf{Chapter \ref{kerrnonlinearity}} presents a proposal of a new scheme for generating large nonlinearities in cavity QED. Large optical Kerr nonlinearities are useful in several contexts, ranging from the realization of nonlinear optics in the few quanta regime to the implementation of quantum information protocols, such as quantum distillation schemes and quantum computation with photons. The generation of strong nonlinearities is also crucial for the realization of strongly interacting many-body models in arrays of coupled cavities, which is our main motivation here. In this chapter we propose a new method for producing Kerr nonlinearities in cavity QED, which is experimentally less demanding than previous proposals and produces comparable nonlinear interactions to the state-of-art EIT scheme. Furthermore, we show that by applying suitable laser pulses at the beginning and end of the evolution of the proposed set-up, we can obtain nonlinear interactions whose strength increases with increasing number of atoms interacting with the cavity mode, leading to effective nonlinear interactions at least two orders of magnitude larger than previously considered possible. 

\item \textbf{Chapter \ref{heisenberg}} presents an analysis of a different application of coupled cavity arrays in the creation of strongly interacting many-body models. While in the previous chapters we used the atoms to boost the interaction of photons in the same cavity, either actively, as in the polaritonic case, or in an undirect manner, as in the generation of large Kerr nonlinearities, in this chapter we focus on a complementary regime, in which the tunneling of photons is used to mediate interactions for atoms in neighboring cavities. We propose the realization of effective spin lattices, more concretely of the anisotropic Heisenberg model (XYZ), with individual atoms in microcavities that are coupled to each other via the exchange of virtual photons. Such a model has a rich phase diagram and is commonly used in the study of quantum magnetism. It is also a limiting case of the fermionic Hubbard model, which is believed to contain the main features of high $T_c$ superconductors. From a quantum information perspective, it can be used to create cluster states, universal resources for quantum computation by individual measurements. The presentation is ended with a discussion on the feasibility of creating the model with current technology in a few promising cavity QED realizations.  
\end{description}

The necessary background for this thesis is familiarity with the basics of quantum mechanics, quantum optics, and quantum information theory. Good references are \cite{Bal98, Per93} for an introduction to quantum theory, \cite{KA83, CDG89, MW95, Lou00} for quantum optics and atomic physics, and \cite{NC00, Pre98} for quantum information theory and quantum computation. In appendix \ref{not} we present the notation used throughout the thesis. 

\part{Entanglement Theory}\label{part1}

\chapter{Entanglement Theory} \label{entanglement}                                        

%\label{ch:entanglementtheory}

\section{Introduction}

Part I of the thesis is focused on the study of different aspects of quantum entanglement. In this chapter we present a collection of definitions and results 
of entanglement theory that will prove useful in the subsequent analysis. Most of the presentation is review oriented, with the exception of a few fragments of original research when indicated. 

We start giving the mathematical definition of entangled states and present two methods to characterize entanglement, the Peres-Horodecki (PPT) criterion and entanglement witnesses. We then present the main ideas behind entanglement distillation, including bound entanglement and entanglement activation. In the sequel, we introduce and outline some properties of entanglement measures that will be considered throughout this thesis: the distillable entanglement, the entanglement cost, the relative entropy of entanglement, and the robustnesses of entanglement. We end up discussing recently established quantum de Finetti theorems, which play an important role in the next chapters.  

\section{Entangled States}

In chapter \ref{introduction} we mentioned the importance of entanglement both for the foundations of quantum theory and for quantum information processing, where entanglement is viewed as a resource. To define what is entanglement it is instructive to start with the definition of \textit{non-entangled} states (also called classically correlated and separable). For finite dimensional Hilbert spaces ${\cal H}_A$ and ${\cal H}_B$ we say that a bipartite quantum state $\rho_{AB} \in {\cal D}({\cal H}_A \otimes {\cal H}_B)$ is non-entangled if there exists a probability distribution $\{ p_i \}$ and two sets of states $\{ \rho_i^A \}$ and $\{ \rho_i^B \}$ acting on ${\cal H}_A$ and ${\cal H}_B$, respectively, such that \cite{Wer89}
\begin{equation} \label{sepstate}
\rho_{AB} = \sum_{i} p_i \rho_i^A \otimes \rho_i^B. 
\end{equation}
An entangled state, in turn, is a state which is not separable\footnote{In this thesis, for simplicity, we will restrict most of the discussion to finite dimensional Hilbert spaces. The interest reader if referred to Refs. \cite{HHHH07, PV07} for an exposition of entanglement in infinite dimensional systems.}$^{,}$\footnote{For the sake of simplicity, once more, we have given the definition of entangled states only for bipartite states. The most general definition to multipartite systems is completely analogous and can be found e.g. in \cite{HHHH07, PV07}.}. 
                                                                                  
The meaning of this definition becomes clear considering the \textsl{distant laboratories paradigm} of entanglement theory. We assume that two parties, e.g. Alice and Bob, are in separated locations and do not have access to any joint quantum interaction. They however can perform any operation allowed by quantum mechanics locally, i.e. any trace preserving completely positive map \cite{NC00}, and can moreover send classical bits to each other. The parties can thus cooperatively implement any physical map consisting of quantum \textit{local operations and classical communication} (LOCC). With this setting in mind, we can rephrase the definition of separable states in an operational way and say that they are the quantum states which can be created by LOCC. Although such states might be correlated, it is clear that all such correlations are due to the classical communication between Alice and Bob and can, therefore, be replaced by shared classical randomness. The correlations contained in entangled states, on the contrary, are inequivalent to classical correlations, as they cannot be created solely by communicating classical information.

\section{Entanglement Characterization} \label{entcharact}

Considering the above definition of entangled states, one immediate question arises. Given the description of a quantum state as a density operator, how can we decide if the state is entangled or not? The very definition of entanglement already gives an algorithmic procedure to determine the separability of a given state $\rho$. Loosely speaking, one could search over the separable states set for a good approximation to $\rho$ (in trace norm for example) and refine successively the quality of approximation. This is however highly inefficient, taking an exponential number of steps in the dimension of the state, and is clearly not suitable for analytical calculations.

%%\cite{BHHA07, CW03, HHH96, Hor97, HH99, HHH06, HLVC00, KCKL00, LKCH00, NK01, Per96, Rud02, Ter02} and numerical \cite{BV04a, BV04b, DPS02, DPS04, DPS05, EHGC04, ITCE04, IT06, Iou06, Per04}

A large amount of work have been devoted to the development of entanglement characterization criteria, both analytical and numerical (see e.g. Refs. \cite{HHHH07, Iou06} for a review of the main methods). An important result in this direction, which has reshaped research efforts in the study of entanglement detection methods, was Gurvits proof that it is $\NP$-hard to decide if a state is entangled\footnote{Taking the the dimension of the state as the input size.} \cite{Gur03}, showing that there is very little hope of an efficient algorithm for the problem in the general case. It is out of the scope of this thesis to give a complete overview of the \textit{separability problem}. Instead in the sequel we discuss two particular methods for entanglement characterization, which have shown to be extremely useful not only for the problem \textit{per se}, but also as analytical tools in entanglement theory. It is actually in this second use that we are most interested in this thesis.      

\subsection{Peres-Horodecki Criterion}

The first and one of the most useful methods for detecting entanglement is the positive partial transpose, or Peres-Horodecki, criterion \cite{HHH96, Per96}. Given bases $\ket{i}_A$ and $\ket{j}_B$ for ${\cal H}_A$ and ${\cal H}_B$, respectively, we can write any operator $X \in {\cal B}({\cal H}_A \otimes {\cal H}_B)$ as 
\begin{equation*} 
X = \sum_{i,j,k,l} x_{i,j,k,l} \ket{i}_A\bra{j} \otimes \ket{k}_B\bra{l}. 
\end{equation*}
We define the partial transpose of $X$ with respect to subsystem $A$ as 
\begin{equation*} 
X^{\Gamma_A} = \sum_{i,j,k,l} x_{i,j,k,l} (\ket{i}_A\bra{j})^T \otimes \ket{k}_B\bra{l},
\end{equation*}
where $T$ is the usual transposition map. Likewise, we define $\Gamma_B$ as the partial transpose map with respect to $B$\footnote{In most cases it will not matter which system we take the map with respect to. To make the notation lighter we will then omit the $A$ index in $\Gamma_A$ and denote $\Gamma_A$ by $\Gamma$.}. As the transposition itself, the partial transposition is basis dependent. The eigenvalues of $X^{\Gamma_{A/B}}$, however, are independent as long as we use a local basis.
 
If we apply the partial transpose map to a generic separable state, given by Eq. (\ref{sepstate}), we find
\begin{equation*}
\rho_{AB}^{\Gamma} = \sum_{i} p_i (\rho_i^A)^T \otimes \rho_i^B, 
\end{equation*}
from which we can easily see that $\rho_{AB}^{\Gamma}$ is a positive semidefinite operator. In Ref. \cite{Per96} it was noted by Peres that in general this is not the case anymore when $\rho_{AB}$ is entangled: there exist states $\rho_{AB}$ for which $\rho_{AB}^{\Gamma}$ has negative eigenvalues. We can therefore use the non-positivity of the partial transpose (PPT) of a state as a sufficient condition for entanglement. Soon after Peres work, the Horodecki family proved that while for systems consisting of two qubits, or a qubit and a qutrit, the non-PPT condition is also necessary for entanglement, this ceases to be the case for any higher dimensions \cite{HHH96}. 

%This result can be best understood within the framework of the next subsection.

\subsection{Entanglement Witnesses} \label{entwit}

Let us denote the set of separable states acting on ${\cal H} = {\cal H}_A \otimes {\cal H}_B$ by ${\cal S}({\cal H})$, or by ${\cal S}$ when ${\cal H}$ is clear from the context. It is a direct consequence of the definition of separable states that they form a closed convex set. Given an entangled state $\rho$, we know that $\rho \notin {\cal S}$ and hence there must exist a hyperplane, i.e. a linear functional on $B({\cal H})$, that separates $\rho$ from ${\cal S}$ \cite{RS72}. An \textit{entanglement witness} for a state $\rho$ is exactly such a hyperplane separating it from the set of separable states \cite{HHH96, Ter02}. 

We define the set of entanglement witnesses ${\cal W}({\cal H})$ as the set of all Hermitian operators $W$ such that  
\begin{equation*}
\tr(W \sigma) \geq 0, \hspace{0.1 cm} \forall \hspace{0.1 cm} \sigma \in {\cal S}.
\end{equation*}
Note that any positive semidefinite operator is contained in ${\cal W}$. These operators are trivial witnesses, since they do not detect the entanglement of any state. It is nonetheless useful to define the set of witnesses containing every operator which is positive on separable states, as in this case ${\cal W}$ is the dual cone\footnote{A subset $C$ of a real vector space is a cone if, and only if, $\lambda x$ belongs to $C$ for every $x \in C$ and every positive real number $\lambda$.}\footnote{Given a cone ${\cal M} \in \mathbb{R}^n$, its dual cone is defined as ${\cal M}^* = \{ y \in \mathbb{R}^n : y^Tx \geq 0 \hspace{0.1 cm} \forall \hspace{0.05 cm} x \in {\cal M} \}$ \cite{BV01}.} of ${\cal S}$ and this allow us to use tools from convex optimization into variational problems involving separable states or entanglement witnesses, as it will be discussed later in this section. 

An interesting aspect of the definition of entanglement witnesses stems from the fact that they are Hermitian operators and can therefore be experimentally measured. Whenever the expectation value of some witness operator takes a value smaller than zero, then we can draw the conclusion that the measured state has been entangled. In several cases the number of measurement rounds needed to measure an entanglement witness is much smaller than the number needed to perform full tomography of the state at hand, showing the usefulness of witnesses operators in the experimental characterization of entanglement. Indeed, multipartite entanglement of up to 8 parties has been experimentally detected by measuring entanglement witnesses in several physical set-ups, including linear optical networks \cite{BEK+04, KST+07} and cold trapped ions \cite{HHR+05}.

By their very definition, entanglement witnesses can characterize the entanglement of every quantum state. The price for this generality is the high complexity of determining a witness for a given entangled state. Based on the $\NP$-hardness of the separability problem, it follows directly that finding witnesses for general entangled states is also $\NP$-hard. Although it is thus very unlikely that an efficient method for calculating entanglement witnesses for every state exists, a great deal of work has been devoted to the construction of particular bipartite and multipartite witness operators (see e.g. \cite{LKCH00, Ter02, CW03, TG05, Tot04, HHHH07} and references therein).

\subsubsection{Convex Optimization with Separable states and Entanglement Witnesses}

In quantum information theory we often encounter convex optimization problems involving separable states and entanglement witnesses. We now show that such problems are encompassed in the framework of convex optimization problems with generalized inequalities (see e.g. Ref. \cite{BV01} for a detailed discussion on convex optimization and generalized inequalities). We can define a partial order over the states acting on ${\cal H}$ as follows:
\begin{equation} \label{genineq}
A \geq_{S} B \hspace{0.4 cm} \text{if} \hspace{0.4 cm} A - B \in \text{cone}({\cal S}),
\end{equation}
where $\text{cone}({\cal S})$ is the convex cone formed by all (not necessarily normalized) separable states. As discussed in Ref. \cite{BV01} (see section 5.9), the Lagrange multiplier associated to such an inequality is an element of the dual of ${\cal S}$, i.e. an operator $X$ such that $X \geq_{S^*} 0$. As ${\cal S}^* = {\cal W}$, we see that the Lagrange multiplier associated to the inequality given by Eq. (\ref{genineq}) is an entanglement witness. This connection allow us to rewrite optimization expressions over separable states as optimization problems over entanglement witnesses and vice-versa, by using e.g. strong duality conditions from convex optimization theory. We will see later on in this thesis two examples of such a procedure, in sections \ref{proofmain} and \ref{distent}.

\section{Entanglement Distillation} \label{entanglementdistillation}

Up to now we have considered entanglement in terms of preparation procedures: to create entanglement the parties need to exchange quantum information. This however does not say anything about the usefulness of such entanglement to information processing (or to violate a Bell's inequality \cite{Bel64}). It turns out that the definition of useful entanglement varies from what application we have in mind. The sets of (entangled) states which are resources for e.g. teleportation \cite{BBC+93}, secure key distribution \cite{BB84, Eke91, HHHO05b}, Bell's inequality violation \cite{Bel64}, or universal quantum computation by single-qubit measurements \cite{RB01}, are in fact different one from another\footnote{With the exception of the set of states violating a Bell's inequality, where it is not known how it compares to the sets of states which are useful for teleportation or key distribution.}.                   
                                                                              
Despite this ambiguity in defining useful entanglement, there is one family of states which stands out as the purest form of entanglement in several applications. These are the \textit{maximally entangled states}, defined by 
\begin{equation}                                                                    
\Phi(K) := \frac{1}{K} \sum_{i=1}^K \sum_{j=1}^K \ket{i, i}\bra{j, j},         
\end{equation}                                                                    
up to local unitaries in $A$ and $B$. Indeed, if Alice and Bob share a maximally entangled state of rank $K$, they can e.g. teleport a $K$-dimensional state with perfect fidelity by communicating $\log(K)$ bits \cite{BW92}, communicate 2$\log(K)$-bits from Alice to Bob (or vice versa) by classically communicating $\log(K)$ bits \cite{BBC+93}, or extract $\log(K)$ secret-bits completely uncorrelated from an eavesdropper by classically communicating over a public, but authenticated, channel \cite{Eke91}. Furthermore, by teleportation, any entangled state acting on $\mathbb{C}^K \otimes \mathbb{C}^K$ can be obtained deterministically from $\Phi(K)$ by LOCC.

%Conversely, none of the these tasks can be performed perfectly if the parties share only one copy of any non-maximally entangled state.    
                                                                                    
Considering the central role of the maximally entangled states as noiseless resources for quantum communication, we are led to the question: under what circumstances can we obtain maximal entanglement from a source of noisy entanglement? This question is not only of theoretical interest, but central for the experimental implementation of quantum information processing. As in practice it is impossible to send quantum information without error, entanglement will inevitably become mixed during distribution, which might not be directly useful. One possible strategy is to bring this entanglement into the desired form by local operations and classical communication. Such a procedure of transforming by LOCC noisy entanglement into maximally entangled states is known as \textit{entanglement distillation}.

Note the similarity of this setting with the one encountered in the theory of (quantum) information transmission through noisy channels. The ability to send error-free information by the use of error correction techniques can be reinterpreted as the possibility of simulating a noiseless channel by a noisy one. Entanglement distillation has a similar flavor. Indeed, after entanglement distillation the parties could use their entanglement to perfectly teleport a quantum state by LOCC, simulating the effect of a noiseless quantum channel. In contrast to standard error correction, however, all the quantum error correction needed is performed in the distillation part of the protocol. Curiously, all the correction procedure could be implemented even before Alice chooses the states to be sent. For this reason, entanglement distillation can be considered as a counter-factual error correction protocol.

Based on classical error correction codes, the first entanglement distillation protocols were proposed in Refs. \cite{Gis96, BBP+96, BDS+96, DEJ+96}, where it was shown that it is indeed possible to distill maximally entangled states from an i.i.d source of states $\rho_{AB}^{\otimes n}$, for some particular choices of $\rho_{AB}$, with a non-zero rate. Although considerable effort has been put in the study of distillation procedures (see Ref. \cite{HHHH07} and references therein), it is still a challenging task to come up with distillation protocols with good rates, specially in the two-way classical communication setting. In section \ref{costdistillable} we revisit this question in our discussion of optimal distillation procedures in the context of entanglement measures. 

\subsection{Bound Entanglement} \label{boundentanglement}

As mentioned in the previous paragraph, entanglement distillation is a difficult concept to grasp in full. In this section we show that the situation is more complex than we might have anticipated even on a qualitative level. Instead of looking at distillation rates, let us consider the binary problem of determining when a quantum state can be distilled. For states of two qubits the answer turns out to be simple: any entangled state can be distilled \cite{HHH97}. For higher dimensions, however, the situation is different and there is not anymore an one-to-one relation between entanglement and distillability \cite{HHH98}.   

To further explore this point, we consider the class of stochastic LOCC (SLOCC) quantum maps. These are the quantum operations which can be implemented by LOCC with non-zero probability. As shown in Ref. \cite{CDKL01}, a direct application of the Choi-Jamiolkowski isomorphism between quantum channels and quantum states \cite{Cho75, Jam72} gives a beautiful characterization of this class: A completely positive map $\Omega$ can be generated with non-zero probability by stochastic LOCC if, and only if, it can be written as
\begin{equation} \label{SEPmap}
\Omega(.) = \sum_k A_k \otimes B_k (.) A_k^{\cal y} \otimes B_k^{\cal y},
\end{equation}
with local operators $A_k$, $B_k$ such that $\sum_{k}A_k^{\cal y}A_k \otimes B_k^{\cal y}B_k \leq \id$. 

We can now reformulate the distillability definition in terms of SLOCC maps. From the equivalence of distillability and entanglement for two qubits states it follows that a state $\rho_{AB}$ is distillable if, and only if, there is an integer $n$ and a stochastic LOCC operation mapping $\rho_{AB}^{\otimes n}$ into a two qubit entangled state \cite{HHH98}\footnote{It will be the main goal of chapter \ref{NPPT} to extend this definition to the non-i.i.d. case. There, instead of an i.i.d. source of the form $\{ \rho^{\otimes n} \}$, we will be interested in figuring out when sequence of states $\{ \omega_n \}$, with correlations among the several copies, can be distilled.}.  

From the characterization of SLOCC maps given by Eq. (\ref{SEPmap}) we find the strongest known necessary condition for distillability: the non-PPTness of the partial transpose. Indeed, taking the partial transposition of $\Omega(\rho_{AB}^{\otimes n})$ and using the form of $\Omega$ given by Eq. (\ref{SEPmap}), we find
\begin{equation*}
\Omega(\rho_{AB}^{\otimes n})^{\Gamma}  = \sum_k (A_k)^T \otimes B_k (\rho_{AB}^{\otimes n})^{\Gamma} A_k^{\cal y} \otimes B_k^{*}.
\end{equation*}
If $\rho_{AB}$ is distillable, $\Omega(\rho_{AB}^{\otimes n})^{\Gamma}$ must be entangled and thus have a non-positive partial transpose. The R.H.S. of the equation above then tells us that this is only possible if $\rho_{AB}^{\Gamma}$ is not positive semidefinite. As there are entangled states with a positive partial transpose \cite{Hor97}, there exist states which, although entangled, cannot be distilled. These are called \textit{bound entangled} states and were first introduced by the Horodecki family in Ref. \cite{HHH98}.  

The phenomenon of bound entanglement shows that there is an inherent irreversibility in the manipulation of entanglement by LOCC. While any entangled state can be created from a maximally entangled state of sufficiently large dimension by LOCC, the reverse process is not always possible. This fact, or more precisely its quantitative version discussed in section \ref{costdistillable}, will be the main motivation for formulating the theory that will be developed in chapter \ref{reversible}. 

Although any distillable state must have a non-positive partial transpose, the converse is not known to be true. In fact, based on several partial indications \cite{DSS+00, DCLB00, Cla06, VD06, PPHH07}, it has been conjectured that bound entangled states with a non-positive partial transpose exist. Unfortunately, despite considerable effort, the conjecture is still open. In chapter \ref{NPPT} we present a new approach to tackle this problem, which albeit giving more evidence on the existence of NPPT bound entanglement, does seem to be strong enough to resolve it.    

\subsection{Entanglement Activation} \label{activation}

The conjectured existence of NPPT bound entangled states raises an intriguing possibility: there might exist two undistillable states $\rho$ and $\sigma$ such that $\rho \otimes \sigma$ is distillable. The distillable entanglement of $\sigma$ would then be activated by $\rho$ (or vice-versa)\footnote{An analogous effect has also been conjectured for quantum channels. It is believed that there are two quantum channels $\Lambda_1 \otimes \Lambda_2$ with zero two-way quantum capacity such that $\Lambda_1 \otimes \Lambda_2$ has a non-zero capacity. Remarkably, this has been recently shown to be the case for the one-way (or zero-way) channel capacity \cite{SY08}.}. This possibility was first linked to NPPT bound entanglement in Ref. \cite{SST01}, where it was shown that if a particular NPPT Werner state\footnote{Werner states are one dimensional families of states formed by the convex combination of normalized projectors onto the symmetric and antisymmetric subspaces of $\mathbb{C}^d \otimes \mathbb{C}^d$ \cite{Wer89}.} is undistillable, then its entanglement can be activated by a PPT, and therefore undistillable, entangled state. In Ref. \cite{EVWW01}, in turn, this result was strengthen by proving that the entanglement of any NPPT state can be activated by a PPT state. It is clear that the converse holds: the existence of such an activation process would imply that NPPT bound entanglement exists, as the set of PPT states is closed under tensoring. 

Although such a strong form of entanglement activation remains as a conjecture, weaker but still meaningful activation phenomena have been proven to exist \cite{HHH99c, EVWW01, VW02, Ish04, IP05, Mas06}. These processes were found in the context of single copy entanglement \textit{quasi-distillation} \cite{HHH99a}. Here we are interested in increasing the fidelity of a given state with the maximally entangled state probabilistically by LOCC. A good figure of merit for this process is the \textit{singlet fraction} \cite{HHH99a}, defined as   
\begin{equation} \label{singletfraction}
F_K(\rho) :=  \sup_{\Omega \in SLOCC} \frac{\tr(\Omega(\rho)\Phi(K))}{\tr(\Omega(\rho))}.
\end{equation}
The idea of entanglement activation and the first result in this direction was given by the Horodecki family in Ref. \cite{HHH99c}. There they presented an entangled state $\rho$ with $F_2(\rho) < 1$ and a PPT entangled state $\sigma$ such that 
\begin{equation*}
\lim_{n \rightarrow \infty} F_2(\rho\otimes \sigma^{\otimes n}) = 1,
\end{equation*}
showing that the entanglement of $\sigma$ can be pumped into $\rho$ to increase its singlet-fraction. 

In Refs. \cite{EVWW01, VW02} other examples of entanglement activation were found. In particular, Vollbrecht and Wolf found a PPT bound entangled state $\sigma$ which can be used to activate any NPPT state, i.e. such that
\begin{equation*}
F_2(\rho\otimes \sigma) > \frac{1}{2},
\end{equation*}
for every NPPT state $\rho$ \cite{VW02}. This result shows that for every integer $n$ there is a state $\rho$ such that $F_2(\rho^{\otimes n})=1/2$ and $F_2(\rho\otimes \sigma) > \frac{1}{2}$\footnote{This follows from the fact that for every $n$ there are NPPT states which are $n$-copy undistillable \cite{DCLB00, DSS+00}.}. 

Finally, the strongest activation result to date has been proven by Masanes in Ref. \cite{Mas06}. There he showed that every entangled state can be activated\footnote{This result is of fundamental significance because it establishes that, in the context of quasi-distillation, \textit{every} entangled state is useful for quantum information processing.}: For every entangled state $\rho$ and every real number $K^{-1} \leq \lambda < 1$ (with $K$ a natural number), there exists a state $\sigma$ with $F_{K}(\sigma) \leq \lambda$ such that $F_{K}(\rho \otimes \sigma) > \lambda$. 

Other examples of entanglement activation in the pure and multipartite cases were considered in Refs. \cite{Ish04, IP05}. We will revisit entanglement activation in chapter \ref{NPPT}, where Masanes theorem is generalized and the relation of such a generalization to the existence of NPPT bound entanglement is explored.

\section{Entanglement Measures} \label{entmeasures}

In the previous sections we have discussed key qualitative features of entanglement. Here we will show that it is also possible to quantify entanglement in a meaningful manner. One such way was already implicit in our discussion on entanglement distillation: We can quantify the amount of entanglement of a state by how much two-dimensional maximally entangled states can be extracted from it by LOCC. As we show in the sequel, this reasoning indeed leads to one of the most meaningful entanglement measures. However this is by no means the only one. In the past years dozens of entanglement measures have been proposed and studied \cite{PV07, HHHH07}. This section gives a short introduction to the basic ideas underlying the construction of entanglement measures and, in particular, to the four measures that will be considered later in this work: the distillable entanglement, the entanglement cost, the relative entropy of entanglement, and the robustness of entanglement. A more in-depth discussion on entanglement measures can be found in Refs. \cite{PV07, HHHH07}. 

A basic principle for quantifying entanglement can be derived from its relation to local operations and classical communication. We might not know exactly what entanglement is, but we do know that we should not be able to create it by LOCC. Thus the first property that we expect from an entanglement measure $E$ is that
\begin{list} {\textbf{1.}}
	\item $E$ should be monotonic decreasing under LOCC operations. 
	
	\vspace{0.4 cm}
	
\end{list}
Different forms of monotonicity have been considered in the literature. In the simplest, $E$ should be monotonic under trace preserving LOCC maps, i.e.
\begin{equation*}
E(\Lambda(\rho)) \leq E(\rho)
\end{equation*}
should hold for every state $\rho$ and every deterministic LOCC operation $\Lambda$. This is arguably the most important requirement for an entanglement measure. A direct consequence of this inequality is the invariance of $E$ under local unitaries. There is another form of monotonicity which, although useful, is not as fundamental as the first (see e.g. Ref. \cite{HHHH07}). Suppose that Alice and Bob implement a LOCC protocol to the state $\rho$ and obtain the states $\{ \rho_i \}$ with probabilities $\{ p_i \}$. Then, \textit{strong} or \textit{full monotonicity} under LOCC requires that the amount of entanglement cannot increase on average, i.e.
\begin{equation}  \label{strongmonotonicity}
E(\rho) \geq \sum_i p_i E(\rho_i).
\end{equation}

Is it an immediate consequence of monotonicity under LOCC that an entanglement measure should be constant on the separable states set. The next property is a refinement of this and postulates the intuitive fact that separable states should not contain any entanglement.
\begin{list} {\textbf{2.}}
	\item $E(\rho) \geq 0$ for every state and $E(\sigma) = 0$ for every separable state $\sigma$. 
   \vspace{0.4 cm}
\end{list}
In the spirit of information theory, we are often interested in the manipulation of entanglement in the limit of a very large number of copies. It is then helpful to have an entanglement measure that does not vary too much for small variations in the state, even if such a state has a very high dimension. This idea is summarized in the following property.
\begin{list} {\textbf{3.}}
	\item For every two states $\rho_n$ and $\sigma_n$ acting on ${\cal H}^{\otimes}$, 
\begin{equation}
\left \vert \frac{E(\rho_n) - E(\sigma_n)}{n} \right \vert \leq f( || \rho_n - \sigma_n ||_1),
\end{equation}
where $f:\mathbb{R} \rightarrow \mathbb{R}$ is a function independent of $n$ such that $\lim_{x \rightarrow 0} f(x) = 0$. This special type of continuity is known as \textit{asymptotic continuity}.
\vspace{0.4 cm}
\end{list}
Another convenient property is normalization.  
\begin{list} {\textbf{4.}}
	\item For every maximally entangled state $\Phi(K)$, $E(\Phi(K)) = \log(K)$.
	\vspace{0.4 cm}
\end{list}
The idea behind normalization is to count the entanglement of maximally entangled states in terms of how many qubits can be teleported by LOCC using the state as a resource. With this convention, $\Phi(2)$ has entanglement equal to one and is hence referred to as an \textit{ebit} (entanglement-bit). 

When analyzing asymptotic entanglement transformations, we usually need to deal with the entanglement of $\rho^{\otimes n}$, for very large $n$. It is therefore a helpful property if the entanglement of several copies $E(\rho^{\otimes n})$ is the same as $n$ times the entanglement of one copy $E(\rho)$. 
\begin{list} {\textbf{5.}}
	\item $E$ should be weakly-additive. For every state $\rho$, $E(\rho \otimes \rho) = 2 E(\rho)$. 
	\vspace{0.4 cm}
\end{list}
Other properties that are sometimes considered and which are very useful in some applications are convexity\footnote{A measure is convex if $E\left( \sum_i p_i \rho_i \right) \leq \sum_i p_i E(\rho_i)$.}, full additivity\footnote{Full additivity means that for every two states $\rho$ and $\sigma$, $E(\rho \otimes \sigma) = E(\rho) + E(\sigma)$.}, strong full additivity\footnote{Strong full additivity means that for every state $\rho_{A_1A_2:B_1B_2}$, $E(\rho_{A_1A_2:B_1B_2}) = E(\rho_{A_1:B_1}) + E(\rho_{A_2:B_2})$.}, faithfulness\footnote{An entanglement measure is faithful if $E(\rho) > 0$ for every entangled state $\rho$.}, monogamy\footnote{The monogamy inequality reads $E(\rho_{AA'B}) \geq E(\rho_{AB}) + E(\rho_{A'B})$, where $AA'$ are Alice's systems and $B$ Bob's.}, and lockability\footnote{An entanglement measure is lockable if it can decrease by a very large amount when one qubit is traced out or dephased \cite{HHHO05}.}. 

It must be stressed that, apart from monotonicity under LOCC, all the properties outlined above are not mandatory for an entanglement measure to be meaningful. Indeed, in this thesis it will become clear that the robustness of entanglement is an example of this fact. 

The properties discussed in this section have been proposed and refined in a number of works, including most notably \cite{VPRK97, VP98, HHH00, Vid00, DHR02, Chr06, PV07, HHHH07}. 

\subsection{Entanglement Cost and Distillable Entanglement} \label{costdistillable}

The \textit{distillable entanglement} of a bipartite state $\rho_{AB}$ is defined as the optimal rate of maximally entangled states that can be distilled from $\rho_{AB}$ by LOCC in the asymptotic limit \cite{BBP+96}. It can be concisely expressed as 
\begin{equation} \label{distillableent}
E_{D}(\rho) := \sup_{\{ K_n \}} \left \{ \liminf_{n \rightarrow \infty} \frac{\log K_n}{n}  : \lim_{n \rightarrow \infty} \left( \inf_{\Lambda \in LOCC} || \Lambda(\rho^{\otimes n}) - \Phi( K_n )||_1 \right) = 0 \right \},
\end{equation}  
where the supremum if taken over all sequences of integers $\{ K_n \}$. It follows directly from its definition that $E_D$ satisfies properties 1, 2, 4 and 5. Furthermore the distillable entanglement is not faithful due to the existence of bound entanglement. All the other properties are not known to hold and some of them are in fact conjectured not to hold. For instance, we have already seen in section \ref{activation} that the existence of NPPT bound entanglement would imply the non-additivity of the distillable entanglement. Using a simple reasoning it was shown \cite{SST01} that the existence of NPPT bound entanglement would also imply that the distillable entanglement is non-convex. 

Despite its very strong operational meaning, the distillable entanglement has been of limited use in practical applications, owing to the lack of methods to calculate it, or even to derive good bounds. This drawback has its roots on our very limited knowledge of efficient distillation protocols involving two-way classical communication. 

One notable exception to this situation is for bipartite pure states. In this case, it was shown by Bennett, Bernstein, Popescu, and Schumacher \cite{BBPS96} that the distillable entanglement is given by the von Neumann entropy of either reduced density matrices, which has since been referred to as the \textit{entropy of entanglement}. In the distillation procedure Alice and Bob locally project their states onto their typical subspaces \cite{NC00}, which with high probability gives them a maximally entangled state of dimension roughly equal to the entropy of entanglement. Note that in this protocol there is no need of classical communication between the parties. 

Another setting where substantial progress has been achieved is in the analysis of the \textit{one-way distillable entanglement}, where only classical communication from Alice to Bob is allowed. Devetak and Winter \cite{DW04, DW05} have derived a closed formula for this quantity (which however still involves both maximization and regularization) and proven the \textit{hashing inequality} \cite{HHH00b}, stating that the coherent information\footnote{The coherent information $I_c(A \rangle B)$ of a bipartite state $\rho_{AB}$ is given by the negative conditional von Neumann entropy $-S(A|B) = S(\rho_B) - S(\rho_{AB})$ and plays a central role in the theory of quantum information transmission.} is an achievable rate in entanglement distillation with one-way classical communication. 

The \textit{entanglement cost} is defined from the reverse process of distillation. It gives the optimal rate of maximally entangled states which needs to be invested in order to form the state $\rho_{AB}$ in the asymptotic limit by LOCC. Analogously to Eq. (\ref{distillableent}), it can be expressed as
\begin{equation} \label{entcost}
E_{C}(\rho) := \inf_{ \{ K_n \} } \left \{ \limsup_{n \rightarrow \infty} \frac{\log K_n}{n} : \lim_{n \rightarrow \infty} \left( \inf_{\Lambda \in LOCC} || \rho^{\otimes n} - \Lambda(\Phi( K_n ))||_1 \right) = 0 \right \}.
\end{equation}
A lot more is known about the entanglement cost than the distillable entanglement. Regarding its properties, $E_C$ is known to satisfy 1, 2, 4, 5 and and to be a faithful, lockable \cite{HHHO05}, non-monogamic \cite{KW04}, and convex measure. It is believed that $E_C$ also satisfy 3 and is full additive. It has also been shown that full additivity and strong monotonicity in the sense of Eq. (\ref{strongmonotonicity}) are equivalent properties for the entanglement cost \cite{BHPV07}.

The entanglement cost has also been determined for pure states and shown to be equal to the entropy of entanglement \cite{BBPS96}. An easy, but wasteful\footnote{In what concerns the classical communication cost of the protocol.}, protocol is the following: Alice locally prepares a compressed version of the state (by applying Schumacher's compression \cite{Schu95} to $\ket{\psi}_{AB}^{\otimes n}$) and teleport the $B$ part of it to Bob. This consumes an amount of classical communication from Alice to Bob proportional to the ebit cost of the protocol. Using a cleverer scheme, it has been shown that a sublinear amount of classical communication, scaling as $\sqrt{n}$ with the number of copies of the state, is both sufficient \cite{LP99} and necessary \cite{HW03, HL04}. 

A remarkable feature of pure state entanglement manipulation is the equality of distillable entanglement and entanglement cost. This shows that pure state bipartite entanglement is a fungible resource: any two entangled states can be reversibly interconverted, as long as the ratio of copies of each one of them matches the ratio of the respective entropies of entanglement. This setting also has analogies with thermodynamics and the second law, which we explore in chapter \ref{reversible}. 

The situation is very different for mixed states. While the entanglement cost is always non-zero for every entangled state \cite{YHHS05}, the distillable entanglement is zero for bound entangled states. There exists therefore an inherent irreversibility in the manipulation of entanglement under LOCC. As anticipated in section \ref{boundentanglement}, this observation will be the main motivation for the new paradigm in entanglement theory that will be explored in chapter \ref{reversible}. 

The entanglement cost is intimately related to another entanglement measure known as the \textit{entanglement of formation} \cite{BDS+96}, defined as the \textit{convex-roof} of the entropy of entanglement, i.e.
\begin{equation}
E_F(\rho) := \min_{\{ p_i, \psi_i\}} \sum_i p_i E(\psi_i), 
\end{equation}
where the minimization is taken over all convex decompositions of $\rho$ into pure states $\rho = \sum_i p_i \ket{\psi_i}\bra{\psi_i}$ and $E(\psi)$ is the entropy of entanglement of $\psi$. It has been proven by Hayden, Horodecki, and Terhal that \cite{HHT01}
\begin{equation*}
E_C(\rho) = \lim_{n \rightarrow \infty} \frac{E_F(\rho^{\otimes n})}{n}.
\end{equation*}
Probably the most important open question in quantum information theory asks if the limit in the equation above is really necessary. It has been conjectured that in fact $E_F(\rho^{\otimes n}) = n E_F(\rho)$ and, hence, $E_F(\rho) = E_C(\rho)$. Based on previous works \cite{MTW04, AB04}, this conjectured additivity of the entanglement of formation has been shown by Shor \cite{Sho04} to be equivalent to other important open problems in quantum information theory, including the additivity of the Holevo quantity \cite{Hol73}, whose regularization gives the capacity of classical information transmission of a quantum channel \cite{Hol98, SW97}\footnote{In Ref. \cite{Sho04} it is actually shown that the full additivity of the entanglement of formation is equivalent to the full additivity of the Holevo capacity (and of the minimum output von Neumann entropy). In Refs. \cite{BHPV07, FW07}, in turn, the weak additivity of the entanglement of formation was shown to be equivalent to its full additivity.}.

For a more detailed discussion on the entanglement of formation, entanglement cost, and distillable entanglement the reader is referred to Refs. \cite{HHHH07, PV07}.

\subsection{Relative Entropy of Entanglement} \label{relentint}

In this section we discuss another measure that has found many applications in entanglement theory. The \textit{relative entropy of entanglement}, introduced by Vedral, Plenio, Rippin, and Knight \cite{VPRK97}, is defined as
\begin{equation} \label{relent}
E_R(\rho) = \min_{\sigma \in {\cal S}} S(\rho || \sigma),
\end{equation}
where $S(\rho || \sigma) := \tr(\rho \log(\rho)) - \tr(\rho \log(\sigma))$ is the quantum relative entropy, or quantum Kullback-Leibler divergence \cite{OP93}. This measure was first proposed as an example of a distance-based entanglement measure in \cite{VPRK97} and then further explored in Ref. \cite{VP98}. It satisfies properties 1, 2, 3, 4 and is convex, non-lockable \cite{HHHO05}\footnote{The relative entropies of entanglement and its regularization are actually the only known non-lockable entanglement measures.}, faithful\footnote{This follows easily form the fact that $S(\rho || \sigma) = 0$ if and only if $\rho = \sigma$. The relative entropy of entanglement satisfies an even stronger form of faithfulness. It is direct consequence of Pinsker's inequality \cite{Pet86} that if $E_R(\rho) \leq \epsilon$, then $\rho$ is $\sqrt{2\epsilon}$ close from a separable state in trace norm.} and non-monogamic\footnote{An example is the anti-symmetric state of three qubits \cite{Ple08}.}. For pure states it is equal to the entropy of entanglement. In Ref. \cite{VW01} it has been proven that the relative entropy of entanglement is not weakly-additive, the antisymmetric Werner state of dimension three being a counterexample. Owing to this fact, in many cases the most interesting quantity to consider is actually the \textit{regularized relative entropy of entanglement}, defined as
\begin{equation}
E_{R}^{\infty}(\rho) = \lim_{n \rightarrow \infty} \frac{E_R(\rho^{\otimes n})}{n}.
\end{equation}  
It is known that $E_R^{\infty}$ satisfy properties 1, 2, 3, 4 and 5 and that it is convex and non-lockable\footnote{All these properties follow easily from its definition and the properties of $E_R$, with the exception of asymptotic continuity, which has been proven for $E_R^{\infty}$ in Ref. \cite{Chr06}}. It is an open problem if this measure satisfies the monogamy inequality. It has been shown in Ref. \cite{BHPV07} that full additivity and strong monotonicity are also equivalent properties for the regularized relative entropy of entanglement, although both still remain unproven. This measure is an upper bound to the distillable secret key rate (and hence also to the distillable entanglement) \cite{HHHO05c} (see section \ref{othermasures}) and an upper bound to the entanglement cost. In chapter \ref{QHT} we prove that $E_{R}^{\infty}$ is faithful. 

It is clear that we can form different measures by changing the set over which the minimization is performed in Equation (\ref{relent}). A popular choice has been the set PPT states. In this case, the regularization of the obtained measured is a sharper bound on the distillable entanglement than $E_R^{\infty}$ and can be calculated for some families of states \cite{AEJ+01, AMVW02}. 

%The relative entopy of entanglement can be used to bound the cost and distillation rates: $E_C(\rho) \leq E_R^{\infty}(\rho) \leq E_D(\rho)$ and $E_F(\rho) \leq E_R(\rho) \leq E_D(\rho)$. It is an interesting question if the inequality $E_C(\rho) \geq E_R(\rho)$ can be proven without having to resort to the additivity vconjecture of the entanglement of formation. 

The relative entropy has a central role in quantum hypothesis testing (see e.g. \cite{ANSV07}). As explained in chapter \ref{QHT}, it gives the optimal rate of the exponential decay of the probability of error in quantum Stein's Lemma \cite{HP91, ON00}. Based on this operational meaning of the relative entropy, we can think that the relative entropy of entanglement could have an analogous interpretation as the optimal decay rate when one tries to discriminate a given entangled state $\rho$ from an unknown separable state $\sigma$. If we are given $\rho$ and a particular separable state $\sigma$, then it follows directly from quantum Stein's Lemma that we can indeed discriminate them with a rate at least as large as the relative entropy of entanglement. However, if we do not know which separable state we have, then it is not clear how to construct a POVM sequence achieving the relative entropy of entanglement. 

In chapter \ref{QHT} we will show that $E_R^{\infty}$ is always an achievable rate in this task. Moreover, it will be proven that the regularized relative entropy of entanglement is actually the optimal rate when the unknown sequence of separable states can also contain correlations between different copies. This is a new result of independent importance, as it provides for the first time an operational interpretation for this quantity. 

The main operational interpretation for $E_R^{\infty}$, however, will come from the new paradigm for entanglement that will be explored in chapter \ref{reversible}.  There it will be shown that $E_R^{\infty}$ is both the cost and distillation functions under asymptotically non-entangling operations, which makes it the \textit{unique} entanglement measure in such a setting. 

\subsection{Robustness of Entanglement} \label{robustnesses}

The robustness of entanglement, introduced in Ref. \cite{VT99} by Vidal and Tarrach, is given by
\begin{equation}
R(\rho) = \min_{\sigma \in {\cal S}, s \in \mathbb{R}} s : \frac{\rho + s \sigma}{1 + s} \in {\cal S}.
\end{equation}
It is thus defined as the minimal amount of separable noise that must be mixed with $\rho$ in order to turn it into a separable state. One can also define an analogous quantity, sometimes called the generalized robustness of entanglement, given by \cite{HN03} 
\begin{equation}
R_G(\rho) = \min_{\sigma \in {\cal D}, s \in \mathbb{R}} s : \frac{\rho + s \sigma}{1 + s} \in {\cal S}.
\end{equation}
In the definition above, instead of mixing the state $\rho$ with separable states, any state can be used. Both measures have been shown to satisfy properties 1 and 2 and to be convex and faithful \cite{VT99}. For a maximally entangled state of dimension $K$ it holds that $R(\Phi(K)) = R_G(\Phi(K)) = K -1$; therefore they do not satisfy property 4. This fact motivates the definition of the their logarithm versions\footnote{In analogy to the definition of the logarithm negativity \cite{VW02b} (see section \ref{othermasures} for the definition of the negativity.).} \cite{Bra05}
\begin{equation}
LR(\rho) = \log(1 + R(\rho)), \hspace{0.5 cm} LR_G(\rho) = \log(1 + R_G(\rho)).
\end{equation}
Both measures are full entanglement monotones\footnote{The proof is a simple combination of the full monotonicity of $R$ and $R_G$ with the concavity of the $\log$, following the main idea of the proof that the logarithm negativity is also a full LOCC monotone \cite{Ple05}}, but are not convex. It can moreover be shown that they fail to be asymptotically continuous\footnote{This can be seen by noticing that the regularization of both $LR$ and $LR_G$ are not equal to the von Neumann entropy for bipartite pure states (both measures are additive on bipartite pure states). As they are LOCC monotones and normalized they cannot be asymptotically continuous, as the regularization of any LOCC monotonic, normalized, and asymptotically continuous measure reduces to the entropy of entanglement on pure states \cite{HHH00}, and this is known not to be the case for the lo robustnesses \cite{HN03}.}. A simple application of the operator monotonicity of the $\log$ function shows that $LR_G \geq E_R$ \cite{HMM06}, from which follows that $LR$ and $LR_G$ are upper bounds to the distillable secret key rate and distillable entanglement. 

The robustness measures are intimately connected to entanglement witnesses. As shown in Ref. \cite{Bra05}, for every entangled state $\rho$ the global robustness can be expressed as follows\footnote{A similar expression can be obtained for the robustness of entanglement \cite{Bra05}.}
\begin{equation*}
R_G(\rho) = - \min_{W \in {\cal W}, W \leq \id} \tr(W \rho).
\end{equation*}
The equation above shows that the global robustness is nothing but the expectation value of the optimal entanglement witness for $\rho$, when we restrict ourselves to witnesses with eigenvalues smaller than one. This variational formula for $R_G$ in terms of entanglement witnesses can be extended to several other quantities \cite{Bra05} and was further analysed in Refs. \cite{CT06, EBA07}. Using this connection to entanglement witnesses, the robustness of entanglement has been linked to Masanes result on the activation of every entangled state explained in section \ref{activation} \cite{Bra07}. 

The relation of the robustnesses of entanglement to entanglement witnesses will be used twice in this thesis. In chapter \ref{QHT} we show that the generalized robustness is intrinsically connected to quantum hypothesis testing and to the regularized relative entropy of entanglement. In chapter \ref{reversible}, in turn, we show that $LR$ is the single-copy entanglement cost under non-entangling operations and that the regularization of $LR_G$ is the exact preparation cost under asymptotically non-entangling maps. Moreover, the connection of $LR_G$ and $E_R^{\infty}$ which is established in chapter \ref{QHT} will be crucial in the proof of the main result of chapter \ref{reversible}, which relates both the cost and the distillation functions under asymptotically non-entangling operations to the regularized relative entropy of entanglement. 

\subsection{Other Measures} \label{othermasures}

Before we end our brief exposion on entanglement measures, it should be pointed out that many interesting measures have not being included here. In particular, there are four measures that merit mention. The first is the negativity (together with its associated log version), defined as the sum of the negative eigenvalues of the partial transpose of $\rho$ \cite{ZHSL98, VW02b, Ple05}. Its importance stems from the fact that it is easily computable. Moreover, for states with a positive bi-negativity\footnote{The bi-negativity of a state is defined as $|\rho^{\Gamma}|^{\Gamma}$ \cite{APE03}.}, it has been shown that the log-negativity is equal to the exact-preparation entanglement cost under PPT operations \cite{APE03}\footnote{There are operations which map any PPT state into another PPT state, including the use of ancilla particles \cite{Rai01}.}.

The second is the squashed entanglement, introduced by Christandl and Winter in Ref. \cite{CW04} and given by
\begin{equation} \label{squasdhedent}
E_{sq}(\rho_{AB}) = \inf \{ \frac{1}{2} I(A:B|E)_{\rho_{ABE}} : \tr_{E}(\rho_{ABE}) = \rho_{AB} \},
\end{equation}
where the infimum is taken over all extensions of $\rho_{AB}$ and $I(A:B|E)_{\rho_{ABE}}$ is the conditional mutual information of $\rho_{ABE}$, given by $(A:B|E) = S(AE) + S(BE) - S(ABE) - S(E)$. The squashed entanglement was first introduced as an example of an entanglement measure for which full additivity can be easily established and is presently the measure that satisfies the largest number of desirable properties. Indeed, from the properties outlined in section \ref{entmeasures}, only faithfulness is not know to hold (while the measure is known to be lockable \cite{CW05}). Recently Oppenheim found a beautiful operational interpretation for the squashed entanglement as the fastest rate at which a quantum state can be sent between two parties who share arbitrary side-information \cite{Opp08}. In section \ref{QMA2} of chapter \ref{complexity} we revisit the squashed entanglement and solve an open question concerning its properties, raised in Refs. \cite{Tuc02, YHW07, HHHH07}: we show that the extension $E$ cannot always be taken to be classical.  
 
The third measure is the distillable secret key. In analogy to the distillable entanglement, it is defined as the maximum number of \textit{secret bits} per copy that can be extracted in the asymptotic limit when the parties share several copies of a trusted bipartite state $\rho_{AB}$ and can communicate classical information over a public but authenticated channel. In Ref. \cite{HHHO05b} a remarkable fact about this measure was found: some bound entangled states have a positive distillable secret key. This shows that distillable privacy and distillable entanglement are not analogous resources and leads to several interesting open problems about this measure. Most notably, it is not known whether secret key can be extracted from every entangled state. This question is intimately related to the open problem concerning the existence of bound information in classical information theory \cite{GW00}. 

The final quantity is not an entanglement measure, as it is not monotonic under LOCC, but has shown to be an important tool in entanglement theory and will be used in chapter \ref{complexity}. To define it, we need to introduce the Henderson-Vedral measure of total correlations contained in a bipartite quantum state \cite{HV01, DW03}, 
\begin{equation}
C^{\leftarrow}(\rho_{AB}) = \max_{\{ M_k\}}  S(\rho_A) - \sum_k p_k S(\rho_k),
\end{equation}
where $\{ M_k \}$ is a POVM on Bob's Hilbert space, $p_k := \tr(\id \otimes M_k \rho_{AB})$, and $\rho_k := \tr_{B}( \id\otimes M_k \rho_{AB})/p_k$. Intuitively, the measure gives the maximum Holevo's accessible information \cite{Hol73} of an ensemble held by Alice that Bob can induce by measuring his part of the state. We can also define a quantity $C^{\rightarrow}(\rho_{AB})$ in a similar manner with the roles of Alice and Bob interchanged. From this measure of correlations we can define an entanglement measure under one-way LOCC as the convex roof of $C^{\leftarrow}(\rho_{AB})$ \cite{YHHS05, Yan06},
\begin{equation} \label{GHV}
G^{\leftarrow}(\rho_{AB}) = \min_{\{ p_i, \rho_i \}} \sum_i p_iC^{\leftarrow}(\rho_i),
\end{equation}
where the minimum is taken over mixed ensembles of $\rho$. The first motivation for defining this quantity is its relation to the entanglement cost. From a remarkable equality relating the entanglement of formation, the Henderson-Vedral measure, and the entropy \cite{KW04}, it can be shown that the entanglement cost is lower bounded by $G^{\leftarrow}$ \cite{YHHS05}. Since the latter can be shown to be is strictly positive for every entangled state, it follows, as mentioned in section \ref{costdistillable}, that every entangled state has a non-zero entanglement cost \cite{YHHS05}. A second application of the measure is in the study of the non-shareability of quantum correlations. Using again the equality proved in Ref. \cite{KW04}, Yang derived the following monogamy inequality \cite{Yan06}
\begin{equation}  \label{Yangmonog}
E_{F}(\rho_{A:B_1,...,B_N}) \geq \sum_{k=1}^N G^{\leftarrow}(\rho_{A:B_k}).
\end{equation}
This inequality shows, in particular, that only separable states $\rho_{AB}$ can, for every $n$, be extended to a state $\rho_{ABB_2...B_n}$ symmetric under the exchange of the $B$ systems. We use inequality \ref{Yangmonog} in chapter \ref{complexity} in our study of proof systems with unentangled proofs.

\section{Quantum de Finetti Theorems} \label{definittisection}

In this section we review recent developments on quantum versions \cite{HM76, KR05, Ren05, CKMR07, Ren07} of the seminal result by Bruno de Finetti on the characterization of exchangeable probability distributions \cite{dFin37}. Although these results are not directly connected to entanglement theory, they will turn out to be indispensable tools in the next three chapters in the study of entanglement in the non-i.i.d. (independent and identically- distributed) regime (for further results on entanglement theory in the non-i.i.d. regime the reader is referred to \cite{BD06, BD07, Mat07} and references therein). Such a regime might show up as a desire of a general theory of entanglement, as in the analysis of entanglement distillation of correlated states in chapter \ref{NPPT}, but also even when one considers the i.i.d. case: In chapter \ref{reversible}, for example, the necessity of dealing with non-i.i.d. states will come from the non-additivity of the relative entropy of entanglement. More concretely, in our study of asymptotic entanglement manipulation under non-entangling operations, we will need to know optimal distinguishability rates of an i.i.d. entangled state from arbitrary, possibly correlated, separable states. The quantum de Finetti theorems, in particular the exponential variant proved by Renner \cite{Ren05, Ren07}, will be crucial in reducing such a problem to an almost i.i.d. setting, where it can be tackled with standard techniques.

Consider a symmetric probability distribution of $k$ variables $Q_k$ that can be extended to a $n$-partite symmetric probability distribution for every $n \geq k$\footnote{A probability distribution that can be extended in such a way is called infinitely exchangeable.}. Then the de Finetti theorem \cite{dFin37} says that there is a measure $\mu$ on the set of probability distributions over one variable such that
\begin{equation*}
Q_k = \int \mu( dq_x) (q_x)^k.
\end{equation*}
In words this theorem tells us that infinitely exchangeable probability distributions are exactly those which can be written as a convex combination of product distributions.  

A quantum generalization of this result was first obtained in Ref. \cite{HM76}. To explain it we define the permutation-symmetric states. These are states acting on ${\cal H}^{\otimes k}$ which are left unchanged under the permutation of the $k$ copies. Employing the standard representation of the symmetric group $S_k$ over ${\cal H}^{\otimes k}$, we say that $\rho$ is permutation-invariant (or permutation-symmetric) if 
\begin{equation*}
	\rho = P_\pi \rho  P_\pi,
\end{equation*}
where $P_\pi$ is the representation in ${\cal H}^{\otimes k}$ of an arbitrary element $\pi$ of $S_k$. The set of permutation-symmetric states is denoted by ${\cal S}_k({\cal H}^{\otimes k})$.
%We define the symmetrization operation $\hat{S}_k$ as 
%\begin{equation}
%%\hat{S}_k(\rho) := \frac{1}{k!} \sum_{\pi \in S_k} P_\pi \rho  P_\pi.
%\end{equation}

We say that a family of density operators $\{ \omega_k \}_{k \in \mathbb{N}}$, with $\omega_k \in {\cal D}({\cal H}^{\otimes k})$, is infinitely exchangeable if $\tr_{k}(\omega_n) = \omega_{k-1}$\footnote{$\tr_k$ denotes the partial trace of the $k$th Hilbert space.} and $\omega_k \in {\cal S}_k({\cal H}^{\otimes k})$. Hudson and Moody \cite{HM76} proved that a sequence of states is infinitely exchangeable if, and only if, there is a measure $\mu$ over ${\cal D}({\cal H})$ such that for every $k$,
\begin{equation*}
\omega_k = \int \mu( d\rho) \rho^{\otimes k}.
\end{equation*}
This result has found applications in quantum statistical mechanics \cite{FV06} and quantum information \cite{Wer90, DPS02, DPS04, BA06}. In particular, Werner \cite{Wer90} used it to present an early proof that only separable bipartite states $\rho_{AB}$ can be extended to a state $\rho_{ABB_2...B_n}$ symmetric under the exchange of the $B$ systems for every $n$. This shows that quantum correlations are not shareable, i.e. entanglement is monogamic\footnote{As we saw in section \ref{othermasures}, Yang's monogamy inequality gives a different and somewhat simpler proof of this fact.}. 

The assumption of infinitely exchangeability, although weaker than the i.i.d. assumption, might still be too strong for some applications and in particular cannot be a priori checked in experiments. This lead us to the more realistic case of finitely exchangeable states. We say that a state $\rho \in {\cal D}({\cal H}^{\otimes k})$ is $n$-exchangeable if there is a symmetric state $\omega_{n} \in {\cal D}({\cal H}^{\otimes k})$ such that $\tr_{\backslash 1,...,k}(\omega_k) = \rho$\footnote{$\tr_{\backslash 1,...,k}$ stands for the partial trace over the $n - k$ last Hilbert spaces.}. From the result of Moody and Hudson we might intuitively expect that if $n \gg k$, then $\rho$ should be close to a convex combination of i.i.d. sates. This is indeed the case and is the content of the finite de Finetti theorem, proved in the classical setting by Diaconis and Freedman \cite{DF80} and in the quantum case first by K\"onig and Renner \cite{KR05} and then by Christandl, K\"onig, Mitchison, and Renner with an improved error bound \cite{CKMR07}. 

\begin{theorem} \label{definetti1}
\cite{CKMR07} Given a permutation invariant state $\rho_n$ over ${\cal H}^{\otimes n}$ and an integer $k \leq n$, there is a probability measure $\mu$ over ${\cal D}({\cal H})$ such that
\begin{equation} 
\left \Vert \tr_{\backslash 1,...,k}(\rho_n) - \int \mu( d\rho) \rho^{\otimes k} \right \Vert_1 \leq 2 \frac{\dim({\cal H})^2k}{n}.
\end{equation}
\end{theorem}

Examples with error of $k/n$ show that the error bound of the equation above cannot be improved \cite{CKMR07}. This is however not good enough for some application, in particular when one is interested in rates. Indeed, if we require that the error in the approximation goes to zero asymptotically with $n$, then we must take $k$ to be sublinear in $n$. This is major drawback e.g. in quantum key distribution, where the finite quantum de Finetti theorem can be used to generalize security proofs from \textit{collective attacks}\footnote{In a collective attack Eve interacts in the same manner with each of the particles sent from Alice and Bob and in the end performs a collective measurement on her quantum systems.} to the most general form of attacks possible, \textit{coherent attacks}\footnote{In a coherent attack Eve might realize any operation allowed by quantum mechanics. For further discussion the reader is invited to Ref. \cite{Ren05}.}. Unfortunately, an application of Theorem \ref{definetti1} to this problem only shows that secret bits can be obtained with a zero rate. 

\subsection{Exponential Quantum de Finetti Theorem} \label{expdf}

A solution to this difficulty is given by the exponential de Finetti theorem proved by Renner \cite{Ren05, Ren07}. The key point is that we can get a much better error bound, exponentially decreasing in $n$, if we are willing to relax the states containing in the convex combination from i.i.d. to states which behave like i.i.d. in what regards their \textit{bulk properties}. As this theorem will be important for establishing the main result of chapter \ref{QHT}, in the rest of this section we recall the necessary definitions to state it precisely. A much more complete analysis of this important result, including its proof, can be found in Refs. \cite{Ren05, Ren07, KM07}. Some of the notation that we use is taken from \cite{HHH+08a, HHH+08b}.  

Let $\text{Sym}({\cal H}^{\otimes n})$ denote the symmetric subspace of ${\cal H}^{\otimes n}$ and $0 \leq r \leq n$. Then, for a $\ket{\theta} \in {\cal H}$, we define the set of $\binom{n}{r}$-i.i.d states in $\ket{\theta}$ as 
\begin{equation*}
{\cal V}({\cal H}^{\otimes n}, \ket{\theta}^{\otimes n - r}) := \{ P_{\pi}( \ket{\phi}^{\otimes n - r}\otimes \ket{\psi_r}) : \pi \in S_n, \ket{\psi_r} \in {\cal H}^{\otimes r}   \}.
\end{equation*}
Thus for every state in ${\cal V}({\cal H}^{\otimes n}, \ket{\theta}^{\otimes n - r})$ we have the state $\ket{\theta}$ in at least $n - r$ of the copies. The set of almost power states in $\ket{\theta}$ is defined as \cite{HHH+08a, HHH+08b}
\begin{equation}
\ket{\theta}^{[\otimes, n, r]} := \text{Sym}({\cal H}^{\otimes n}) \cap \text{span} ({\cal V}({\cal H}^{\otimes n}, \ket{\theta}^{\otimes n - r})). 
\end{equation}

\begin{theorem} \label{expdefinetti}
\cite{Ren05,Ren07, KM07} For any state $\ket{\psi_{n + k}} \in \text{Sym}({\cal H}^{\otimes n + k})$ there exists a measure $\mu$ over ${\cal H}$ and for each pure state $\ket{\theta} \in {\cal H}$ another pure state $\ket{\psi^{\theta}_n} \in \ket{\theta}^{[\otimes, n, r]}$ such that
\begin{equation}
\left \Vert \tr_{1, ..., k}(\ket{\psi_{n}}\bra{\psi_{n}}) - \int \mu(d\ket{\theta}) \ket{\psi^{\theta}_{n}}\bra{\psi^{\theta}_{n}} \right \Vert_1 \leq n^{\dim({\cal H})} 2^{- \frac{k(r + 1)}{n + k}}.
\end{equation}
\end{theorem}

The generalization of Theorem \ref{expdefinetti} to permutation-symmetric mixed states goes as follows. First, we use the fact that every permutation-symmetric mixed state $\rho_{n + k}^S$ acting on ${\cal H}_S^{\otimes n + k}$ has a symmetric purification $\ket{\psi}^{SE}_{n + k} \in ({\cal H}_S \otimes {\cal H}_E)^{\otimes n + k}$, with $\dim({\cal H}_E) = \dim({\cal H}_S)$ (see e.g. Lemma 4.2.2 of Ref. \cite{Ren05}). Then we apply Theorem \ref{expdefinetti} to $\ket{\psi}^{SE}_{n + k}$ and use the contractiveness of the trace norm under the partial trace \cite{NC00} to find 
\begin{equation} \label{df1}
\left \Vert \tr_{1,...,k}(\rho_{n + k}) - \int \mu(d\sigma)\rho_{\sigma} \right \Vert_1 \leq n^{\dim({\cal H})^2} 2^{- \frac{k(r + 1)}{n + k}}
\end{equation}
where
\begin{equation} \label{df2}
\rho_{\sigma} := \tr_{E}(\ket{\psi_{n}^{\ket{\theta}}}\bra{\psi_{n}^{\ket{\theta}}}),
\end{equation}
with $\sigma := \tr_E(\ket{\theta}\bra{\theta})$ and
\begin{equation} \label{df3}
\mu(d\sigma) := \int_{\ket{\theta} \supset \sigma} \mu(d\ket{\theta}).
\end{equation}
In the equation above $\ket{\theta} \supset \sigma$ means that the integration is taken with respect to the purifying system $E$ and runs over all purifications of $\sigma$.

%In chapter 4 we will also use the following Lemma, proved in Ref. \cite{Ren05}. 
%\begin{lemma}
%\cite{Ren05} Let $\ket{\Psi} \in \ket{\theta}^{[\otimes, n, r]}$. Then there exists an orthornomal family $\{ \ket{\Psi^s} \}_{s \in {\cal S}}$ of %vectors from ${\cal V}({\cal H}^{\otimes n}, \ket{\theta}^{\otimes n - r})$ with cardinality $|{\cal S}| \leq 2^{n h(\frac{r}{n})}$ such that %$\ket{\Psi} \in \text{span}\{ \ket{\Psi^s} \}_{s \in {\cal S}}$.
%\end{lemma}
%In the Lemma $h(x) := -x\log x - (1 - x)\log(1 - x)$ is the binary Shannon entropy. 

There is a simple, but useful, characterization of states $\ket{\Psi} \in \ket{\theta}^{[\otimes, n, r]}$, derived in Ref. \cite{Ren05}, which we now recall. Let $\{ \ket{i} \}_{i=1}^{d}$ be an orthonormal basis for ${\cal H}$, with $\ket{1} = \ket{\theta}$. Some thought reveal that any $\ket{\Psi} \in \ket{\theta}^{[\otimes, n, r]}$ can be written as
\begin{equation}
\ket{\Psi} = \sum_{k=0}^{r} \sum_{l=1}^{\binom{n}{k}} \binom{n}{k}^{-1/2} \beta_k  P_{\pi_{l, k}} \ket{\Psi_k} \otimes \ket{\theta}^{\otimes n - k}, 
\end{equation}
where $\ket{\Psi_k}$ are permutation-symmetric states in $({\cal H}/\{\ket{\theta}\})^{\otimes n - k}$, $\pi_{l, k}$ are permutations such that
\begin{equation*}
\sum_{l=1}^{\binom{n}{k}} \binom{n}{k}^{-1/2} P_{\pi_{l, k}} \ket{\Psi_k} \otimes \ket{\theta}^{\otimes n - k}
\end{equation*}
is permutation-invariant, and $\beta_k$ complex coefficients satisfying
\begin{equation*}
\sum_{k=0}^r |\beta_k|^2 = 1.
\end{equation*}

For $r < n/2$, the total number of distinct terms can be bounded as follows
\begin{equation}
\sum_{k=0}^r \binom{n}{k} \leq \sum_{k=0}^r 2^{n h(k/n)} \leq \sum_{k=0}^r 2^{n h(r/n)} = (r + 1) 2^{n h(r/n)}, 
\end{equation}
where $h(x) := -x\log x - (1 - x)\log(1 - x)$ is the binary Shannon entropy. The first inequality follows from Formula 12.40 of \cite{CT91} and the second from the fact that $h(x)$ is monotonic increasing in the interval $[0, 1/2]$.

\subsubsection{Chernoff-Hoeffding Bound for Almost Power States}

The states $\tr_{E}(\ket{\psi^{\theta}_n}\bra{\psi^{\theta}_n})$ behave like $\tr_E(\ket{\theta}\bra{\theta})^{\otimes n}$ in many respects. One example is the case where the same POVM is measured on all the $n$ copies. Let us first recall a variant of the Chernoff-Hoeffding bound for product probability distributions. 

\begin{lemma} \label{classicalCH}
\cite{Ren05} Let $P_X$ be a probability distribution on ${\cal X}$ and let $x$ be chosen according to the $n$-fold product distribution $(P_X)^n$. Then, for any $\delta > 0$, 
\begin{equation*} 
\text{Pr}_{x} [|| \lambda_x - P_X  ||_1 > \delta] \leq 2^{- n (\frac{\delta^2}{2 \ln 2} - |{\cal X}|\frac{\log(n + 1)}{n})}.
\end{equation*}
were $||.||_1$ is the trace distance of two probability distributions and $|{\cal X}|$ is the cardinality of ${\cal X}$.
\end{lemma}

Let $\{ M_{\omega} \}_{\omega \in {\cal W}}$ be a POVM on ${\cal H}$ and define its induced probability distribution on $\ket{\theta}$ by $P_{M}(\ket{\theta}\bra{\theta}) = \{ \bra{\theta}M_{\omega}\ket{\theta} \}_{\omega \in {\cal W}}$. Theorems 4.5.2 of Ref. \cite{Ren05} and its reformulation as Lemma 2 of Ref. \cite{HHH+08a} show the following.
\begin{lemma} \label{hof}
\cite{Ren05, HHH+08a} Let $\ket{\Psi_n}$ be a vector from $\ket{\theta}^{[\otimes, n, r]}$ with $0 \leq r \leq \frac{n}{2}$ and $\{ M_{\omega} \}_{\omega \in {\cal W}}$ be a POVM on ${\cal H}$.  
\begin{equation} \label{testtest}
Pr \left(  \Vert P_M(\ket{\theta}\bra{\theta}) - P_M(\ket{\Psi_n}\bra{\Psi_n}) \Vert_1 > \delta \right) \leq 2^{- n \left(\frac{\delta^2}{4} - h\left(\frac{r}{n}  \right)  \right) + | {\cal W} |\log(\frac{n}{2} + 1)} 
\end{equation}
where $P_M(\ket{\Psi_n}\bra{\Psi_n})$ is the frequency distribution of outcomes of $M^{\otimes n}$ applied to $\ket{\Psi_n}\bra{\Psi_n}$, and the probability is taken over those outcomes. 
\end{lemma}

This Lemma shows that apart from the factor $h(r/n)$, which in an usual application of Lemma \ref{hof} is taken to be vanishing small, the statistics of the frequency distribution obtained by measuring an almost power state along $\ket{\theta}$ is the same as if we had $\ket{\theta}^{\otimes n}$.

\chapter{Entanglement Distillation from Correlated Sources} \label{NPPT}

\section{Introduction}

In the previous chapter we have studied the concept of entanglement distillation and the classification of density operators into distillable and non-distillable. There an important underlying assumption was that Alice and Bob had several identical and independent distributed (i.i.d.) copies of a given bipartite state. In this chapter we revisit this problem now considering non-i.i.d. sequences. We focus on sequences of states for which all single-copy reduced density matrices are the same and prove that whether this sequence is distillable or not is only a property of this reduction. 

This generalization is interesting for two reasons. The first is that in many situations the i.i.d. assumption is unjustified. For example, in a common scenario in key-distribution based on entanglement distillation, Alice sends quantum bits to Bob with the aim of establishing a shared state of the form $\rho_{AB}^{\otimes n}$, which can then be distilled into a maximally entangled state from which a secure key can be obtained. A well-known difficulty with this approach is that an eavesdropper might perform coherent attacks over the particles sent to Bob in such a way that the final state shared by the parties is not i.i.d. (see e.g. \cite{Ren05}). In this context it would thus be helpful to have a criterion of when a sequence of correlated states is distillable.  

The second reason comes from the desire to better understand the set of undistillable states. As will be shown later in this chapter, by extending the definition of undistillability to states with correlations among its copies, we can obtain a simpler characterization of the \textit{convex hull} of the undistillable states set, in terms of the intersection of a nested sequence of convex sets. As examples of the usefulness of this characterization, we employ it to prove stronger versions of Masanes result on the activation of every entangled state \cite{Mas06} and to give new supporting evidence on the existence of NPPT bound entanglement (see section \ref{boundentanglement}).

The organization of this chapter is the following. In section \ref{defmain} we present some definitions and state the main results. Section \ref{proof1} will be devoted to the proofs of Theorem \ref{T=coC} and Corollary \ref{posdist} while in section \ref{proof2} we prove Theorem \ref{actcor}. Finally in section \ref{NPPT2} we discuss the connection of these results to the conjecture on the existence of NPPT bound entanglement. 

\section{Definitions and Main Results} \label{defmain}

As discussed in section \ref{entanglementdistillation}, we say that a bipartite state $\rho$ is undistillable if for every $k \in \mathbb{N}$, $F_2(\rho^{\otimes k})= 1/2$, where $F_2$ is the singlet-fraction, defined by Eq. (\ref{singletfraction}). The set of $k$-undistillable states over ${\cal H} = {\cal H}_A \otimes {\cal H}_B$, where ${\cal H}_{A/B}$ are finite dimensional Hilbert spaces, is formed by all the states $\rho \in {\cal D}({\cal H})$ such that $F_2(\rho^{\otimes k}) = 1/2$ and is denoted by ${\cal C}_k({\cal H})$. The set of undistillable states ${\cal C}$ is the intersection of the sets ${\cal C}_k$. 

Our first definition is the generalization of such sets to the case where correlations among the several copies of the state might be present. We consider the worst case scenario and say that a state $\rho$ is \textit{copy-correlated $k$-undistillable} if there is a $1$-undistillable state $\omega_k \in {\cal D}({\cal H}^{\otimes k})$ such that $\tr_{\backslash m}(\omega_k) = \rho$ for every $1 \leq m \leq k$. In other words, if we can add correlations to $\rho^{\otimes k}$, forming the state $\omega_k$, such that no two qubit entanglement can be extracted from $\omega_k$, we say that $\rho$ is copy-correlated $k$-undistillable.

Let us define the symmetrization operation $\hat{S}_k$ by 
\begin{equation}
\hat{S}_k(\rho) := \frac{1}{k!} \sum_{\pi \in S_k} P_\pi \rho  P_\pi.
\end{equation}
As in chapter \ref{entanglement}, we denote the set of permutation-symmetric states over ${\cal H}^{\otimes n}$ by ${\cal S}_k({\cal H}^{\otimes k})$. It is clear that if there is 1-undistillable extension $\omega_k$ of $\rho$, then $\hat{S}_k(\omega_k)$ is also a valid extension. It hence follows that w.l.o.g. we can define the set of copy-correlated $k$-undistillable states as follows. 

\begin{definition}
A bipartite state $\rho \in {\cal D}({\cal H})$ is copy-correlated $k$-undistillable if it has a permutation-symmetric extension $\omega_k \in {\cal D}({\cal H}^{\otimes k})$ which is single-copy undistillable. We denote the set of all such states by ${\cal T}_k$, i.e.
\begin{eqnarray}
	{\cal T}_k &:=& \{ \rho \in {\cal D}({\cal H}) 
	: \exists \hspace{0.1 cm} \omega_k \in 
	{\cal S}_k({\cal H}^{\otimes k}) 
	\cap {\cal C}_1({\cal H}^{\otimes k}) \hspace{0.3 cm} 
	\text{s.t.} \hspace{0.3 cm} 
	\rho = \tr_{\backslash 1}(\omega_k) \}.
\end{eqnarray}

\end{definition}

In the same way as one defines undistillability as $k$-undistillability for all $k$, one can introduce an analogous definition in the 
copy-correlated case. 

\begin{definition}
A state $\rho \in {\cal D}({\cal H})$ is copy-correlated undistillable if it is copy-correlated $k$-undistillable for every $k \in \mathbb{N}$. We denote the set of all such states by ${\cal T}$, i.e.
\begin{equation}
{\cal T} := \bigcap_{k \in \mathbb{N}} {\cal T}_k.
\end{equation}
\end{definition}

In words, a state $\rho$ belongs to ${\cal T}$ if for every number of copies of 
the state one can add correlations among them so that no useful entanglement 
can be establish at all. 

Equipped with these definitions, we can now
state the main results of this chapter. The first concerns the relationship
between copy-correlated undistillable states and the
undistillable states in the ordinary i.i.d. sense. The set of copy-correlated undistillable states turns out to be equal to the convex hull of 
the set of undistillable states. Since the latter 
set is possibly non-convex
(a property related to the existence of NPPT bound 
entanglement), the convex hull of this set might, however, 
be different from the set itself.

\begin{theorem} \label{T=coC}
The set of copy-correlated undistillable states is equal to the convex-hull 
of the set of undistillable states.
\begin{equation} 
{\cal T} = co({\cal C}).
\end{equation}
\end{theorem}

Up to now our definitions have focused on the negative of entanglement distillation. If $\rho \in {\cal T}$, then 
for every number of copies of $\rho$ we can add correlations to them in a way that makes impossible to extract any two qubit 
entangled state by SLOCC. However, in the i.i.d. case we know that if a state is distillable then not only a two qubit 
entangled state can be extracted from it, but actually maximally entangled states can be distilled at a non-zero rate. Although 
it is an open question if we can extract a non-zero rate of EPR pairs from every correlated sequence associated to every $\rho \notin {\cal T}$, it is indeed possible to extract
maximally entangled states from every copy-correlated distillable state. This is the content of the next corollary of Theorem \ref{T=coC}.

\begin{corollary}\label{posdist}
Let $\rho\not\in {\cal T}({\cal H})$. Then, for any sequence of states $\{ \omega_n \}_{n \in \mathbb{N}}$ with 
reductions equal to $\rho$, every integer $D$, and every $\lambda \in [1/D, 1)$, there 
is an integer $n$ such that   
\begin{equation}
F_D(\omega_n) > \lambda.
\end{equation}
\end{corollary}

We now turn to the second main result. It can be seen as a generalization of Masanes activation result \cite{Mas06} and indicates the power of copy-correlated $k$-undistillable states to serve as activators in single-copy quasi-distillation. 

\begin{theorem} \label{actcor}
For every entangled state $\rho \in {\cal D}({\cal H})$ and every $k \in \mathbb{N}$ there is a copy-correlated $k$-undistillable state $\sigma$ such that the joint state $\rho \otimes \sigma$ is single-copy distillable, i.e.
\begin{equation}  
F_{2}(\rho \otimes \sigma) > \frac{1}{2}.
\end{equation}
\end{theorem}

As ${\cal C}_1 = {\cal T}_1$, the main result of Ref. \cite{Mas06} is a particular case of Theorem \ref{actcor}. There is an immediate 
corollary of the previous result which we can state as follows.

\begin{corollary}
For every entangled state $\rho \in {\cal D}({\cal H})$ and any 
$\varepsilon > 0$ there is a single-copy undistillable state $\sigma$ 
such that
\begin{enumerate}
	\item We can find a probability distribution 
	$\{ p_i \}$ and a set of undistillable states $\{ \rho_i \}$ satisfying
\begin{equation}	
	\bigl \Vert \sigma - \sum_i p_i \rho_i	\bigr \Vert_1 \leq \varepsilon,
\end{equation}	
  \item The joint state $\rho \otimes \sigma$ is single-copy distillable.
\end{enumerate}
\end{corollary}

This corollary is a direct consequence of Theorem \ref{actcor} 
and a standard result of convex analysis stating that a family of 
closed convex sets $\{ A_i \}$ such that $A_{i + 1} \subseteq A_i$ 
converge to their intersection with respect to the Hausdorff 
distance \cite{Kur58}\footnote{The Hausdorff distance between two sets $A, B \in \mathbb{R}^n$ is given by 
\begin{equation*}
H(A, B) = \max_{a \in A} \min_{b \in B} || a - b ||_1.
\end{equation*}
}.

Motivated by these findings, we are led to the following conjecture. 

\begin{conjecture} \label{conj1}
For every entangled state $\rho \in {\cal D}({\cal H})$ there is an undistillable 
state $\sigma$ such that the joint state $\rho \otimes \sigma$ is 
single-copy distillable, i.e.
\begin{equation}  
F_{2}(\rho \otimes \sigma) > \frac{1}{2}.
\end{equation}
\end{conjecture}

This statement would imply that the distillable entanglement is not additive and would prove
the existence of NPPT bound entanglement, as the tensor product of two PPT states cannot be distillable. In section \ref{NPPT2} we discuss the limitations of our methods for solving conjecture \ref{conj1}.

\section{Proof of Theorem \ref{T=coC} and Corollary \ref{posdist}} \label{proof1}

We will now proceed to prove Theorem \ref{T=coC}.
To make the discussion simpler, before we proof it in full we first derive in the next Lemma a weaker, but already illustrative, characterization of the elements of ${\cal T}$. An important element in its proof is the finite quantum de Finetti theorem discussed in chapter \ref{entanglement} (Theorem \ref{definetti1}). 

\begin{lemma} \label{lem1}
A  state $\sigma\in {\cal D}({\cal H})$ 
belongs to ${\cal T}$ if, and only if, there 
exists a probability measure $\mu$
over the state space ${\cal D}({\cal H})$  such that
\begin{equation} \label{pin}
\sigma = \int \mu(d\rho)\rho \hspace{0.3 cm} \text{and} \hspace{0.2 cm} \pi_{k} := \int \mu(d\rho) \rho^{\otimes k} \hspace{0.1 cm} \in \hspace{0.2 cm} {\cal C}_1({\cal H}^{\otimes k}) 
\end{equation}
for every $k \in \mathbb{N}$.
\end{lemma}

\begin{proof}
Let $\sigma \in {\cal T}$. Then, for each $k \in \mathbb{N}$, 
there exists a permutation-symmetric state
$\omega_k \in {\cal C}_1({\cal H}^{\otimes k})$ such that $\tr_{\backslash 1}(\omega_k) = \sigma$. From Theorem \ref{definetti1} it follows that for each $k \geq 1$ there exists a probability measure $\mu_k$ such that
\begin{equation*}
\Vert \tr_{k + 1, \dots, k^2}(\omega_{k^2}) - \int \mu_k(d\rho) \rho^{\otimes k}  \Vert_1 \leq \frac{4 d^2}{k}.
\end{equation*} 
where $d := \dim({\cal H})$. Let us define $\pi^k_{j} := \tr_{j + 1, \dots, k^2} (\omega_{k^2})$. From the contractiveness of the trace norm under partial tracing, we have that for each $j \leq k$,
\begin{equation} \label{EEE}
\Vert \pi^k_j - \int \mu_k(d\rho) \rho^{\otimes j}  \Vert_1 \leq 
\frac{4 d^2}{k}.
\end{equation} 
Moreover, for every $0 \leq j \leq k < \infty$: $\tr_{\backslash 1}(\pi^k_j) = \sigma$ and $\pi^k_j\in{\cal C}_1({\cal H}^{\otimes j})$, as the partial trace can be done locally. As Eq. (\ref{EEE}) is true for every $k$ and the set ${\cal C}_1$ is closed, we find that for each $j$ there is a probability measure $\nu_j$ such that 
\begin{equation} \label{ex1}
\int \nu_j(d\rho) \rho^{\otimes j} \in {\cal C}_1, \hspace{0.2 cm}\int \nu_j(d\rho) \rho = \sigma.
\end{equation}
Consider the sequence of probability measures $\{ \nu_j \}$. The state space over ${\cal H}$ can be regarded as a compact subset of $\mathbb{R}^n$ and it thus follows that 
there exists a subsequence $\{ \nu_{j'} \}$ converging weakly to a measure $\nu$ (see e.g. Theorem 9.3.3 of Ref. \cite{Dud02}). One may always trace out $j-n$ copies in Eq. (\ref{ex1}) to find
\begin{equation*} 
\int \nu_j(d\rho) \rho^{\otimes n} \in {\cal C}_1, \hspace{0.2 cm}\int \nu_j(d\rho) \rho = \sigma.
\end{equation*} 
Then, from Lemma \ref{stormer}, which is stated and proved in the sequel, we find that for all $n$, $\int \nu_j(d\rho) \rho^{\otimes n} $ converges to $\int \nu(d\rho) \rho^{\otimes n}$ in trace norm. Using again that the set of single copy undistillable states is closed, we find that for every $n$,
\begin{equation*} 
\int \nu(d\rho) \rho^{\otimes n} \in {\cal C}_1
\end{equation*} 
and 
\begin{equation*} 
\int \nu(d\rho) \rho = \sigma.
\end{equation*} 
The converse direction follows straightforwardly from the definition of ${\cal T}$.
\end{proof}
\vspace{0.3 cm}

Let $(S, d)$ be a metric space. For a real valued function of $S$, the Lipschitz seminorm is defined by $||f||_L := \sup_{x \neq y} |f(x) - f(y)|/d(x, y)$. Call the supremum norm $||f||_{\infty} := \sup_{x}|f(x)|$. Then, the bounded Lipschitz seminorm is defined by
\begin{equation}
||f||_{BL} := ||f||_L + ||f||_{\infty}.
\end{equation}
Given two probabilities measures $\mu$ and $\nu$ over $S$, the L\'evy-Prohorov metric is defined by
\begin{equation*}
\beta(\mu, \nu) = \sup \left \{ \left | \int f \mu - \int f \nu \right | : ||f ||_{BL} \leq 1 \right\}. 
\end{equation*}

%\begin{lemma} \label{problemma}
%(Theorem 11.3.3 of Ref. \cite{Probability}) for any separable metric space $(S, d)$ and probablity measures $\mu_n$ and $\mu$ on $S$, $\mu_n %\rightarrow \nu$ weakly iff $\beta(\mu_n, \mu) \rightarrow 0$.
%\end{lemma}

\begin{lemma} \label{stormer}
Given a sequence of probability measures $\mu_k$ over $D({\cal H})$ converging weakly to $\mu$, it holds true that for every $n \in \mathbb{N}$,
\begin{equation} \label{final}
\lim_{k \rightarrow \infty}\left \Vert \int \mu_{k}(d\rho) \rho^{\otimes n}  - \int \mu(d\rho) \rho^{\otimes n} \right \Vert_1 = 0,
\end{equation}
\end{lemma}
\begin{proof}
The Lemma is a simple application of Theorem 11.3.3 of Ref. \cite{Dud02}, stating that the sequence of probability measures $\mu_k$ converge weakly to $\mu$ iff $\beta(\mu_k, \mu) \rightarrow 0$. As, by assumption, $\mu_k \rightarrow \mu$ weakly, we have
\begin{equation} \label{prob22}
\beta(P_k, P) \rightarrow 0. 
\end{equation}
Consider the functions:
\begin{equation*}
f_{X, n}(\rho) := \tr(\rho^{\otimes n} X), 
\end{equation*}
where $|| X ||_{\infty} \leq 1$. It is easy to see that $|| f_{X, n} ||_{BL} \leq K(n)$, for some bounded function of $n$ only. The Lemma then follows from Eq. (\ref{prob22}) and the variational characterization of the trace norm. 
\end{proof}
\vspace{0.3 cm}

Considering the characterization ${\cal T}$ just obtained, it should now be clear how we are going to prove Theorem \ref{T=coC}. The main idea is that we can perform tomography by LOCC in some of the copies of the state $\int \mu(d \rho) \rho^{\otimes n}$ and, conditioned on the estimated state, filter a particular one. Then, as this resulting state should be single copy undistillable and we can obtain an arbitrary large number of copies of any of the states appearing in the convex combination as an outcome of such a filtering, we find that $\mu$ can only be supported on undistillable states. 
In order to make the above handwaving argument rigorous, we now introduce the concept of an informationally complete POVM.

\subsection{Informationally Complete POVM} \label{infcomp}

An informationally complete POVM in ${\cal B}(\mathbb{C}^{m})$ is defined as a set of positive semi-definite operators $A_i$ forming a resolution of the identity and such that $\{ A_i \}$ forms a basis for ${\cal B}(\mathbb{C}^{m})$. Informationally complete POVMs can be explicitly constructed in every dimension (see e.g. \cite{KR05}).

We say that a family $\{ A_i \}$  of elements from 
${\cal B}(\mathbb{C}^m)$ is a dual of the a 
family $\{ A_i^* \}$ if for all $X \in {\cal B}(\mathbb{C}^m)$,
\begin{equation} \label{dualbasis}
X = \sum_{i} \tr[A_i X] A_i^*.
\end{equation}
The above equation implies in particular that the operator 
$X$ is fully determined by the expectations values $\tr[A_i X]$. 
Another useful property is that for every informationally complete POVM in ${\cal B}(\mathbb{C}^{m})$ there is a real number $K_m$\footnote{For example, in the family of informationally complete POVM constructed in Ref. \cite{KR05}, $K_m \leq m^4$.} such that for every two states $\rho$ and $\sigma$,
\begin{equation} \label{eq3.1.5}
|| \rho - \sigma ||_1 \leq K_m || p_{\rho} - p_{\sigma}||_1,
\end{equation}
with $p_{\rho} = \tr(A_i \rho)$ and $p_{\sigma} = \tr(A_i \sigma)$. We will use this relation in chapter \ref{QHT}.

For proving Theorem \ref{T=coC} we make use of the simple observation that if $\{ A_i \}$ and $\{ B_j \}$ are informationally complete POVMs on ${\cal B}(\mathbb{C}^m)$ and 
${\cal B}(\mathbb{C}^l)$, then $\{ M_{i,j} \}$, defined by 
\begin{equation}\label{Prod}
	M_{i,j} := A_i \otimes B_j, 
\end{equation}	
is an informationally complete POVM on ${\cal B}(\mathbb{C}^m\otimes \mathbb{C}^l)$.

%Before now turning to the proof of Theorem 1, there is
%one last ingredient that we need for our argument: It may
%be viewed as a variant of a Chernoff bound. Note that this
%is a statement on classical probability distributions, not
%on quantum states.
 
%\begin{lemma} [Variant of Chernoff's bound \cite{Tomas}] \label{renner}
%Let $P_X$ be a probability distribution on ${\cal X}$ and let $x$ be chosen according to the $n$-fold product distribution $(P_X)^n$. Then, for any %$\delta > 0$, 
%\begin{equation} 
%\text{Pr}_{x} [|| \lambda_x - P_X  ||_1 > \delta] \leq 2^{- n (\frac{\delta^2}{2 \ln 2} - |{\cal X}|\frac{\log(n + 1)}{n})}.
%\end{equation}
%Here, $||.||_1$ is the trace distance of two probability distributions and $|{\cal X}|$ is the cardinality of ${\cal X}$. 
%\end{lemma}

\vspace{0.5 cm}

\begin{proof} (Theorem \ref{T=coC})

We proceed by showing that both $\text{co}({\cal C}) \subseteq {\cal T}$ and ${\cal T} \subseteq \text{co}({\cal C})$ hold true. If $\rho \in \text{co}({\cal C})$, then there are undistillable states $\rho_i$ and a probability distribution $p_i$ such that $\rho = \sum_i p_i \rho_i$. As $\sum_i p_i \rho_i^{\otimes k} \in {\cal T}_k$ for every $k$, we find that indeed $\text{co}({\cal C}) \subseteq {\cal T}$.

Let us then focus on the converse inclusion. To this aim, let $\pi \in {\cal T}$. Then for all $n \in \mathbb{N}$ there exists
a $\pi_n$ given by Eq. (\ref{pin}) such that $\tr_{\backslash 1}[\pi_n] = \pi$. 

We now show that the probability measure $\mu$ obtained from Eq. (\ref{pin}) is up to a set of measure zero 
supported only on undistillable states. We do this proving that for every 
$n \in \mathbb{N}$, the probability measure $\mu(\rho)$ 
vanishes for all  $n$-distillable states, except again in a set of measure zero. 
The main idea goes as follows. We consider
$\pi_{n+m}$ and construct a SLOCC that performs
measurements of an informationally complete
POVM in the last $m$ systems. Based on this information,
we perform a further operation on the first $n$ systems  
with the aim of distilling the entanglement of the state at hand when it is 
distillable or filtering it out when it is not.

More specifically,  
for each $n,m\in\mathbb{N} $ we define
the SLOCC map 
$\Lambda_{m, n} : 
{\cal B}({\cal H}^{\otimes (m + n)}) \rightarrow 
{\cal B}(\mathbb{C}^2 \otimes \mathbb{C}^2)$ as follows: 

\begin{itemize}

\item We first measure the 
informationally-complete POVM $\{ M_{i,j} \}=:
\{ M_{k} \}$ of Eq.\ (\ref{Prod})
individually on each of 
the last $m$ bipartite systems, 
where $k$ is the joint index labeling the outcomes. 
This is clearly a LOCC operation as the parties can 
implement their measurements individually and 
communicate the outcomes obtained to each other. 
In this way, one can estimate
an empirical probability distribution $P_m(k)$ 
from the relative frequency of the outcomes 
$k$ of the POVM. 
\item 
Then, using Eq.\ (\ref{dualbasis}), we form the operator
\begin{equation*}
X_m = \sum_k P_m(k) M_k^* \in {\cal B}({\cal H}).
\end{equation*}
As this might not be a valid density operator, we define $\sigma_m\in {\cal D}({\cal H})$ 
as a state which is closest in trace norm to $X_m$.
This is done based on the measurement outcomes
obtained above. If $\sigma_m$ defined in this way is 
not unique, we select one from the respective set of 
solutions. The state
$\sigma_m$ can now either be $n$-distillable or
$n$-undistillable. 

\item In the first case, so if
$\sigma_m\in {\cal D}({\cal H})$ is $n$-distillable,
we apply the trace preserving LOCC map 
$\Omega: {\cal B}({\cal H}^{\otimes n})
\rightarrow {\cal B}(\mathbb{C}^2\otimes \mathbb{C}^2)$ on the remaining $n$ systems which 
minimizes the following linear function: $\tr[\Omega(\sigma_m^{\otimes n})(\id/2 - \phi_2)]$. 
The map $\Omega$ can be identified with the optimal trace preserving quasi-distillation map
for $\sigma_m^{\otimes n}$.

In the second case, so if $\sigma_m\in {\cal D}({\cal H})$ is $n$-undistillable, we discard 
the state and replace it by the zero operator on ${\cal H}^{\otimes n}$.

\end{itemize}
This procedure defines our family of SLOCC
operations $\Lambda_{m, n} : 
{\cal B}({\cal H}^{\otimes (m + n)}) \rightarrow 
{\cal B}(\mathbb{C}^2 \otimes \mathbb{C}^2)$. Note that as $\pi_{n + m} \in {\cal C}_1({\cal H}^{\otimes{(n + m)}})$,
\begin{equation} \label{eq52}
\tr[\Lambda_{n, m}(\pi_{n + m})(\id/2 - \phi_2)] \geq 0,
\end{equation}
for all $m, n$.

From the law of large numbers \cite{Dud02} we can infer that the probability that the trace norm difference
of the estimated state with the real state is larger than $\varepsilon$, for any $\varepsilon > 0$, 
goes to zero when $m \rightarrow \infty$. So we find 
that the family of functions, defined  
for states $\rho\in {\cal D}({\cal H})$ as
\begin{equation*}
f_{m}(\rho) :=  \tr[\Lambda_{m, n}(\rho^{\otimes{(n + m)}})(\id/2 - \phi_2)],
\end{equation*}
converge pointwise to 
\begin{equation*}
f(\rho) := 
\begin{cases}
\tr[\Xi_{\rho}(\rho^{\otimes n})(\id/2 - \phi_2)], & \text{if} \hspace{0.2 cm} \rho \hspace{0.2 cm} \text{is $n$-distillable} \\
0 & \text{otherwise},
\end{cases}
\end{equation*}
where $\Xi_{\rho}:{\cal B}({\cal H}^{\otimes n}) \rightarrow 
{\cal B}(\mathbb{C}^2 \otimes \mathbb{C}^2)$ is an optimal LOCC map for $\rho^{\otimes n}$, i.e. a trace preserving LOCC operation that 
minimizes $\tr[\Xi(\rho^{\otimes n})(\id/2 - \phi_2)]$.

To proceed, we first note the upper bound $|f_{m}(\rho)| = | \tr[\Lambda_{m, n}(\rho^{\otimes{(n + m)}})(\id/2 - \phi_2)]| \leq 1$, for every $\rho\in{\cal D}({\cal H})$. As the functions $\{ f_m \}$ are measurable and dominated by the unit constant function, which is integrable on the state space, we can apply Lebesgue dominated convergence theorem \cite{RS72} to get from Eq.\ (\ref{eq52}) 
\begin{eqnarray} \label{limlebegue}
0 & \leq & \lim_{m \rightarrow \infty} \int \mu(d\rho) \tr[\Lambda_{m, n}(\rho^{\otimes{(n + m)}})(\id/2 - \phi_2)]\\ &=&  \lim_{m \rightarrow \infty} \int \mu(d\rho) f_m(\rho)  = \int \mu(d\rho) \lim_{m \rightarrow \infty} f_m(\rho) \nonumber \\ 
&=& \int \mu(d\rho) f(\rho) \nonumber\\
&=&
\int_{{\cal D}({\cal H})\backslash{\cal C}_{n}({\cal H})} 
\mu(d\rho) \tr[\Xi_{\rho}(\rho^{\otimes n})(\id/2 - \phi_2)],\nonumber
\end{eqnarray}
where ${\cal D}({\cal H})\backslash{\cal C}_{n}({\cal H})$ is the set of
the $n$-distillable states.
By definition, we have that for each 
$n$-distillable state $\rho$, $\tr[\Omega(\rho^{\otimes n})(\id/2 - \phi_2)] < 0$\footnote{Actually, what we have by definition is that there exists a SLOCC map $\Lambda$, not necessarily trace preserving, such that $\tr[\Lambda(\rho^{\otimes n})(\id/2 - \phi_2)] < 0$. However, given a SLOCC $\Lambda$ satisfying the previous equation, we can easily construct a deterministic LOCC operations that also satisfies it as follows. We try to implement $\Lambda$ in some input state. This works out with some finite probability. If we fail to implement $\Lambda$, then we throw away the state at hand and prepare the product state $\ket{0, 0}\bra{0, 0}$.}. We hence find from Eq.\ (\ref{limlebegue}) that 
$\mu$ can be non-zero only in a zero measure subset of the set 
of $n$-distillable states. As this is true for an arbitrary $n$, 
we find that $\mu$ must be supported on the set of 
undistillable states.
\end{proof}
\vspace{0.3 cm}
With the proof of Theorem \ref{T=coC} fresh in mind we can now easily prove Corollary \ref{posdist}.

\begin{proof} (Corollary \ref{posdist})
We prove the result by contradiction. Suppose conversely that 
for every $n \in \mathbb{C}{N}$ and every 
SLOCC $\Omega$, $\tr[\Omega(\omega_n)(\lambda\id - \phi_D)] \geq 0$. 
Then we can follow the proof of Theorem \ref{T=coC} to show that $\rho \in {\cal T}$, in contradiction with the assumption that it is not. 

The key point is to notice that Theorem \ref{T=coC} also holds if we replace the single copy undistillability condition $F_2(\omega_n) = 1/2$  by 
$F_D(\omega_n) \leq \lambda$, for any integer $D$ and $\lambda \in [1/D, 1)$. We only have to modify the third step of the SLOCC map we defined as follows: we now discard the state if the estimated state $\sigma_m$ is such that $\tr[\Omega(\sigma_m^{\otimes n})] \leq \lambda$ for every SLOCC map $\Omega$, or apply the optimal SLOCC map $\Omega$ minimizing $\tr[\Omega(\sigma_m^{\otimes n})
(\lambda \id - \phi_D)]$ otherwise. The proof then proceeds in a completely analogous way.
\end{proof}
\vspace{0.3 cm}

It should be pointed out that we cannot say anything about the rate of distillation from copy-correlated states. Although Corollary \ref{posdist} shows that we can distill an arbitrary good approximation of a maximally entanglement state of arbitrary dimension, it does not say anything about the distillable rate. Indeed, from the proof of Theorem \ref{T=coC} and Corollary \ref{posdist} we can find out an explicit distillation protocol for this task: The parties trace out several copies, perform state tomography in some of the remaining, and apply the optimal distillation map of the estimated state in the others. It turns out that this protocol has a zero rate, as in the first step we need to trace out a much larger number of copies than the ones we keep in, in order to apply the finite quantum de Finetti Theorem. It seems possible that this difficulty can be overcome by employing the exponential quantum de Finetti theorem, discussed in section \ref{expdf}. This is however left as an open problem for future research.

\section{Proof of Theorem \ref{actcor}} \label{proof2}

We now turn to the proof of Theorem \ref{actcor}. We start by proving two auxiliary Lemmata, 
which give a characterization for the elements of the dual cones of the sets ${\cal S}_k({\cal H}^{\otimes k})$ 
and ${\cal T}_k({\cal H}^{\otimes k})$, which will again sometimes
be abbreviated as ${\cal S}_k$ and ${\cal T}_k$. 

\begin{lemma}\label{S(Q)=0}                  
If $Q \in ({\cal S}_k)^*$, then $\hat{S}_{k}(Q) \geq 0$.                                    
\end{lemma} 
\begin{proof}                        
As $Q \in ({\cal S}_{k})^*$, we have that for every positive semi-definite operator $X \geq 0$ acting on ${\cal H}^{\otimes k}$, $\tr[ X \hat{S}_{k}(Q)] = \tr[\hat{S}_{k}(X)Q] \geq 0$. This can only be true if $\hat{S}_{k}(Q) \geq 0$.
\end{proof}

\begin{lemma} \label{dual}     
For each $k \in \nn$ and for every element                       
$X$ of $\hspace{0.1 cm}{\cal T}_{k}^{*}$, there exist an 
SLOCC map $\Lambda$ and an operator $Q \in ({\cal S}_{k})^*$ such that 
\begin{eqnarray} \label{NF1}                                     
	&& X\otimes \id^{\otimes{( k - 1)}} = \Lambda(\id/2 - \phi_2) + Q.                                       
\end{eqnarray}
\end{lemma}

\begin{proof} 
For any two closed convex cones $A$ and $B$ defined on a finite dimensional Hilbert space we have $(A \cap B)^{*} = A^{*} + B^{*}$ \cite{HL97}. It is easily seen that $\text{cone}({\cal S}_k \cap {\cal C}_1) = \text{cone}({\cal S}_k) \cap \text{cone}({\cal C}_1)$,
where
the {\it conic hull} is defined for a set $C$ as
\begin{equation*}
	\text{cone}(C) := 
	\biggl\{ \sum_j \lambda_j W_j: 
	\hspace{0.2 cm} \lambda_j \geq 0, \hspace{0.1 cm} 
	W_j \in C \biggr\}.
\end{equation*}
Therefore,                                    
\begin{eqnarray*}                        
	({\cal S}_{k} \cap {\cal C}_{1})^{*} &=& 
	[\text{cone}({\cal S}_{k} \cap {\cal C}_{1})]^{*} 
	\nonumber\\
	&=& 
	[\text{cone}({\cal S}_{k})\cap \text{cone}({\cal C}_1)]^*
	 = {\cal S}_{k}^* + {\cal C}_{1}^*.
\end{eqnarray*}                           
This in turn implies that every element $Y$ of $({\cal S}_{k} \cap {\cal C}_{1})^{*}$ can be written as the right hand side of Eq.\ 
(\ref{NF1}).
We find that if $X \in {\cal T}_{k}^{*}$, then $X \otimes {\id}^{\otimes ( k - 1) }$ is an element of $({\cal S}_{k} \cap {\cal C}_{1})^{*}$. Indeed,
\begin{eqnarray*}                                
	\tr[X \rho] \geq 0 \hspace{0.15 cm}                             
	\forall \rho \in {\cal T}_{k} &\Rightarrow &
	\tr[X \tr_{\backslash 1}(\pi)] \geq 0 \hspace{0.15 cm} \forall \pi \in 
	{\cal S}_{k} \cap {\cal C}_{1}  \nonumber\\
	&\Rightarrow& \tr[(X \otimes {\id}^{\otimes{( k - 1)}}) \pi] \geq 0 
	\hspace{0.15 cm} \forall \pi \in {\cal S}_{k} \cap {\cal C}_{1}.
	\nonumber\\
\end{eqnarray*}                                                
Hence, any element of the dual cone of ${\cal T}_k$ can be written as a sum of an
element of the dual cone of ${\cal S}_k$ and an element
of the dual cone of ${\cal C}_1$, which can always be expressed as $\Lambda(\id/2 - \phi_2)$ for some SLOCC map $\Lambda$. 
\end{proof}          
\vspace{0.3 cm}

The next Lemma is the key element to prove Theorem \ref{actcor}. It makes a connection between separability and the structure of the dual sets $({\cal T}_k)^*$. Before we turn to its formulation and proof, let us introduce some notation which will render the discussion more transparent. In this Lemma, we will
set ${\cal H} := \mathbb{C}^{2d} \otimes \mathbb{C}^{2d}$. If we have a tensor product between a $d\times d$-system and a $2\times 2$ system, the latter is thought to be embedded in a $d\times d$-dimensional system. We reserve $\id$ for the identity operator acting on ${\cal H}$. 
The identity operator acting on $\mathbb{C}^m \otimes \mathbb{C}^m$, for every other $m$ different from $2d$, will be denoted by $\id_{m^2}$.

\begin{lemma}
\label{crucial}
Let $\sigma \in {\cal D}(\mathbb{C}^d \otimes \mathbb{C}^d)$ and 
$k \in \mathbb{C}{N}$. If 
\begin{equation}
	\sigma \otimes (\id_{4}/2 - \phi_2) \in ({\cal T}_k)^{*},
\end{equation}	
then $\sigma$ is separable. 
\end{lemma}
\begin{proof}
By Lemma \ref{dual} we can write
\begin{equation}
\sigma \otimes (\id_{4}/2 - \phi_2) \otimes \id^{\otimes{( k - 1)}} = \Lambda(\id_{4}/2 - \phi_2) + Q,
\end{equation}
for some SLOCC map $\Lambda$ and an operator $Q \in {\cal S}_k^*$. Applying the symmetrization 
operator $\hat{S}_k$ to both sides of the previous equation, multiplying both sides from the left 
with $\id \otimes (\id_{d^2} \otimes \ket{0, 0}\bra{0, 0})^{\otimes{( k - 1)}}$ and taking the 
partial trace with respect to all systems except the first $\mathbb{C}^{2d}\otimes \mathbb{C}^{2d}$-dimensional subsystem, we find
\begin{equation} \label{mas}
	\sigma \otimes (\id_{4}/2 - \phi_2) = 
	\Upsilon(\id_{4}/2 - \phi_2) + P,
\end{equation}
where
\begin{eqnarray*}
	P &:= &\tr_{\backslash 1}[\id \otimes 
	(\id_{d^2} \otimes \ket{0, 0}
	\bra{0, 0})^{\otimes{( k - 1)}} \hat{S}(Q)], \\
	\Upsilon(.) &:=& 
	\tr_{\backslash 1}
	[\id \otimes (\id_{d^2} \otimes 
	\ket{0, 0}\bra{0, 0})^{\otimes{( k - 1)}} 
	(\hat{S}\circ \Lambda)(.)], 
\end{eqnarray*}	
By Lemma \ref{S(Q)=0}, it follows that $P \geq 0$. As the symmetrization and the projection onto a product state can be done by SLOCC, we find that $\Gamma$ is a SLOCC map. The statement of the Lemma then follows from Ref. \cite{Mas06}, where it was shown that Eq.\ (\ref{mas}) implies the separability of $\sigma$. 
\end{proof}

Theorem \ref{actcor} can now be established from Lemma \ref{crucial}, together with the activation protocol used in 
Refs.\ \cite{Mas06, VW02}.

\begin{proof} (Theorem \ref{actcor})

Following Refs.\cite{Mas06, VW02}, let us consider states $\sigma\in {\cal D}(({\cal H}_{A_{2}}\otimes {\cal H}_{A_{3}}) \otimes ({\cal H}_{B_{2}}\otimes {\cal H}_{B_{3}}))$, where ${\cal H}_{A_{2}} = {\cal H}_{B_{2}} =  \mathbb{C}^d,\,\,\, {\cal H}_{A_{3}} = {\cal H}_{B_{3}} = \mathbb{C}^2$.

To prove the Theorem we have to show that for all $k \in \mathbb{N}$ there exists a $\rho \in {\cal T}_{k}\subset {\cal D}(\mathbb{C}^d\otimes \mathbb{C}^d)$ 
and a SLOCC operation $\Lambda$ such that 
\begin{equation*}\label{Hold}
	\tr[\Lambda(\sigma \otimes \rho)\phi_2] > 	
	\tr[\Lambda(\sigma \otimes \rho)]/2. 
\end{equation*}
This condition can also be expressed as
\begin{equation*} \label{Lambda1}
\tr[\Lambda(\rho \otimes \sigma)(\id / 2 - \phi_2)] < 0.
\end{equation*}

We choose the SLOCC map $\Lambda$ as follows: The parties perform a local measurement -- on subsystems $A_1A_2$, $B_1B_2$ -- in a basis of maximally entangled states, post-selecting when both systems are projected onto the projectors associated with $\Phi_{A_{1}A_{2}}$ and $\Phi_{B_{1}B_{2}}$, respectively. Using the Jamilkowski isomorphism \cite{Jam72} and performing a little algebra we find that, for this particular choice of $\Lambda$, Eq. (\ref{Lambda1}) reads \cite{Mas06}
\begin{equation} \label{sat}
	\tr[ \rho \sigma^T \otimes (\id/2 - \phi_2)] < 0.
\end{equation}
To complete the proof it suffices to note that by Lemma \ref{crucial}, if $\sigma$ is entangled then there must exists a state $\rho \in {\cal T}_k$ satisfying Eq.\ (\ref{sat}). Indeed, if this were not true, then $\sigma^{T}\otimes (\id/2 - \phi_2)$ would have to belong to the dual cone of ${\cal T}_k$, which was shown in Lemma \ref{crucial} to imply the separability of $\sigma$. 
\end{proof}

\section{On the Existence of NPPT Bound Entanglement} \label{NPPT2}

Before we conclude this chapter, we would like to comment on the applicability of this 
approach to the conjecture on the existence of bound entangled states with a non-positive partial transpose, and in 
particular to Conjecture \ref{conj1}. It is clear that if we could prove the validity of 
Lemma \ref{crucial} for the full set ${\cal T}$, then Conjecture \ref{conj1} 
would be true. Indeed, if the activation procedure outlined in the proof of Theorem \ref{actcor} works for a convex combination of undistillable states, then it has to work at least for one of the states appearing in the convex combination, as it is made explicit by the linearity of Eq.\ (\ref{sat}).

However, although the presented methods seem applicable to this question, a significant further step will be necessary,
and a direct extension of  Lemma \ref{crucial} to ${\cal T}$ does not seem to work. From Theorem \ref{T=coC} it is a simple exercise in convex analysis to show that 
\begin{equation*}
{\cal C}^* = {\cal T}^{*} = \overline{\bigcup_{k \in \mathbb{C}{N}} {\cal T}_k^*}.
\end{equation*}
If we now assume that $\sigma \otimes (\id_4/2 - \phi_2) \in {\cal T}^*$, then for every $\varepsilon > 0$, 
there exists an integer $n_{\varepsilon}$ such that
\begin{equation*}
\sigma \otimes (\id_4/2 - \phi_2) + \varepsilon \id \in ({\cal T}_{k_{\varepsilon}})^*.
\end{equation*}
If we followed the steps taken in the proof of Lemma \ref{crucial}, we would find, instead of Eq.\ (\ref{mas}), the following:
\begin{equation*}
\sigma \otimes (\id_4/2 - \phi_2) + (n_{\varepsilon} - 1) \varepsilon \id = \Omega_{\varepsilon}(\id/2 - \phi_2) + P_{\varepsilon},
\end{equation*}
where $P_{\varepsilon} \geq 0$ and $\Omega_{\varepsilon} $ is a SLOCC operation for every $\varepsilon > 0$. Hence, in order to be able to carry over with the approach similar to one outlined in Ref.\ \cite{Mas06}, we would have to be able to show that we can choose the sequence 
$\{ n_{\varepsilon} \}$ to be such that
\begin{equation*} 
\lim_{\varepsilon \rightarrow 0} (n_{\varepsilon} - 1) \varepsilon = 0.
\end{equation*}
Although it could well be the case that such relation hold, we could not find a way either to prove it nor to disprove it, despite considerable
effort.

From a different perspective, it seems that the rate of convergence of an arbitrary element of ${\cal T}^*$ by elements of the inner approximations given by ${\cal T}_k^*$ matters in our problem. Note that it is exactly the closure in
\begin{equation}\label{CH}
{\cal T}^* = \overline{\bigcup_{k \in \mathbb{N}} {\cal T}_k^*}
\end{equation}
the responsible for this behavior. Indeed, Lemma \ref{crucial} can straightforwardly be applied if we require only that
\begin{equation*}
\sigma \otimes (\id_4/2 - \phi_2) \in \bigcup_{k \in \mathbb{N}} {\cal T}_k^*.
\end{equation*}
So the question of the existence of NPPT bound entanglement can be related to the question of the necessity of the closure in Eq.\ (\ref{CH}).

\chapter{A Generalization of Quantum Stein's Lemma} \label{QHT}

\section{Introduction}

Hypothesis testing refers to a general set of tools in statistics and probability theory for making decisions based on experimental data from random variables. In a typical scenario, an experimentalist is faced with two possible hypotheses and must decide based on experimental observation which of them was actually realized. There are two types of errors in this process, corresponding to mistakenly identifying one of the two options when the other should have been chosen. A central task in hypothesis testing is the development of optimal strategies for minimizing such errors and the determination of compact formulae for the minimum error probabilities. 

Substantial progress has been achieved both in the classical and quantum settings for i.i.d processes (see e.g. \cite{Che52, CL71, Bla74, HP91, ON00, Hay02, OH04, Nag06, ANSV07, Hay07} and references therein). We review some of the most important developments in the sequel. The non-i.i.d. case, however, has proven harder and much less is known. The main result of this chapter is a particular instance of quantum hypothesis testing of non-i.i.d. sources for which the optimal decaying rates of the error probabilities can be fully determined. To the best of the author's knowledge, the solution of such a problem was not known even in the classical setting. 

This result, in spite of not referring directly to entanglement, has important consequences to entanglement theory. In this chapter we present two applications: the derivation of an operational interpretation of the regularized relative entropy of entanglement and a proof that this measure is faithful (see section \ref{entmeasures}). This latter consequence provides in particular a new proof of the fact that every entangled state has a strictly positive entanglement cost \cite{YHHS05}. The main application of the findings of this chapter will be presented in chapter \ref{reversible}, where they will be the key technical element for proving reversibility of entanglement manipulation under asymptotically non-entangling operations. 

The structure of this chapter is as follows. In subsection \ref{iid} we review a few results of hypothesis testing of i.i.d. sources. In section \ref{mainr} we present some definitions and state the main results of this chapter. Sections \ref{proofmain} and \ref{coroll} are devoted to prove Theorem \ref{maintheorem} and Corollary \ref{faithful}, respectively. Finally, in section \ref{constraint} we discuss some partial progress in generalizing the findings of this chapter to situations where one encounters constraints on the available projective operator value measurements (POVMs), such as hypothesis testing using separable or PPT POVMs. 

\subsection{Hypothesis Testing for i.i.d. Sources} \label{iid}

Suppose we have access to a source of i.i.d. random variables chosen from one of two possible probability distributions. Our aim is to decide which probability distribution actually corresponds to the random variable. In the quantum generalization of such a problem, we are faced with a source that emits several i.i.d. copies of one of two quantum states $\rho$ and $\sigma$, and we should decide which of them is being produced. Since the quantum setting also encompasses the classical, we will focus on the former, noting the differences between the two when necessary. 

In order to learn the identity of the state the observer measures a two outcome POVM $\{ M_n, \id - M_n \}$ given $n$ realizations of the unknown state. If he obtains the outcome associated to $M_n$ ($\id - M_n$) then he concludes that the state was $\rho$ ($\sigma$). The state $\rho$ is seen as the null hypothesis, while $\sigma$ is the alternative hypothesis. There are two types of errors here: 
\begin{itemize}
	\item Type I: The observer finds that the state was $\sigma$, when in reality it was $\rho$. This happens with probability $\alpha_n(M_n) := \tr(\rho^{\otimes n}(\id - M_n))$.
	\item Type II: The observer finds that the state was $\rho$, when it actually was $\sigma$. This happens with probability $\beta_n(M_n) := \tr(\sigma^{\otimes n}M_n)$.
\end{itemize}
There are a few distinct settings that might be considered, depending on the importance we attribute to the two types of errors. 

In symmetric hypothesis testing, we wish to minimize the average probability of error. Suppose we have $\rho^{\otimes n}$ with probability $p$ and $\sigma^{\otimes n}$ with probability $1 - p$. The optimal average probability of error is thus
\begin{equation*}
P_e(n) = \min_{0 \leq M_n \leq \id} p\tr(\sigma^{\otimes n}M_n) + (1 - p)\tr(\rho^{\otimes n}(\id - M_n)).
\end{equation*}
As $n$ grows, $P_e(n)$ goes to zero as the states become increasingly more distinguishable. Indeed, for sufficiently large $n$ the decay of the average error is exponential in $n$, and the relevant information concerning the distinguishability of the two states is in its exponent. Such an exponent was calculated for probability distributions in the fifties by Chernoff \cite{Che52} and is given by an expression that now carries his name. The generalization of Chernoff's distance to the quantum case remained open for a long time, until the works of Nussbaum and Szkola \cite{NS06} and Audenaert et al \cite{ACM+07}. Together they proved \cite{ANSV07}
\begin{equation}
\lim_{n \rightarrow \infty} - \frac{\log(P_e(n))}{n} = - \log \left( \inf_{s \in [0, 1]}\tr(\rho^{1-s}\sigma^s) \right), 
\end{equation}
which is exactly the formula originally obtained by Chernoff, after changing density operators by probability distributions. Note that the formula is independent of the a priori probability $p$. 

In some situations the costs associated to the two types of errors can be very different. In one such setting, the probability of type II error should be minimized to the extreme, while only requiring that the probability of type I error is bounded by a small parameter $\epsilon$. The relevant error quantity in this case can be written as
\begin{equation*}  
\beta_n(\epsilon) := \min_{0 \leq M_n \leq \id} \{ \beta_n(M_n) : \alpha_n(M_n) \leq \epsilon \}.
\end{equation*}
Quantum Stein's Lemma \cite{HP91, ON00} tell us that for every $0 \leq \epsilon \leq 1$,
\begin{equation} \label{Stein}
\lim_{n \rightarrow \infty} - \frac{\log(\beta_n(\epsilon))}{n} = S(\rho || \sigma). 
\end{equation}
This fundamental result gives a rigorous operational interpretation for the relative entropy and was proven in the quantum case by Hiai and Petz \cite{HP91} and Ogawa and Nagaoka \cite{ON00}. Different proofs have since be given in Refs. \cite{Hay02, OH04, ANSV07}. The relative entropy is also the asymptotic optimal exponent for the decay of $\beta_n$ when we require that $\alpha_n \stackrel{n \rightarrow \infty}{\longrightarrow} 0$ \cite{OH04}. 

Stein's Lemma is a particular case of the more powerful Hoeffding-Blahut-Csisz\'ar-Longo bound \cite{CL71, Bla74}. Its quantum version has been established recently by Hayashi \cite{Hay07}, Nagaoka \cite{Nag06}, and Audenaert, Nussbaum, Szkola, and Verstraete \cite{ANSV07}. Let us define the type I and type II rate limits as
\begin{equation*}
\alpha_R(\{M_n\}_{n \in \mathbb{N}}) := \lim_{n \rightarrow \infty} \left( - \frac{\log \alpha_n(M_n)}{n} \right), \hspace{0.3 cm} \beta_R(\{M_n\}_{n \in \mathbb{N}}) := \lim_{n \rightarrow \infty} \left( - \frac{\log \beta_n(M_n)}{n} \right).
\end{equation*}
Given a constraint of the form $\alpha_R \geq r$, the quantum HBSL bound gives the maximum achievable value of $\beta_R$: 
\begin{equation}
e_Q(r) := \sup_{s \in [0, 1)} \frac{-rs - \log \tr \sigma^s \rho^{1- s}}{1 - s}.
\end{equation}

Before we finish our short discussion on previous results, we comment on three extensions of quantum Stein's Lemma. The first, which was actually already proved in the seminal paper by Hiai and Petz on the subject \cite{HP91}, states that quantum Stein's Lemma is also true when instead of $\rho^{\otimes n}$, we have reduced states $\rho_n$ of an ergodic state\footnote{See Ref. \cite{HP91} for the definition of a ergodic state.}. The second, which is sometimes referred to as quantum Sanov's Theorem \cite{Hay02, BDK+05}, concerns the situation in which instead of a single state as the null hypothesis, we have a family of states ${\cal K} \subseteq {\cal D}({\cal H})$. In this case it has be shown that that rate limit of type II error is given by $\inf_{\rho \in {\cal K}}S(\rho || \sigma)$ \cite{Hay02, BDK+05}. The third and most general is the so called information spectrum approach to hypothesis testing. As discussed in Ref. \cite{NH07}, this method delivers the achievability and strong converse\footnote{The strong converse rate is defined as the minimum rate of the exponential decay of type II error for which the probability of type I error goes to one asymptotically.} optimal rate limits in terms of divergence spectrum rates for arbitrary sequence of states. Despite its generality, this method has the drawback that no direct connection to the relative entropy is established and that in general the achievability and strong converse rates are different. 

The main result of this chapter has a similar flavor to the above-mentioned generalizations. We will be interested in the case where the \textit{alternative hypothesis} is not only composed of a single i.i.d. state, but is actually formed by a family of non-i.i.d. states satisfying certain conditions to be specified in the next section. We will then show that a quantity similar to the regularized relative entropy of entanglement, in which the minimization is taken over the elements of the alternative hypothesis set, is the optimal rate limit for type II error. In the next section we provide the necessary definitions and state the result in full. 

\section{Definitions and Main Results} \label{mainr}

Given a closed set of states ${\cal M} \subseteq {\cal D}({\cal H})$, we define 
\begin{equation} \label{relent1}
E_{{\cal M}}(\rho) := \min_{\sigma \in {\cal M}} S(\rho || \sigma),
\end{equation}
and
\begin{equation} \label{LRM}
LR_{\cal M}(\rho) := \min_{\sigma \in {\cal M}} S_{\max}(\rho || \sigma),
\end{equation}
where
\begin{equation}
S_{\max}(\rho || \sigma) := \inf \{ s : \rho \leq 2^s \sigma   \} 
\end{equation}
is the maximum relative entropy introduced by Datta \cite{Dat08a, Dat08b}. Note that if we take ${\cal M}$ to be the set of
separable states, then $E_{\cal M}$ and $LR_{\cal M}$ reduce to the relative entropy of entanglement and the logarithm 
global robustness of entanglement (see sections \ref{relentint} and \ref{robustnesses}). This connection is the reason for the nomenclature in Eqs. (\ref{relent1}) and (\ref{LRM}). We will also need a \textit{smooth version} of $LR_{\cal M}$, defined as
\begin{equation} \label{smooth}
LR_{{\cal M}}^{\epsilon}(\rho) := \min_{\tilde{\rho} \in B_{\epsilon}(\rho)} LR_{\cal M}(\tilde{\rho}), 
\end{equation}
where $B_{\epsilon}(\rho) := \{ \tilde{\rho} \in {\cal D}({\cal H}) : || \rho- \tilde{\rho} ||_1 \leq \epsilon \}$\footnote{We note that smooth versions of other non asymptotically continuous measures, such as the min- and max-entropies \cite{RW04, Ren05}, have been proposed and shown to be useful in non-asymptotic and non-i.i.d. information theory. See e.g. Refs. \cite{RW04, Ren05}.}. 

Let us specify the set of states over which the alternative hypothesis vary. We will consider any family of sets $\{ {\cal M}_n \}_{n \in \mathbb{N}}$, with ${\cal M}_n \in {\cal D}({\cal H}^{\otimes n})$, satisfying the following properties
\begin{enumerate}
	\item \label{cond1} Each ${\cal M}_n$ is convex and closed. 
	\item \label{cond2} Each ${\cal M}_n$ contains the maximally mixed state $\id^{\otimes n}/\dim({\cal H})^n$.
	\item \label{cond3} If $\rho \in {\cal M}_{n + 1}$, then $\tr_{n + 1}(\rho) \in {\cal M}_{n}$.
	\item \label{cond4} If $\rho \in {\cal M}_{n}$ and $\sigma \in {\cal M}_m$, then $\rho \otimes \sigma \in {\cal M}_{n + m}$.
	\item \label{cond5} If $\rho \in {\cal M}_n$, then $P_{\pi}\rho P_{\pi} \in {\cal M}_n$, for every $\pi \in S_n$\footnote{As in section 2.6, $P_{\pi}$ is the standard representation in ${\cal H}^{\otimes n}$ of an element $\pi$ of the symmetric group $S_n$.}. 
\end{enumerate}
We define the regularized version of the quantity given by Eq. (\ref{relent1}) as\footnote{To show that the limit exists in Eq. (\ref{regu1}) we use the fact that if a sequence $(a_n)$ satisfies $a_n \leq cn$ for some constant $c$ and $a_{n + m} \leq a_n + a_m$, then $a_n/n$ is convergent \cite{HHT01, BNS98}. Using properties \ref{cond2} and \ref{cond4} it is easy to see that our sequence satisfy the two conditions.}
\begin{equation}  \label{regu1}
E_{\cal M}^{\infty}(\rho) := \lim_{n \rightarrow \infty} \frac{1}{n} E_{{\cal M}_n}(\rho^{\otimes n}).
\end{equation}
We now turn to the main result of this chapter. Suppose we have one of the following two hypothesis: 
\begin{enumerate}
	\item For every $n \in \mathbb{N}$ we have $\rho^{\otimes n}$.
	\item For every $n \in \mathbb{N}$ we have an unknown state $\omega_n \in {\cal M}_n$, where $\{ {\cal M}_n \}_{n \in \mathbb{N}}$ is a family of sets satisfying properties \ref{cond1}-\ref{cond5}.
\end{enumerate}
The next theorem gives the optimal rate limit for the type II error when one requires that type I error vanishes asymptotically.

\begin{theorem} \label{maintheorem}
Given a family of sets  $\{ {\cal M}_n \}_{n \in \mathbb{N}}$ satisfying properties \ref{cond1}-\ref{cond5} and a state $\rho \in {\cal D}({\cal H})$, for every $\epsilon > 0$ there exists a sequence of POVMs $\{ A_n, \id - A_n \}$ such that
\begin{equation*}
\lim_{n \rightarrow \infty} \tr((\id - A_n) \rho^{\otimes n}) = 0 
\end{equation*}
and for all sequences of states $\{ \omega_n \in {\cal M}_n \}_{n \in \mathbb{N}}$,
\begin{equation*}
- \frac{\log \tr(A_n \omega_n)}{n} + \epsilon \geq E_{\cal M}^{\infty}(\rho) = \lim_{\epsilon \rightarrow 0} \limsup_{n \rightarrow \infty} \frac{1}{n} LR_{{\cal M}_n}^{\epsilon}(\rho^{\otimes n}).
\end{equation*}
Conversely, if there is a $\epsilon > 0$ and sequence of POVMs $\{ A_n, \id - A_n \}$ satisfying 
\begin{equation*}
 - \frac{\log( \tr(A_n \omega_n))}{n} - \epsilon \geq E_{\cal M}^{\infty}(\rho)
\end{equation*}
for all sequences $\{ \omega_n \in {\cal M}_n \}_{n \in \mathbb{N}}$, then
\begin{equation*}
\lim_{n \rightarrow \infty} \tr((\id - A_n) \rho^{\otimes n}) = 1. 
\end{equation*}
\end{theorem}

Theorem \ref{maintheorem} gives the promised operational interpretation for the regularized relative entropy of entanglement. Taking $\{ {\cal M}_n \}_{n \in \mathbb{N}}$ to be the sets of separable states over ${\cal H}^{\otimes n}$, it is a simple exercise to check that they satisfy conditions \ref{cond1}-\ref{cond5}. Therefore, we conclude that $E_R^{\infty}(\rho)$ gives the rate limit of the type II error when we try to decide if we have several realizations of $\rho$ or a sequence of \textit{arbitrary} separable states. This rigorously justify the use of the regularized relative entropy of entanglement as a measure of distinguishability of quantum correlations from classical correlations. 

Another noteworthy point of Theorem \ref{maintheorem} is the formula
\begin{equation} \label{globrobrelent}
\lim_{n \rightarrow \infty} \frac{1}{n} E_{{\cal M}_n}(\rho^{\otimes n}) = \lim_{\epsilon \rightarrow 0} \limsup_{n \rightarrow \infty} \frac{1}{n} LR_{{\cal M}_n}^{\epsilon}(\rho^{\otimes n}).
\end{equation}
Taking once more $\{ {\cal M}_n \}$ as the sets of separable states over ${\cal H}^{\otimes n}$, this equation shows that the regularized relative entropy of entanglement is a smooth asymptotic version of the log global robustness of entanglement. We hence have a connection between the robustness of quantum correlations under mixing and their distinguishability to classical correlations.

An interesting Corollary of Theorem \ref{maintheorem} is the following.
\begin{corollary} \label{faithful}
The regularized relative entropy of entanglement is faithful. For every entangled state $\rho \in {\cal D}({\cal H}_1 \otimes ... \otimes {\cal H}_n)$,
\begin{equation}
E_R^{\infty}(\rho) > 0.
\end{equation}
\end{corollary}

As mentioned in section \ref{relentint}, the entanglement cost is lower bounded by the regularized relative entropy of entanglement. Corollary \ref{faithful} then gives a new proof of the fact that the entanglement cost is strictly larger than zero for every entangled state, a result first proved by Yang, Horodecki, Horodecki, and Synak-Radtke in Ref. \cite{YHHS05}. 

In the next two sections we provide the proofs of Theorem \ref{maintheorem} and Corollary \ref{faithful}. In section \ref{constraint} we then discuss on possible generalizations of Theorem \ref{maintheorem} to situations where we have constraints on the POVMs available.
\section{Proof of Theorem \ref{maintheorem}} \label{proofmain}

We now prove Theorem \ref{maintheorem}. In order to make the presentation clearer, we break it into several shorter pieces. We start by establishing the following Proposition.

\begin{proposition} \label{relenteqrob}
For every family of sets  $\{ {\cal M}_n \}_{n \in \mathbb{N}}$ satisfying properties \ref{cond1}-\ref{cond5} and every state $\rho \in {\cal D}({\cal H})$,
\begin{equation} \label{relenteqrobeq}
E_{\cal M}^{\infty}(\rho) = \lim_{\epsilon \rightarrow 0} \limsup_{n \rightarrow \infty} \frac{1}{n} LR_{{\cal M}_n}^{\epsilon}(\rho^{\otimes n}).
\end{equation}
\end{proposition}
We then employ it to prove
\begin{proposition} \label{maincompact} 
For every family of sets  $\{ {\cal M}_n \}_{n \in \mathbb{N}}$ satisfying properties \ref{cond1}-\ref{cond5} and every state $\rho \in {\cal D}({\cal H})$,
\begin{equation}
\lim_{n \rightarrow \infty} \min_{\omega_n \in {\cal M}_n} \tr( \rho^{\otimes n} - 2^{yn}\omega_n)_+ = 
\begin{cases}
0, & y > E_{\cal M}^{\infty}(\rho),\\
1, & y < E_{\cal M}^{\infty}(\rho).
\end{cases}
\end{equation}
\end{proposition}

Before proving these two results, let us show how Proposition \ref{maincompact} implies Theorem \ref{maintheorem}.

\vspace{1 cm}

\begin{proof} (Theorem \ref{maintheorem})
Consider the following family of convex optimization problems
\begin{equation*}
\lambda_n(\pi, K) := \max_{A, \sigma} \tr(A \pi) : 0 \leq A \leq \id, \hspace{0.2 cm} \tr(A \sigma) \leq \frac{1}{K} \hspace{0.2 cm} \forall \hspace{0.1 cm} \sigma \in {\cal M}_n.
\end{equation*}
The statement of the theorem can be rewritten as
\begin{equation} \label{refor}
\lim_{n \rightarrow \infty} \lambda_n(\rho^{\otimes n}, 2^{n y}) = 
\begin{cases}
0, & y > E_{\cal M}^{\infty}(\rho), \\
1, & y < E_{\cal M}^{\infty}(\rho).
\end{cases}
\end{equation}
In order to see that Eq. (\ref{refor}) holds true, we go to the dual formulation of $\lambda_n(\pi, K)$. We first rewrite it as 
\begin{equation*}
\lambda_n(\pi, K) := \max_{A, \sigma} \tr(A \pi) : 0 \leq A \leq \id, \hspace{0.2 cm} \tr\left((\id/K - A) \sigma\right) \geq 0 \hspace{0.2 cm} \forall \hspace{0.1 cm} \sigma \in \text{cone}({\cal M}_n).
\end{equation*}
Then, we note that the second constraint is a generalized inequality (since the set $\text{cone}({\cal M}_n)$ is a convex proper cone\footnote{See section \ref{entwit} for more details.}) and write the problem as
\begin{equation} \label{primal}
\lambda_n(\pi, K) := \max_{A, \sigma} \tr(A \pi) : 0 \leq A \leq \id, \hspace{0.2 cm} (\id/K - A)  \in ({\cal M}_n)^{*}.
\end{equation}
The Lagrangian of $\lambda_n(\pi, K)$ is given by
\begin{equation*}
L(\pi, K, A, X, Y, \mu) = \tr(A \pi) + \tr(XA) + tr(Y(\id - A)) + \tr((\id/K - A)\mu), 
\end{equation*}
where $X, Y \geq 0$ and $\mu \in \text{cone}({\cal M}_n)$ are Lagrange multipliers. It is easy to find a strictly feasible solution for the primal optimization problem given by Eq. (\ref{primal}). Therefore, by Slater's condition \cite{BV01} $\lambda_n(\pi, K)$ is equal to its dual formulation, which reads
\begin{equation*}
\lambda_n(\pi, K) = \min_{Y, \mu} \tr(Y) + \tr(\mu)/K : \pi \leq Y + \mu, \hspace{0.2 cm} Y \geq 0, \hspace{0.2 cm} \mu \in \text{cone}({\cal M}_n). 
\end{equation*}
Using that $\tr(A)_+ = \min_{Y} \tr(Y) : Y \geq 0, Y \geq A$, we find
\begin{equation*}
\lambda_n(\pi, K) = \min_{\mu} \tr(\pi - \mu)_+ + \tr(\mu)/K : \mu \in \text{cone}({\cal M}_n), 
\end{equation*}
which can finally be rewritten as
\begin{equation*}
\lambda_n(\pi, K) = \min_{\mu, b} \tr(\pi - b\mu)_+ + b/K : \mu \in {\cal M}_n, \hspace{0.2 cm} b \geq 0. 
\end{equation*}
Let us consider the asymptotic behavior of $\lambda_n(\rho^{\otimes n}, 2^{n y})$. Take $y = E_{\cal M}^{\infty}(\rho) + \epsilon$, for any $\epsilon > 0$. Then we can choose, for each $n$, $b = 2^{n(E_{\cal M}^{\infty}(\rho) + \frac{\epsilon}{2})}$, giving
\begin{equation*} 
\lambda_n(\rho^{\otimes n}, 2^{ny}) \leq \min_{\mu \in {\cal M}_n} \tr(\rho^{\otimes n} - 2^{n(E_{\cal M}^{\infty}(\rho) + \frac{\epsilon}{2})}\mu)_+ + 2^{-n\frac{\epsilon}{2}}.
\end{equation*}
From Proposition \ref{maincompact} we then find that $\lambda_n(\rho^{\otimes n}, 2^{ny}) \rightarrow 0$. 

We now take $y = E_{\cal M}^{\infty}(\rho) - \epsilon$, for any $\epsilon > 0$. The optimal $b$ for each $n$ has to satisfy $b_n \leq 2^{y n}$, otherwise $\lambda_n(\rho^{\otimes n}, 2^{ny})$ would be larger than one, which is not true. Therefore,
\begin{equation*} 
\lambda_n(\rho^{\otimes n}, 2^{ny}) \geq \min_{\mu \in {\cal M}_n} \tr(\rho^{\otimes n} - 2^{n(E_{\cal M}^{\infty}(\rho) - \epsilon)}\mu)_+,
\end{equation*}
which goes to one again by Proposition \ref{maincompact}.
\end{proof}

\subsection{Proof of Proposition \ref{relenteqrob}}

The next two Lemmata will play an important role in the proof of Proposition \ref{relenteqrob}. The first, due to Ogawa and Nagaoka, appeared in Ref. \cite{ON00} as Theorem 1 and was the key element for establishing the strong converse of quantum Stein's Lemma.

\begin{lemma} \label{ON}
\cite{ON00} Given two quantum states $\rho, \sigma \in {\cal D}({\cal H})$ and a real number $\lambda$, 
\begin{equation}
\tr(\rho^{\otimes n} - 2^{\lambda n}\sigma^{\otimes n})_+ \leq 2^{- n (\lambda s - \psi(s))},
\end{equation}
for every $s \in [0, 1]$. The function $\psi(s)$ is defined as 
\begin{equation}
\psi(s) := \tr( \log( \rho^{1 + s}\sigma^{-s})).
\end{equation}
\end{lemma}
Note that $\psi(0) = 0$ and $\psi'(0) = S(\rho || \sigma)$. Hence, if $\lambda > S(\rho || \sigma)$, $\tr(\rho^{\otimes n} - 2^{\lambda n}\sigma^{\otimes n})_+$ goes to zero exponentially fast in $n$.

The next Lemma, due to Datta and Renner \cite{DR08}, appeared in Ref. \cite{DR08} as Lemma 5 and will be used in the proof of Proposition \ref{relenteqrob} and several times in the proof of Proposition \ref{maincompact}.

\begin{lemma} \label{DR}
\cite{DR08} Let $\rho, Y, \Delta$ be positive semidefinite operators such that $\rho \leq Y + \Delta$. Then there exists a positive semidefinite operator $\tilde{\rho}$, with $\tr(\tilde{\rho}) \leq \tr(\rho)$, such that 
\begin{equation}
|| \tilde{\rho} - \rho ||_1 \leq 4\sqrt{\tr(\Delta)},
\end{equation}
\begin{equation}
F(\tilde{\rho}, \rho) \geq 1 - \tr(\Delta), \footnote{\normalfont $F(A, B)$ is the fidelity of two positive semidefinite operators, given by $F(A, B) := (\tr(\sqrt{A^{1/2}BA^{1/2}}))^{2}$.}
\end{equation}
and
\begin{equation}
\tilde{\rho} \leq Y.
\end{equation}
\end{lemma}

\vspace{1 cm}

\begin{proof} (Proposition \ref{relenteqrob})

We start showing that 
\begin{equation*}
E_{\cal M}^{\infty}(\rho) \leq \lim_{\epsilon \rightarrow 0} \limsup_{n \rightarrow \infty} \frac{1}{n} LR_{{\cal M}_n}^{\epsilon}(\rho^{\otimes n}).
\end{equation*} 
Let $\rho_n^{\epsilon} \in B_{\epsilon}(\rho^{\otimes n})$ be an optimal state for $\rho^{\otimes n}$ in Eq. (\ref{smooth}). For every $n$ there is a state  $\sigma_n \in {\cal M}_n$ such that $\rho_n^{\epsilon} \leq s_n \sigma_n$, with $LR_{{\cal M}_n}^{\epsilon}(\rho^{\otimes n}) = LR_{{\cal M}_n}(\rho_n^{\epsilon}) = \log(s_n)$. It follows easily from the operator monotonicity of the 
$\log$ function \cite{Bat96} that if $\rho \leq 2^k \sigma$ (where 
$\rho$ and $\sigma$ are states), then $S(\rho || \sigma) \leq k$. Hence, 
\begin{equation*}
        \frac{1}{n}E_{{\cal M}_n}(\rho_n^{\epsilon}) \le \frac{1}{n}S(\rho_n^{\epsilon} || \sigma_n) \leq \frac{1}{n}LR_{{\cal M}_n}(\rho_n^{\epsilon}) = \frac{1}{n}LR_{{\cal M}_n}^{\epsilon}(\rho^{\otimes n}).
        \label{1456}
\end{equation*}
In Ref. \cite{SH06} it was shown that the minimum of the relative entropy with respect to any convex set containing the maximally mixed state is asymptotically continuous. Due to properties \ref{cond1} and \ref{cond2} of the sets ${\cal M}_n$ we thus find that the measures $E_{{\cal M}_n}$ are asymptotically continuous. Then, as $\rho_n^{\epsilon} \in B_{\epsilon}(\rho^{\otimes n})$, 
\begin{equation*}
        \frac{1}{n}E_{{\cal M}_n}(\rho^{\otimes n}) \leq \frac{1}{n}LR_{{\cal M}_n}^{\epsilon}(\rho^{\otimes n}) + f(\epsilon),
        \label{1456}
\end{equation*}
where $f: \mathbb{R} \rightarrow \mathbb{R}$ is such that $\lim_{\epsilon \rightarrow 0} f(\epsilon) = 0$. Taking the limits $n \rightarrow \infty$ and $\epsilon \rightarrow 0$ in both sides of the equation above,
\begin{equation*}
E^{\infty}_{{\cal M}}(\rho) = \lim_{\epsilon \rightarrow 0} \limsup_{n \rightarrow \infty}  \frac{1}{n}E_{{\cal M}_n}(\rho^{\otimes n})\leq \lim_{\epsilon \rightarrow 0} \limsup_{n \rightarrow \infty} \frac{1}{n}LR_{{\cal M}_n}^{\epsilon}(\rho^{\otimes n}).
\end{equation*}

To show the converse inequality, let $y_k := E_{{\cal M}_k}(\rho^{\otimes k})
+ \varepsilon = S(\rho^{\otimes k} || \sigma_k) + \varepsilon$ ($\sigma_k$ 
is an optimal state for $\rho^{\otimes k}$ in $E_{{\cal M}_k}(\rho^{\otimes k})$) with 
$\varepsilon > 0$. We can write
\begin{equation}\label{inehj}
\rho^{\otimes kn} \leq 2^{y_k n} \sigma_k^{\otimes n} + (\rho^{\otimes kn} - 2^{y_k n} \sigma_k^{\otimes n})_+.
\end{equation}
From Lemma \ref{ON} we have
\begin{equation*}
\lim_{n \rightarrow \infty}(\rho^{\otimes kn} - 2^{y_k n} \sigma_k^{\otimes n})_+ = 0.
\end{equation*}
Applying Lemma \ref{DR} to Eq. (\ref{inehj}) we then find that there is a sequence of states $\rho_{n, k}$ such that 
\begin{equation*}
\lim_{n \rightarrow \infty} || \rho^{\otimes kn} - \rho_{n, k} ||_1 = 0
\end{equation*}
and
\begin{equation*}
\rho_{n, k} \leq g(n)2^{y_k n} \sigma_k^{\otimes n},
\end{equation*}
where $g: \mathbb{R} \rightarrow \mathbb{R}$ is such that $\lim_{n \rightarrow \infty} g(n) = 1$. It follows that for every $\delta > 0$ there is a sufficiently large $n_0$ such that for all $n \geq n_0$, $\rho_{n, k} \in B_{\delta}(\rho^{\otimes kn})$. Moreover, from property \ref{cond4} of the sets we find $\sigma_k^{\otimes n} \in {\cal M}_{kn}$. Hence, for every $\delta > 0$,
\begin{equation*}
\limsup_{n \rightarrow \infty} \frac{1}{n}LR_{{\cal M}_{nk}}^{\delta}(\rho^{\otimes nk}) \leq \limsup_{n \rightarrow \infty} \frac{LR_{{\cal M}_{kn}}(\rho_{n, k})}{n} \leq y_k = E_{{\cal M}_k}(\rho^{\otimes k}) + \varepsilon.
\end{equation*}
As this is true for every $\varepsilon, \delta > 0$, it follows that\footnote{Here we use that for every $k \in \mathbb{N}$,
\begin{equation} \label{foot1}
\limsup_{n \rightarrow \infty} \frac{1}{nk}LR_{{\cal M}_{nk}}^{\delta}(\rho^{\otimes nk}) = \limsup_{n \rightarrow \infty} \frac{1}{n}LR_{{\cal M}_{n}}^{\delta}(\rho^{\otimes n}).
\end{equation} 
The $\leq$ inequality follows straighforwardly. For the $\geq$ inequality, let $\{ n' \}$ be a subsequence such that 
\begin{equation*}
M := \lim_{n' \rightarrow \infty} \frac{1}{n'}LR_{{\cal M}_{n'}}^{\delta}(\rho^{\otimes n'})
\end{equation*}
is equal to the R.H.S. of Eq. (\ref{foot1}). Let $n_k'$ be the first multiple of $k$ larger than of $n'$. Then, 
\begin{eqnarray*}
\limsup_{n \rightarrow \infty} \frac{1}{nk}LR_{{\cal M}_{nk}}^{\delta}(\rho^{\otimes nk}) &\geq& \limsup_{n_k' \rightarrow \infty} \frac{1}{n_k'} LR_{{\cal M}_{n_k'}}^{\delta}(\rho^{\otimes n_k'}) \nonumber \\ &\geq& \limsup_{n_k' \rightarrow \infty} \frac{1}{n_k'} LR_{{\cal M}_{n'}}^{\delta}(\rho^{\otimes n'}) \nonumber \\ &=& M.  
\end{eqnarray*}
The last inequality follows from $LR_{{\cal M}_{n}}^{\delta}(\pi) \geq LR_{{\cal M}_{n-l}}^{\delta}(\tr_{1,..l}(\pi))$, which is a consequence of property \ref{cond3} of the sets.} 
\begin{equation*}
\lim_{\delta \rightarrow 0} \limsup_{n \rightarrow \infty} \frac{1}{n} LR_{{\cal M}_{n}}^{\delta}(\rho^{\otimes n}) \leq \frac{1}{k}E_{{\cal M}_k}(\rho^{\otimes k}).
\end{equation*}
Finally, since the above equation is true for every $k \in \mathbb{N}$, we find the announced result.
\end{proof}
\vspace{0.3 cm}

There is another related quantity that we might consider in this context, in which $\epsilon$ and $n$ are not independent. Define
\begin{equation}
LG_{{\cal M}}(\rho) := \inf_{\{ \epsilon_n\}} \left \{ \limsup_{n \rightarrow \infty} \frac{1}{n} LR_{{\cal M}_n}^{\epsilon_n}(\rho^{\otimes n}) : \lim_{n \rightarrow \infty} \epsilon_n = 0 \right \}. 
\end{equation}
The proof of Proposition \ref{relenteqrob} can be straightforwardly adapted to show
\begin{corollary} \label{LGmeasure}
For every family of sets  $\{ {\cal M}_n \}_{n \in \mathbb{N}}$ satisfying properties \ref{cond1}-\ref{cond5} and every quantum state $\rho \in {\cal D}({\cal H})$,
\begin{equation}
LG_{{\cal M}}(\rho) = E_{\cal M}^{\infty}(\rho).
\end{equation}
\end{corollary}

This result will be used in the next chapter.

\subsection{Proof of Proposition \ref{maincompact}}

We now turn to the proof of Proposition \ref{maincompact}, which is the main technical contribution of this chapter. Before we start with the proof in earnest, we provide an outline of the main steps which will be taken, in order to make the presentation more transparent. 

The first step is to note that when $y > E_{\cal M}^{\infty}(\rho)$, it follows directly from Proposition \ref{relenteqrob} that 
\begin{equation} \label{ouline}
\lim_{n \rightarrow \infty} \min_{\omega_n \in {\cal M}_n} \tr( \rho^{\otimes n} - 2^{yn}\omega_n)_+ = 0,
\end{equation}
while for any $y < E_{\cal M}^{\infty}(\rho)$, this limit is strictly larger than zero. This shows that that $E_{\cal M}^{\infty}(\rho)$ is the strong converse rate in the hypothesis testing problem we are analysing.

It is more involved to show that $E_{\cal M}^{\infty}(\rho)$ is also an achievable rate, i.e. that the limit is one for every $y < E_{\cal M}^{\infty}(\rho)$. The difficulty is precisely that the alternative hypothesis is non-i.i.d. in general. Indeed, if $\omega_n$ were i.i.d., then the result would follows directly from the achievability part of quantum Stein's Lemma \cite{HP91}. Most of the proof is devoted to circumvent this problem. 

The main ingredient will be the exponential de Finetti theorem due to Renner \cite{Ren05, Ren07}. We proceed by means of a contradiction. Assuming conversely that the limit in Eq. (\ref{ouline}) is $0 < \mu < 1$ and using Lemma \ref{DR}, we can find a state $\rho_n$ with non-negligible fidelity with $\rho^{\otimes n}$ such that
\begin{equation*}
\rho_n \leq 2^{yn}\omega_n,
\end{equation*}
for every $n$, where $\omega_n \in {\cal M}_n$ is the optimal state in the minimization of Eq. (\ref{ouline}). Due to property \ref{cond5} of the sets, we can take $\omega_n$ and thus also $\rho_n$ to be permutation-symmetric. Tracing a sublinear number of copies $m$ and using the exponential de Finetti theorem we then find, 
\begin{equation*}
\int \mu(d\sigma) \pi_{\sigma} \leq 2^{yn}\tr_{1,...,m}(\omega_n),
\end{equation*}
where each $\pi_{\sigma}$ is close to an almost power state along $\sigma$. Since the state appearing in the L.H.S. of this Equation has a non-negligible fidelity with $\rho^{\otimes n - m}$ and the states $\pi_{\sigma}$ behave like $\sigma^{\otimes n - m}$ when we measure the same POVM on all $n - m$ copies, the integral must have a non-negligible support on a neighborhood of $\rho$. This then allows us to write 
\begin{equation*}
\pi_{\tilde{\rho}} \leq n^k 2^{yn}\tr_{1,...,m}(\omega_n),
\end{equation*}
for some constant $k$ and an approximation $\tilde{\rho}$ of $\rho$. Then, using the operator monotonicity of the $\log$, the properties of almost power states, which will be shown to be similar to the properties of power states in what concerns the measures $E_{{\cal M}_k}$, and the asymptotic continuity of both $E_{{\cal M}_k}$ and $E_{\cal M}^{\infty}$, we find from this equation that
\begin{equation*}
E_{\cal M}^{\infty}(\rho) \leq y.
\end{equation*}
As we assume $y < E_{\cal M}^{\infty}(\rho)$, we will arrive in a contradiction, showing that the limit in Eq. (\ref{ouline}) must indeed be one. 

\vspace{0.5 cm}

\begin{proof} (Proposition \ref{maincompact})

Let us start showing that if $y = E_{\cal M}^{\infty}(\rho) + \epsilon$, then
\begin{equation}\label{onepart}
\lim_{n \rightarrow \infty} \min_{\omega_n \in {\cal M}_n} \tr( \rho^{\otimes n} - 2^{yn}\omega_n)_+ = 0.
\end{equation}
By Proposition \ref{relenteqrob} there is a $\delta_0 > 0$ such that
\begin{equation} \label{dif1}
\left | E_{\cal M}^{\infty}(\rho)  - \limsup_{n \rightarrow \infty} \frac{1}{n} LR_{{\cal M}_n}^{\delta}(\rho^{\otimes n}) \right | \leq \epsilon/2,
\end{equation}
for every $\delta \leq \delta_0$. Let $\rho_{n, \delta} \in B_{\delta}(\rho^{\otimes n})$ be an optimal state in Eq. (\ref{smooth}) for $\rho^{\otimes n}$. Then there must exists a $\sigma_n \in {\cal M}_n$ such that 
\begin{equation*}
\rho_{n, \delta} \leq 2^{LR_{{\cal M}_n}^{\delta}(\rho^{\otimes n})} \sigma_n,
\end{equation*}
from which follows that for every $\lambda \geq LR_{{\cal M}_n}^{\delta}(\rho^{\otimes n})/n$,
\begin{equation*}
\min_{\omega_n \in {\cal M}_n} \tr( \rho^{\otimes n} - 2^{\lambda n}\omega_n)_+ \leq \min_{\omega_n \in {\cal M}_n} \tr( \rho_{n, \delta} - 2^{{\lambda} n}\omega_n)_+ + \delta \leq \delta.
\end{equation*}
 
From Eq. (\ref{dif1}) and our choice of $y$ we then find that for every $\delta > 0$ there is a sufficiently large $n_0$ such that for all $n \geq n_0$,
\begin{equation*}
\min_{\omega_n \in {\cal M}_n} \tr( \rho^{\otimes n} - 2^{y n}\omega_n)_+ \leq \delta,
\end{equation*}
from which Eq. (\ref{onepart}) follows.

Let us now prove that if $y = E_{\cal M}^{\infty}(\rho) - \epsilon$, then 
\begin{equation}\label{otherpart}
\lim_{n \rightarrow \infty} \min_{\omega_n \in {\cal M}_n} \tr( \rho^{\otimes n} - 2^{yn}\omega_n)_+ = 1.
\end{equation}
We start first proving the weaker statement that the limit in the L.H.S. of Eq. (\ref{otherpart}) goes to $1 - \lambda$, with $0 \leq \lambda < 1$. We assuming conversely that this is not the case and that the limit is zero. For each $n$ we have
\begin{equation} \label{ine11}
\rho^{\otimes n} \leq 2^{yn}\omega_n + (\rho^{\otimes n} - 2^{yn}\omega_n)_+
\end{equation}
Applying Lemma \ref{DR} to Eq. (\ref{ine11}) we find that there are states $\tilde{\rho}_n$ such that $||\rho^{\otimes n} - \tilde{\rho}_n  ||_1 \rightarrow 0$ and $\tilde{\rho}_n \leq g(n)2^{yn}\omega_n$, for a function $g$ satisfying $\lim_{n\rightarrow \infty}g(n) = 1$. It follows that 
\begin{equation*}
\frac{1}{n}LR_{{\cal M}_n}(\tilde{\rho}_n) \leq y
\end{equation*}
and that for every $\delta > 0$ and sufficiently large $n$, $\tilde{\rho}_n \in B_{\delta}(\rho^{\otimes n})$. Therefore, for very $\delta > 0$,
\begin{equation*}
\limsup_{n \rightarrow \infty} \frac{1}{n}LR_{{\cal M}_n}^{\delta}(\rho^{\otimes n}) \leq \limsup_{n \rightarrow \infty} \frac{1}{n}LR_{{\cal M}_n}(\tilde{\rho}_n) \leq y = E_{{\cal M}}^{\infty}(\rho) - \epsilon,
\end{equation*}
in contradiction to Eq. (\ref{relenteqrobeq}) of Proposition \ref{relenteqrob}.

In the rest of the proof we show that if $0 < \lambda < 1$, we also find a contradiction, which will lead us to conclude that $\lambda = 0$, as desired. 

Let $\{ \sigma_n \in {\cal M}_n \}_{n \in \mathbb{N}}$ be a sequence of optimal solutions in the minimization of Eq. (\ref{otherpart}). We assume conversely that
\begin{equation*}
\limsup_{n \rightarrow \infty} \tr( \rho^{\otimes n} - 2^{yn}\sigma_n)_+ = 1 - \lambda < 1.
\end{equation*}
Note that from the monotonicity of the trace of the positive part under trace preserving CP maps and property \ref{cond5} of the sets $\{ {\cal M}_n \}_{n \in \mathbb{N}}$, we can assume w.l.o.g. that the states $\sigma_n$ are permutation-symmetric.
 
For each $n \in \mathbb{N}$ we have $\rho^{\otimes n} \leq 2^{yn}\sigma_n + ( \rho^{\otimes n} - 2^{yn}\sigma_n)_+$. Applying Lemma \ref{DR} once more we see that for every $n \in \mathbb{N}$ there is a state $\rho_n$ such that
\begin{equation} \label{fidelity10}
F(\rho_n, \rho^{\otimes n}) \geq \lambda 
\end{equation}
and
\begin{equation}\label{fidelitynext}
\rho_n \leq  \lambda^{-1}2^{yn}\sigma_n.
\end{equation}
From the monotonicity of the fidelity under trace preserving CP maps and the permutation-invariance of $\sigma_n$ and $\rho^{\otimes n}$ we can also take $\rho_n$ to be permutation-symmetric. 

In the next paragraphs we employ the exponential de Finetti theorem \cite{Ren05, Ren07} to reduce the problem to the case in which the state appearing in the L.H.S. of Eq. (\ref{fidelitynext}) is close to an almost power state. Then in the sequel we reduce the problem to the i.i.d. case, which we can more easily handle. 

From Theorem \ref{expdefinetti} and Eqs. (\ref{df1}), (\ref{df2}), and (\ref{df3}) of chapter \ref{entanglement} we find, with $k = \alpha n$, $0 < \alpha < 1$, and $r := 11 d^2 \log(n) / \alpha$, where $d := \dim({\cal H})$,
\begin{equation} \label{defin}
\rho_{(1-\alpha) n} := \tr_{1,...,\lfloor \alpha n \rfloor}(\rho_n) = \int_{\sigma \in D({\cal H})} \int_{\ket{\theta} \supset \sigma} \mu(d\ket{\theta}) \tr_{E}\left( \ket{\psi^{\ket{\theta}}_{(1-\alpha)n}} \bra{\psi^{\ket{\theta}}_{(1-\alpha)n}} \right) + X_n, 
\end{equation}
where $\ket{\psi^{\ket{\theta}}_{(1-\alpha)n}} \in \ket{\theta}^{[\otimes, (1-\alpha)n, 11d^2\alpha^{-1} \log(n)]}$ and $|| X_n||_1 \leq n^{d^2} 2^{- \frac{\alpha n 11 d^2 \log(n)}{\alpha n}} = n^{- 10d^2}$.  As in section \ref{expdf} of chapter \ref{entanglement}, the notation $\ket{\theta} \supset \sigma$ means that the integration is taken with respect to the environment Hilbert space $E$ and over all the purifications of $\sigma$. For simplicity of notation we define 
\begin{equation*}
\pi_{n}^{\ket{\theta}} := \tr_E(\ket{\psi^{\ket{\theta}}_{(1-\alpha)n}}\bra{\psi^{\ket{\theta}}_{(1-\alpha)n}}).
\end{equation*}

Let us show that due of Eq. (\ref{fidelity10}), $\mu$ must have positive measure in a set of purifications of states in a neighborhood of $\rho$. We make use of the fact that $\pi_{n}^{\ket{\theta}}$, with $\ket{\theta} \supset \sigma$, behaves almost in the same manner as $\sigma^{\otimes (1 - \alpha)n}$ if we measure a POVM on each of the $(1 - \alpha)n$ copies\footnote{See section \ref{expdf} for further discussion on this point.}. First we note that for every POVM element $0 \leq A \leq \id$\footnote{The first inequality and the following equality follow directly from the definition of the trace of the positive part. The second inequality is a well-known relation between the trace norm difference and the fidelity of two quantum states. The third inequality follows from the monotonicity of the fidelity under trace preserving CP maps and, finally, the last inequality is a consequence of Eq. (\ref{fidelity10}).},
\begin{eqnarray} \label{chaintnf}
|\tr(A\rho_{(1 - \alpha)n}) - \tr(A \rho^{\otimes (1 - \alpha)n})| &\leq& \tr(\rho_{(1 - \alpha)n} - \rho^{\otimes (1 - \alpha)n})_+ \nonumber \\ &=& \frac{1}{2}|| \rho_{(1 - \alpha)n} - \rho^{\otimes (1 - \alpha)n} ||_1 \nonumber \\ &\leq& \sqrt{1 - F(\rho_{(1 - \alpha)n}, \rho^{\otimes (1 - \alpha)n})^2} \nonumber \\ &\leq& \sqrt{1 - F(\rho_n, \rho^{\otimes n})^2} \nonumber \\ &\leq& \sqrt{1 - \lambda^2}.
\end{eqnarray}
We now consider a two outcome POVM $\{A_n, \id - A_n \}$ formed as follows. We measure an informationally complete\footnote{See section \ref{infcomp} of chapter \ref{NPPT} for the definition of an informationally complete POVM.} POVM $\{ M_k \}_{k=1}^{L}$ in each of the $(1 - \alpha)n$ systems, obtaining an empirical frequency distribution $p_{k, n}$ of the possible outcomes $\{ k \}_{k=1}^L$. Using this probability distribution, we form the operator
\begin{equation*}
L_n := \sum_{k=1}^L p_{k, n} M_k^{*},
\end{equation*}
where $\{ M_k^* \}$ is the dual of the family $\{ M_k \}$\footnote{See section \ref{infcomp} for the definition of the dual set of operators of an informationally complete POVM. Note in particular that there is a number $K_d \leq d^4$ such that $|| L_n - \rho ||_1 \leq K_d || p_{k, n} - \tr(\rho M_k) ||_1$}. If $|| L_n - \rho ||_1 \leq n^{-\frac{1}{6}}$\footnote{Note that here, as in a few other parts in the rest of the proof, the $1/6$ exponent is not fundamental. We however choose an explicitly value for it to make the presentation simpler.} we accept, otherwise we reject. We denote the POVM element associated with the event that we accept by $A_n$. From Lemma \ref{classicalCH} and Eq. (\ref{eq3.1.5}) of chapter \ref{NPPT} we find 
\begin{equation*}
\Pr \left( || L_n -\rho ||_1 \leq n^{-1/6}\right) \geq \Pr \left( || p_{k, n} - \tr(\rho M_k)||_1 \leq d^{-4} n^{-1/6}\right) \geq 1 - O(2^{ - d^{-9}(1 - \alpha) n^{2/3}}).
\end{equation*}
Thus
\begin{equation*} 
|1 - \tr(A_n \rho^{\otimes(1 - \alpha)n})| \leq O(2^{ -d^{-9}(1 - \alpha)n^{2/3}}).
\end{equation*}
By a similar argument, we find from Lemma \ref{hof} of chapter \ref{entanglement}, in turn, that for every $\ket{\theta}$ such that $\tr_E(\ket{\theta}\bra{\theta}) \notin B_{n^{-1/8}}(\rho)$,
\begin{equation*}
\tr(A_n \pi_{n}^{\ket{\theta}}) \leq O(2^{-d^{-8}(1 - \alpha)n^{2/3}}),
\end{equation*}
From Eq. (\ref{chaintnf}) we can write
\begin{equation*}
|1 - \int_{\sigma \in B_{n^{-1/8}}(\rho)} \int_{\ket{\theta} \supset \sigma} \mu(d\ket{\theta}) \tr(A_n \pi_{n}^{\ket{\theta}}) | \leq \sqrt{1 - \lambda^2} + O(2^{-d^{-8}(1 - \alpha)n^{2/3}}),
\end{equation*}
and therefore
\begin{eqnarray} \label{boundint}
\int_{\sigma \in B_{n^{-1/8}}(\rho)} \int_{\ket{\theta} \supset \sigma} \mu(d\ket{\theta}) &\geq& \int_{\sigma \in B_{n^{-1/8}}(\rho)} \int_{\ket{\theta} \supset \sigma} \mu(d\ket{\theta}) \tr(A_n \pi_{n}^{\ket{\theta}}) \nonumber \\ &\geq& 1 - \sqrt{1 - \lambda^2} + O(2^{-d^{-8}(1 - \alpha)n^{2/3}}) \nonumber \\ &\geq& (1 - \sqrt{1 - \lambda^2})/2,
\end{eqnarray}
for sufficiently large $n$. 

Define the set
\begin{equation*}
{\cal L} := \{ \ket{\theta} \in {\cal H}_S\otimes {\cal H}_E : \tr_E(\ket{\theta}\bra{\theta}) \in B_{n^{-1/8}}(\rho)\}.
\end{equation*}
On one hand, it follows from Eq. (\ref{boundint}) that there must exist a $\ket{\tilde{\theta}} \in {\cal L}$ such that\footnote{This can be seen by contradiction. Suppose that for every state in ${\cal L}$ the L.H.S. of Eq. (\ref{eqnet}) is smaller than the R.H.S. Let us take the maximal set  $\{ \ket{\phi_i} \}_{i=1}^N$ of pure states in ${\cal L}$ satisfying $|| \ket{\phi_i} - \ket{\phi_j} ||_1 \geq n^{-4}$ for every $i, j$. This set has cardinality at most $(1 + n^{-4}/4)^{2d^2}/(n^{-4}/4)^{2d^2} < n^{8d^2}/2$ (see e.g. Lemma II.4 of Ref. \cite{HLSW04}). Then 
\begin{equation*}
\int_{\ket{\theta} \in {\cal L}}\mu(d\ket{\theta}) \leq \sum_{i=1}^N \int_{\ket{\theta} \in B_{n^{-4}}(\ket{\phi_i})}\mu(d\ket{\theta}) \leq N n^{-8d^2} (1 - \sqrt{1 - \lambda^2}) < (1 - \sqrt{1 - \lambda^2})/2,
\end{equation*}
in contradiction to Eq. (\ref{boundint}).}
\begin{equation} \label{eqnet}
V_{n} := \int_{\ket{\theta} \in B_{n^{-4}}(\ket{\tilde{\theta}})} \mu(d\ket{\theta}) \geq (1 - \sqrt{1 - \lambda^2})n^{-8d^2}.
\end{equation}
On the other, from Eqs. (\ref{fidelitynext}) and (\ref{defin}),
\begin{equation*}
\int_{\ket{\theta} \in B_{n^{-4}}(\ket{\tilde{\theta}})} \mu(d\ket{\theta}) \pi_{n}^{\ket{\theta}} + X_n \leq \lambda^{-1}2^{yn}\tr_{1,...,\lfloor \alpha n \rfloor}(\sigma_n).
\end{equation*}
Dividing both sides of this equation by $V_{n}$,
\begin{equation}\label{eq200}
\varsigma_n := \int_{\ket{\theta} \in B_{n^{-4}}(\tilde{\theta})} \nu(d\ket{\theta}) \pi_{n}^{\ket{\theta}} \leq \lambda^{-1}(1 - \sqrt{1 - \lambda^2})^{-1} n^{8d^2} 2^{yn}\tr_{1,...,\lfloor \alpha n \rfloor}(\sigma_n) + (X'_n)_+,
\end{equation}
where $X'_n := X_n / V_{n}$, which satisfies $\tr(X'_n)_+ \leq || X'_n  ||_1 \leq n^{8d^2}n^{- 10 d^2} \leq n^{- 2d^2}$, and $\nu := \mu/V_{n}$ is a probability density. 

Applying Lemma \ref{DR} to Eq. (\ref{eq200}), we find there is a state $\varsigma'_n$ such that $|| \varsigma_n - \varsigma'_n ||_1 \leq 2n^{- 2d^2}$ and \begin{equation*}
\varsigma'_n \leq \lambda^{-1}(1 - \sqrt{1 - \lambda^2})^{-1} n^{8d^2} 2^{yn}\tr_{1,...,\lfloor \alpha n \rfloor}(\sigma_n).
\end{equation*} 

From property \ref{cond3} of the sets we see that $\tr_{1,...,\lfloor \alpha n \rfloor}(\sigma_n) \in {\cal M}_{(1 - \alpha)n}$. Using the operator monotonicity of the $\log$ \cite{Bat96},
\begin{equation*}
\frac{1}{n}E_{{\cal M}_{(1 - \alpha)n}}(\varsigma'_n) \leq y +9d^2\frac{\log(n)}{n},
\end{equation*}
for sufficiently large $n$. From the asymptotic continuity of $E_{{\cal M}_k}$, in turn,
\begin{equation} \label{cont}
\frac{1}{n}E_{{\cal M}_{(1 - \alpha)n}}(\varsigma_n) \leq y + 9d^2\frac{\log(n)}{n} + f(2n^{- 2d^2}),
\end{equation}
where $f : \mathbb{R} \rightarrow \mathbb{R}$ is a function such that $\lim_{x \rightarrow 0}f(x) = 0$. 

We should now already have an idea of how the contradiction will appear. Indeed, $\varsigma_n$ is a state formed by almost power states along purifications of states close to $\rho$. We can then expect that for very large $n$, $\frac{1}{n}E_{{\cal M}_{(1 - \alpha)n}}(\varsigma_n)$ should be close to $\frac{1}{n}E_{{\cal M}_{(1 - \alpha)n}}(\rho^{\otimes n})$, from which we could readily find a contradiction from Eq. (\ref{cont}). In the next paragraphs we prove that this intuition is indeed correct.

As pointed out in section \ref{expdf}, every $\ket{\psi^{\ket{\theta}}_{(1 - \alpha)n}}$ in the convex combination forming $\varsigma_n$ can be written as 
\begin{equation} \label{rep}
\ket{\psi^{\ket{\theta}}_{(1 - \alpha)n}} = \sum_{k = 0}^{r} \sum_{l=1}^{\binom{(1 - \alpha)n}{k}} \binom{(1 - \alpha)n}{k}^{-1/2} \beta_k^{\ket{\theta}}  P_{\pi_{l, k}} \ket{\Psi_k^{\ket{\theta}}} \otimes \ket{\theta}^{\otimes (1 - \alpha)n - k}, 
\end{equation}
where the each of the terms satisfies the conditions presented in section \ref{expdf} and, as before, $r = 11 d^2 \log(n) / \alpha$. Note in particular that each $\ket{\Psi_k^{\ket{\theta}}}$ is such that
\begin{equation}
\ket{\Psi_k^{\ket{\theta}}} = \sum_{m_1,...,m_k} c_{m_1,...,m_k}^{\ket{\theta}} \ket{m_1} \otimes ... \otimes \ket{m_k},
\end{equation}
with $|\braket{m_l}{\theta}| = 0$ for every $l$. 

Define the states $\ket{\phi^{\ket{\theta}}_{(1 - \alpha)n}} := \ket{\overline{\psi}^{\ket{\theta}}_{(1 - \alpha)n}}/ ||\ket{\overline{\psi}^{\ket{\theta}}_{(1 - \alpha)n}}||$, with
\begin{equation} \label{approx1}
\ket{\overline{\psi}^{\ket{\theta}}_{(1 - \alpha)n}} := \sum_{k : |\beta_k^{\ket{\theta}}| \geq n^{-1}} \sum_{l=1}^{\binom{(1 - \alpha)n}{k}} \binom{(1 - \alpha)n}{k}^{-1/2} \beta_k^{\ket{\theta}}  P_{\pi_{l, k}} \ket{\tilde{\Psi}_k^{\ket{\theta}}} \otimes \ket{\tilde{\theta}}^{\otimes (1 - \alpha)n - k}, 
\end{equation}
where 
\begin{equation}
\ket{\tilde{\Psi}_k^{\ket{\theta}}} := \sum_{m_1,...,m_k} c_{m_1,...,m_k}^{\ket{\theta}} (\ket{m_1} - \braket{\tilde{\theta}}{m_1}\ket{\tilde{\theta}}) \otimes ... \otimes (\ket{m_k} - \braket{\tilde{\theta}}{m_k}\ket{\tilde{\theta}}).
\end{equation}

A simple calculation shows that for any $\ket{\theta} \in B_{n^{-4}}(\ket{\tilde{\theta}})$\footnote{Define the auxiliary non-normalized state
\begin{equation*} 
\ket{\zeta^{\ket{\theta}}_{(1 - \alpha)n}} := \sum_{k : |\beta_k^{\ket{\theta}}| \geq n^{-1}} \sum_{l=1}^{\binom{(1 - \alpha)n}{k}} \binom{(1 - \alpha)n}{k}^{-1/2} \beta_k^{\ket{\theta}}  P_{\pi_{l, k}} \ket{\Psi_k^{\ket{\theta}}} \otimes \ket{\tilde{\theta}}^{\otimes (1 - \alpha)n - k}. 
\end{equation*}
Then 
\begin{equation*}
|\braket{\zeta^{\ket{\theta}}_{(1 - \alpha)n}}{\psi^{\ket{\theta}}_{(1 - \alpha)n}}|^2 = \sum_{k : |\beta_k^{\ket{\theta}}| \geq n^{-1}} |\beta_k^{\ket{\theta}}|^2 |\braket{\theta}{\tilde{\theta}}|^{2(1 - \alpha)n - 2k} \geq (1 - n^{-2})(1 - n^{-4})^{(1 - \alpha)n} \geq 1 - n^{-3}, 
\end{equation*}
where we used that $\ket{\Psi_k} \in ({\cal H} / \{ \ket{\theta} \})^{\otimes k}$. Similarly, 
\begin{equation*}
|\braket{\zeta^{\ket{\theta}}_{(1 - \alpha)n}}{\overline{\psi}^{\ket{\theta}}_{(1 - \alpha)n}}|^2 = \sum_{k : |\beta_k^{\ket{\theta}}| \geq n^{-1}} |\beta_k^{\ket{\theta}}|^2 |\braket{\Psi_k}{\tilde{\Psi}_k}|^2 \geq (1 - n^{-2}) (1 - d^2\log(n)/n^3) \geq 1 - n^{-2}, 
\end{equation*}
for sufficiently large $n$.},   
\begin{equation*}
|| \ket{\psi^{\ket{\theta}}_{(1 - \alpha)n}}\bra{\psi^{\ket{\theta}}_{(1 - \alpha)n}} -   \ket{\phi^{\ket{\theta}}_{(1 - \alpha)n}}\bra{\phi^{\ket{\theta}}_{(1 - \alpha)n}} ||_1 \leq n^{-1}.
\end{equation*}
Defining 
\begin{equation} \label{tau}
\tau_n := \int_{\ket{\theta} \in B_{n^{-4}}(\rho)} \nu(d\ket{\theta}) \tr_{E}(\ket{\phi^{\ket{\theta}}_{(1 - \alpha)n}} \bra{\phi^{\ket{\theta}}_{(1 - \alpha)n}}),
\end{equation}
we see that $|| \tau_n - \varsigma_n || \leq n^{-1}$, for sufficiently large $n$. From Eq. (\ref{cont}) and the asymptotic continuity of $E_{{\cal M}_k}$,
\begin{equation} \label{correct5} 
\frac{1}{n}E_{{\cal M}_{(1 - \alpha)n}}(\tau_n) \leq y + 9d^2\frac{\log(n)}{n} + 2f(n^{-1}).
\end{equation}

Let us denote the maximum $k$ appearing in Equation (\ref{approx1}) by $k_{\max}^{\ket{\theta}}$. Consider the term $\ket{\tilde{\Psi}_{k_{\max}}^{\ket{\theta}}} \otimes \ket{\tilde{\theta}}^{\otimes (1 - \alpha)n - k_{\max}^{\ket{\theta}}}$. It is clear that in all the other terms of the superposition the state $\ket{\tilde{\theta}}$ will appear at least in one of the first $k_{\max} \leq r$ registers. As each $\ket{\tilde{\Psi}_k^{\ket{\theta}}}$ lives in $({\cal H}/ \{ \ket{\tilde{\theta}}\})^{\otimes k}$, 
\begin{eqnarray*}
(\ket{\tilde{\theta}}\bra{\tilde{\theta}})^{\otimes (1 - \alpha)n - r} &\leq& \binom{(1 - \alpha)n}{r} |\beta_{k_{\max}}^{\ket{\theta}}|^{-2} \tr_{1,...,r}(\ket{\phi^{\ket{\theta}}_{(1 - \alpha)n}}\bra{\phi^{\ket{\theta}}_{(1 - \alpha)n}}) \nonumber \\ &\leq&  2^{nh(r/n)}n^2 \tr_{1,...,r}(\ket{\phi^{\ket{\theta}}_{(1 - \alpha)n}}\bra{\phi^{\ket{\theta}}_{(1 - \alpha)n}}),
\end{eqnarray*}
where the last inequality follows from the fact that $|\beta_{k_{\max}}^{\ket{\theta}}|^{-2} \leq n^{2}$ and the bound $\binom{n}{k} \leq 2^{nh(k/n)}$ \cite{CT91}. From Eq. (\ref{tau}) we then find
\begin{equation*}
\tilde{\rho}^{\otimes (1 - \alpha)n - r} \leq n^22^{nh(r/n)} \tr_{1,...,r}(\tau_n),
\end{equation*}
where $\tilde{\rho} := \tr_{E}(\ket{\tilde{\theta}}\bra{\tilde{\theta}}) \in B_{n^{-1/8}}(\rho)$. 

Define $\tilde{\tau}_n := \tr_{1,...,r}(\tau_n)$ and let $\omega_n \in {\cal M}_{(1 - \alpha)n - r}$ be such that
\begin{equation*} 
E_{{\cal M}_{(1 - \alpha)n - r}}(\tilde{\tau}_n) = S(\tilde{\tau}_n || \omega_n).
\end{equation*}
Let $\lambda = E_{{\cal M}_{(1 - \alpha)n - r}}(\tilde{\tau}_n) + n\epsilon/2$. For every integer $m$ we can write
\begin{eqnarray} \label{beforelast}
\tilde{\rho}^{\otimes ((1 - \alpha)n - r)m} &\leq& n^{2m}2^{nh(r/n)m}\tilde{\tau}_n^{\otimes m} \nonumber \\ &\leq& n^{2m}2^{nh(r/n)m}2^{\lambda m} \omega_n^{\otimes m} + n^{2m}2^{nh(r/n)m}(\tilde{\tau}_n^{\otimes m} - 2^{\lambda m} \omega_n^{\otimes m})_+.
\end{eqnarray}
From Lemma \ref{ON} we have 
\begin{equation*}
\tr(\tilde{\tau}_n^{\otimes m} - 2^{\lambda m} \omega_n^{\otimes m})_+ \leq 2^{- \epsilon n m/2}.
\end{equation*}
Then, noting that $2^{- \epsilon/2 n}  n^22^{nh(r/n)} \leq 2^{- \epsilon/4 n}$ for sufficiently large $n$, we can apply Lemma \ref{DR} to Eq. (\ref{beforelast}) to find that there is a state $\tilde{\rho}_{m n}$ such that $|| \tilde{\rho}_{m, n} -  \tilde{\rho}^{\otimes ((1 - \alpha)n - r)m} ||_1 \leq 2^{- \epsilon/8 n m}$ and
\begin{equation*}
\tilde{\rho}_{m, n} \leq g(n)(n^22^{nh(r/n)})^m 2^{\lambda m} \omega_n^{\otimes m}, 
\end{equation*}
for a function $g(n)$ such that $\lim_{n \rightarrow \infty}g(n) = 1$. From the operator monotonicity of the $\log$ \cite{Bat96} and the asymptotic continuity of the quantities $E_{{\cal M}_k}$, once more, for sufficiently large $n$ it holds
\begin{eqnarray*}
\frac{1}{m} E_{{\cal M}_{((1 - \alpha) - r)m}}(\tilde{\rho}^{\otimes ((1 - \alpha)n - r)m}) &\leq&  2 n h(r/n) + E_{{\cal M}_{(1 - \alpha)n - r}}(\tilde{\tau}_n) + n\epsilon \nonumber \\ &\leq & 2 n h(r/n) + E_{{\cal M}_{(1 - \alpha)n}}(\tau_n) + n\epsilon,
\end{eqnarray*}
where the last inequality follows from property \ref{cond3} of the sets ${\cal M}_k$. From property \ref{cond4}, on the other hand, 
\begin{equation*}
E_{\cal M}^{\infty}(\tilde{\rho}) \leq \frac{1}{((1 - \alpha)n - r)m} E_{{\cal M}_{((1 - \alpha) - r)m}}(\tilde{\rho}^{\otimes ((1 - \alpha)n - r)m}),
\end{equation*}
and hence
\begin{equation*}
E_{\cal M}^{\infty}(\tilde{\rho}) \leq 2 \frac{n}{(1 - \alpha)n - r} h(r/n) + \frac{1}{(1 - \alpha)n - r}E_{{\cal M}_{(1 - \alpha)n}}(\tau_n) + \frac{n}{(1 - \alpha)n - r}\epsilon.
\end{equation*}
In Ref. \cite{Chr06} it was proven that the regularization of the minimum relative entropy over any family of sets satisfying properties \ref{cond1}, \ref{cond2}, \ref{cond3} and \ref{cond4} is asymptotically continuous. It then follows that 
\begin{equation*}
E_{\cal M}^{\infty}(\rho) \leq 2 \frac{n}{(1 - \alpha)n - r} h(r/n) + \frac{1}{(1 - \alpha)n - r}E_{{\cal M}_{(1 - \alpha)n}}(\tau_n) + \frac{n}{(1 - \alpha)n - r}\epsilon + f(n^{-1/8}).
\end{equation*}
Using the bound for $E_{{\cal M}_{(1 - \alpha)n}}(\tau_n)$ given by Eq. (\ref{correct5}) and taking the limit $n \rightarrow \infty$ in the equation  above we finally find
\begin{equation*}
E_{\cal M}^{\infty}(\rho) \leq \frac{1}{1 - \alpha}E_{\cal M}^{\infty}(\rho) - \frac{\epsilon}{2(1 - \alpha)}.
\end{equation*}
As $\alpha$ can be taken to be arbitrarily small, we find a contradiction since $\epsilon > 0$. 
\end{proof}

\section{Proof of Corollary \ref{faithful}} \label{coroll}

In this section we prove that the regularized relative entropy of entanglement is faithful. The idea 
is to combine Theorem \ref{maintheorem} with the exponential de Finetti theorem. 

\vspace{0.5 cm}

\begin{proof} (Corollary \ref{faithful})
For simplicity we present the proof for the bipartite case. Its extension to the multipartite setting is completely analogous. In the following paragraphs we prove that for every entangled state $\rho \in {\cal D}({\cal H}_A \otimes {\cal H}_B)$, there is a $\mu(\rho) > 0$ and a sequence of POVM elements $0 \leq A_n \le \id$ such that
\begin{equation*}
\lim_{n \rightarrow \infty} \tr(A_n \rho^{\otimes n}) = 1, 
\end{equation*}
and for all sequences of separable states $\{ \omega_n \}_{n \in \mathbb{N}}$,
\begin{equation*}
- \frac{\log \tr(A_n \omega_n)}{n} \geq \mu(\rho),
\end{equation*}
From Theorem \ref{maintheorem} it will then follows that $E_R^{\infty}(\rho) \geq \mu(\rho) > 0$.

To form the sequence of POVM elements we use a recurrent construction in this thesis, employed before in the proofs of Theorems \ref{maincompact} and \ref{T=coC}. We apply the symmetrization operation $\hat{S}_n$ to the $n$ individual Hilbert spaces, trace out the first $\alpha n$ systems ($0 < \alpha < 1$), and then measure a LOCC informationally complete POVM $\{ M_k \}_{k=1}^{L}$ in each of the remaining $(1 - \alpha)n$ systems, obtaining an empirical frequency distribution $p_{k, n}$ of the possible outcomes $\{ k \}_{k=1}^L$. Using this probability distribution, we form the operator
\begin{equation*}
L_n := \sum_{k=1}^L p_{k, n} M_k^{*},
\end{equation*}
where $\{ M_k^* \}$ is the dual set of the family $\{ M_k \}$. If 
\begin{equation*}
|| L_n - \rho ||_1 \leq \epsilon,
\end{equation*}
where
\begin{equation} \label{epsfar}
\epsilon := \min_{\sigma \in {\cal S}} || \rho - \sigma ||_1,\footnote{Note that as $\rho$ is entangled, $\epsilon > 0$.}
\end{equation}
we accept, otherwise we reject. Let $A_n := \hat{S}_n(\id^{\otimes \alpha n} \otimes \tilde{A}_n)$ be the POVM element associated to the event that we accept, where $\tilde{A}_n$ is the POVM element associated to measuring $\{ M_k \}_{k=1}^{L}$ on each of the $(1 - \alpha)n$ copies and accepting. 

First, by the law of large numbers \cite{Dud02} and the definition of informationally complete POVMs, it is clear that $\lim_{n \rightarrow \infty} \tr(A_n \rho^{\otimes n}) = 1$. It thus remains to show that $\tr(A_n \omega_n) = \tr(\id^{\otimes \alpha n} \otimes \tilde{A}_n) \hat{S}_n(\omega_n))  \leq 2^{- \mu n}$, for a positive number $\mu$ and every sequence of separable states $\{ \omega_n \}_{n \in \mathbb{N}}$.  

Applying Theorem \ref{expdefinetti} with $k = \alpha n$ and $r = \beta n$ to $\tr_{1,...,\alpha n}(\hat{S}_n(\omega_n))$, we find that there is a probability measure $\nu$ such that
\begin{equation} \label{neweqcorr} 
\tr_{1,...,\alpha n}(\hat{S}_n(\omega_n)) = \int_{\sigma \in D({\cal H})} \int_{\ket{\theta} \supset \sigma} \nu(d\ket{\theta}) \pi_{n}^{\ket{\theta}} + X_n, 
\end{equation}
where $|| X_n ||_1 \leq 2^{\alpha \beta n}$,
\begin{equation*}
\pi_n^{\ket{\theta}} := \tr_{E}\left( \ket{\psi^{\ket{\theta}}_{(1-\alpha)n}} \bra{\psi^{\ket{\theta}}_{(1-\alpha)n}} \right),
\end{equation*}
and $\ket{\psi^{\ket{\theta}}_{(1-\alpha)n}} \in \ket{\theta}^{[\otimes, (1-\alpha)n, \beta n]}$. 

In the next paragraphs we show that only an exponentially small portion of the volume of $\nu$ is on a neighborhood of purifications of $\rho$. To this aim we employ Theorem A.1 of Ref. \cite{KR05}, which for our proposes can be stated as follows. Consider the random process given by the probability distribution
\begin{equation*}
p(i_1,...,i_{(1 - \alpha)n}) := \tr(\id^{\otimes \alpha n} \otimes M_{k_{1}}\otimes M_{k_{2}} \otimes ... \otimes M_{k_{(1 - \alpha)n}}\hat{S}_n(\omega_n)),
\end{equation*}
i.e. the random process associated to the measurement of the POVM $\{ M_k \}$ in each of the $(1 - \alpha)n$ parties of $\hat{S}_n(\omega_n)$. Let us define the post-selected states
\begin{equation} \label{definet}
\pi_{i_1,...,i_{(1 - \alpha)n}} := \frac{\tr_{\backslash 1}(\id^{\otimes \alpha n} \otimes  M_{k_{1}}\otimes M_{k_{2}} \otimes ... \otimes M_{k_{(1 - \alpha)n}}\hat{S}_n(\omega_n))}{\tr(\id^{\otimes \alpha n} \otimes M_{k_{1}}\otimes M_{k_{2}} \otimes ... \otimes M_{k_{(1 - \alpha)n}}\hat{S}_n(\omega_n))}
\end{equation}
and let $L_n^{i_1,...,i_{(1 - \alpha)n}}$ be the estimated state when the sequence of outcome $\{i_1,...,i_{(1 - \alpha)n}\}$ is obtained. Define ${\cal R}$ as the set of all outcome sequences such that 
\begin{equation*}
|| L_n^{i_1,...,i_{(1 - \alpha)n}} - \pi_{i_1,...,i_{(1 - \alpha)n}} ||_1 \geq \delta.
\end{equation*}
Then there is a $M > 0$ (only depending on the dimension of ${\cal H}$) such that \cite{KR05}
\begin{equation} \label{KRen}
\sum_{(i_1,...,i_{(1 - \alpha)n - 1}) \in {\cal R}} p(i_1,...,i_{(1 - \alpha)n - 1}) \leq  2^{-M(1 - \alpha)n\delta^2}.
\end{equation}

Since we are measuring local POVMs, the operation $\pi \mapsto \tr_{\backslash 1}(\hat{S}_n(\pi) \id^{\otimes \alpha n} \otimes \tilde{A}_n)$ is a sthocastic LOCC map. It hence follows from Eq. (\ref{neweqcorr}) that 
\begin{eqnarray} \label{eqcor1}
\tr_{\backslash 1}(\hat{S}_n(\omega_n) \id \otimes \tilde{A}_n) &=& \int_{\sigma \in B_{2\epsilon}(\rho)} \int_{\ket{\theta} \supset \sigma} \nu(d\ket{\theta}) \tr_{\backslash 1}(\pi_{n}^{\ket{\theta}}\id \otimes \tilde{A}_n) \nonumber \\ &+& \int_{\sigma \in \notin B_{2\epsilon}(\rho)} \int_{\ket{\theta} \supset \sigma} \nu(d\ket{\theta}) \tr_{\backslash 1}(\pi_{n}^{\ket{\theta}}\id \otimes \tilde{A}_n) \nonumber \\ &+& \tr_{\backslash 1}(X_n\id \otimes \tilde{A}_n) \in \text{cone}({\cal S}).
\end{eqnarray}
As $|| X_n || \leq 2^{-\alpha \beta n}$, we find $|| \tr_{\backslash 1}(X_n\id \otimes \tilde{A}_n) ||_1 \leq 2^{-\alpha \beta n}$. Furthermore, from Lemma \ref{hof} of chapter \ref{entanglement}, 
\begin{equation*}
|| \tr_{\backslash 1}(\pi_{n}^{\ket{\theta}}\id \otimes \tilde{A}_n) ||_1 = \tr(\pi_{n}^{\ket{\theta}}\id \otimes \tilde{A}_n) \leq n^{d^2}2^{-(\epsilon - h(\beta))(1 - \alpha)n}.
\end{equation*}
if $\tr_E(\ket{\theta}\bra{\theta}) \notin B_{2\epsilon}(\rho)$. Thus
\begin{eqnarray*}
\tr_{\backslash 1}(\hat{S}_n(\omega_n) \id \otimes \tilde{A}_n) = \int_{\sigma \in B_{2\epsilon}(\rho)} \int_{\ket{\theta} \supset \sigma} \nu(d\ket{\theta}) \tr_{\backslash 1}(\pi_{n}^{\ket{\theta}}\id \otimes \tilde{A}_n) + \tilde{X}_n \in \text{cone}({\cal S}).
\end{eqnarray*}
with $\tilde{X_n}$ given by the sum of the two last terms in Eq. (\ref{eqcor1}), which satisfies $|| \tilde{X}_n ||_1 \leq 2^{-\alpha \beta n} + n^{d^2}2^{-(\epsilon - h(\beta))(1 - \alpha)n}$.

For each $\tr_{\backslash 1}(\pi_{n}^{\ket{\theta}}\id \otimes \tilde{A}_n)$, with $\tr_E(\ket{\theta}\bra{\theta}) \in B_{2\epsilon}(\rho)$, we can write 
\begin{equation*}
\tr_{\backslash 1}(\pi_{n}^{\ket{\theta}}\id \otimes \tilde{A}_n) = \tr_{\backslash 1}(\pi_{n}^{\ket{\theta}}\id \otimes B_n) + \tr_{\backslash 1}(\pi_{n}^{\ket{\theta}}\id \otimes (\tilde{A}_n - B_n)),
\end{equation*}
where $B_n$ is the sum of the POVM elements products for which the post-selected state is $\delta$-close from the empirical state. From Eq. (\ref{KRen}) we find $\tr(\pi_{n}^{\ket{\theta}} \id \otimes (\tilde{A_n} - B_n)) \leq 2^{-M (1 - \alpha)\delta^2 n}$. Therefore, 
\begin{eqnarray} \label{hatx}
\tr_{\backslash 1}(\hat{S}_n(\omega_n) \id \otimes \tilde{A}_n) &=& \int_{\sigma \in D({\cal H})} \int_{\ket{\theta} \supset \sigma \in B_{2\epsilon}(\rho)} \nu(d\ket{\theta}) \tr(\pi_{n}^{\ket{\theta}}\id \otimes B_n)\rho^{\ket{\theta}} \nonumber \\ &+& \hat{X}_n \in \text{cone}({\cal S}).
\end{eqnarray}
where $\hat{X}_n$ is such that $|| \hat{X}_n ||_1 \leq 2^{-\alpha \beta n} + n^{d^2}2^{-(\epsilon - h(\beta))(1 - \alpha)n} + 2^{-M (1 - \alpha)\delta^2 n}$ and
\begin{equation*}
\rho^{\ket{\theta}} := \frac{\tr_{\backslash 1}(\pi_{n}^{\ket{\theta}} \id \otimes B_n)}{\tr(\pi_{n}^{\ket{\theta}}\id \otimes B_n)}.
\end{equation*}
Note that we have $|| \rho^{\ket{\theta}} - \rho|| \leq \delta$ for every $\rho^{\ket{\theta}}$ appearing in the integral of Eq. (\ref{hatx}). Define 
\begin{equation*}
\Lambda := \int_{\sigma \in D({\cal H})} \int_{\ket{\theta} \supset \sigma \in B_{2\epsilon}(\rho)} \nu(d\ket{\theta}) \tr(\pi_{n}^{\ket{\theta}}\id \otimes B_n). 
\end{equation*}
Then, 
\begin{equation} \label{wehereweused}
\left \Vert  \Lambda^{-1}\int_{\sigma \in D({\cal H})} \int_{\ket{\theta} \supset \sigma \in B_{2\epsilon}(\rho)} \nu(d\ket{\theta}) \tr(\pi_{n}^{\ket{\theta}}\id \otimes B_n)\rho^{\ket{\theta}} - \rho \right \Vert \leq \delta, 
\end{equation}
From Eqs. (\ref{epsfar}) and (\ref{wehereweused}) it follows that $\Lambda^{-1}\int_{\sigma \in D({\cal H})} \int_{\ket{\theta} \supset \sigma \in B_{2\epsilon}(\rho)} \nu(d\ket{\theta}) \tr(\pi_{n}^{\ket{\theta}}\id \otimes B_n)\rho^{\ket{\theta}}$ is at least $\epsilon - \delta$ far away from the separable states set. Using Eq. (\ref{hatx}) we thus find that 
\begin{equation*}
\Lambda \leq (\epsilon - \delta)^{-1}(2^{-\alpha \beta n} + n^{d^2}2^{-(\epsilon - h(\beta))n} + n2^{-((1 - \alpha)n - 1)\delta^2M^{-2}}),
\end{equation*}
With this bound we finally see that
\begin{eqnarray*}
\tr(\omega_n A_n) &=& \tr(\hat{S}_n(\omega_n) \id \otimes \tilde{A}_n) \nonumber \\ &=& \Lambda + \tr(\hat{X}) \nonumber \\ &\leq& (1 + (\epsilon - \delta)^{-1})(2^{-\alpha \beta n} + n^{d^2}2^{-(\epsilon - h(\beta))n} + n2^{-((1 - \alpha)n - 1)\delta^2M^{-2}}) \nonumber \\ &\leq& 2^{- \mu n},
\end{eqnarray*}
for appropriately chosen $\alpha, \beta \in [0, 1]$ and $\mu > 0$.
\end{proof}
\vspace{0.3 cm}

In the proof above the only property of the set of separable states that we used, apart from the five properties required for Theorem \ref{maintheorem} to hold, was its closedness under SLOCC. It is an interesting question if such a property is really needed, or if actually the positiveness of the rate function is a generic property of any $\rho \notin {\cal M}$ for every family of sets satisfying Theorem \ref{maintheorem}. The following example shows that for some choices of sets $\{ {\cal M}_k \}$, the rate function can be zero for a state $\rho \notin {\cal M}$. In fact, in our example the rate function is zero for every state.   

A bipartite state $\sigma_{AB}$ is called $n$-extendible if there is a state $\tilde{\sigma}_{AB_1...B_n}$ symmetric under the permutation of the $B$ systems and such that $\tr_{B_2,...,B_n}(\tilde{\sigma}) = \sigma$. Let us denote the set of $n$-extendible states acting on ${\cal H} = {\cal H}_A \otimes {\cal H}_B$ by ${\cal E}_k({\cal H})$. It is clear that the sets $\{ {\cal E}_k({\cal H}^{\otimes n}) \}_{n \in \mathbb{N}}$ satisfy conditions \ref{cond1}-\ref{cond5} and therefore we can apply Theorem \ref{maintheorem} to them. Corollary \ref{faithful} however does not hold in this case, as the sets are not closed under two-way LOCC, even though they are closed under one-way LOCC. In fact, the statement of the corollary fails dramatically in this case as it turns out that the measures $E_{{\cal E}_k}^{\infty}$ are zero for every state. This can be seen as follows: Given a state $\rho$, let us form the $k$-extendible state
\begin{equation*}
\tilde{\rho}_{AB_1,...,B_k} := \id_A \otimes \hat{S}_{B_1,...,B_k}\left(\rho_{AB} \otimes \left(\frac{\id}{d^2}\right)^{\otimes k - 1} \right)
\end{equation*}
We have $\tilde{\rho}_{AB_1,...,B_k} \geq \rho_{AB} \otimes \frac{\id}{d^2}^{\otimes k - 1}/k$. Then, from the operator monotonicity of the $\log$,
\begin{equation*}
E_{{\cal E}_k}(\rho) \leq S(\rho || \tr_{B_2,...,B_n}(\tilde{\rho})) \leq k. 
\end{equation*} 
As the upper bound above is independent of $n$, we then find
\begin{equation*} 
E_{{\cal E}_k}^{\infty}(\rho) = \lim_{n \rightarrow \infty} \frac{1}{n} E_{{\cal E}_k}^{\infty}(\rho^{\otimes n}) \leq \lim_{n \rightarrow \infty} \frac{k}{n} = 0. 
\end{equation*}

Note that as ${\cal E}_1$ is contained in the set of one-way undistillable states ${\cal C}_{\text{one-way}}$, the same is true for $E_{{\cal C}_{\text{one-way}}}^{\infty}$, i.e. it is identically zero. It is interesting that an one-way distillable state cannot be distinguished with an exponential decreasing probability of error from one-way undistillable states if we allow these to be correlated among several copies, while any entangled state can be distinguished from arbitrary sequences of separable states with exponentially good accuracy. Moreover, as the set of PPT states satisfy conditions \ref{cond1}-\ref{cond5} and is closed under SLOCC, every NPPT state can be exponentially well distinguished from a sequence of PPT states. It is an intriguing open question if the same holds for distinguishing a two-way distillable state from a sequence of two-way undistillable states. Due to the conjecture existence of NPPT bound entanglement, property \ref{cond4} might fail and therefore we do not know what happens in this case.

\section{Constrained Hypothesis Testing} \label{constraint}

Up to now we have assumed that the observer has access to the most general type of measurements allowed by quantum mechanics. However, in many situations there are further constraints that reduces the class of available measurements. This might happen due to experimental limitations or even due to basic physical constraints such as super-selection rules and locality. Indeed, LOCC discrimination procedures play a fundamental role in the distant lab paradigm of quantum information theory (see e.g. \cite{BDF+99, WSHV00, Rai01, VSPM01, VP03, HOSS04, Win05, HMM06, HMT06, MW07} and references therein).  

It turns out that the characterization of LOCC POVMs is a formidable task, and only little is presently known about it. In this respect it is helpful to consider larger classes of POVMs, which although cannot always be implemented by LOCC, approximates the class of LOCC and hence put locality constraints on the type of measurements available. Two such classes are the separable POVMs, composed of non-normalized separable states as POVM elements, and the PPT POVMs, in which each POVM element is PPT \cite{BDF+99, Rai01, VP03}. 

Here we would like to comment on a possible generalization of Theorem \ref{maintheorem} to the case of constrained POVMs. Instead of considering the condition $0 \leq A_n \leq \id$ for the POVM elements, we will consider a constraint of the form $0 \leq_{{\cal P}} A_n \leq \id$, where $\leq_{\cal P}$ is a generalized inequality characterizing a particular type of constraint. For example, choosing ${\cal P}$ to be the set of PPT positive operators and non-normalized separable states we find that $A_n$ is a PPT and separable POVM element, respectively. Note that we could demand more and set the POVM elements to be such that $0 \leq_{{\cal P}} A_n \leq_{{\cal P}} \id$. Operationally this means that we require that the two POVM elements $A_n$ and $\id - A_n$ satisfy the constraint given by ${\cal P}$. We will not consider this more demanding setting, however, as the technique we suggest in the sequel does not seem to be strong enough to deal with it. 

We conjecture that the following is true.

\begin{conjecture}  \label{conj}
Given a family of sets  $\{ {\cal M}_n \}_{n \in \mathbb{N}}$ satisfying properties \ref{cond1}-\ref{cond5}, a proper convex cone ${\cal P}$, and a state $\rho \in {\cal D}({\cal H})$, for every $\epsilon > 0$ there exists a sequence of ${\cal P}$-constrained POVM elements $0 \leq_{{\cal P}} A_n \leq \id$ such that
\begin{equation*}
\lim_{n \rightarrow \infty} \tr((\id - A_n) \rho^{\otimes n}) = 0 
\end{equation*}
and for all sequences of states $\{ \omega_n \in {\cal M}_n \}_{n \in \mathbb{N}}$,
\begin{equation*}
- \frac{\log \tr(A_n \omega_n)}{n} + \epsilon \geq \lim_{\epsilon \rightarrow 0} \limsup_{n \rightarrow \infty} \frac{1}{n} LR_{{\cal M}_n, {\cal P}}^{\epsilon}(\rho^{\otimes n}),
\end{equation*}
where
\begin{equation*}
LR_{{\cal M}_n, {\cal P}}^{\epsilon}(\rho^{\otimes n}) := \min_{\rho_{n, \epsilon} \in B_{\epsilon}(\rho^{\otimes n})}\min_{\sigma \in {\cal M}_n} S_{\max}^{\cal P}(\rho_{n, \epsilon} || \sigma),
\end{equation*}
with
\begin{equation*}
S_{\max}^{\cal P}(\rho || \sigma) := \min \{ s : \rho \leq_{{\cal P}^*} 2^s \sigma   \}. 
\end{equation*}
Conversely, if there is a $\epsilon > 0$ and sequence of ${\cal P}$-constrained POVM elements $0 \leq_{{\cal P}} A_n \leq \id$ satisfying 
\begin{equation*}
 - \frac{\log( \tr(A_n \omega_n))}{n} - \epsilon \geq \lim_{\epsilon \rightarrow 0} \limsup_{n \rightarrow \infty} \frac{1}{n} LR_{{\cal M}_n, {\cal P}}^{\epsilon}(\rho^{\otimes n})
\end{equation*}
for all sequences $\{ \omega_n \in {\cal M}_n \}_{n \in \mathbb{N}}$, then
\begin{equation*}
\lim_{n \rightarrow \infty} \tr((\id - A_n) \rho^{\otimes n}) = 1. 
\end{equation*}
\end{conjecture}

Note that in the conjecture ${\cal P}^*$ is the dual cone of ${\cal P}$. For example, if we take ${\cal P}$ to be the set of separable states ${\cal S}$, ${\cal S}^*$ is the set of entanglement witnesses, while if ${\cal P}$ is taken to be the set of PPT operators, then the dual cone is the set of operators of the form $P + Q^{\Gamma}$, for any two positive semidefinite operators $P, Q$ \cite{LKCH00}.
                        
Although we cannot prove the conjecture in full, based on the findings of this chapter we believe it should be true at least for some particular choices of ${\cal P}$. Indeed, it is possible to establish that
\begin{equation*}
LR_{{\cal M}, {\cal P}}^{\infty}(\rho) := \lim_{\epsilon \rightarrow 0} \limsup_{n \rightarrow \infty} \frac{1}{n} LR_{{\cal M}_n, {\cal P}}^{\epsilon}(\rho^{\otimes n})
\end{equation*}
gives the strong converse rate, i.e. that the second half of the conjecture holds true.

To see this let us define the quantity
\begin{equation}
\tr(X)_{+, {\cal P}} := \max_{0 \leq_{\cal P} A \leq \id} \tr(A X).
\end{equation}
Then following the first part of the proof of Theorem \ref{maintheorem}, we find that if $y = LR_{{\cal M}, {\cal P}}^{\infty}(\rho) + \epsilon$, 
\begin{equation*}
\lim_{n \rightarrow \infty} \min_{\sigma_n \in {\cal M}_n} \tr(\rho^{\otimes n} -2^{yn}\sigma_n)_{+, {\cal P}} = 0,
\end{equation*}
for every $\epsilon > 0$. This equation then readily implies the second part of the conjecture, i.e. the impossibility of tests with a better rate than $LR_{{\cal M}, {\cal P}}^{\infty}(\rho)$. 

The direct part cannot be easily obtained by adapting the proof of Theorem \ref{maintheorem} from the cone of positive semidefinite operators to ${\cal P}$, as many properties that are particular to the former, such as the fact that it is self-dual\footnote{A set ${\cal M}$ is self-dual if ${\cal M}^{*} = {\cal M}$.}, are employed. It is an interesting open question to find out if conjecture \ref{conj} holds true in general, or at least for some particular choices of ${\cal P}$, such as the set of PPT and separable POVM elements.

\chapter{A Reversible Theory of Entanglement} \label{reversible}

\section{Introduction}

A basic feature of many physical settings is the existence of constraints on physical operations and processes that are 
available. These restrictions generally imply the existence of resources that can be consumed to overcome the constraints. 
Examples include an auxiliary heat bath in order to decrease the entropy of an isolated thermodynamical system \cite{Cal85} 
or prior secret correlations for the establishment of secret key between two parties who can only operate locally and communicate by a public
channel \cite{Mau98}. As discussed in chapter \ref{entanglement}, in quantum information theory one often considers the scenario in which two or more distant parties want to exchange quantum information, but are restricted to act locally on their quantum systems and 
communicate classical bits. We found in this context that a resource of intrinsic quantum character, entanglement, allows the parties to completely overcome the limitations caused by the locality requirement on the quantum operations available. 

Resource theories are considered in order to determine when a physical system, or a state thereof, contains a given resource; to characterize the possible conversions from a state to another when one has access only to a restricted class of operations which cannot create the resource for free; and to quantify the amount of such a resource contained in a given system. 

One may try to answer the above questions at the level of 
individual systems, which is usually the situation encountered in experiments. However it is natural to expect that a 
simplified theory will emerge when instead one looks at the bulk properties of a large number of 
systems. The most successful example of such a theory is arguably 
thermodynamics. This theory was initially envisioned to describe the physics of large systems 
in equilibrium, determining their bulk properties by a very simple 
set of rules of universal character. This was reflected in 
the formulation of the defining axiom of thermodynamics, the 
second law, by Clausius \cite{Cla1850}, Kelvin \cite{Kel1849} and Planck \cite{Pla1900} in terms of quasi-static 
processes and heat exchange. However, the apparently universal 
applicability of thermodynamics suggested a deeper mathematical 
and structural foundation. Indeed, there is a long history of
examinations of the foundations underlying the second law, starting with Carath\'eodory work 
in the beginning of last century \cite{Car1909}. Of particular interest in the present context is the work
of Giles \cite{Gil64} and notably Lieb and Yngvason \cite{LY99} 
stating that there exists a total ordering of equilibrium 
thermodynamical states that determines which state transformations 
are possible by means of an adiabatic process\footnote{The very concept of adiabatic processes 
already deserves a more precise definition. Usually we refer to a process as adiabatic when no 
exchange of heat from the system to its environment is involved. This is not completely satisfactory as we need the concept of heat here, which again lacks an unambiguous definition. Following Lieb and Yngvason \cite{LY99}, we will call a process adiabatic if after its execution the only change in the environment is that a weight has been lifted or lowered}. In this sense, 
thermodynamics can be seen as a resource theory of \textit{order}, dictating 
which transformations are possible between systems with different 
amounts of (dis)order by operations which cannot order out systems 
(adiabatic processes).

A remarkable aspect of these foundational 
works \cite{Gil64, LY99} is that from simple, abstract, axioms they were able to show the existence of an \textit{entropy} function 
$S$ fully determining the achievable transformations by adiabatic processes: given two equilibrium states $A$ and $B$, $A$ can be 
converted by an adiabatic process into $B$ if, and only if, 
$S(A) \leq S(B)$. As pointed out by Lieb and Yngvason\cite{LY02},
it is a strength of this abstract approach that it allows 
to uncover a thermodynamical structure in settings that 
may at first appear unrelated. 

Early studies in quantum information indicated that entanglement theory could be one such setting. 
Possible connections between entanglement theory and 
thermodynamics were first noted earlier on when it was found that for 
bipartite pure states a very similar situation to the 
second law holds in the asymptotic limit of an arbitrarily 
large number of identical copies of the state. As discussed 
in section \ref{costdistillable}, given two bipartite pure states $\ket{\psi_{AB}}$ and 
$\ket{\phi_{AB}}$, the former can be converted into the 
latter by LOCC if, and only if, $E(\ket{\psi_{AB}}) \geq E(\ket{\phi_{AB}})$, 
where $E$ is the entropy of entanglement \cite{BBPS96}.

However, as we have already seen, for mixed entangled states there are bound
entangled states that require a non-zero 
rate of pure state entanglement for their creation by LOCC, 
but from which no pure state entanglement can be extracted 
at all \cite{HHH98}. As a consequence 
no unique measure of entanglement exists in the general case 
and no unambiguous and rigorous direct connection to thermodynamics 
appeared possible, despite various interesting 
attempts which we review in the next section.

In this chapter we introduce a class of maps that can be identified as the counterpart of adiabatic processes. This, together 
with the technical tools developed in chapter \ref{QHT}, will allow
us to establish a theorem completely analogous to the 
Lieb and Yngvason formulation of the second law of thermodynamics for entanglement manipulation.
Similar considerations may also be applied to
resource theories in general, including e.g. theories 
quantifying the non-classicality, the non-Gaussian 
character, and the non-locality of quantum states. Indeed, 
the main conceptual message of this chapter is an unified 
approach to deal with resources manipulation, by abstracting 
from theories where a resource appears due to the existence 
of some constraint on the type of operations available, to 
theories where the class of operations is derived
from the resource considered, and chosen to be
such that the latter cannot be freely generated. As 
it will be shown, such a shift of focus leads to a 
much simpler and elegant theory, which at the same
time still gives relevant information about the 
original setting. 

The structure of this chapter is the following. In section \ref{qnslet} we review qualitative analogies of the second law in entanglement theory and define the class of non-entangling maps, while in subsection \ref{related} we comment on previous related work. In section \ref{defmain5}, in turn, we present some definitions and the main results of this chapter. Section \ref{proofmaint5} is devoted to the Proof of Theorem \ref{maintheorem5} and Corollary \ref{corref}. We revisit the choice of the operations employed in section \ref{mustcan}, and derive the single-copy cost under non-entangling maps. Finally, in section \ref{cafslt} we discuss the connection of our framework to works on the foundations of thermodynamics, more specifically to the axiomatic approach of Lieb and Yngvason to the second law.

\section{A Qualitative Analogue of the Second Law for Entanglement Theory} \label{qnslet}

Studies on the connections of entanglement theory and 
thermodynamics date back to the earlier foundational works on the subject \cite{PR97, HH98, HHH98b, PV98}. 
There it was noted that the basic postulates of quantum mechanics and the 
definition of entangled states imply that 
\begin{center}
\textbf{(a)}  \textit{entanglement cannot be created by local operations 
and classical communication}.
\end{center}
It was argued that this is a basic law of quantum information 
processing and can be seen as a weak qualitative analogue of 
the second law of thermodynamics, once we make the identification of entanglement with 
order and of LOCC maps with adiabatic processes. 

Local operations and classical communication are the fundamental 
class of operations to be considered in the distant lab paradigm, 
for which the definition of entanglement emerges most naturally. 
However, in view of principle (a) it is important to note that LOCC 
is not the largest class that cannot generate entanglement out of separable states. Consider, for instance, the 
class of separable operations, introduced in section \ref{boundentanglement}. While it is clear that a separable map cannot generate entanglement, 
it turns out that there are separable operations which cannot be implemented by LOCC \cite{BDF+99}.  

We might now wonder if separable maps 
are the largest class of quantum operations which cannot create entanglement. As proven 
in Ref. \cite{CDKL01}, this is indeed the case if we allow the 
use of ancillas. That is, if we require that $\Omega \otimes \id_d$, 
where $\id_d$ is a identity map which is applied to a $d$-dimensional ancilla state, 
does not generate entanglement for an arbitrary $d$, then $\Omega$ must be a separable
superoperator. However, for the following it will be important to note that there is yet 
a larger class of operations for which no entanglement can be generated
if we do not require that our class of quantum maps is closed 
under tensoring with the identity as above. 
 
\begin{definition}
Let $\Omega: {\cal D}(\mathbb{C}^{d_1} \otimes ... \otimes \mathbb{C}^{d_m}) \rightarrow 
{\cal D}(\mathbb{C}^{d'_1} \otimes ... \otimes \mathbb{C}^{d'_m})$ be a quantum operation. 
We say that $\Omega$ is a separability-preserving or a non-entangling 
map if for every fully separable state $\sigma \in {\cal D}(\mathbb{C}^{d_1} 
\otimes ... \otimes \mathbb{C}^{d_m})$\footnote{ \normalfont A multipartite state $\sigma$ is fully separable if it can be written as 
\begin{equation*}
\sigma = \sum_i p_i \sigma_i^1 \otimes ... \otimes  \sigma_i^n,
\end{equation*}
for local states $\{ \sigma_i^k \}$ and a probability distribution $\{ p_i \}$.}, $\Omega(\sigma)$ is a fully separable state. We 
denote the class of such maps by $SEPP$.
\end{definition}

It is clear from its very definition that $SEPP$ is the largest 
class of operations which cannot create entanglement. An example 
of a completely positive map which is separability-preserving, yet is 
not a separable operation is the swap operator. In fact the class $SEPP$
is even strictly larger than the convex hull of separable operations and the composition of 
separable operations with the swap operator \cite{VHP05}. 

There is a quantitative version of law (a), which roughly speaking, states that 
\begin{center}
\textbf{(b)}  \textit{entanglement cannot be increased by local operations 
and classical communication}.
\end{center}
Although (b) is clearly stronger than the first version discussed, 
it is not as fundamental as (a), since we must assume there is a
way to quantify entanglement, something that cannot be done in a 
completely unambiguous manner. Here we will focus on three specific 
entanglement quantifiers as the underlying quantitative notion of 
entanglement needed for (b): the robustnesses of entanglement and
the relative entropy of entanglement, discussed in sections \ref{robustnesses} and \ref{relentint}, 
respectively. There it was pointed out that all these three measures satisfy (b). 
 
The choices for these measures here comes from the fact that, using them 
to quantify entanglement, LOCC is again not the largest class of operations 
for which (b) is true. Indeed, for both the robustness, global robustness, 
and relative entropy of entanglement, non-entangling maps are once more the largest 
such class. 

At this point it is instructive to briefly look at the 
role of adiabatic operations in the formulation of the 
second law of thermodynamics. As it was said before, in 
its more recent formulation, the second law is formulated 
concisely as the characterization of state transformations 
that are possible by adiabatic processes. It follows both 
from the approach of Giles \cite{Gil64} and Lieb and 
Yngvason \cite{LY99} that 
\begin{itemize}
	\item the class formed by all 
adiabatic processes is the largest class of operations 
which cannot decrease the entropy of an isolated equilibrium 
system 
\end{itemize}

From this we can thus speculate that the counterpart in entanglement 
theory of an adiabatic process is not a LOCC map, but rather 
a non-entangling operation. In a sense this is indeed the case and will
be the driving motivation for the framework we develop in the sequel. However, there is still a last ingredient we must take into account
before defining the class of operations we will be working with. 

Once more, we can find the motivation from thermodynamics. There the need of the thermodynamical limit is fundamental. Indeed, for finite size 
systems there are fluctuations which may look like violating the second law (see 
e.g. \cite{WSM+02}), which shows that the concept of an adiabatic process is only meaningful the thermodynamical limit.

We can define an analogous definition in entanglement
theory and consider the class of asymptotically non-entangling maps. We 
define it precisely in section \ref{defmain5}, but here we would like to anticipate 
that this class is formed by sequences of maps $\{ \Lambda_n \}_{n \in \mathbb{N}}$, 
where each $\Lambda_n$ generates at most an $\epsilon_n$ amount of entanglement, and
such that $\epsilon_n$ goes to zero when $n$ grows. Note that this class is
the largest class that \textit{asymptotically} satisfies the basic law of the non-increase of entanglement. 

We argue that this type of 
operations should be regarded as the fully counterpart of adiabatic processes in
entanglement theory. The main result of this chapter 
is a theorem complete analogous to the formulation by Lieb and Yngvason \cite{LY99} of 
the second law of thermodynamics for entanglement transformations 
under asymptotically non-entangling maps, quantitatively supporting this claim. 

In turn, we also gain new insight into the irreversibility 
found under LOCC manipulations. It is because the class of 
non-entangling operations is strictly larger than LOCC that 
there is no total order for bipartite entanglement under 
local quantum processing. Finally, we identify the key role of the regularized relative entropy of 
entanglement as the unique entanglement measure in this setting. Before we present the main result of 
this chapter, in the next subsection we first review previous 
works related to the theme of this chapter.  

\subsection{Previous Work and Related Approaches} \label{related}

The phenomenon of bound entanglement has attracted 
considerable interest in recent years, and different possible 
explanations for it have been considered. It is tempting to 
think that it is the loss of classical information when going 
to mixed states that causes irreversibility (see e.g. 
\cite{HV00, EFP+00}). Nonetheless, arguments to the contrary 
were given in Refs. \cite{HHH+03, HHO03}, where it 
was shown that when one is concerned in concentrating pure 
states from operations in which no pure state ancilla can be added - a problem complementary to the distillation of 
entanglement - reversibility holds even for mixed states. 
Hence, it does not seem plausible to attribute irreversibility 
to the loss of classical information.  In fact, the key
is not the loss of information but the fact that the recovery
of this information via measurements may lead to irreversible
losses. This depends on the set of operations that is available
to recover the information. Indeed, the results of the
present chapter provide strong arguments supporting this viewpoint
as it shows that for entanglement itself, loss of information 
does not necessarily imply irreversibility.

In Ref. \cite{VK02}, the applicability of Giles axiomatic 
approach \cite{Gil64} to entanglement theory was studied. It 
was shown that for pure state bipartite entanglement the 
same axioms used in the derivation of the second law of 
thermodynamics hold true. Therefore, one can derive the 
uniqueness of the entropy of entanglement following the 
steps taken by Giles in the derivation of the entropy in 
the context of the second law \cite{Gil64}. One of Giles 
postulates is that if two states $A$ and $B$ are both 
adiabatic accessible from another state $C$, either $A$ 
is adiabatic accessible to $B$ or vice-versa (if not both) 
\cite{Gil64}. In Ref. \cite{MFV04}, it was pointed out 
that this property does not hold in asymptotic mixed state entanglement 
transformations under LOCC, showing the inapplicability of 
Giles approach in the mixed state scenario. 

Various approaches have been considered to enlarge the class of
operations in a way that could lead to reversibility of entanglement
manipulation under such a set of operations. Two closely related
but different routes have been taken here. 

A first approach was considered in \cite{Rai01, EVWW01, APE03}. 
There, entanglement manipulation was studied under the 
class of operations that maps every PPT state into another 
PPT state (including the use of ancillas)\footnote{This class of PPT operations was 
introduced by Rains in the seminal paper \cite{Rai01}.}. It was realized
in \cite{EVWW01} that every state with a non-positive partial 
transpose becomes distillable under PPT preserving operations (see 
section \ref{activation} for more on this point). This eliminates the phenomenon of bound entanglement in a qualitative
level thereby suggesting the possibility of reversibility in this setting.
This was taken as a motivation for further studies, e.g.
\cite{APE03}, where it was shown that under PPT
maps the antisymmetric states of arbitrary dimension can 
be reversibly interconverted into pure state entanglement, clearly 
showing a nontrivial example of mixed state reversibility. 
Unfortunately, no other example have been found so far 
and, hence, the reversibility under the class of PPT 
operations remains as an open question. Moreover, this 
approach suffers from the fundamental drawback of 
considering as a resource just a subset of the set of 
entangled states.

In a second approach one considers every PPT state as a free resource
in an LOCC protocol. Then again, every state with a 
non-positive partial transpose becomes distillable 
\cite{EVWW01}. However, in Ref. \cite{HOH02} 
it was shown, under some unproven but reasonable assumptions, that 
in this scenario one still has irreversibility.
   
The possibility of having reversible entanglement transformations under enlarged classes of operations was also 
analysed in Ref. \cite{HOH02}. In this work the authors considered the analogy \textit{entanglement-energy} in the relation of entanglement theory to thermodynamics, first raised in Refs. \cite{HH98, HHH98b}, complementary to the \textit{entanglement-entropy} analogy \cite{PR97, PV98} adopted here, 
to argue that a fully thermodynamical theory of entanglement could in principle be established even considering the existence of bound 
entanglement. However, as already mentioned, under some assumptions on the properties of an entanglement 
measure there defined, it was shown that one probably does not encounter exactly the
setting envisioned. Interestingly, it was proven that if one has reversibility under a class of operations that includes mixing, then the unique measure of entanglement governing state transformations is the regularized relative entropy from the 
set of states which are closed under the class of operations allowed. This result nicely fits into the findings of the paper: as the class 
of non-entangling maps includes mixing, we will indeed find the regularized relative entropy of entanglement as the unique entanglement quantifier.

\subsubsection{Nice Resources}

There is yet another line of research into which our framework is connected: the quest for identifying the \textit{nice resources}\footnote{The term nice resources is taken from Ref. \cite{Smi07}.} of quantum information theory, which allow 
for a simpler theory over the unassisted case. The idea here is not to consider what resources are useful from the point of view of information processing, but actually the ones that are nice in the sense of leading to a marked simplification in the theory under consideration.   

The first example of such a nice resource is unlimited entanglement between sender and receiver for communication over a noisy quantum channel. In has been proven in Refs. \cite{BSS+99, BSS+02} that this leads to a remarkably simple formula for the quantum and classical capacities (which are actually related by a factor of two), which in this case is single-letterized, meaning that no regularization is needed, and a direct generalization of Shannon's capacity formula for classical noisy channels.   

A more recent example is the use of symmetric side channels for sending quantum information. By the no cloning theorem \cite{Die82, WZ82} we know that it is not possible to reliably send quantum information through a channel which distributes the information symmetrically between the receiver and the environment. It has been shown in Refs. \cite{SSW06, Smi07} that nonetheless such channels are nice resources, as it is possible to derive a single-letter and convex expressions for the symmetric-side-channel-assisted quantum and private channel capacities. Such an approach has recently lead to an important breakthrough in quantum information theory, as it was used by Smith and Yard to show that the quantum channel capacity is not additive \cite{SY08}. 

A third example is of course the use of PPT operations and PPT states in entanglement theory, as discussed in the previous section and in section \ref{activation} of chapter \ref{entanglement}.

\section{Definitions and Main Results} \label{defmain5}

As mentioned before, in this chapter we will be interested in entanglement manipulation under operations that asymptotically cannot 
generate entanglement. The next definition determines precisely the class of maps employed. 

\begin{definition}\label{epsilonSEPPreserving}
Let $\Omega: {\cal D}(\mathbb{C}^{d_1} \otimes ... \otimes \mathbb{C}^{d_m}) 
\rightarrow {\cal D}(\mathbb{C}^{d'_1} \otimes ... \otimes \mathbb{C}^{d'_m})$ 
be a quantum operation. We say that $\Omega$ is an 
$\epsilon$-non-entangling (or $\epsilon$-separability-preserving) map if for every 
separable state $\sigma \in {\cal D}(\mathbb{C}^{d_1} \otimes ... \otimes \mathbb{C}^{d_m})$, 
\begin{equation} \label{3456}
R_{G}(\Omega(\sigma)) \leq \epsilon,
\end{equation}
where $R_G$ is the global robustness of entanglement with respect to fully separable states (see section \ref{robustnesses}). We denote the set of $\epsilon$-non-entangling 
maps by $SEPP(\epsilon)$. 
\end{definition}

With this definition an asymptotically non-entangling operation is given by a sequence of CP trace preserving maps 
$\{ \Lambda_n \}_{n \in \mathbb{N}}$, $\Lambda_n : {\cal D}((\mathbb{C}^{d_1} \otimes ... \otimes \mathbb{C}^{d_m})^{\otimes n}) \rightarrow {\cal D}((\mathbb{C}^{d'_1} \otimes ... \otimes \mathbb{C}^{d'_m})^{\otimes n})$, 
such that each $\Lambda_n$ is $\epsilon_n$-non-entangling and $\lim_{n \rightarrow \infty}\epsilon_n = 0$. 

The use of the global robustness to measure of the amount of entanglement generated is not arbitrary. The reason for this choice will be explained in section \ref{mustcan}.

Having defined the class of maps we are going to use to manipulate entanglement, we can define the cost and distillation functions, in analogy to the definitions of section \ref{costdistillable} for the LOCC case. For simplicity of notation we set $\phi_2 := \Phi(2)$.

\begin{definition}
We define the entanglement cost under asymptotically non-entangling maps of a state $\rho \in {\cal D}(\mathbb{C}^{d_1}\otimes ... \otimes \mathbb{C}^{d_m})$ as
\begin{eqnarray} \label{costane}
E_{C}^{ane}(\rho) := \inf_{ \{ k_n, \epsilon_n \} } \left \{ \limsup_{n \rightarrow \infty} \frac{k_n}{n}  :  \lim_{n \rightarrow \infty} \left( \min_{\Lambda \in SEPP(\epsilon_n)} || \rho^{\otimes n} - \Lambda(\phi_2^{\otimes k_n})||_1 \right) = 0, \hspace{0.1 cm}  \lim_{n \rightarrow \infty} \epsilon_n = 0  \right \}, \nonumber
\end{eqnarray}
where the infimum is taken over all sequences of integers $\{ k_n \}$ and real numbers $\{ \epsilon_n \}$. In the formula above $\phi_2^{\otimes k_n}$ stands for $k_n$ copies of a two-dimensional maximally entangled state shared by the first two parties and the maps $\Lambda_n : {\cal D}((\mathbb{C}^{2}\otimes \mathbb{C}^{2})^{\otimes k_n}) \rightarrow {\cal D}((\mathbb{C}^{d_1}\otimes ... \otimes \mathbb{C}^{d_m})^{\otimes n})$ are $\epsilon_n$-non-entangling operations. 
\end{definition}

\begin{definition} \label{distane}
We define the distillable entanglement under asymptotically non-entangling maps of a state $\rho \in {\cal D}(\mathbb{C}^{d_1}\otimes ... \otimes \mathbb{C}^{d_m})$ as 
\begin{eqnarray}
E_{D}^{ane}(\rho) := \sup_{\{ k_n, \epsilon_n \} } \left \{ \liminf_{n \rightarrow \infty} \frac{k_n}{n} : \lim_{n \rightarrow \infty} \left( \min_{\Lambda \in SEPP(\epsilon_n} || \Lambda(\rho^{\otimes n}) - \phi_2^{\otimes k_n}||_1 \right) = 0, \hspace{0.1 cm} \lim_{n \rightarrow \infty} \epsilon_n = 0  \right \}, \nonumber
\end{eqnarray}
where the infimum is taken over all sequences of integers $\{ k_n \}$ and real numbers $\{ \epsilon_n \}$. In the formula above $\phi_2^{\otimes k_n}$ stands for $k_n$ copies of a two-dimensional maximally entangled state shared by the first two parties and the maps $\Lambda_n : {\cal D}((\mathbb{C}^{d_1}\otimes ... \otimes \mathbb{C}^{d_m})^{\otimes n}) \rightarrow {\cal D}((\mathbb{C}^{2}\otimes \mathbb{C}^{2})^{\otimes k_n})$ are $\epsilon_n$-non-entangling operations. 
\end{definition}

Note that when we do not specify the state of the other parties we mean that their state is trivial (one-dimensional). Note furthermore that the fact that initially only two parties share entanglement is not a problem as the class of operations we employ include the swap operation. We are now in position to state the main result of this chapter.

\begin{theorem} \label{maintheorem5}
For every multipartite state $\rho \in {\cal D}(\mathbb{C}^{d_1} \otimes ... \otimes \mathbb{C}^{d_m})$,
\begin{equation}
E_C^{ane}(\rho) = E_D^{ane}(\rho) = E_R^{\infty}(\rho).
\end{equation}
\end{theorem}

We thus find that under asymptotically 
non-entangling operations entanglement can be interconverted 
reversibly. The situation is analogous to the Giles-Lieb-Yngvason 
formulation of the second law of thermodynamics. Indeed, Theorem \ref{maintheorem5} 
readily implies

\begin{corollary} \label{corref}
For two multipartite states $\rho \in {\cal D}(\mathbb{C}^{d_1} \otimes ... \otimes \mathbb{C}^{d_m})$ and  $\sigma \in {\cal D}(\mathbb{C}^{d'_1} \otimes ... \otimes \mathbb{C}^{d'_{m'}})$, there is a sequence of quantum operations $\Lambda_n$ such that
\begin{eqnarray}
	\Lambda_n \in SEPP(\epsilon_n),  \hspace{1.5 cm} \lim_{n \rightarrow \infty}\epsilon_n = 0, 
\end{eqnarray}
and\footnote{ \normalfont In Eq. (\ref{sublinmain})  $o(n)$ stands for a sublinear term in $n$.} 
\begin{eqnarray}\label{sublinmain}	
\lim_{n \rightarrow \infty} || \Lambda_n(\rho^{\otimes n}) - \sigma^{\otimes n - o(n)} ||_1 = 0
\end{eqnarray}
if, and only if, 
\begin{equation}
E_R^{\infty}(\rho) \geq E_R^{\infty}(\sigma).
\end{equation}
\end{corollary}

In addition we have also identified the regularized relative entropy of 
entanglement as the natural counterpart in entanglement theory to the entropy function 
in thermodynamics. As shown in chapter \ref{QHT}, the regularized relative
entropy measures how distinguishable an entangled state is 
from an arbitrary sequence of separable states. Therefore, 
we see that under asymptotically non-entangling operations, 
the amount of entanglement of any multipartite state is completely 
determined by how distinguishable the latter is from a state 
that only contains classical correlations. Furthermore, in 
Corollary \ref{LGmeasure} of chapter \ref{QHT} we saw that 
\begin{equation} \label{LG5}
LG(\rho) := \inf_{\{ \epsilon_n\}} \left \{ \limsup_{n \rightarrow \infty} \frac{1}{n} LR_{G}^{\epsilon_n}(\rho^{\otimes n}) : \lim_{n \rightarrow \infty} \epsilon_n = 0 \right \} = E_R^{\infty}(\rho). 
\end{equation}
Hence we also find that the amount of entanglement is equivalently uniquely 
defined in terms of the robustness of quantum correlations 
to noise in the form of mixing.   

The approach of considering the manipulation of a given resource 
under the largest class of operations that cannot create it is appealing and can 
also be applied to other resources \cite{BP08b}. In fact, although for 
concreteness we present the proof specifically to entanglement as resource, much of it does not concern only entanglement theory, but can be applied 
to rather general resource theories. 

An interesting question is if the second law of thermodynamics can also be derived from such considerations. Although to derive the second law in its full generality one of course needs further assumptions, a simplified version of it can indeed be obtained. Consider the von Neumann entropy and the class of operations that do not decrease it, which are the well studied doubly sthocastic maps. We can study the manipulation of non-maximally mixed states, which are the resource states that cannot be created for free, by doubly sthocastic operations and asks whether one has reversibility of asymptotic transformations in this setting. It turns out that, as shown in Ref. \cite{HHO03}, we do indeed have reversible transformations. The proof we present in the sequel can actually also be applied to this case, although we would only be able to establish reversibility under operations that asymptotically do not decrease the entropy. 

\section{Proof of Theorem \ref{maintheorem5}} \label{proofmaint5}

\subsection{The Entanglement Cost under Asymptotically non-Entangling Maps}

We start showing that the entanglement quantified by the global log-robustness 
cannot increase by more than a factor proportional to $\log(1 + \epsilon)$ under 
$\epsilon$-non-entangling maps.

\begin{lemma}\label{epsilonincrease}
If $\Lambda \in SEPP(\epsilon)$, then
\begin{equation} \label{143222}
LR_G(\Lambda(\rho)) \leq \log(1 + \epsilon) + LR_G(\rho).
\end{equation}
\end{lemma}
\begin{proof}
Let $\pi$ be an optimal state for $\rho$ in the sense that
\begin{equation*}
\rho + R_G(\rho) \pi = (1 + R_G(\rho))\sigma,
\end{equation*}
where $\sigma$ is a separable state. We have that
\begin{equation*}
\Lambda(\rho) + R_G(\rho) \Lambda(\pi) = (1 + R_G(\rho))\Lambda(\sigma),
\end{equation*}
with $R_G(\Lambda(\sigma)) \leq \epsilon$. Setting $Z$ to be a state for which $\Lambda(\sigma) + \epsilon Z$ is separable, we find 
\begin{equation*}
\Lambda(\rho) + R_G(\rho) \Lambda(\pi) + \epsilon(1 + R_G(\rho)) Z = (1 + R_G(\rho))\Lambda(\sigma) + \epsilon(1 + R_G(\rho)) Z \hspace{0.1 cm}\in \text{cone}({\cal S}),
\end{equation*}
from which Eq. (\ref{143222}) follows.
\end{proof}
\vspace{0.2 cm}

\begin{proposition} \label{propcost}
For every multipartite state $\rho \in {\cal D}(\mathbb{C}^{d_1}\otimes ... \otimes \mathbb{C}^{d_2})$,
\begin{equation}
E_{C}^{ane}(\rho) =  E_{R}^{\infty}(\rho).
\end{equation}
\end{proposition}

\begin{proof}
Let $\Lambda_n \in SEPP(\epsilon_n)$ be an optimal sequence of maps for the 
entanglement cost under asymptotically non-entangling 
maps, i.e.
\begin{equation}
\lim_{n \rightarrow \infty} || \Lambda_n(\phi_2^{\otimes k_n}) - \rho^{\otimes n}||_1 = 0,\;\;\;\;
\;\;\;\; \lim_{n\rightarrow\infty} \epsilon_n = 0,
\end{equation}
and
\begin{equation}
\limsup_{n \rightarrow \infty} \frac{k_n}{n} = E_{C}^{ane}(\rho).
\end{equation}
Then, from Lemma \ref{epsilonincrease},
\begin{eqnarray}
\frac{1}{n} LR_G(\Lambda_n(\phi_2^{\otimes k_n}))  &\leq&  \frac{1}{n} LR_G(\phi_2^{\otimes k_n}) + \frac{1}{n} \log(1 + \epsilon_n) \nonumber \\ &=& \frac{k_n}{n} + \frac{1}{n} \log(1 + \epsilon_n),
\end{eqnarray}
where the last equality follows from $R_G(\phi_2^{\otimes k_n}) = 2^{k_n} - 1$. Hence, as $\lim_{n \rightarrow \infty}\epsilon_n = 0$,
\begin{eqnarray}
LG(\rho) &\leq& \limsup_{n \rightarrow \infty} \frac{1}{n}LR_G(\Lambda_n(\phi_2^{\otimes k_n})) \nonumber \\ &\leq& \limsup_{n \rightarrow \infty} \left( \frac{k_n}{n} + \frac{1}{n} \log(1 + \epsilon_n) \right) \nonumber \\ &=& E_C^{ane}(\rho).
\end{eqnarray}
To show the converse inequality, we consider maps of the form:
\begin{equation}
\Lambda_n(A) = \tr(A \Phi(K_n)) \rho_n + \tr(A (\id - \Phi(K_n)))\pi_n,
\end{equation}
where (i) $\{ \rho_n \}$ is an optimal sequence of approximations for $\rho^{\otimes n}$ 
achieving the infimum in $LG(\rho)$ \footnote{Note that the infimum might not be achievable by any sequence $\{ \rho_n \}$. In this case, for every $\mu > 0$ we can find a sequence $\{ \rho_n^{\mu} \}$ such that $\lim_{n \rightarrow \infty} \frac{LR_G(\rho_n^{\mu})}{n} = LG(\rho) + \mu$, proceed as in the case where the infimum can be achieved and let $\mu \rightarrow 0$ in the end, obtaining the same results.}, (ii) $\log(K_n) = \lfloor \log(1 + R_{G}(\rho_n)) \rfloor$, and (iii) $\pi_n$ is a state such that
\begin{equation} \label{redone1}
\frac{\rho_n + (K_n-1) \pi_n}{K_n} \in {\cal S},
\end{equation}
which always exists as $K_n \geq 2^{\log(1 + R_G(\rho_n))} = R_G(\rho_n) + 1$. As $\pi_n$ and $\rho_n$ are states, each $\Lambda_n$ is completely positive and trace-preserving. 

The next step is to show that each $\Lambda_n$ is a $1/(K_n-1)$-separability-preserving map. 
From Eq. (\ref{redone1}) we find
\begin{equation} 
\frac{\pi_n + (K_n - 1)^{-1} \rho_n}{1 + (K_n - 1)^{-1}}  \in {\cal S},
\end{equation}
and, thus,
\begin{equation}
R_G(\pi_n) \leq \frac{1}{K_n - 1}.
\end{equation}

From Eq. (\ref{redone1}) we have that
\begin{equation}
\Lambda_n(\ket{0, 0}\bra{0, 0}) = \frac{\rho_n + (K_n - 1)\pi_n }{K_n} \in {\cal S}
\end{equation}
and
\begin{equation} \label{ppp}
R_G \left(\Lambda_n\left(\frac{\id - \Phi(K_n)}{K_n^2 - 1}\right)\right)
 = R_G(\pi_n) \leq \frac{1}{K_n - 1}.
\end{equation}
From the form of $\Lambda_n$ we can w.l.o.g. restrict our attention to isotropic input states. Any such state $I(q)$ can be written as
\begin{equation}
        I(q) = q I_b + (1 - q) \frac{\id - \Phi(K)}{K^2 - 1},  
\end{equation}
where $I_b$ is the separable isotropic state at the 
boundary of the separable states set and $0 \leq q \leq 1$. From the convexity of $R_G$,
\begin{equation}
R_G(\Lambda_n(I(q))) \leq q R_G(\Lambda_n(I_b)) + (1 - q) R_G\left( \Lambda_n \left(\frac{\id - \Phi(K)}{K^2 - 1} \right) \right) \leq \frac{1}{K_n - 1}, 
\end{equation}
where we used Eq. (\ref{ppp}) and
\begin{equation}
R_G(\Lambda_n(I_b)) = R_G(\Lambda_n(\ket{0, 0}\bra{0, 0})) = 0.
\end{equation}
We hence see that indeed $\Lambda_n$ is a $1/(K_n - 1)$-separability-preserving map.

From Corollary \ref{faithful} of chapter \ref{QHT} we find that $LG(\rho) = E_R^{\infty}(\rho) > 0$ for every entangled state. Therefore 
\begin{equation*}
\lim_{n\rightarrow\infty} (K_n - 1)^{-1} \leq \lim_{n\rightarrow\infty} R_G(\rho_n)^{-1} = 0,
\end{equation*}
where the last equality follows from Eq. (\ref{LG5}). Moreover, as 
\begin{equation}
\lim_{n \rightarrow \infty} || \rho^{\otimes n} - \Lambda_n(\Phi(K_n)) ||_1 = \lim_{n \rightarrow \infty} || \rho^{\otimes n} - \rho_n ||_1 = 0,
\end{equation}
it follows that $\{ \Lambda_n \}$ is an admissible sequence of maps 
for $E_C^{ane}(\rho)$ and, thus,
\begin{eqnarray}
E_{C}^{ane}(\rho) &\leq&  \limsup_{n \rightarrow \infty} \frac{1}{n} \log(K_n) \nonumber \\ &=& \limsup_{n \rightarrow \infty} \frac{1}{n} \lfloor \log(1 + R_G(\rho_n)) \rfloor \nonumber \\ &=& LG(\rho).
\end{eqnarray}
\end{proof}

\subsection{The Distillable Entanglement under non-Entangling Operations} \label{distent}

Before we turn to the proof of the main proposition of this 
section, we state and prove an auxiliary lemma which will 
be used later on. It can be considered the analogue for non-entangling 
maps of Theorem 3.3 of Ref. \cite{Rai01}, which deals with 
PPT maps. It is also similar to the representation for $\lambda_n(\pi, K)$ we developed in 
the proof of Theorem \ref{maintheorem} in chapter \ref{QHT}. 

\begin{lemma} \label{singfractione}
For every multipartite state $\rho \in {\cal D}(\mathbb{C}^{d_1}\otimes ... \otimes \mathbb{C}^{d_n})$ the singlet-fraction 
under non-entangling maps,
\begin{equation}
F_{sep}(\rho; K) := \max_{\Lambda \in SEPP} \tr(\Phi(K)\Lambda(\rho)),
\end{equation}
where $\Phi(K)$ is a $K$-dimensional maximally entangled state shared by the first two parties, satisfies
\begin{equation} \label{sfsep}
F_{sep}(\rho; K) = \min_{\sigma \in \text{cone}({\cal S})}  \tr (\rho - \sigma)_+ + \frac{1}{K}\tr(\sigma).
\end{equation}
\end{lemma}

\begin{proof}
Due to the $UU^*$-symmetry of the maximally entangled state and the fact that the composition of a $SEPP$ operation with the twirling map is again a separability-preserving operation, we can w.l.o.g. perform the maximization over $SEPP$ maps of the form
\begin{equation}
        \Lambda(\rho) = \tr(A \rho) \Phi(K) + \tr((\id - A)\rho) 
        \frac{\id - \Phi(K)}{K^2 - 1}.
\end{equation}
Since $\Lambda$ must be completely positive we have $0 \leq A \leq \id$. 
As $\Lambda(\rho)$ is an isotropic state for every input state 
$\rho$, it is separable iff $\tr(\Lambda(\rho)\Phi(K)) \leq 1/K$ 
\cite{HH99}. Hence, we find that $\Lambda$ is 
separability-preserving iff for every separable state $\sigma$,
\begin{equation}
\tr(A \sigma) \leq \frac{1}{K}.
\end{equation}
The singlet fraction is thus given by 
\begin{equation} \label{singfracuse}
        F_{sep}(\rho; K) = \max_{A} [\tr(A \rho) : 0 \leq A \leq \id, 
        \hspace{0.2 cm} \tr(A \sigma) \leq 1/K, \hspace{0.2 cm} \forall 
        \hspace{0.1 cm} \sigma \in {\cal S}]. 
\end{equation}
The R.H.S. of this equation is a convex optimization problem 
and we can find its dual formulation, in analogy to 
what was done in the proof of Theorem \ref{maintheorem}. Let us form 
the Lagrangian of the problem, 
\begin{equation}
        L(A, X, Y, \sigma) =  - \tr(A \rho) - \tr(X A) - \tr(Y(\id - A)) - 
        \tr((\id/K - A) \sigma),
\end{equation}
where $X, Y \geq 0$ are Lagrange multipliers associated to 
the constraints $0 \leq A \leq \id$, and $\sigma \in \text{cone}({\cal S})$ 
is a Lagrange multiplier (an unnormalized separable state) associated to the constraint $\tr(A \sigma) 
\leq 1/K \hspace{0.2 cm} \forall \hspace{0.1 cm} \sigma 
\in {\cal S}$\footnote{See e.g. section \ref{entwit}}. The dual problem is then given by 
\begin{equation}
        F_{sep}(\rho; K) = \min_{Y, \sigma} [\tr(Y) + \frac{1}{K} 
        \tr(\sigma) : \sigma \in {\cal S}_u, \hspace{0.1 cm} 
        Y \geq 0, \hspace{0.1 cm} Y \geq \rho - \sigma]. 
\end{equation}
Using that $\tr(A)_{+} = \min_{Y \geq A, \hspace{0.01 cm} Y \geq 0} \tr(Y)$, we then find Eq. (\ref{sfsep}).
\end{proof}
\vspace{0.3 cm}

It turns out that for the distillation part we do not need to allow any generation of 
entanglement from the maps. In analogy to Definition \ref{distane}, we can define the distillable entanglement
under non-entangling maps as 
\begin{equation}
E_{D}^{ne}(\rho) := \sup_{\{ k_n \}} \left \{ \liminf_{n \rightarrow \infty} \frac{k_n}{n}  : \lim_{n \rightarrow \infty} \left( \min_{\Lambda \in SEPP} || \Lambda(\rho^{\otimes n}) - \phi_2^{\otimes k_n}||_1 \right) = 0 \right \}.
\end{equation}

Using Lemma \ref{singfractione} and Theorem \ref{maintheorem} proved in chapter \ref{QHT} we can then easily establish the following proposition. 

\begin{proposition} \label{disne}
For every multipartite entangled state $\rho \in {\cal D}(\mathbb{C}^{d_1}\otimes ... \otimes \mathbb{C}^{d_n})$,
\begin{equation}
E_{D}^{ne}(\rho) = E_R^{\infty}(\rho).
\end{equation}
\end{proposition}
\begin{proof}
On the one hand, from Theorem \ref{maintheorem} we have
\begin{equation} \label{maintdis}
\lim_{n \rightarrow \infty} \min_{\sigma_n \in {\cal S}}\tr(\rho^{\otimes n} - 2^{ny}\sigma_n)_+ = 
\begin{cases}
0, & y > E_R^{\infty}(\rho) \\
1, & y < E_R^{\infty}(\rho).
\end{cases}
\end{equation}
On the other, from Lemma \ref{singfractione} we find
\begin{equation} \label{sfne}
F_{sep}(\rho^{\otimes n}; 2^{ny}) := \min_{\sigma \in {\cal S}, b \in \mathbb{R}}  \tr (\rho^{\otimes n} - 2^{n b}\sigma)_+ + 2^{- (y - b)n}.
\end{equation}

Let us consider the asymptotic behavior of $F_{sep}(\rho^{\otimes n}, 2^{n y})$. Take $y = E_{R}^{\infty}(\rho) + \epsilon$, for any $\epsilon > 0$. Then we can choose, for each $n$, $b = 2^{n(E_{R}^{\infty}(\rho) + \frac{\epsilon}{2})}$, giving
\begin{equation*} 
F_{sep}(\rho^{\otimes n}, 2^{ny}) \leq \min_{\sigma \in {\cal S}} \tr(\rho^{\otimes n} - 2^{n(E_{\cal M}^{\infty}(\rho) + \frac{\epsilon}{2})}\sigma)_+ + 2^{-n\frac{\epsilon}{2}}.
\end{equation*}
We then see from Eq. (\ref{maintdis}) that $\lim_{n \rightarrow \infty} F_{sep}(\rho^{\otimes n}, 2^{ny}) = 0$, from which follows that $E_D^{ne}(\rho) \leq E_{R}^{\infty}(\rho) + \epsilon$. As $\epsilon$ is arbitrary, we find $E_D^{ne}(\rho) \leq E_{R}^{\infty}(\rho)$. 

Conversely, let us take $y = E_{\cal M}^{\infty}(\rho) - \epsilon$, for any $\epsilon > 0$. The optimal $b$ for each $n$ has to satisfy $b_n \leq 2^{y n}$, otherwise $F_{sep}(\rho^{\otimes n}, 2^{ny})$ would be larger than one, which is not true. Therefore,
\begin{equation*} 
F_{sep}(\rho^{\otimes n}, 2^{ny}) \geq \min_{\sigma \in {\cal S}} \tr(\rho^{\otimes n} - 2^{n(E_{\cal M}^{\infty}(\rho) - \epsilon)}\sigma)_+,
\end{equation*}
which goes to one again by Eq. (\ref{maintdis}). This then shows that $E_D^{ne}(\rho) \geq E_{R}^{\infty}(\rho) - \epsilon$. Again, as $\epsilon > 0$ is arbitrary, we find  $E_D^{ne}(\rho) \geq E_{R}^{\infty}(\rho)$.
\end{proof}

\vspace{0.3 cm}
The proof of the other half of Theorem \ref{maintheorem5} follows easily from Proposition \ref{disne} and the following Lemma. 

\begin{lemma} \label{almmonrelent}
If $\Lambda \in SEPP(\epsilon, {\cal H})$, then
\begin{equation} \label{1432}
E_R(\Lambda(\rho)) \leq \log(1 + \epsilon) + E_R(\rho).
\end{equation}
\end{lemma}
\begin{proof}
Let $\sigma$ be an optimal separable state for $\rho$ in the relative entropy of 
entanglement. Then, if $\Omega$ is a $\epsilon$-separability preserving map and $Z$ a state such that $\Omega(\sigma) + \epsilon Z$ is separable,
\begin{eqnarray*}
E_R(\rho) &=& S(\rho || \sigma) \nonumber \\ &\geq& S( \Omega(\rho) || \Omega(\sigma) ) \nonumber \\
&\geq& S( \Omega(\rho) || \Omega(\sigma) + \epsilon Z ) \nonumber \\ 
&=& S( \Omega(\rho) || (\Omega(\sigma) + \epsilon Z)/(1 + \epsilon) ) - \log(1 + \epsilon) \nonumber \\ 
&\geq& E_R(\Omega(\rho)) - \log(1 + \epsilon), 
\end{eqnarray*}
The first inequality follows from the monotonicity of the 
relative entropy under trace preserving CP maps and the second inequality from the
operator monotonicity of the $\log$.
\end{proof}
\vspace{0.3 cm}

Indeed, as any sequence of non-entangling maps is obviously asymptotically non-entangling, we have $E_D^{ane}(\rho) \geq E_D^{ne}(\rho) = E_R^{\infty}(\rho)$, where the last equality follows from Proposition \ref{disne}. To prove the converse inequality $E_D^{ane}(\rho) \leq E_R^{\infty}(\rho)$, we use 
Lemma \ref{almmonrelent} and the monotonicity of the relative entropy under separability-preserving maps. Let $\Lambda_n \in SEPP(\epsilon_n)$ be an optimal sequence of maps for the distillable entanglement under asymptotically non-entangling maps in the sense that
\begin{equation*}
\lim_{n \rightarrow \infty} || \Lambda_n(\rho^{\otimes n}) - \phi_2^{\otimes k_n}||_1 = 0\;\;\;\;
\;\;\;\; \lim_{n\rightarrow\infty} \epsilon_n = 0,
\end{equation*}
and
\begin{equation*}
\liminf_{n \rightarrow \infty} \frac{k_n}{n} = E_{D}^{ane}(\rho).
\end{equation*}
From Lemma \ref{almmonrelent},
\begin{eqnarray*}
       \frac{1}{n} E_R(\Lambda_n(\rho^{\otimes n}))
          &\leq&  \frac{1}{n} E_R(\rho^{\otimes n}) + \frac{1}{n} 
         \log(1 + \epsilon_n).
\end{eqnarray*}
Hence, as $\lim_{n \rightarrow \infty}\epsilon_n = 0$ and 
from the asymptotic continuity of relative entropy of entanglement,
\begin{eqnarray*}
        E_D^{as}(\rho) &=& \liminf_{n \rightarrow \infty}\frac{1}{n}E_R(\Lambda_n(\rho^{\otimes n})) \\ 
        &\leq& \liminf_{n \rightarrow \infty} \frac{1}{n}E_R(\rho^{\otimes n}) 
        + \liminf_{n \rightarrow \infty} \frac{1}{n} \log(1 + \epsilon_n)\nonumber \\ 
        &=& E_R^{\infty}(\rho).
\end{eqnarray*}

\subsection{Proof Corollary \ref{corref}}

Finally, we can now easily establish Corollary \ref{corref}.
\vspace{0.2 cm}

\begin{proof} (Corollary \ref{corref})
Assume there is a sequence of quantum maps $\{ \Lambda_n \}_{n \in \mathbb{N}}$ satisfying the three conditions of the corollary. Then, 
\begin{eqnarray*}
E_{R}^{\infty}(\sigma) &=& \lim_{n \rightarrow \infty} \frac{1}{n}E_R(\Lambda_n(\rho^{\otimes n})) \\
&\leq& \frac{1}{n}E_R(\rho^{\otimes n}) + \frac{\log(1 + \epsilon_n)}{n} \\
&=& E_R^{\infty}(\rho).
\end{eqnarray*}
The first equality follow from the asymptotic continuity of $E_R$ and the following inequality from Lemma \ref{almmonrelent}.

To show the other direction, let us assume that $E_R^{\infty}(\rho) \geq E_R^{\infty}(\sigma)$. As $E_R^{\infty}(\rho) = E_D^{ane}(\rho)$, there is a sequence of maps $\{ \Lambda_n \}_{n \in \mathbb{N}}$, $\Lambda_n : {\cal D}((\mathbb{C}^{d_1}\otimes ... \otimes \mathbb{C}^{d_m})^{\otimes n}) \rightarrow {\cal D}((\mathbb{C}^{2} \otimes \mathbb{C}^{2})^{\otimes k_n})$, such that 
\begin{eqnarray*}
	\Lambda_n \in SEPP(\epsilon_n),  \hspace{1 cm} \lim_{n \rightarrow \infty}\epsilon_n = 0, \nonumber 
\end{eqnarray*}
\begin{eqnarray} \label{tnbo1}	
\lim_{n \rightarrow \infty} || \Lambda_n(\rho^{\otimes n}) - \phi_2^{\otimes k_n} ||_1 = 0 \nonumber
\end{eqnarray}
and
\begin{equation} \label{relentrho}
\lim_{n \rightarrow \infty} \frac{k_n}{n} = E_R^{\infty}(\rho)\footnote{We can always find a sequence for which the limit exists by using the optimal sequence such that $\limsup_{n \rightarrow \infty} \frac{k_n}{n} = E_R^{\infty}(\rho)$ and increasing the value of the $k_n$'s which are not close to the limit value.}
\end{equation}
Moreover, as $E_R^{\infty}(\sigma) = E_C^{ane}(\sigma)$, there is another sequence of maps $\{ \Omega_n \}_{n \in \mathbb{N}}$, $\Omega_n : {\cal D}((\mathbb{C}^{2} \otimes \mathbb{C}^{2})^{\otimes k'_n}) \rightarrow {\cal D}((\mathbb{C}^{d'_1}\otimes ... \otimes \mathbb{C}^{d'_{m'}})^{\otimes n})$, satisfying
\begin{eqnarray*}
	\Omega_n \in SEPP(\epsilon'_n),  \hspace{1 cm} \lim_{n \rightarrow \infty}\epsilon'_n = 0, \nonumber 
\end{eqnarray*}
\begin{eqnarray}	\label{tnbo2}
\lim_{n \rightarrow \infty} || \Omega_n(\phi_2^{\otimes k'_n}) - \sigma^{\otimes n} ||_1 = 0 \nonumber
\end{eqnarray}
and
\begin{equation} \label{relentsigma}
\lim_{n \rightarrow \infty} \frac{k'_n}{n} = E_R^{\infty}(\sigma).
\end{equation}

From Eqs. (\ref{relentrho}) and (\ref{relentsigma} there is a sequence $\delta_{n_0}$ converging to zero when $n_0 \rightarrow \infty$ such that for every $n \geq n_0$,
\begin{equation*}
k_n \geq (E_R^{\infty}(\rho) - \delta_{n_0}/2)n, \hspace{0.3 cm} k'_n \leq (E_R^{\infty}(\sigma) + \delta_{n_0}/2)n.
\end{equation*}
Then, for every $n \geq n_0$, $k_n \geq - \delta_{n_0}n + k'_n$. From Eq. (\ref{relentsigma}) we thus find that for sufficiently large $n \geq n_0$,
\begin{equation*}
k_n =  k'_{n - o(n)} + r_n,
\end{equation*}
with $r_n$ a positive integer. Here $o(n)$ stands for a sublinear term in $n$. 

Let us now consider the sequence of maps $\{ \Omega_n \circ \tr_{1,...,r_n}\circ \Lambda_n \}_{n \in \mathbb{N}}$. From Eqs. (\ref{tnbo1}, \ref{tnbo2}) we find
\begin{eqnarray}
\lim_{n \rightarrow \infty} || \Omega_{n - o(n)} \circ \tr_{1,...,r_n} \circ \Lambda_n(\rho^{\otimes n}) - \sigma^{\otimes n - o(n)} ||_1  &\leq& \lim_{n \rightarrow \infty} || \Lambda_n(\rho^{\otimes n}) - \phi_2^{\otimes k_n} ||_1 \nonumber \\ &+& || \Omega_{n - o(n)}(\phi_2^{\otimes n - o(n)}) - \sigma^{\otimes n - o(n)} ||_1 = 0. \nonumber
\end{eqnarray}
Moreover, from Lemma \ref{epsilonincrease} we see that for every separable state $\sigma$,
\begin{eqnarray*}
LR_G(\Omega_{n - o(n)} \circ \tr_{1,...,r_n} \circ \Lambda_n(\sigma)) &\leq& LR_G(\Lambda_n(\sigma)) + \log(1 + \epsilon'_n) \\ &\leq& \log(1 + \epsilon_n) + \log(1 + \epsilon'_n),
\end{eqnarray*}
where we used $\Omega_{n - o(n)} \circ \tr_{1,...,r_n} \in SEPP(\epsilon'_n)$ and $\Lambda_n \in SEPP(\epsilon_n)$. Hence, $\Omega_{n - o(n)} \circ \tr_{1,...,o(n)} \circ \Lambda_n \in SEPP(\epsilon_n + \epsilon'_n + \epsilon_n \epsilon'_n)$.
\end{proof}

\section{How Much Entanglement Must and Can be Generated?}  \label{mustcan}

We are now in position to understand the choice of the global robustness as the measure to quantify the amount of entanglement generated. The reason we need to allow some entanglement to be generated is that we relate the entanglement cost to the regularized relative entropy of entanglement by using the connection of the latter to the asymptotic global robustness. The amount of entanglement generated is then due to the fact that the optimal mixing state in the global robustness might be entangled. Before we analyse more carefully if we indeed need to allow for some entanglement to be generated, let us analyse if we can quantify it by some other measure, instead of the global robustness.

Suppose we required alternatively only that 
\begin{equation} \label{tnapp}
\lim_{n \rightarrow \infty} \max_{\sigma \in {\cal S}} \min_{\pi \in {\cal S}} || \Lambda_n(\sigma) - \pi ||_1 = 0,
\end{equation}
instead of $\lim_{n \rightarrow \infty}\max_{\sigma \in {\cal S}} R_G(\Lambda_n(\sigma)) = 0$. Then the achievability part in Proposition \ref{disne} would still hold, as for every quantum operation $\Lambda$,
\begin{equation*}
\max_{\sigma \in {\cal S}} R_G(\Lambda(\sigma)) \leq \max_{\sigma \in {\cal S}} \min_{\pi \in {\cal S}} || \Lambda(\sigma) - \pi ||_1.  
\end{equation*}
However this is not sufficient. We still have to make sure that the cost is larger than the distillation function, which should be finite. It is easy to see that Eq. (\ref{tnapp}) ensures that both the distillation and cost functions are zero for separable states. It turns out however that the distillable entanglement is infinity for every entangled state! We hence have a bizarre situation in which even though entanglement cannot be created for free, it can be amplified to the extreme whenever present, no matter in which amount. The key to see this is to consider the analogue of $F_{sep}$, given by Eq. (\ref{sfne}), when we only require that the map satisfies Eq. \ref{tnapp}. Following the proof of Lemma \ref{singfractione} we can easily see that the singlet-fraction under maps $\Lambda$ satisfying 
\begin{equation*} 
 \max_{\sigma \in {\cal S}} \min_{\pi \in {\cal S}} || \Lambda(\sigma) - \pi ||_1 \leq \epsilon
\end{equation*}
is given by
\begin{equation} 
F_{sep}(\rho; K ; \epsilon) = \min_{\sigma \in \text{cone}({\cal S})}  \tr (\rho - \sigma)_+ + \tr(\sigma)(\frac{1}{K} + \epsilon),
\end{equation}
which for $\rho^{\otimes n}$ can be rewritten as
\begin{equation} 
F_{sep}(\rho^{\otimes}; 2^{ny}; \epsilon_n) = \min_{\sigma \in {\cal S}, b \in \mathbb{R}}  \tr (\rho^{\otimes n} - 2^{bn}\sigma)_+ + 2^{-(y - b)n} + 2^{-( (\log(1/\epsilon_n)/n) - b)n}.
\end{equation}
It is clear that the optimal $b$ must be such that $b < \min(b, \log(1/\epsilon_n)/n)$, as otherwise $F_{sep}(\rho^{\otimes}; 2^{ny}; \epsilon)$ would be larger than zero. Therefore, if $y > \log(1/\epsilon_n)/n$, 
\begin{equation} \label{sgtn}
F_{sep}(\rho^{\otimes}; 2^{ny}; \epsilon_n) \geq \min_{\sigma \in {\cal S}}  \tr (\rho^{\otimes n} - \epsilon_n^{-1}\sigma)_+.
\end{equation} 
By Theorem \ref{maintheorem}, $F_{sep}(\rho^{\otimes}; 2^{ny}; \epsilon_n)$ goes to one for every $y$, a long as $\epsilon_n$ goes to zero slower than $2^{- n E_R^{\infty}(\rho)}$, which imply that the associated distillable entanglement is unbounded. Note that the same happens if we use \textit{any} asymptotically continuous measure to bound the amount of entanglement generated.

If instead we require that
\begin{equation}
\max_{\sigma \in {\cal S}} \min_{\pi \in {\cal S}} || \Lambda(\sigma) - \pi ||_1 \leq \epsilon/ \dim({\cal H}),
\end{equation}
or even that 
\begin{equation}
\max_{\sigma \in {\cal S}} \min_{\pi \in {\cal S}} || \Lambda(\sigma) - \pi ||_{\infty} \leq \epsilon/ \dim({\cal H}),
\end{equation}
then we would find that the associated $\epsilon$-singlet-fraction would satisfy
\begin{equation} 
\tilde{F}_{sep}(\rho; K ; \epsilon) = \min_{\sigma \in \text{cone}({\cal S})}  \tr (\rho - \sigma)_+ + \tr(\sigma)\frac{1 + \epsilon}{K}.
\end{equation}
In this case it is easy to see that the distillable entanglement would be bounded and we would recover a sensible situation. It is 
interesting and rather mysterious to the author that although it seems that some entanglement must be generated to have reversibility, only very little can actually be afforded before the theory becomes trivial.

\subsection{Single-copy and Asymptotic Entanglement Cost under non-Entangling Maps}

In this subsection we show that we can exactly determine the single-copy entanglement cost under non-entangling operations. This will bring more insight into the necessity of using asymptotically non-entangling maps for having reversibility.

\begin{definition}
We define the $\epsilon$-single-copy entanglement cost under non-entangling maps as
\begin{equation} \label{sbec}
E_{C,1}^{ne,\epsilon}(\rho) := \inf \{ \log K : \min_{\Lambda \in SEPP} || \Lambda(\Phi(K)) - \rho ||_1 \leq \epsilon \}.
\end{equation}
%The $\epsilon$-single-copy distillable entanglement under non-entangling maps, in turn, is given by
% \begin{equation} \label{scdist}
%E_{D,1}^{ne,\epsilon}(\rho) := \sup \{ \log K : \min_{\Lambda \in SEPP} || \Phi(K) - \Lambda(\rho) ||_1 \leq \epsilon \}.
%\end{equation}
\end{definition}

Consider the log robustness, defined in section \ref{robustnesses}. In the notation introduced in chapter \ref{QHT}, we can express it as
\begin{equation}
LR(\rho) = \min_{\sigma \in {\cal S}} S_{\max, {\cal S}}(\rho || \sigma),
\end{equation}
with
\begin{equation}
S_{\max, {\cal S}}(\rho || \sigma) :=  \min \{ s : \rho \leq_{\cal S} 2^s \sigma \} .
\end{equation}
Still following the approach of chapter \ref{QHT}, we can define its $\epsilon$ smooth version as
\begin{equation} \label{elogrob}
LR^{\epsilon}(\rho) := \min_{\tilde{\rho} \in B_{\epsilon}(\rho)} LR(\tilde{\rho}).
\end{equation}

\begin{proposition}
For every $0 \leq \epsilon < 1$ and every $\rho \in {\cal D}(\mathbb{C}^{d_1}\otimes ... \otimes \mathbb{C}^{d_n})$,
\begin{equation}
LR^{\epsilon}(\rho) \leq  E_{C,1}^{ne,\epsilon}(\rho) \leq \lfloor LR^{\epsilon}(\rho) \rfloor.
\end{equation}
\end{proposition}
\begin{proof}
Let $\Lambda$ be an optimal non-entangling map in Eq. \ref{sbec}. Then, by Eq. \ref{elogrob} and the monotonicity of $LR$ under non-entangling maps,
\begin{equation}
LR^{\epsilon}(\rho) \leq LR(\Lambda(\Phi(2^{E_{C,1}^{ne,\epsilon}(\rho)}))) \leq LR(\Phi(2^{E_{C,1}^{ne,\epsilon}(\rho)})) = E_{C,1}^{ne,\epsilon}(\rho).
\end{equation}

To show the other inequality, we consider a formation map of the form
\begin{equation}
\Lambda(*) := \tr(\Phi(K)*) \rho_{\epsilon} + \tr((\id - \Phi(K))*)\pi, 
\end{equation}
where $\rho_{\epsilon}$ is an optimal state for $\rho$ in Eq. \ref{elogrob}, $K := \lfloor 2^{LR^{\epsilon}(\rho)}  \rfloor$ and $\pi$ is a separable state such that 
\begin{equation}
\frac{1}{1 + K}\left( \rho_{\epsilon} + K\pi\right) \in {\cal S}.
\end{equation}
As $|| \Lambda(\Phi(K)) - \rho ||_1 \leq \epsilon$, we find that indeed $E_{C,1}^{ne,\epsilon}(\rho) \leq \lfloor LR^{\epsilon}(\rho) \rfloor$.
\end{proof}

\vspace{0.3 cm}

Using this proposition it is straightforward to show that the entanglement cost under non-entangling maps is given by
\begin{equation}
E_C^{ne}(\rho) := \sup_{\{ \epsilon_n \}} \left \{ \lim_{n \rightarrow \infty} \frac{1}{n} LR^{\epsilon_n}(\rho^{\otimes n}) : \lim_{n \rightarrow \infty} \epsilon_n = 0  \right \}.
\end{equation}
From Eq. \ref{globrobrelent} we then see that the question whether we must allow the generation of some entanglement in order to have a reversible theory reduces to the question whether the two robustnesses become the same quantity after smoothening and regularization.

\section{Connection to the Axiomatic Formulation of the Second Law of Thermodynamics} \label{cafslt}

In this section we comment on the similarities and differences of entanglement manipulation under 
asymptotically non-entangling operations and the axiomatic approach of Giles \cite{Gil64} and more 
particularly of Lieb and Yngvason \cite{LY99} for the second law of thermodynamics. 

Let us start by briefly recalling the axioms used by Lieb 
and Yngvason \cite{LY99} in order to derive the second law. 
Their starting point is the definition of a system as a collection of points called state space and denoted 
by $\Gamma$. The individual points of a state space are 
the states of the system. The composition of two state 
spaces $\Gamma_1$ and $\Gamma_2$ is given by their Cartesian 
product. Furthermore, the scaled copies of a given system is defined as follows: if 
$t > 0$ is some fixed number, the state space $\Gamma^{(t)}$ 
consists of points denoted by $t X$ with $X \in \Gamma$. 
Finally, a preorder $\prec$ on the state 
space satisfying the following axioms is assumed:
\begin{enumerate}
        \item $X \prec X$.
        \item $X \prec Y$ and $Y \prec Z$ implies $X \prec Z$.
        \item If $X \prec Y$, then $t X \prec tY$ for all $t > 0$.
        \item $X \prec (t X, (1 - t)X)$ and $(t X, (1 - t)X) \prec X$.
        \item If, for some pair of states, $X$ and $Y$,
        \begin{equation}
        (X, \epsilon Z_0) \prec (Y, \epsilon Z_1)
        \end{equation}
        holds for a sequence of $\epsilon$'s tending to zero and some states $Z_0$, $Z_1$, then $X \prec Y$.
        \item $X \prec X'$ and $Y \prec Y'$ implies $(X, Y) \prec (X', Y')$.
\end{enumerate}
Lieb and Yngvason then show that these axioms, together with the 
comparison hypothesis, which states that
\begin{center}
\textbf{Comparison Hypothesis:} for any two states $X$ 
and $Y$ either $X \prec Y$ or $Y \prec X$,
\end{center}
are sufficient to prove the existence of a single 
valued entropy function completely determining the 
order induced by the relation $\prec$. 

In the context of entanglement transformations, we 
interpret the relation $\rho \prec \sigma$ as the 
possibility of asymptotically transforming $\rho$ 
into $\sigma$ by asymptotically non-entangling
maps. Then, the composite state $(\rho, \sigma)$ is 
nothing but the tensor product $\rho \otimes \sigma$. 
Moreover, $t \rho$ takes the form of $\rho^{\otimes t}$, 
which is a shortcut to express the fact that if 
$\rho^{\otimes t} \prec \sigma$, then asymptotically 
$t$ copies of $\rho$ can be transformed into one of 
$\sigma$. More concretely, we say that 
\begin{equation}
\rho^{\otimes t} \prec \sigma^{\otimes q},
\end{equation}
for positive real numbers $t, q$ if there is a
sequence of integers $n_t, n_q$ and of $SEPP(\epsilon_n)$ maps
$\Lambda_n$ such that
\begin{equation*}
\lim_{n \rightarrow \infty} || \Lambda_{n}(\rho^{\otimes n_t}) - \sigma^{\otimes n_q - o(n)}||_1 = 0,
\end{equation*}
\begin{equation*}
\lim_{n \rightarrow \infty} \epsilon_n = 0, \hspace{0.3 cm} \lim_{n \rightarrow \infty} \frac{n_t}{n} = t, \hspace{0.3 cm} \text{and}\hspace{0.3 cm} \lim_{n \rightarrow \infty} \frac{n_q}{n} = q.
\end{equation*}
With this definition it is straightforward to observe that properties 
1, 3, and 4 hold true for entanglement manipulation under asymptotically non-entangling maps. Property 2 can be shown to hold, in turn, by noticing that, from Lemma \ref{epsilonincrease}, if $\Lambda \in SEPP(\epsilon)$ and $\Omega \in SEPP(\delta)$, then $\Lambda \circ \Omega \in SEPP(\epsilon + \delta + \delta \epsilon)$. Therefore the composition of two asymptotically non-entangling maps is again asymptotically non-entangling. That property 5 is also true is proven in the following lemma.

\begin{lemma}
If for two states $\rho$ and $\sigma$,
\begin{equation} \label{axi1}  
        \rho \otimes \pi_1^{\otimes \epsilon} \prec 
        \sigma \otimes \pi_2^{\otimes \epsilon},
\end{equation} 
holds for a sequence of $\epsilon$'s tending to zero and two states $\pi_0$, $\pi_1$, 
then $\rho \prec \sigma$.
\end{lemma}
\begin{proof}
Eq. (\ref{axi1}) means that for every $\epsilon > 0$ there is a sequence of maps $\Lambda_n \in SEPP(\epsilon_n)$ such that
\begin{equation*}
\lim_{n \rightarrow \infty} || \Lambda_{n}(\rho^{\otimes n} \otimes \pi_1^{\otimes n_{\epsilon}}) - \sigma^{\otimes n - o(n)} \otimes \pi_2^{\otimes n'_{\epsilon} - o(n)}||_1 = \lim_{n \rightarrow \infty} \delta_n = 0,
\end{equation*}
\begin{equation*}
\lim_{n \rightarrow \infty} \epsilon_n = 0, \hspace{0.3 cm} \lim_{n \rightarrow \infty} \frac{n_{\epsilon}}{n} = \epsilon, \hspace{0.3 cm} \text{and} \hspace{0.3 cm} \lim_{n \rightarrow \infty} \frac{n'_{\epsilon}}{n} = \epsilon.
\end{equation*}
We have
\begin{eqnarray}
\frac{1}{n}E_R(\rho^{\otimes n}) + \frac{1}{n}E_R(\pi_1^{\otimes n_{\epsilon}}) &\geq& \frac{1}{n} E_R(\rho^{\otimes n} \otimes \pi_1^{\otimes n_{\epsilon}}) \nonumber \\ &\geq& \frac{1}{n}E_R(\Lambda_{n}(\rho^{\otimes n} \otimes \pi_1^{\otimes n_{\epsilon}})) - \frac{\log(1 + \epsilon_n)}{n} \nonumber \\ &\geq& \frac{1}{n}E_R(\sigma^{\otimes n - o(n)} \otimes \pi_2^{\otimes n'_{\epsilon} - o(n)}) - f(\delta_{\epsilon}) - \frac{\log(1 + \epsilon_n)}{n} \nonumber \\
&\geq& \frac{1}{n}E_R(\sigma^{\otimes n - o(n)}) - f(\delta_{\epsilon}) - \frac{\log(1 + \epsilon_n)}{n}, 
\end{eqnarray}
where $f: \mathbb{R} \rightarrow \mathbb{R}$ is such that $\lim_{x \rightarrow 0}f(x) = 0$. The first inequality follows from the subadditivity of $E_R$, the second from Lemma \ref{almmonrelent}, the third from the asymptotic continuity of $E_R$, and the last from the monotonicity of $E_R$ under the partial trace. 

As $E_R(\pi_2) \leq \log(\dim({\cal H}))$, where ${\cal H}$ is the Hilbert space in which $\pi_2$ acts on, we find 
\begin{equation}
\frac{1}{n}E_R(\rho^{\otimes n}) \geq \frac{1}{n}E_R(\sigma^{\otimes n - o(n)}) - f(\delta_{\epsilon}) - \frac{\log(1 + \epsilon_n)}{n} - \frac{n_{\epsilon}}{n}\log(\dim({\cal H})).
\end{equation}
Taking the limit $n \rightarrow \infty$,
\begin{equation}
E_R^{\infty}(\rho) \geq E_R^{\infty}(\sigma)  - \epsilon.
\end{equation}
Taking $\epsilon \rightarrow 0$ we find that $E_R^{\infty}(\rho) \geq E_R^{\infty}(\sigma)$. The Lemma then follows from Corollary \ref{corref}.
\end{proof}

The Comparison Hypothesis, in turn, follows from Corollary \ref{corref}: it expresses the total order 
induced by the regularized relative entropy of 
entanglement. 

We cannot decide if the theory we are considering for entanglement satisfy 
axiom 6. This is fundamentally linked to the 
possibility of having \textit{entanglement catalysis} \cite{JP99} 
under asymptotically non-entangling transformations. As shown in Theorem 2.1 of Ref. \cite{LY99}, 
given that axioms 1-5 are true, axiom 6 is equivalent to
\begin{enumerate}
        \item $(X, Y) \prec (X', Y)$ implies $X \prec X'$,
\end{enumerate} 
which is precisely the non-existence of catalysis. Interestingly, we can link 
such a possibility in the bipartite case to an important open problem in entanglement theory, the 
full additivity of the regularized relative entropy of entanglement. In turn, 
the later was shown in Ref. \cite{BHPV07} to be equivalent to the 
full monotonicity under \textit{LOCC} of $E_R^{\infty}$. 

\begin{lemma}
The regularized relative entropy of entanglement is fully additive for bipartite states, i.e. for every two states 
$\rho \in {\cal D}(\mathbb{C}^{d_1}\otimes \mathbb{C}^{d_2})$ and $\pi \in {\cal D}(\mathbb{C}^{d'_1}\otimes \mathbb{C}^{d'_2})$,
\begin{equation} \label{relentadd}
E_{R}^{\infty}(\rho \otimes \pi) = E_{R}^{\infty}(\rho) + E_{R}^{\infty}(\pi),
\end{equation}
if, and only if, there is no catalysis for entanglement manipulation under asymptotically 
non-entangling maps.
\end{lemma}
\begin{proof}
If Eq. (\ref{relentadd}) holds true and $\rho \otimes \pi \prec \sigma \otimes \pi$, then
\begin{equation}
E_{R}^{\infty}(\rho) + E_{R}^{\infty}(\pi) = E_{R}^{\infty}(\rho \otimes \pi) \geq E_{R}^{\infty}(\sigma \otimes \pi) = E_{R}^{\infty}(\sigma) + E_{R}^{\infty}(\pi),
\end{equation}
and thus, as $E_{R}^{\infty}(\rho) \geq E_{R}^{\infty}(\sigma)$, we find from Corollary \ref{corref} that $\rho \prec \sigma$. 

Conversely, assume that there is no catalysis. Then from the discussion above we find that axiom 6 holds true. For every bipartite pure state $\ket{\psi}$, the regularized relative entropy of entanglement is equal to the von Neumann entropy of the reduced density matrix $S(\psi_A)$. It hence follows that for every bipartite state $\rho$, there is a bipartite pure state $\ket{\psi}$ such that $E_R^{\infty}(\rho) = E_R^{\infty}(\psi)$. 

Let $\ket{\psi}$ and $\ket{\phi}$ be such that $E_R^{\infty}(\rho) = E_R^{\infty}(\psi)$ and $E_R^{\infty}(\pi) = E_R^{\infty}(\phi)$. From Corollary \ref{corref} we have $\rho \prec \psi$, $\pi \prec \phi$ and vice versa. Then, by axiom 6 we find that $\rho \otimes \pi \prec \psi \otimes \phi$ and $\psi \otimes \phi \prec \rho \otimes \pi$, from which we find, once more from Corollary \ref{corref}, that $E_R^{\infty}(\rho \otimes \pi) = E_R^{\infty}(\psi \otimes \phi)$. The lemma is a consequence of the additivity of $E_R^{\infty}$ on two pure states (which follows from the fact that for pure states the measure is equal to the entropy of entanglement). 
\end{proof}

It is an open question if we can extend the lemma to the multipartite setting. The difficulty in this case is that we do not have a simple formula for $E_R^{\infty}$ of pure states and hence do not know if the measure is additive for two multipartite pure states.

\part{Quantum Complexity of Many-Body Physics}\label{part2}

\chapter{Quantum Complexity Theory} \label{complexity}

\section{Introduction}

Quantum computation appears to offer an exponential algorithmic speed-up over classical computation. The most notable example, although by no means the only one (see e.g. \cite{Sim97, Hal02, CCD+03, AJL06, AAEL07}), is Shor's polynomial quantum algorithm for factoring \cite{Sho97}, a problem for which no polynomial classical algorithm is believed to exist. A second striking application of quantum computers is the efficient simulation of the dynamics of quantum many-body physics. This possibility, first raised by Feynman in Ref. \cite{Fey82}, has since been refined in several works, e.g. \cite{Llo96, AT03, BACS07}, and, if realized, would represent a major breakthrough in the study of quantum many-body physics. 

The quantum simulation of many-body systems has been analysed in two lines of investigation. In the first one studies the use of experimentally well controlled quantum systems to simulate the dynamics of particular many-body models of interest: the former system is employed as a quantum simulator to the latter. This approach is much less stringent than to build a full working quantum computer, although its applicability is limited to simulating the specific models that can be created in the physical set-up under consideration. In part III of the thesis we discuss these type of quantum simulators in more detail. In the second line, which is the focus of part II of the thesis, one analyses what could be achieved with a quantum computer. Here one is interested in finding quantum algorithms for calculating not only dynamical, but also static properties of many-body systems, as well as in delineating the limitations of quantum computation to this aim. 

The goal of better understanding what can be efficiently computed by quantum means is a motivation for the study of quantum complexity theory \cite{KSV02, Wat08}. This theory analyses the issues related to the amount of quantum resources, such as the number of quantum bits or the number of basic quantum operations, needed to solve computational problems. As our best understanding of the physical structure of the universe is quantum mechanical, quantum complexity theory is actually more fundamental than its classical counterpart in the identification of what can in principle be efficiently computed or proved.  

In the following chapters we study the use of quantum computers for calculating properties of many-body systems. In chapter \ref{partition} we show that a quantum computer is helpful in approximating partition functions and spectral densities of quantum local Hamiltonians. In chapter \ref{QCMA}, in turn, we prove, under certain complexity theoretical assumptions, that not even a quantum computer can efficiently calculate approximations to the ground state energy of one-dimensional Hamiltonians with an inverse polynomial spectral gap\footnote{The spectral gap is the difference of the energy of the first excited state to the ground state energy of the Hamiltonian.}. 

The objective of the present chapter is two-fold. First we review some definitions and results of quantum complexity theory that will be used in the following two chapters. Second we make a connection with part I of the thesis by presenting an example of a very recent trend: the application of ideas from entanglement theory in quantum complexity theory \cite{ABD+08, LTW08}. We apply Yang's monogamy inequality \cite{Yan06} to solve a problem recently posed by Aaronson \textit{et al} concerning the role of unentangled proofs in quantum proof systems \cite{ABD+08}. On the way to establish the result we also solve an open problem of entanglement theory, raised in Refs. \cite{Tuc02, YHW07, HHHH07}: we show that in the calculation of the squashed entanglement the extension of the system cannot always be taken to be classical, i.e. the squashed entanglement is not equal to the classical squashed entanglement (see section \ref{othermasures}). The analysis we carry out also suggests a curious feature of quantum correlations: locally accessible entanglement appears to be much more monogamic than global entanglement. 

The organization of this chapter is the following. In section \ref{BQP} we present the definition of the class of problems solved in polynomial time by a quantum computer (\BQP), while in subsection \ref{PhaseEst} we review an important quantum algorithmic primitive: phase estimation \cite{KSV02, NC00}. In section \ref{DQC1} we overview the one-clean-qubit model of quantum computation \cite{KL98, SJ08}. Section \ref{QMA} is devoted to the discussion of the quantum analogue of \NP\footnote{More precisely of the probabilistic version of \NP, \MA \hspace{0.05 cm} \cite{AB08}.}, \QMA, and to review Kitaev's result that the determination of the ground state energy of local Hamiltonians is \QMA-complete \cite{KSV02}, together with further developments on it \cite{KKR06, OT05, AGIK07, SCV08, BT08}. Finally, in section \ref{QMA2} we define the classes \QMA($k$) and \BellQMA($k$) \cite{KMY03, ABD+08} and show, using tools from entanglement theory, that, for every fixed $k$, \BellQMA($k$) = \QMA. 
  
\section{Bounded Error Quantum Polynomial} \label{BQP}

In order to define the class of problems efficiently solved by a quantum computer with a small probability of error we first introduce some standard terminology in complexity theory \cite{AB08}. We follow closely the presentation of Ref. \cite{Wat08}. 

Complexity classes are most naturally defined in terms of decision problems, which are computational problems that have a binary answer. In the context of quantum complexity theory it is also helpful to consider promise problems, which are decision problems for which the input is assumed to be drawn from some subset of all possible input strings. More specifically, let $\Sigma = \{ 0, 1\}$ be the binary alphabet and $\Sigma^*$ the set of all finite binary strings. A promise problem is a pair $(\A_{\YES}, \A_{\NO}) \subseteq \Sigma^*$, with $\A_{\YES}, \A_{\NO}$ two disjoint sets. The set $\A_{\YES}$ contains the \YES-instances of the problem, while $\A_{\NO}$ is composed of the \NO-instances. Given a $L \in \Sigma^*$ with the promise that it belongs either to $\A_{\YES}$ or to $\A_{\NO}$, we should decide which is the case. 

Classical complexity classes are defined in terms of a Turing machine, which formalizes the notion of uniform computation\footnote{The reader is referred to Ref. \cite{AB08} for a detailed account on Turing machines.}. An important class is the one containing all the promise problems which can be solved in polynomial time by a probabilistic Turing machine with a small probability of error.

\begin{definition}
\BPP: A promise problem $A = (\A_{\YES}, \A_{\NO})$ is in \BPP \hspace{0.05 cm} (Bounded-error, Probabilistic, Polynomial time) if, and only if, there exists a polynomial-time probabilistic Turing machine $M$ that accepts every string $L \in \A_{\YES}$ with probability at least $a(n)$, and accepts every string $L \in \A_{\NO}$ with probability at most $b(n)$, with $a(n) - b(n) \geq 1/\poly(n)$. If $a = 1$ and $b = 0$, then $L \in \P$.
\end{definition}
Note that the error bounds are not fundamental. Indeed, by repeating the computation a sufficient, but polynomial, number of times one can obtain overwhelming statistical evidence of the correct answer and hence reduce the error probabilities exponentially \cite{AB08}.  

A quantum circuit acting on $n$ qubits is a unitary operation formed by the concatenation of quantum gates, unitaries which act non-trivially 
only on a single or two qubits\footnote{A detailed presentation of quantum circuits and quantum computation can be found in Ref. \cite{NC00}.}. The number of quantum gates of a circuit defines its size. By the Solovay-Kitaev theorem \cite{KSV02, NC00}, any quantum circuit can be approximated by another circuit composed of gates from a finite set of gates and with only a poly-logarithmic blow up in the size of the original circuit. Such sets of single and two-qubit unitaries capable of approximately generating any other unitary are called universal sets \cite{NC00}. 

We will model quantum algorithms by families of quantum circuits formed by gates from some universal set. For each input string $L_n$ of $n$ bits, we should be able to find, in polynomial time in a Turing machine, a classical description of a polynomial quantum circuit $Q_{L_n}$\footnote{Such a requirement captures the uniformity requirement in the quantum setting. General families of quantum circuits define the class of quantum computation with a classical advice, \BQP/\poly \hspace{0.05 cm} \cite{Wat08}.}$^{,}$\footnote{In terms of a sequence of single and two qubit gates drawn from some discrete universal set.}, such that a measurement in the computational basis of the first qubit of the state $Q_{L_n}\ket{0}^{\otimes \poly(n)}$ reveals with high probability the answer to the problem. 

\begin{definition} 
\BQP: Let $A = (\A_{\YES}, \A_{\NO})$ be a promise problem and let $a, b : \mathbb{N} \rightarrow [0, 1]$ be functions. Then $A \in \BQP(a, b)$ if, and only if, for every instance $L_n \in A$ of $n$ bits, there exists a polynomial-time generated quantum circuit $Q_{L_n}$ acting on $q(n) = \poly(n)$ qubits such that
\begin{enumerate}
	\item If $L_n \in \A_{\YES}$, then $\tr(Q_{L_n} (\ket{0}\bra{0})^{\otimes q(n)} Q_{L_n}^{\cal y} (\ket{1}\bra{1} \otimes \id^{\otimes q(n) - 1})) \geq a(n)$.
	\item If $L_n \in \A_{\NO}$, then $\tr(Q_{L_n} (\ket{0}\bra{0})^{\otimes q(n)} Q_{L_n}^{\cal y} (\ket{1}\bra{1} \otimes  \id^{\otimes q(n) - 1})) \leq b(n)$.
\end{enumerate}
We define $\BQP = \BQP(2/3, 1/3)$.
\end{definition}
As in the case of \BPP, it is possible reduce the error bounds by repeating the computation several times. As long as $a(n) - b(n) \geq 1/\poly(n)$, we have $\BQP(a, b) = \BQP(1 - 2^{-\poly(n)}, 2^{-\poly(n)})$ \cite{KSV02, Wat08}.

It is widely believed that $\BPP \subset \BQP$, based on several quantum algorithms for problems for which no efficient classical algorithm is known. As it holds that $\BQP \subseteq \PSPACE$ \hspace{0.05 cm} \cite{BV97}\footnote{\PSPACE \hspace{0.05 cm} is the class of problems that can be solved in polynomial space, but possibly in exponential time \cite{AB08}.}, any proof that quantum computation is superior to classical would imply $\P \neq \PSPACE$ and thus have major implications in complexity theory. This should be seen as an indication of the hardness of establishing $\BPP \subset \BQP$ unconditionally. 

The notion of hardness and completeness used in classical complexity classes \cite{AB08} can also be applied to \BQP. We say that a problem $L$ is \BQP-hard if any problem in \BQP \hspace{0.05 cm} can be solved by a probabilistic classical Turing machine in polynomial time given access to an oracle to $L$ at unit cost\footnote{Here we are using the the notion of Cook reductions. In Karp reductions, in turn, a problem $L$ is hard for \BQP \hspace{0.05 cm} if for every problem $M$ in $\BQP$, there is a polynomial time reduction mapping each instance of $M$ to an instance of $L$ \cite{Wat08}.}. A problem is \BQP-complete if it is \BQP-hard and is itself contained in \BQP. There are many known \BQP-complete problems, such as the simulation of the dynamics of local Hamiltonians, including one dimensional translational invariant ones \cite{VC08}, and certain additive approximations\footnote{See chapter \ref{partition} for the definition of additive approximations.} to quadratically signed weight enumerators \cite{KL99}, to the Jones polynomial of the plat closure of braids at any primitive root of unity \cite{FKW02, FLW02, AJL06}, to the Tutte polynomial of planar graphs \cite{AAEL07}, to some mixing properties of sparse graphs \cite{JW06}, and to the contraction of tensor networks \cite{AL08}\footnote{Interestingly, the exact evaluation of all these quantities are known to be $\sharp \P$-hard.}.   

\subsection{Phase Estimation Quantum Algorithm} \label{PhaseEst}

In this subsection we review the phase estimation algorithm \cite{Kit95, CEMM98, KSV02}, a quantum algorithmic primitive that has found many applications, including Shor's factoring algorithm \cite{Sho97}, quantum algorithms for evaluating NAND formulas \cite{Amb07}, and $\BQP$-complete algorithms for certain problems concerning local Hamiltonians \cite{AL99, WZ06} and sparse graphs \cite{JW06}.  

Let $U$ be a unitary acting on $n$ qubits. Suppose we are given access to black boxes for controlled-$U^{2^j}$, with $j \in \{ 1, ..., m \}$, and an eigenvector of $U$, $\ket{\psi}$, satisfying $U \ket{\psi} = e^{2 \pi i \phi }\ket{\psi}$. Our goal is to obtain a $m$-bit precision estimate of the associated eigenvalue $\phi$. The phase estimation quantum algorithm can be used to find such a $m$-bit approximation to $\phi$, with probability larger than $1 - \epsilon$, using, in addition to the state $\ket{\psi}$ and one call to each of the $m$ oracles realizing the controlled-$U^{2^j}$, $m + O(\log(1/\epsilon))$ qubit ancillas initialized in the state $\ket{0}$ and $\poly(n, m, \log(1/\epsilon))$ quantum gates \cite{CEMM98}. The quantum circuit implementing the algorithm is a composition of the quantum Fourier transform with calls to the oracles implementing the controlled-$U^{2^j}$s and can be found in Ref. \cite{CEMM98}, together with a detailed discussion on its execution. 

A possible application of the phase estimation algorithm, which we employ in chapter \ref{partition}, is to estimate an eigenvalue of a $O(\log(n))$-local Hamiltonian\footnote{A Hamiltonian which can be written as a sum of terms in which at most $O(\log(n))$ have a non-trivial interaction.}. Assume we are given an eigenvector $\ket{\psi}$ of a $O(\log(n))-$local Hamiltonian $H$, with $|| H || \leq 1$ ( where $|| . ||$ stands for the operator norm \cite{NC00}) and would like to find an approximation to the associated eigenvalue $\lambda$ up to polynomial accuracy. To do so, we simulate the Hamiltonian for time $t$, realizing an approximation to $U_t = e^{i t H}$. As long as $t = \poly(n)$, we can implement $U_t$ in polynomial time in a quantum computer up to polynomial accuracy\footnote{Measured in the operator norm.}. Using the phase estimation algorithm it is thus possible to compute, with high probability, $\lambda \pm 1/\poly(n)$. See e.g. Ref. \cite{WZ06} for details. 

\section{One-Clean-Qubit Model} \label{DQC1}

The class \BQP \hspace{0.05 cm} is interesting not only because it represents what could be efficiently calculated if a quantum computer is built, but also because it constitutes a new class conjectured to contain hard problems for classical computation, which appears to be incomparable to the standard complexity classes. As a result, we can give strong evidence that some problems are classically intractable by showing that they are hard for quantum computation, something that could be more difficult employing classical complexity theory. In this section we look at another class of problems defined in terms of quantum computation that is also believed to contain classically intractable problems. 

The one-clean-qubit model of quantum computation, introduced by Knill and Laflamme \cite{KL98}, considers a variant of the standard quantum computation model in which all but one qubit are initialized in the maximally mixed state, and measurements can only be performed in the end of the computation. The original motivation for studying this class was in the context of NMR quantum computing, where it is hard to initialize the qubits in a pure state \cite{BCC+07}. 

\begin{definition} 
\DQC: Let $A = (\A_{\YES}, \A_{\NO})$ be a promise problem and let $a, b : \mathbb{N} \rightarrow [0, 1]$ be functions. Then $A \in \DQC$ if, and only if, for every $n$ bit string $L_n \in A$, there exists a polynomial-time generated quantum circuit $Q_{L_n}$ acting on $q(n) = \poly(n)$ qubits such that
\begin{enumerate}
	\item If $L_n \in \A_{\YES}$, then $\tr(Q_{L_n} \left(\ket{0}\bra{0}\otimes \frac{\id^{\otimes q(n) - 1}}{2^{q(n) - 1}}\right) Q_{L_n}^{\cal y} (\ket{1}\bra{1} \otimes \id^{\otimes q(n) - 1})) \geq a(n)$.
	\item If $L_n \in \A_{\NO}$, then $\tr(Q_{L_n} \left(\ket{0}\bra{0} \otimes \frac{\id^{\otimes q(n) - 1}}{2^{q(n) - 1}}\right) Q_{L_n}^{\cal y} (\ket{1}\bra{1} \otimes  \id^{\otimes q(n) - 1})) \leq b(n)$.
\end{enumerate}
\end{definition}

Similarly to \BQP \hspace{0.05 cm} we can reduce the errors exponentially by repeating the computation a polynomial number of times. However, here such amplification of the success probability must be done outside the model, i.e. unlike for \BQP, it is in general not possible to decrease the error probabilities by adding more gates to the circuit or more maximally mixed ancilla qubits. As shown in Ref. \cite{SJ08}, the model remains the same if we allow up to $\log(n)$ input clean qubits or if we initialize the clean qubit in a partially depolarized state \cite{KL99}. 

Since in the definition of the class we allowed for free classical computation, it is clear that $\BPP \subseteq \DQC$. It is also easy to see that $\DQC \subseteq \BQP$. In fact both inclusions are believed to be strict, i.e. \DQC \hspace{0.05 cm} is conjectured to contain classically hard problems and to be strictly weaker than all of quantum computation. 

A reason for conjecturing that $\BPP \subset \DQC$ is the existence of problems in $\DQC$ for which no efficient classical solution is known. These, which turn out to be complete for the class, include certain additive approximations to the trace of quantum circuits \cite{KL98, She06}, to quadratically signed weight enumerators \cite{KL99}, and to the Jones polynomial of the trace closure of braids \cite{AJL06, SJ08}. In chapter \ref{partition} we show that certain additive approximations to partition functions and approximations to the spectral density of local Hamiltonians are \DQC-hard.  

\section{Quantum Merlin-Arthur} \label{QMA}

The class $\NP$ (non-Deterministic Polynomial Time)\footnote{This name originates from an alternative definition of $\NP$ as the class of problems solved in polynomial time by a non-deterministic Turing machine. \cite{AB08}}, formed by the problems for which a $\YES$ answer can be checked in polynomial time, is one of the most fundamental in complexity theory. Given the intuitive fact that to check the correctness of a proof to a statement is in general much easier than to actually find a proof, it is widely believe that $\P \neq \NP$, which is arguably the most important open question in complexity theory \cite{AB08}. Moreover, the theory of $\NP$-completeness has key importance in the whole of computer science first as a tool to attest the hardness of certain problems and second as building block to breakthrough results such as the PCP Theorem \cite{AB08}. 

Given the distinguished role of $\NP$ in classical complexity theory it is interesting to investigate the analogous quantum class. In order to so, first we should introduce a probabilistic version of $\NP$, for which the quantum generalization is most meaningful. 
\begin{definition}
$\MA$: A promise problem $A = (\A_{\YES}, \A_{\NO})$ is in $\MA$ (Merlin-Arthur) if, and only if, for every $L_n \in A$ of $n$ bits, there exists a polynomial-time probabilistic Turing machine $M$ and two functions $a, b$, with $a(n) - b(n) \geq 1/\poly(n)$ such that 
\begin{itemize}
	\item (completeness) If $L_n \in \A_{\YES}$, then there is a $\poly(n)$-size witness which makes $M$ to accept with probability larger than $a(n)$.
	\item (soundness) If $L_n \in \A_{\NO}$, then no witness makes $M$ to accept with probability larger than $b(n)$.
\end{itemize}
\end{definition}
As before, the soundness and completeness errors can be decreased to $2^{- \poly(n)}$ \cite{AB08}. The name Merlin-Arthur follows from the interpretation of the class as a game played by two parties. Merlin, who has infinite computational power tries to convince Arthur, who is limited to perform polynomial time classical computation, that $L \in \A_{\YES}$ by sending a proof of this fact. Case the instance is indeed in $\A_{\YES}$, Merlin should be able to convince Arthur with probability larger than $a(n)$, while if the problem is actually in $\A_{\NO}$, then Arthur should be convinced with probability at most $b(n)$, no matter which witness Merlin sends to him.   

The class $\QMA$ is a generalization of $\MA$, where now Merlin sends a quantum state of polynomially many qubits as a witness to Arthur, who has access to a quantum computer to verify its correctness \cite{Wat00, KSV02}.

\begin{definition} \label{QMAdef}
\QMA: Let $A = (\A_{\YES}, \A_{\NO})$ be a promise problem and let $a, b : \mathbb{N} \rightarrow [0, 1]$ be functions. Then $A \in \QMA(a, b)$ if, and only if, for every $n$ bit string $L_n \in \A$, there exists a polynomial-time generated quantum circuit $Q_{L_n}$ acting on $N = \poly(n)$ qubits and a function $m = \poly(n)$ such that
\begin{enumerate}
	\item (completeness) If $L_n \in A_{\YES}$, there exists a state $\ket{\psi} \in (\mathbb{C}^2)^{\otimes m}$ satisfying 
\begin{equation*}	
\tr(Q_{L_n} \left((\ket{0}\bra{0})^{\otimes N - m}\otimes \ket{\psi}\bra{\psi} \right) Q_{L_n}^{\cal y} (\ket{1}\bra{1} \otimes \id^{\otimes N - 1})) \geq a(n).
\end{equation*}
	\item (soundness) If $L_n \in \A_{\NO}$, for every state $\ket{\psi} \in (\mathbb{C}^2)^{\otimes m}$, 
\begin{equation*}	
\tr(Q_{L_n} \left((\ket{0}\bra{0})^{\otimes N - m}\otimes \ket{\psi}\bra{\psi} \right) Q_{L_n}^{\cal y} (\ket{1}\bra{1} \otimes \id^{\otimes N - 1})) \leq b(n).
\end{equation*}	
\end{enumerate}
We define $\QMA = \QMA(2/3,1/3)$. 
\end{definition}
It is also possible to amplify the soundness and completeness of the protocol exponentially, i.e. $\QMA(a, b) = \QMA(1 - 2^{-\poly(n)}, 2^{-\poly(n)})$ for any $a(n) - b(n) \geq 1/\poly(n)$ \cite{KSV02}. An interesting result by Marriott and Watrous shows that it is possible to realize such an amplification procedure without enlarging the witness size \cite{MW05}. 

There is an interesting intermediate class between $\MA$ and $\QMA$, first considered in Ref. \cite{AN02}, in which the witness Merlin sends to Arthur is classical, but Arthur has a quantum computer to check it.
\begin{definition} \label{QCMAdef}
\QCMA: Let $A = (\A_{\YES}, \A_{\NO})$ be a promise problem and let $a, b : \mathbb{N} \rightarrow [0, 1]$ be functions. Then $A \in \QMA(a, b)$ if, and only if, for every $n$ bit string $L_n \in \A$, there exists a polynomial-time generated quantum circuit $Q_{L_n}$ acting on $N = \poly(n)$ qubits and a function $m = \poly(n)$ such that
\begin{enumerate}
	\item (completeness) If $L_n \in A_{\YES}$, there exists a computational basis state $\ket{\psi} = \ket{a_1, ..., a_m}$, with $a_i \in \{ 0, 1 \}$, satisfying 
\begin{equation*}	
\tr(Q_{L_n} \left((\ket{0}\bra{0})^{\otimes N - m}\otimes \ket{\psi}\bra{\psi} \right) Q_{L_n}^{\cal y} (\ket{1}\bra{1} \otimes \id^{\otimes N - 1})) \geq a(n).
\end{equation*}
	\item (soundness) If $L_n \in \A_{\NO}$, for any computational basis state $\ket{\psi} = \ket{a_1, ..., a_m}$, with $a_i \in \{ 0, 1 \}$,
\begin{equation*}	
\tr(Q_{L_n} \left((\ket{0}\bra{0})^{\otimes N - m}\otimes \ket{\psi}\bra{\psi} \right) Q_{L_n}^{\cal y} (\ket{1}\bra{1} \otimes \id^{\otimes N - 1})) \leq b(n).
\end{equation*}	
\end{enumerate}
We define $\QCMA = \QCMA(2/3,1/3)$. 
\end{definition}

The relation of $\QCMA$ and $\QMA$ was explored in Ref. \cite{AK06}, where a \textit{quantum} oracle separation was presented. As $\QMA$, it is also believed that $\QCMA$ is strictly larger than $\MA$. We will revisit the $\QCMA$ in chapter \ref{QCMA}, in which we consider the computational complexity of calculating ground-states of pol-gapped quantum local Hamiltonians.

\subsection{The Local Hamiltonian Problem} \label{lochamprob}

The theory of $\NP$-completeness was started with the celebrated Cook-Levin Theorem, which shows that $\SAT$ is $\NP$-complete\footnote{In $\SAT$ we are given a set of clauses on $n$ variables and and should determine if there is an assignment that satisfy all the clauses \cite{AB08}.}. In this section we review a seminal result by Kitaev which can be considered the quantum counterpart of Cook-Levin Theorem: the $\QMA$-completeness of the \locHam \hspace{0.05 cm} problem \cite{KSV02}. After Kitaev's result, other $\QMA$-complete problems have been identified, including checking if a given unitary is close to the identity \cite{JWB03}, deciding if a set of density operators approximates the reductions of a given global density operator \cite{Liu06}, and its fermionic version as the $N$-representability problem \cite{LCV07}.

\begin{definition}
\locHam: We are given a $k$-local Hamiltonian\footnote{\normalfont A $k$-local Hamiltonian acting on ${\cal H}^{\otimes n}$ is defined as a sum of terms where each is an Hermitian operator acting on at most $k$ qubits.} on $n$ qubits $H = \sum_{j=1}^r H_j$ with $r = \poly(n)$. Each $H_j$ has a bounded operator norm $|| H_j || \leq \poly(n)$ and its entries are specified by $\poly(n)$ bits. We are also given two constants $a$ and $b$ with $b - a \geq 1/\poly(n)$. In $\YES$ instances, the smallest eigenvalue of $H$ is at most $a$. In $\NO$ instances, it is larger than $b$. We should decide which one is the case.
\end{definition}

Kitaev's original theorem is that 5-\locHam \hspace{0.05 cm} is $\QMA$-complete \cite{KSV02}. After it, it was shown in Ref. \cite{KR03} that 3-\locHam \hspace{0.05 cm} is $\QMA$-complete and then in Ref. \cite{KKR06} that already 2-\locHam \hspace{0.05 cm} is $\QMA$-complete. The latter result is in a sense optimal since 1-\locHam \hspace{0.05 cm} clearly belongs to $\P$. These results, however, do not say anything about the dimensionality of the Hamiltonian. Such a line of investigation was pursued in Ref. \cite{OT05}, where it was shown that 2-\locHam \hspace{0.05 cm} with nearest neighbors interactions of particles arranged in a two dimensional square lattice is $\QMA$-complete. A surprising result of Aharonov, Gottesman, and Kempe \cite{AGIK07} is that 2-\locHam \hspace{0.05 cm} for particles arranged in a line is already $\QMA$-complete\footnote{For their construction one needs to employ 12 state particles, instead of qubits.}, in sharp contrast to the case of classical Hamiltonians, for which a solution can be classically computed in polynomial time. More studies on the complexity of the \locHam \hspace{0.05 cm} problem, with different assumptions on the type of Hamiltonians considered, were reported in Refs. \cite{BDOT06, SV07, Kay07, BT08, Kay08, VC08, SCV08}. We discuss some of them in chapter \ref{QCMA}.

We now outline the main steps to prove that 5-\locHam \hspace{0.05 cm} is $\QMA$-complete, as we employ such a construction in chapters \ref{partition} and \ref{QCMA}. We actually use elements of the proof by Kempe, Kitaev, and Regev \cite{KKR06} that 2-\locHam \hspace{0.05 cm} is $\QMA$-complete, but stick to 5-local Hamiltonians for simplicity. That 5-\locHam \hspace{0.05 cm} belongs to $\QMA$ is simple. A witness is the ground state of the Hamiltonian, whose energy can be calculated to polynomial accuracy by the phase estimation algorithm discussed in section \ref{PhaseEst} \cite{KSV02}. The hardness part, as usual, is the more involved direction. 

Let $A \in \QMA(1 - \delta, \delta)$. We would like to find an encoding of every $L \in A$ into a local Hamiltonian such that its minimum eigenvalue is smaller than $a$ case $L \in \A_{\YES}$ and larger than $b$ if $L \in \A_{\NO}$, with $b - a \geq 1/\poly(n)$. Given a $r$-bit string $L_r \in A$, let $Q = U_T... U_2 U_1$ be the quantum verifier circuit composed of $T = \poly(r)$ single and two qubit unitaries $U_j$ and operating on $N = \poly(r)$ qubits. We assume that the last $N - m$ qubits are initialized in the zero state, while the first $m$ qubits contains the proof, and the output of the circuit is written in the first qubit. 

The encoding Hamiltonian is defined on the space of $n := N + T$ qubits and is divided into two register. The first $N$ qubits encodes the computation, whereas the last $T$ represents the possible values of the clock, responsible for keeping track of the correct temporal order of the circuit. The Hamiltonian reads
\begin{equation} \label{HamiltonianQMA}
H_{\text{out}} + J_{\text{in}} H_{\text{in}} + J_{\text{prop}} H_{\text{prop}} + J_{\text{clock}} H_{\text{clock}},     
\end{equation}   
where
\begin{equation} \label{Hkit1}
H_{\text{out}} = (T + 1)\ket{0}\bra{0}_1 \otimes \ket{1}\bra{1}_T^c, \hspace{1 cm} H_{\text{in}} = \sum_{i=m+1}^N \ket{1}\bra{1}_i \otimes \ket{000}\bra{000}_{1,2,3}^c,
\end{equation}
\begin{equation} \label{Hkit2}
H_{\text{clock}} = \sum_{i=1}^{T-1}\id \otimes \ket{01}\bra{01}_{i, i + 1}^c, \hspace{0.5 cm} H_{\text{prop}} = \sum_{i=1}^T H_{\text{prop}, i},
\end{equation}
with
\begin{eqnarray} \label{Hkit3}
H_{\text{prop}, i} &=&  \id \otimes \ket{100}\bra{100}_{i-1, i, i+1}^c - U_i \otimes \ket{110}\bra{100}_{i-1, i, i+1}^c  \nonumber \\ &-&  U_i^{\cal y} \otimes \ket{100}\bra{110}_{i-1, i, i+1}^c + \id \otimes \ket{110}\bra{110}_{i-1, i, i+1}^c. 
\end{eqnarray}
For $H_{\text{prop}, 1}$ and $H_{\text{prop}, T}$ we modify the above Hamiltonian accordingly and disconsider the clock terms that would correspond to $i = 0, T + 1$ (see e.g. \cite{OT05}). The coefficients $J_{\text{in}}, J_{\text{prop}}$ and $J_{\text{clock}}$ will be chosen to be $\poly(n)$ factors whose exact expression will be determined later on. As noted above, the Hamiltonian is defined in a space consisting of two registers. The first, containing $N$ qubits, stores the computation data, while the second, of $T$ qubits and labeled by a $c$ superscript, represents the possible values of the clock\footnote{We are counting time in the unary representation, e.g. $11110$ denotes 4.}. As each $U_j$ is either a single or a two-qubit gate, the final Hamiltonian is 5-local. 

Each term in Eq. (\ref{HamiltonianQMA}) has a specific purpose. First, $H_{\text{clock}}$ forces the clock state to be of the form $\ket{1}^{\otimes l} \otimes \ket{0}^{T - l}$, for some $0 \leq l \leq T$, by adding an energy penalty to any basis state of the clock that contains the sequence $01$. The role of $H_{\text{in}}$ is to make sure that when the clock is zero, then the ancilla qubits are properly initialized in the $\ket{0}$ state. The term $H_{\text{out}}$ checks that the output qubit indicates acceptance by the verifying circuit $Q$. Finally, $H_{\text{prop}}$ checks that the propagation follows $Q$. 

Let us first consider the case in which $Q$ accepts with probability larger than $1 - \delta$ for some witness $\ket{\xi}$. Define
\begin{equation*}
\ket{\eta_{\xi}} = \frac{1}{\sqrt{T + 1}} \sum_{t=0}^T U_t ... U_1\ket{0, \xi}\otimes \ket{1}^{\otimes t} \otimes \ket{0}^{\otimes T - t}.
\end{equation*}
A simple calculation shows that $\bra{\eta_{\xi}}H_{\text{prop}}\ket{\eta_{\xi}} = \bra{\eta_{\xi}}H_{\text{clock}}\ket{\eta_{\xi}} = \bra{\eta_{\xi}}H_{\text{in}}\ket{\eta_{\xi}} = 0$ and $\bra{\eta_{\xi}}H_{\text{out}}\ket{\eta_{\xi}} \leq \delta$. Hence the ground state energy of Hamiltonian (\ref{HamiltonianQMA}) is smaller than $\delta$.

For the other direction, let us assume that $Q$ accepts with probability at most $\delta$ on every possible witness. To prove that the minimum eigenvalue will larger than $\delta$ by a non-negligible amount, we follow the perturbation theory approach of Ref. \cite{KKR06}, also pursued in Ref. \cite{OT05}. 

Consider a Hamiltonian $\tilde{H} = H + V$, where $H$ is the unperturbed Hamiltonian, which has a large spectral gap $\Delta$, and $V$ is a small perturbation. We assume that the ground state energy of $H$ is zero\footnote{The following discussion can easily be adapted to the situation where the ground state energy is non-zero, as in Refs. \cite{KKR06, OT05}. Here, for simplicity, we consider only the zero ground state energy case.}. Let $\Pi_{-}$, $\Pi_+$ be the projectors onto the groundspace of $H$ and its orthogonal complement, respectively, and define $X_{\pm \mp} = \Pi_{\pm} X \Pi_{\mp}$, for an operator $X$. Consider also the \textsl{self-energy} operator, $\Sigma_-(z)$ \cite{KKR06}, whose series expansion is given by\footnote{See section 6.1 of Ref. \cite{KKR06} for the definition of $\Sigma_-(z)$.} 
\begin{equation} \label{Sigma-}
\Sigma_-(z) = V_{--} +  \sum_{k=0}^{\infty} V_{-+}(G_{++}V_{++})^k G_{++} (V_{++}G_{++})^k V_{+-},
\end{equation}
where $G_{++}(z) = (z \id_{++} - H_{++})^{-1}$ is the resolvent of $H_{++}$ \cite{KKR06}.

\begin{theorem} \label{perturbation}
\cite{KKR06, OT05} Consider a Hamiltonian $\tilde{H} = H + V$, where $H$ has zero ground-state energy and spectral gap of $\Delta$ above its ground-space, and $V$ is such that $||V|| \leq \Delta/2$. Let $\tilde{H}|_{< \Delta/2}$ be the restriction of $\tilde{H}$ to the eigenspace of eigenvalues less than $\Delta/2$. If there is an $\epsilon > 0$ and a Hamiltonian $H_{\eff}$, with spectrum in $[a, b]$ for $a < b < \Delta/2 - \epsilon$, such that
\begin{equation*}
|| \Sigma_-(z) - H_{\eff} || \leq \epsilon
\end{equation*}
for every $z \in [a - \epsilon, b + \epsilon]$, then 
\begin{equation*} 
| \lambda_{j}(\tilde{H}|_{< \Delta/2}) - \lambda_j(H_{\eff})  | \leq \epsilon,
\end{equation*}
for every $j \in \{0, ...,  2^{\dim(\tilde{H}|_{< \Delta/2}))} - 1\}$, where $\lambda_j(X)$ is the $j$-th smallest eigenvalue of $X$. 
\end{theorem}

Let us apply this theorem to $\tilde{H} = H + V$ with $H = J_{\text{clock}}H_{\text{clock}} + J_{\text{prop}} H_{\text{prop}}$ and $V = J_{\text{in}} H_{\text{in}} + H_{\text{out}}$. A moment of thought reveals that the ground-space of $H$ is spanned by \cite{KKR06}
\begin{equation} \label{etabasis}
\ket{\eta_i} = \frac{1}{\sqrt{T + 1}} \sum_{t=0}^T U_t ... U_1\ket{i}\otimes \ket{1}^{\otimes t} \otimes \ket{0}^{\otimes T - t},
\end{equation}
for $i \in \{ 0, ..., 2^N - 1 \}$. To find a lower bound on the spectral gap of $H$ we set $J_{\text{clock}} = J_{\text{prop}}$ and follow the proof of Lemma 3.11 in Ref. \cite{AvDK+07} to obtain that $\Delta \geq J_{\text{prop}}\Omega(T^{-3})$. Let $\Pi$ be the projector onto the ground-space of $H$. Since $H_{++} = (\id - \Pi)H (\id - \Pi)$ and $\id_{++} = (\id - \Pi)$, we find
\begin{equation*}
||G_{++}(z)|| = ||(z \id_{++} - H_{++})^{-1}|| \leq |(z - \Delta)^{-1}| \leq T^{3}J_{\text{prop}}^{-1}/2,
\end{equation*}
for $z \leq \Delta/2$.
Then, from Eq. (\ref{Sigma-}) and the bound $|| V_{+-} || = ||  (\id - P) V P  || \leq || V ||$,
\begin{equation*}
\Sigma_-(z) = (T + 1)\Pi H_{\text{out}} \Pi + J_{\text{in}} \Pi H_{\text{in}} \Pi + O(T^3 J_{\text{in}}^2 J_{\text{prop}}^{-1}),  
\end{equation*}
Choosing $J_{prop} = \delta^{-2} T^3J_{\text{in}}^2$, we find that $|| \Sigma_-(z) - (T + 1)\Pi H_{\text{out}} \Pi + J_{\text{in}} \Pi H_{\text{in}} \Pi || \leq O(\delta^2)$ for every $z \leq \Delta/2$. Applying Theorem \ref{perturbation} with $H_{\eff} = (T + 1)\Pi H_{\text{out}} \Pi + J_{\text{in}} \Pi H_{\text{in}} \Pi$, $\epsilon = \delta^2$, $a = 0$ and $b = J_{\text{in}} > 1$, it follows that the spectrum of $H_{\eff}$ approximates the spectrum of Hamiltonian (\ref{HamiltonianQMA}) for energies below $\Delta/2$ to accuracy $O(\delta^2)$.  

The ground-space of $\Pi H_{\text{in}} \Pi$ is spanned by states of the form \ref{etabasis}, with $\ket{i} = \ket{0, j}$, where $j$ is a computational basis state on the last $m$ states. We can now apply Theorem \ref{perturbation} again, this time to $H = J_{\text{in}} \Pi H_{\text{in}} \Pi$ and $V = (T + 1)\Pi H_{\text{out}} \Pi$. Note that any eigenvector of $H$ not in its groundspace has energy at least $J_{\text{in}}/ (T + 1)$. Choosing $J_{\text{in}} = \delta^{-2}(T + 1) = \poly(n)$ and performing a simple calculation, we find from Theorem \ref{perturbation} that the minimum eigenvalue of Hamiltonian (\ref{HamiltonianQMA}) is larger than $1 - \delta - O(\delta^2)$, which completes the reduction. For more details see Refs. \cite{KKR06, OT05}. 

A simpler proof that the reduction works can be obtained employing the projection Lemma of Ref. \cite{KKR06}. However, the perturbation theory approach we followed gives us more information: Not only the ground-state energy of the perturbed Hamiltonian is close to the unperturbed one, but actually the whole low-lying energy eigenspectrum is close to the one of the original Hamiltonian. As we shall see in chapter \ref{partition}, this will be important in order to establish $\DQC$-hardness of certain approximations of quantum partition functions and spectral densities. 

\subsubsection{$\QMA$-completeness of the One-Dimensional Local Hamiltonian Problem} \label{QMAcomp1D}

Before we conclude our discussion on the \locHam \hspace{0.05 cm} problem, we give an overview of the result of Ref. \cite{AGIK07} that one-dimensional \locHam \hspace{0.05 cm} for 12-dimensional particles is $\QMA$-complete. The construction and proof of correctness are both rather involved and we do not attempt to provide them in full here. However, as we use this construction in chapter \ref{QCMA}, we present certain aspects of it that are relevant for the future discussion. 

The first innovation in the construction of Ref. \cite{AGIK07} is the manner in which time is dealt with. As in a Turing machine, it is assumed there is a head which runs through the one dimensional chain back and forth, performing operations in the registers on its way. More concretely, each site of the chain contains both control and data registers, forming a 12 level system. There is a set of transition rules which, based on the position in the chain and on the state of the control register of the site in which the head sits on, determines the next step of the computation. The active sites, in which the computation is performed, consists of a block of sites which move along the chain during the computation. The time of the circuit is hence encoded in the position of the active sites in the chain. For a circuit acting on $n$ qubits and consisting of $T$ gates, we consider a chain of size $nT$, where the initial state is stored in the first $n$ sites and the rest is initialized in an \textit{unborn} state. The computation is then performed by a complicated movement of the head, determined by the transition rules, which performs operations on the data and control registers of the active sites and move them along the chain. In the end, the final state of the circuit will be in the data registers of the last $n$ sites of the chain \cite{AGIK07}, while the remaining sites will be in a \textit{dead} state. 

The local Hamiltonian encoding $\QMA$ problems proposed in \cite{AGIK07} is similar to the one of Eq. (\ref{HamiltonianQMA}) and reads
\begin{equation} \label{HamiltonianQMA1D}
\tilde{H}_{\text{out}} + J_{\text{in}} \tilde{H}_{\text{in}} + J_{\text{prop}} \tilde{H}_{\text{prop}} + J_{\text{penalty}} \tilde{H}_{\text{penalty}}.
\end{equation} 
The first three terms have a similar role to the the first three terms in Eq. (\ref{HamiltonianQMA}), although they have a somehow different form (see Ref. \cite{AGIK07} for details). An important modification here is that we now have a penalty term, instead of a clock term, acting on the control registers of the sites, which penalizes forbidden configurations (e.g. configurations containing more than one head).  

A complication in this proposal is that not all illegal configurations can be checked by $\tilde{H}_{\text{penalty}}$. However, this can be compensated by the propagation term $\tilde{H}_{\text{prop}}$, as it can be shown that illegal configurations that are not detected by $\tilde{H}_{\text{penalty}}$ will evolve into configurations that are detected by it at later times. This is formalized in the Clairvoyance Lemma of Ref. \cite{AGIK07}, which we now briefly explain. 

Let us consider the ground-space ${\cal K}$ of $\tilde{H}_{\text{penalty}}$, which is spanned by valid configurations states of the control registers, and its orthogonal complement ${\cal K}^{\bot}$. It turns out that these subspaces are also invariant under $\tilde{H}_{\text{prop}}$, as this does not map illegal configurations into legal ones and vice-versa. The Clairvoyance Lemma \cite{AGIK07} tells us that the minimum energy of $\tilde{H}_{\text{prop}} + \tilde{H}_{\text{penalty}}$ restricted to the subspace ${\cal K}^{\bot}$ is $\Omega(n^{-6}T^{-3})$. Thus only legal configurations are not penalized by their joint action.

The subspace ${\cal K}$ itself can be divided into two subspaces, consisting of the ground-space ${\cal L}$ of $\tilde{H}_{\text{prop}}$ and its orthogonal complement ${\cal L}^{\bot}$ in ${\cal K}$. As $\tilde{H}_{\text{penalty}}$ is zero on ${\cal K}$, the subspace ${\cal L}$ is also the groundspace of the sum $\tilde{H}_{\text{prop}} + \tilde{H}_{\text{penalty}}$. Furthermore, using the same bounds as in the 5-local Hamiltonian case, it can be shown that the minimum eigenvalue of $\tilde{H}_{\text{prop}}$ in ${\cal L}^{\bot}$ is $\Omega(n^{-6}T^{-3})$. We can then follow the approach outlined before, looking at the terms $\tilde{H}_{\text{out}}$ and $\tilde{H}_{\text{in}}$ in the subspace ${\cal L}$, in which the form of the resulting Hamiltonian is very similar to the 5-local case and the proof of correctness of the reduction can be carried over exactly in the same manner.

\section{Unentangled Quantum Proofs: $\QMA$(k) and $\BellQMA$(k)} \label{QMA2}

In this section we consider the following problem: If instead of sending one quantum proof to Arthur, Merlin sends several proofs with the promise they are not entangled, is the class of problems which can be proved by Merlin larger than $\QMA$? This question was first analysed by Kobayashi, Matsumoto, and Yamakami in Ref. \cite{KMY03}, where the following classes were introduced. 

\begin{definition} \label{QMAkdef}
\QMA{\normalfont(k)}: Let $A = (\A_{\YES}, \A_{\NO})$ be a promise problem and let $a, b : \mathbb{N} \rightarrow [0, 1]$ be functions. Then $A \in \QMA(k, a, b)$ if, and only if, for every $n$ bit string $L_n \in \A$, there exists a polynomial-time generated quantum circuit $Q_{L_n}$ acting on $N = \poly(n)$ qubits and a function $m = \poly(n)$ such that
\begin{enumerate}
	\item (completeness) If $L_n \in \A_{\YES}$, then there exists a set of states \{ $\ket{\psi_j} \in (\mathbb{C}^2)^{\otimes m} \}_{j=1}^k$ satisfying 
\begin{equation*}	
\tr(Q_{L_n} \left((\ket{0}\bra{0})^{\otimes N - km}\otimes \ket{\psi_1}\bra{\psi_1} \otimes ... \otimes \ket{\psi_k}\bra{\psi_k} \right) Q_{L_n}^{\cal y} (\ket{1}\bra{1} \otimes \id^{\otimes N - 1})) \geq a(n).
\end{equation*}
	\item (soundness) If $L_n \in \A_{\NO}$, then for every set of states $\{ \ket{\psi_j} \in (\mathbb{C}^2)^{\otimes m} \}_{j=1}^k$, 
\begin{equation*}	
\tr(Q_{L_n} \left((\ket{0}\bra{0})^{\otimes N - km}\otimes \ket{\psi_1}\bra{\psi_1} \otimes ... \otimes \ket{\psi_k}\bra{\psi_k} \right) Q_{L_n}^{\cal y} (\ket{1}\bra{1} \otimes \id^{\otimes N - 1})) \leq b(n).
\end{equation*}	
\end{enumerate}
Moreover, if each of the $k$ witnesses consists of at most $r$ qubits, then we say that $A \in \QMA_r(k, a, b)$. Finally, we define $\QMA(k) = \QMA(k, 2/3,1/3)$. 
\end{definition}

While it is clear that $\QMA(k, a, b) \subseteq \QMA(1, a, b)$, the converse inclusion is an open question and is conjectured not to hold. Indeed, a naive attempt to simulate a protocol for $\QMA(k, a, b)$ with a single Merlin fails in general, as he might cheat by entangling the $k$ witnesses, and no direct method for testing the non-existence of such entanglement is presently known. We hence seem to have a task for which the promise of not having any entanglement is actually helpful, which goes in the opposite way to the usual quantum information paradigm of entanglement as a resource. 

Recently two breakthrough results have shown the inherent significance of the classes $\QMA$(k). First Blier and Tapp proved that the problem  3-$\COLORING$, which is known to be NP-complete \cite{AB08}, belongs to $\QMA_{\log(n)}(2, 1, 24n^{-6})$ \cite{BT07}. This result is remarkable because a classical proof of similar size would imply $\P = \NP$. Even a single quantum proof of $O(\log(n))$ qubits is unlikely, as it would imply $\NP \subseteq \BQP$\footnote{Suppose 3-$\COLORING \in \QMA_{O(\log(n))}(1, 1, 1/\poly(n))$. Then using Marriott and Watrous \cite{MW05} amplification procedure it follows that 3-$\COLORING \in \QMA_{O(\log(n))}(1, 1 - 2^{- \poly(n)}, 2^{-\poly(n)})$, which is in $\BQP$ by Theorem 3.6 of \cite{MW05}.}. A caveat of this result, however, is that cheating Merlins can only be identified with a polynomially small probability. In Ref. \cite{ABD+08} Aaronson \textit{et al} showed that such a drawback can be circumvented by using $\sqrt{n}$ unentangled witnesses instead of two, as they proved that 3$\SAT$ belongs to $\QMA_{\log(n)}(\poly\log(n)\sqrt{n}, 1, 1/\poly(n))$. Again, it is not expected that 3$\SAT$ has sublinear classical proofs, as this would imply a subexponential classical algorithm to it. 

Given these intriguing results on the power of multiple unentangled quantum proofs, a throughout characterization of the classes $\QMA(k)$ seems in order. It turns out that even basic questions concerning $\QMA(k)$, which for $\QMA$ can readily be solved, are still open problems for $k \geq 2$. For example, it is not known whether it is possible to perform error amplification for $\QMA(k)$, $k \geq 2$, nor the relation between the classes $\QMA(k)$ among themselves and with $\QMA$. 

In Ref. \cite{ABD+08} an application of entanglement theory to these open problems was found. It was shown that assuming the superadditivity of the entanglement of formation (see section \ref{costdistillable}), $\QMA(k)$ can be amplified to exponentially small error for every $k \geq 1$, i.e. $\QMA(k, a, b) = \QMA(k, 1- 2^{-\poly(n)}, 2^{-\poly(n)})$ for $b - a \geq 1/\poly(n)$, and that $\QMA(k) = \QMA(2)$ for $k \geq 2$. For the argument all that is actually needed is a superadditive entanglement measure which satisfies properties 1. and 4. of section \ref{entmeasures}\footnote{For the argument of Ref. \cite{ABD+08}, these two properties can be replaced by the bound $E(\rho) \leq \log(D)$, where $D$ is the dimension of the Hilbert space in which $\rho$ acts on, which is a consequence of properties 1. and 4.}, and a stronger form of faithfulness, termed polynomial faithfulness \cite{ABD+08}, which requires that if $E(\rho) \leq \epsilon$, then there is a separable state which is $f(\epsilon)\poly(\log(D))$-close to $\rho$ in trace norm, where $D$ is the dimension of the Hilbert space in which $\rho$ acts on and $f$ is a real function which goes to zero when $\epsilon$ does so. The reason why the conjecture on the superadditivity of the entanglement of formation had to be used is that no entanglement measure satisfying these four properties is presently known. This shows that if the additivity conjecture of quantum information theory is true, then $E_F$ has quite unique properties, which can be seen as yet another reason for the notorious hardness of solving the conjecture.  

In Ref. \cite{ABD+08} an interesting subclass of $\QMA(k)$ was considered. In the class $\BellQMA(k)$, not only Merlin is restricted to send $k$ unentangled proofs, but Arthur is also constrained only to perform separate measurements on each witness and then postprocess classically the obtained outcomes, i.e., Arthur is restricted to realize a Bell experiment on the witnesses. Even though no entangling measurements can be performed, the absence of entanglement among the witnesses could a priori still be helpful. In Ref. \cite{ABD+08} the relation of $\BellQMA(k)$ with $\QMA(k)$ was left as an open problem. In the rest of this section we consider such a question and prove the following theorem. 

\begin{theorem} \label{teoBellQMA}
For every constant $k$, $\BellQMA(k) = \QMA$. 
\end{theorem}

To motivate the main idea of the proof of Theorem \ref{teoBellQMA}, let us consider the following strategy to show that $\QMA = \QMA(k)$ (which as will be seen shortly fails). Suppose we had an entanglement measure $E$ which (i) is polynomially faithful, (ii) bounded in the sense that for every $\rho$, $E(\rho) \leq \log(D)$, where $D$ is the dimension of the Hilbert space in which $\rho$ acts on, and (iii) satisfies the monogamy inequality
\begin{equation} \label{mon}
E(\rho_{A:BB'}) \geq E(\rho_{AB}) + E(\rho_{AB'}).
\end{equation}
Then we could use it to show that $\QMA(2, a, b) \subseteq \QMA$, for every $a, b$ with $a - b \geq 1/\poly(n)$\footnote{We would actually be able to prove that $\QMA(k, a, b) \subseteq \QMA$, for every $k \geq 2$, but for simplicity we only discuss the $k=2$ case.}. Indeed, consider a problem in $\QMA(2, a, b)$ for which Merlin sends two witnesses of size $r = \poly(n)$ qubits each to Arthur. We then consider the following protocol in $\QMA$ for the same problem: Arthur expects a single witness of size $(N+1)r$, where $N = \poly(n)$, and divide it into $N+1$ registers of $r$ qubits each. He then symmetrizes the last $N$ registers by applying a random permutation from ${\cal S}_{N}$. Finally, he discards all the registers except the first two and applies the verification procedure of $\QMA(2, a, b)$ to them. Case the solution of the problem is $\YES$, Merlin can send the state $\ket{\psi_1} \otimes \ket{\psi_2}^{\otimes N}$, where $\ket{\psi_1} \otimes \ket{\psi_2}$ is a witness for $\QMA(2, a, b)$, making Arthur to accept with probability larger than $a$, i.e. the protocol for $\QMA$ has completeness $a$. In order to analyse the soundness of the protocol, let $\rho_{1,2,...,N}$ be Arthur's state after the randomization step. Then, from properties (ii) and (iii) of $E$ we find that $r \geq E(\rho_{1:2,...,N}) \geq N E(\rho_{1:2})$. Hence, $E(\rho_{1:2}) \leq r/N = 1/\poly(n)$ and by property (iii), in turn, $\rho_{1:2}$ is $1/\poly(n)$-close from a separable state in trace norm. The verification procedure of $\QMA(2, a, b)$ that is applied to $\rho_{1:2}$ works with a soundness error smaller than $b + 1/\poly(n)$ and, therefore, the problem is in $\QMA(1,a, b + 1/\poly(n))$, which is equal to $\QMA$ since $a - b \geq 1/\poly(n)$. 

The strategy above fails because the three properties required from $E$ are mutually exclusive. Let us consider the following state in $(\mathbb{C}^N)^{\otimes N}$
\begin{equation} \label{psianti}
\ket{\psi_{N}} = \frac{1}{\sqrt{N}} \sum_{\pi \in {\cal S}_N} (-1)^{\text{sgn}(\pi)} \ket{\pi(1)} \otimes ... \otimes \ket{\pi(N)},
\end{equation}
where $\text{sgn}(\pi)$ is the sign of the permutation $\pi$ \cite{Sag01}. The state $\ket{\psi_N}\bra{\psi_N}$ is permutation-symmetric and its two party reductions are all equal to the anti-symmetric Werner state, which is $\Omega(1)$ away from any separable state in the trace norm. For any measure that satisfies (ii) and (iii), however, we have $E(\tr_{\backslash 1,2}(\ket{\psi_N}\bra{\psi_N})) \leq \log(N)/N$. Therefore, $E$ cannot be polynomially faithful. 

This simple observation is the key to solve an open problem concerning the squashed entanglement \cite{CW04}, raised in Refs. \cite{Tuc02, YHW07, HHHH07}: the squashed entanglement is not equal in general to the classical squashed entanglement. Indeed, as the squashed entanglement satisfies properties (ii) and (iii) above (see section \ref{othermasures}), by the discussion of the previous paragraph it cannot satisfy (ii). The classical squashed entanglement, in which the infimum of Eq. (\ref{squasdhedent}) is taken only over classical extensions of $\rho_{AB}$\footnote{A classical extension of $\rho_{AB}$ is a state of the form $\rho_{ABE} = \sum_k p_k \rho_k^{AB} \otimes \ket{k}^E\bra{k}$, for some orthonormal basis $\{ \ket{k} \}$ and such that $\tr_E(\rho_{ABE}) = \rho_{AB}$.}, is readily seen to be given by \cite{HHHH07}   
\begin{equation*} 
E_{sq}^c(\rho) = \frac{1}{2}\min_{\{ p_i, \rho_i \}} \sum_i p_i I(A:B)_{\rho_i},
\end{equation*}
where the minimum is taken over all convex combinations $\{ p_i, \rho_i \}$ of $\rho$ and $I(A:B)$ is the mutual information, given by $I(A:B) = S(A) + S(B) - S(AB)$. From the expression above it follows that $E_{sq}^c$ satisfy (ii). To show that it also satisfy (i), let $\{ p_i, \rho_i \}$ be a optimal ensemble for $\rho$ and note that 
\begin{eqnarray*}
E_{sq}^c(\rho) &=& \frac{1}{2}\sum_i p_i I(A:B)_{\rho_i} \nonumber \\ &=& \frac{1}{2}\sum_i p_i S(\rho_i || \tr_B(\rho_i) \otimes \tr_A(\rho_i)) \nonumber \\ &\geq& \frac{1}{2}\sum_i p_i E_R(\rho_i) \nonumber \\ &\geq& \frac{1}{2} E_R(\rho),
\end{eqnarray*}
where we used that $I(A:B)_{\rho} = S(\rho || \rho_A \otimes \rho_B)$ and the convexity of the relative entropy of entanglement (see section \ref{relentint}). From Pinsker's inequality \cite{Pet86} it is clear that $E_R$ is polynomially faithful and hence so it is $E_{sq}^c$. Thus $E_{sq} \neq E_{sq}^c$, as the two measures have different properties. 

After this short digression on the squashed entanglement let us turn back to the proof of Theorem \ref{teoBellQMA}. From the definition of the classes $\BellQMA(k)$ it follows that we do not need to require the state given to Arthur to be separable or close to separable in trace norm; it suffices that it behaves similarly to a separable state when separate measurements are performed on each witness. Hence we are interested in preventing the existence of entanglement that can be locally accessed. It turns out that, in a sense, such entanglement is much more monogamic than global entanglement. As a first example of this fact, consider the state given by Eq. \ref{psianti}. As already mentioned, it violates the monogamy inequality for any polynomially faithful measure, because the two party reduced density matrices are anti-symmetric Werner states, which are $\Omega(1)$ away from any separable state in trace norm. However, anti-symmetric Werner states become increasingly indistinguishable from a separable state under LOCC measurements for large dimensions and, hence, we could well have a measure satisfying the monogamy inequality and being polynomially faithful under a suitable norm that only looks at local measurements. In Ref. \cite{HLW06} it was shown that being far away from a separable state in trace norm and at the same time almost indistinguishable from a separable state by LOCC is actually a generic feature of mixed states in high dimensions. One might then take such examples as evidences to conjecture that there might be an entanglement measure satisfying properties (i) to (iii) discussed above, if we only require polynomial faithfulness with respect to the \textit{separable norm}, defined as \cite{Win08}
\begin{equation*}
|| X ||_S := \max_{0 \leq M \leq \id, M \in \text{cone}({\cal S})} \tr(M X).
\end{equation*}
All the examples known by the author show that squashed entanglement and the regularized relative entropy of entanglement could indeed be examples of such a measure. Note that if we had a measure with such properties then we could readily show that $\BellQMA(k) = \QMA$\footnote{Essentially following the strategy outlined before for $\QMA(k)$ using properties (i-iii).}. Although we do now know any such measure, it turns out that for this particular application Yang's monogamy inequality, given by Eq. (\ref{Yangmonog}), suffices. Using it we can prove the following Proposition. 

\begin{proposition} \label{propBellqma}
If $L \in \BellQMA_r(k, a, b)$, with $a - b \geq 1/\poly(n)$, then $L \in \BellQMA_{\poly(r)}(k-1, 1 - 2^{-\poly(n)}, 2^{-\poly(n)})$.
\end{proposition}

Before we turn to its proof, let us show how it implies Theorem \ref{teoBellQMA}. Applying the Proposition above $k$ times recursively, one can show that a problem in $\BellQMA_r(k, a, b)$ also belongs to $\QMA_{\underbrace{\poly(\poly(...\poly(n)))}_{k}}(1, 1 - 2^{-\poly(n)}, 2^{-\poly(n)})$, which is equal to $\QMA$, as $k$ is a constant independent of $n$.

\vspace{0.5 cm}

\begin{proof} (Proposition \ref{propBellqma})
Let $L \in \BellQMA_r(k, a, b)$. We consider that there are $k$ Merlins who send the proofs to $k$ Arthurs, who perform a measurement on the respective proofs and jointly postprocess the outcomes in a classical computer to decide whether to accept or not. If $x \in L$,  Arthurs $1$ to $k$ accept with probability larger than $a$, while if $x \notin L$, they accept with probability at most $b$. Consider the following protocol for $L$, which we show in the sequel to be in $\BellQMA_{\poly(r)}(k-1, a , b + 1/\poly(n))$: Arthur 1 expects from Merlin a witness of size $(N+1)r$, while Arthur's $2$ to $k-1$ expect witnesses of size $r$ qubits each. Arthur 1 divides the first witness into $N+1$ groups of size $r$ and denotes them by $A,B_1,...,B_N$. He then symmetrizes the $B$'s registers by performing a random permutation from ${\cal S}_N$ to them and traces out the registers $B_2$ to $B_N$. Finally he performs the measurement from the verification procedure for $\BellQMA_r(k, a, b)$ to the state on $A$ and $B_1$ as the first two proofs, while the remaining $k-2$ Arthur's also perform the measurements from the protocol for $\BellQMA_r(k, a, b)$.

Before we analyse the soundness and completeness of the protocol, let us first derive some bounds that will prove useful. Let $\rho_{A:B_1,...,B_N}$ be Arthur $1$ state after the randomization. From Eq. (\ref{Yangmonog}),  
\begin{eqnarray*}
r\geq E_F(\rho_{A:B_1,...,B_N}) \geq N G^{\leftarrow}(\rho_{AB}),
\end{eqnarray*}
Letting $\{ q_k, \rho_k \}$ be an optimal decomposition for $\rho_{AB}$ in $G^{\leftarrow}(\rho_{AB})$,
\begin{equation*}
G^{\leftarrow}(\rho_{AB}) = \sum_k q_k C^{\leftarrow}(\rho_k) \leq \frac{r}{N}.  
\end{equation*}
Let ${\cal I}$ be the set of indices $k$ for which $C^{\leftarrow}(\rho_k) \geq \epsilon$. Then,
\begin{equation} \label{equsewill1}
\sum_{k \in {\cal I}} q_k \leq \frac{r}{N \epsilon}. 
\end{equation}

For any $k \notin {\cal I}$, we have $C^{\leftarrow}(\rho_k) \leq \epsilon$. Then, for every POVM $\{ M_k \}$ on the $B$ register (which is the only register of the $B$ group that was not traced out) and every $k \notin {\cal I}$,
\begin{equation*}
S\left( \sum_j p_j \rho_{k, j} \otimes \ket{j}\bra{j}  \left \Vert  \rho_{k}^A \otimes \sum_j p_j \ket{j}\bra{j}\right. \right) = S(\rho_k^A) - \sum_j p_j S(\rho_{k, j}) \leq C^{\leftarrow}(\rho_k) \leq \epsilon,
\end{equation*}
where $p_j = \tr(\id \otimes M_j \rho_k)$, $\rho_{k, j}= \tr_B( \id \otimes M_j \rho_k)/p_j$ and $\rho_k^A = \tr_{B}(\rho_{k})$. Applying Pinsker's inequality to the equation above we get
\begin{equation*} 
\sum_{j} p_j || \rho_k^A - \rho_{k, j} ||_1 \leq \sqrt{2 \epsilon}.
\end{equation*}
Denoting by ${\cal I}_k$ the set of all indices $j$ for which $|| \rho_k^A - \rho_{k, j} ||_1 > \delta$, we find 
\begin{equation} \label{equsewill2}
\sum_{j \in {\cal I}_k} p_j  \leq \frac{\sqrt{2 \epsilon}}{\delta}.
\end{equation}

We now turn to bound the completeness and soundness parameters. If $x \in L$, Merlin 1 can send the proof $(\ket{\psi_1} \otimes \ket{\psi_2})$, while Merlin $2$ to $k-1$ send $\{ \ket{\psi_k} \}_{j=3}^k$, where $\ket{\psi_1} \otimes ... \otimes \ket{\psi_k}$ is a proof for $\BellQMA(k, a, b)$. The Arthurs will then accept with probability at least $a$. 

If $x \notin L$, we can bound the probability that the Arthurs accept as follows. Let $\{ M_k \}$ be the measurement that Arthur $2$ would realize in the verification procedure of $\BellQMA(k)$. Now such a measurement is performed in the register $B$ of Arthur $1$ and, using Eqs. (\ref{equsewill1}, \ref{equsewill2}), we find that the probability of acceptance is bounded as follows
\begin{equation*}
\sum_{k \in {\cal I}} q_k + \sum_{k \notin {\cal I}}\left( \sum_{j \in {\cal I}_k} p_j +  \sum_{j \notin {\cal I}_k} p_j (b + \delta) \right) \leq \frac{r}{N \epsilon} +  \frac{\sqrt{2 \epsilon}}{\delta} + b + \delta, 
\end{equation*}
for every $\delta, \epsilon > 0$. Such quantity can be taken to be as small as $b + 1/\poly(n)$ for suitable $0 \leq \delta, \epsilon \leq 1$ and $N = \poly(n)$. The proposition follows from the fact that $\BellQMA(k)$ can be amplified to exponentially small error by asking Merlin to send a polynomial number of copies of the original witness \cite{ABD+08}. 
\end{proof}

As a final comment, we could also define the classes $\LOCCQMA(k)$ in an analogous manner to $\BellQMA(k)$, but now allowing Arthur to perform general LOCC measurements. It is conceivable that also in this case $\LOCCQMA(k) = \QMA$. To prove this, however, we would need a monogamic entanglement measure which is polynomially faithful with respect to the separable norm, a possibility which we leave as an open problem.

\chapter{Quantum Algorithms for Partition Functions and Spectral Densities} \label{partition}

\section{Introduction}

One of the most prominent uses of a quantum computer is the simulation of the dynamics of quantum many-body Hamiltonians. Many times in the study of many-body physics, however, we are more interested in static properties than in the dynamical ones. Examples include the ground state energy, two-point correlation functions of the ground or thermal states, and the partition function of classical and quantum Hamiltonians. An interesting question is thus whether quantum computation is of any help in the calculation of such properties. For the ground state energy, it turns out that, unless $\BQP = \QMA$ or $\NP \subseteq \BQP$, its estimation to polynomial accuracy is hard even for quantum computation\footnote{Already for one dimensional quantum local Hamiltonians \cite{AGIK07} or two dimensional classical local Hamiltonians \cite{Bar82}.}. The determination of the partition function seems even harder: For classical Hamiltonians its exact evaluation is $\sharp\P$-hard\footnote{$\sharp \P$ is a class of function problems that given a boolean function, asks how many inputs are mapped to a fixed output. It can be considered as a quantitative version of $\NP$, in which instead of asking if a problem has a solution or not, one would like to know how many solutions there are.}, while even certain approximations are $\NP$-hard \cite{JS93}. Of course the complexity of finding approximations to a certain quantity depends crucially on the accuracy required. In this respect we can ask: Is there any quantum algorithm which delivers a non-trivial approximation to static properties of many-body systems? By non-trivial we mean an approximation which is unlikely to be achieved in polynomial time in a classical computer.  

For partition functions of classical Hamiltonians, such a question was posed more than ten years ago as a challenge to quantum algorithms (see e.g. \cite{LB97}) and has been the focus of intensive research recently. In Ref. \cite{AAEL07}, Aharonov, Arad, Eban, and Landau proposed a quantum algorithm to calculate additive approximations to the Tutte polynomial of any planar graph. By an additive approximation we mean an estimate to a quantity $X$ in the range $[X - \Delta/\poly(n), X + \Delta/\poly(n)]$, where the additive window of the approximation $\Delta$ is a parameter that can be easily calculated from the input of the problem. From a well-known relation between the Tutte polynomial and the partition function of the Potts model\footnote{The Pott's model is a generalization of the Ising model to more than two dimensions, see e.g. \cite{Wu82}}, their algorithm can be used to calculate an additive approximation to the the latter at any temperature and any planar configuration of the spins. Additive approximations should however always be considered with care. If the additive window of the approximation is much larger than the quantity being approximated, then no information is gained and the approximation is useless. In Ref. \cite{AAEL07} a convincing argument was given that, at least in some cases, the approximations obtained are indeed meaningful. It was shown that for some choices of parameters of the problem, it is $\BQP$-hard to find the additive approximation in question. Unfortunately, in order to show this result to the partition function of the Pott's model, complex temperatures have to be employed. An open problem left in \cite{AAEL07} is if one can show $\BQP$-hardness for additive approximations of physical instances, corresponding to positive real temperatures. 

Independently of \cite{AAEL07}, Van den Nest, D\"ur, Raussendorf, and Briegel \cite{VDB07,VDB08,VDRB08} presented different quantum algorithms for additive approximations of partition functions of the Pott's and several other classical models. For some complex-parameters regimes, they could also show that the problem solved is $\BQP$-hard. For the case of physical parameters, however, no hardness result could be established either. Also implicitly in the work of Bombin and Martin-Delgado \cite{BMD08} is a quantum algorithm for additive approximations of certain 3-body Ising models. Finally, Arad and Landau \cite{AL08} recently proposed a quantum algorithm for additive approximations of tensor network contractions and, from it, found quantum algorithms that additively approximate partition functions of a host of classical spin chains. 

In this chapter we consider the same problem from a different perspective. We present a simple quantum algorithm for finding additive approximations of \textit{quantum} Hamiltonian partition functions. We then prove that the approximation we obtain is non-trivial by showing a hardness result for physical instances of the problem, corresponding to positive temperatures. We, however, only prove a weaker hardness statement, namely that the approximation obtained for certain $O(\log(n))$-local quantum Hamiltonians is $\DQC$-hard (see section \ref{DQC1} for a definition of the class $\DQC$). For classical Hamiltonians, in turn, we show that such an approximation can be obtained in $\P$. Therefore the use of quantum Hamiltonians appears as a crucial ingredient in our approach. For establishing that the problem is $\DQC$-hard, we consider an unexplored application of Kitaev's construction of encoding a quantum circuit into a local Hamiltonian, and show that it can also be applied to problems in the one-clean-qubit model of quantum computation. Finally, the techniques we introduce can also be used to find $\DQC$-hard quantum algorithms for the spectral density of local Hamiltonians, which complements a recent classical polynomial algorithm to the problem found by Osborne \cite{Osb06} for one dimensional quantum Hamiltonians.

%Partition functions of classical Hamiltonians can in many cases be well approximated by Monte-Carlo calculations []. The use of similar techniques in quantum Monte Carlo calculations, although also useuful in many instances, has achieved much less success. The constructions we employ shed new light into the hardness of performing Monte Carlo calculations for quantum partition functions on a classical computer and the use of a quantum computer thereto. We show that some natural extensions of standard Monte Carlo techniques to quantum local Hamiltonians, even to one dimensional 2-local Hamiltonians, turn out to be $\DQC$-hard, while they can be realized in a quantum computer. 

The organization of this chapter is the following. In section \ref{defmainj} we present some definitions and state the main results. In section \ref{quantalgorithm}, in turn, we show how a quantum computer could be used to solve the problems \ref{partprob} and \ref{specprob}. Finally, in section \ref{DQC1hardsection} we prove Theorem \ref{DQC1hard} which shows that the problems solved are $\DQC$-hard.

\section{Definitions and Main Results} \label{defmainj}

Given a Hamiltonian $H$ on $n$ qubits, its partition function at inverse temperature $\beta := 1/ k_B T$, where $k_B$ is Boltzmann constant, is given by
\begin{equation*}
Z(H, \beta) := \tr(e^{-\beta H}) = \sum_{k=1}^{2^n} e^{- \beta \lambda_k},
\end{equation*}
where $\lambda_k = \lambda_k(H)$ is the $k$-th smallest eigenvalue of $H$. We also denote by $\lambda_{\min}(H)$ the minimum eigenvalue of $H$. The eigenvalue density of $H$ is defined as 
\begin{equation*}
\mu_{H}(x) := \frac{1}{2^n} \delta(E_j - x), 
\end{equation*}
where $\delta$ is Dirac's delta function, and the eigenvalue counting function as 
\begin{equation*}
N_{H}(a, b) := \int_{a}^b \mu_{H}(x) dx = \frac{1}{2^n}\sum_{k \hspace{0.01 cm} : \hspace{0.01 cm} a \leq \lambda_k \leq b} 1,
\end{equation*}
for any two real numbers $b \geq a$. It gives the proportion of eigenvalues of $H$ that are in the interval $[a, b]$. 

The first two problems we consider refer to the determination of the eigenvalue counting function and a certain additive approximation of the partition function of local Hamiltonians.  

\begin{definition} \label{partprob}
The problem \partition \hspace{0.08 cm} is defined as follows. We are given a local Hamiltonian acting on $n$ qubits $H = \sum_{j=1}^r H_j$, with $r = \poly(n)$, and three real number $1/\poly(n) \leq \beta \leq \poly(n)$, $\delta > 0$ and $\epsilon > 0$. Each $H_j$ acts on at most $m$ qubits and has bounded operator norm $||H_j|| \leq \poly(n)$. We are also given an lower bound $\lambda$ to the ground state energy of $H$, i.e. $\lambda \leq \lambda_{\min}(H)$. We should find a number $\chi$ such that, with probability larger than $\delta$,
\begin{equation*}
\left \vert \chi - \frac{Z(H, \beta)}{Z(H, 0)e^{-\beta \lambda }} \right \vert \leq \epsilon.
\end{equation*}
\end{definition}

\vspace{0.5 cm}

Note that the approximation is likely to be more useful in the high temperature limit, when $Z(H, \beta) \approx Z(H, 0) = 2^n$.

\begin{definition} \label{specprob}
The problem \spec $\hspace{0.08 cm}$ is defined as follows. We are given a quantum Hamiltonian acting on $n$ qubits $H = \sum_{j=1}^r H_j$, with $r = \poly(n)$ and real numbers $a$ and $b$, with $b - a \geq 1/\poly(n)$, $\delta > 0$ and $\epsilon > 0$. Each $H_j$ acts on at most $m$ qubits and has bounded operator norm $||H_j|| \leq \poly(n)$. We should find a number $\chi$ such that, with probability larger than $\delta$,
\begin{equation*}
\left \vert \chi - N_{H}(a, b) \right \vert \leq \epsilon.
\end{equation*}
\end{definition}

Out first result show that the two problems can be efficiently solved in a quantum computer.

\begin{theorem} \label{quantalgorithm}
There exist quantum algorithms polynomial in $(n, 1/\epsilon, \log(1/\delta), 2^m)$ for \partition \hspace{0.08 cm} and \spec. When restricted to classical Hamiltonians, there are classical polynomial algorithms for both problems.
\end{theorem}

The next result demonstrates that the problems are likely to be hard for classical computation. 

\begin{theorem} \label{DQC1hard}
\partition \hspace{0.08 cm} and \spec \hspace{0.08 cm} with $\delta > 1/2$, $\epsilon = 1/\poly(n)$, $m = O(\log(n))$, and $\lambda = \sum_{j=1}^r \lambda_{\min}(H_j)$ are $\DQC$-hard. 
\end{theorem}

An interesting open question is whether one can prove $\DQC$-hardness for strictly local Hamiltonians ($m = O(1)$). Although we believe it should indeed be possible, no construction has been found so far. An even more challenging open problem would be to prove the hardness result for one dimensional Hamiltonians. This would be particularly interesting in the case of \spec \hspace{0.08 cm}, as Osborne \cite{Osb06} recently found a classical algorithm for it, in the case of one dimensional systems, polynomial in $n$ but exponential in $\epsilon^{-1}$. $\DQC$-hardness for one dimensional Hamiltonians would then imply that, unless $\DQC = \P$, there is no classical algorithm polynomial in $n$ and $\epsilon^{-1}$, i.e. the costly error scaling of Osborne's algorithm is unavoidable.

%A very successful tool in practice for calculating partition functions of classical Hamiltonians is the Markov chain Monte-Carlo method []. The idea is that one can construct a Markov chain over the states of the spins which is garanted to converge to the Gibbs distribution when the number of steps in the chain grows.  

\section{Proof of Theorem \ref{quantalgorithm}}

The algorithms are very simple and are based on the idea that one can uniformly sample an eigenvalue of classical Hamiltonians efficiently and, with the quantum phase estimation algorithm, also of quantum Hamiltonians. Let us first present the classical algorithm and then extend the result to quantum Hamiltonians. 

Given a classical local Hamiltonian $H$ of $n$ spins, for each $k \in \{1, ..., r\}$, with $r = \poly(n)$, we pick $n$ uniform random bits $\eta_k := (i_1^k, ..., i_{n}^k)$ and compute $\chi_k := e^{-\beta H(\eta_k)} / e^{-\beta \lambda}$. We then compute the random variable  
\begin{equation} \label{av1}
\mu_v := \frac{1}{v}\sum_{k=1}^{v} \chi_k.
\end{equation}
It is clear that $\mu_v$ converges to the expectation value of the random process given by
\begin{equation} \label{exp1}
\mu := \sum_{i_1=0}^1... \sum_{i_n=0}^1  \frac{1}{2^n} \frac{e^{-\beta H(i_1, ..., i_n)}}{e^{-\beta \lambda}} = \frac{Z(H, \beta)}{Z(H, 0) e^{-\beta \lambda}}, 
\end{equation}
which is exactly the quantity we are interested in. To bound the probability that Eq. (\ref{av1}) deviates substantially from the expectation value given by Eq. (\ref{exp1}), we use Chernoff-Hoeffding's inequality: Given the sum of $r$ independent random variables $\chi_1, ..., \chi_n$ attaining values in the interval $[0, 1]$ (which is the case since $\lambda$ is a lower bound to the ground state energy), it follows that
\begin{equation*}
Pr\left( \left \vert \frac{\sum_{k=1}^v \chi_k}{v} - \mu  \right \vert \geq \epsilon  \right) \leq e^{- 2 v \epsilon^2},
\end{equation*}
which is exponential small as long as $\epsilon = 1/\poly(n)$. 

For the problem \spec $\hspace{0.1 cm}$, the reasoning is basically the same. For $k \in \{1,..., r \}$ we again choose $n$ uniform random bits $\eta_k := (i_1^k, ..., i_{n}^k)$ and compute 
\begin{equation*}
\chi_k := 
\begin{cases}
1 & H(\eta_k) \in [a, b] \\
0 & \text{otherwise}
\end{cases}
\end{equation*}
By Chernoff-Hoeffding's inequality, once more, the sum 
\begin{equation*}
\frac{1}{v}\sum_{k=1}^{v} \chi_k
\end{equation*}
with $v = \poly(n)$ is $\epsilon$-close to $N_H(a, b)$ with exponentially high probability as long as $\epsilon = 1 / \poly(n)$.

The quantum algorithms for quantum $O(\log(n))$-local Hamiltonians use the same idea as above, the only difference is that we have to find an efficient way to compute a random eigenvalue of a quantum local Hamiltonian. To this aim we make use of the quantum phase estimation algorithm, discussed in section \ref{PhaseEst}. As explained in section \ref{PhaseEst}, given a quantum $O(\log(n))$-local Hamiltonian $H$ and an eigenvector $\ket{u}$ of it, the phase estimation algorithm gives in $\poly(n)$ time an estimate of the corresponding eigenvalue $u$ to accuracy $1/\poly(n)$, with probability exponentially close to one. If instead of inputing $\ket{u}$, we input the maximally mixed state, then we obtain, with exponentially small probability of error, $\lambda_i \pm \delta$, where the indice $i$ is taken at random uniformly from $\{ 1, ..., 2^n \}$ and $\delta = 1/\poly(n)$. The analysis then goes as in the classical case. The only detail we now must take care of is to choose $\delta$ sufficiently small such that $\delta\beta = 1/\poly(n)$, in order to ensure that the errors in the estimated eigenvalues of the Hamiltonian only give an error of at most $1/\poly(n)$ to the quantity of interest.

\section{Proof of Theorem \ref{DQC1hard}} \label{DQC1hardsection}

In this section we prove Theorem  \ref{DQC1hard}. The main idea is to use Kitaev's construction of encoding a quantum circuit into a local Hamiltonian (see section \ref{lochamprob}). The key difference, however, is that while in the original $\QMA$-completeness result the solution of the problem is encoded in the ground state energy of the Hamiltonian, we will encode the solution in the average energy of the low-lying spectrum of the Hamiltonian.

Consider a problem $L \in \DQC$. Then, as discussed in section \ref{DQC1}, for every $r$ bit instance $x \in L$, there is a quantum circuit $U_x = U_T...U_2U_1$ composed of $T = \poly(r)$ quantum gates and operating on $N = \poly(r)$ qubits such that if $x \in L_{\YES}$, 
\begin{equation*}
\mu_{\YES} := \tr \left( \left( U_x\ket{0}\bra{0}\otimes \frac{\id}{2^{N-1}}U_x^{\cal y} \right) \ket{1}\bra{1}\otimes \id \right) \geq a
\end{equation*}
while if $x \in L_{\NO}$, 
\begin{equation} \label{muno}
\mu_{\YES} = \tr \left( \left( U_x\ket{0}\bra{0}\otimes \frac{\id}{2^{q(n)}}U_x^{\cal y} \right) \ket{1}\bra{1}\otimes \id \right) \leq b,
\end{equation} 
where $a - b = \Omega(1/\poly(n))$. It is hence clear that if we can estimate $\mu_{\YES}$ to polynomial accuracy, then we can determine the solution of the problem. 

Let us define the $O(\log(n))$-local Hamiltonian associated to the circuit, whose average energy of its low-lying spectrum gives $\mu_{\YES}$ up to the desired precision. The Hamiltonian, acting on $n = N + \log(T)$ qubits, is actually identical to the one used in Kitaev's original $\QMA$-completeness \cite{KSV02} proof and is given by 
\begin{equation} \label{HamiltonianDQC1}
H_{\DQC} = H_{\text{out}} + J_{\text{in}} H_{\text{in}} + J_{\text{prop}} H_{\text{prop}},     
\end{equation}   
where
\begin{equation*}
H_{\text{out}} = (T + 1)\ket{0}\bra{0}_1 \otimes \ket{T}\bra{T}, \hspace{1 cm} H_{\text{in}} = \ket{1}\bra{1}_1 \otimes \ket{0}\bra{0},
\end{equation*}
and
\begin{equation*}
H_{\text{prop}} = \sum_{i=1}^T H_{\text{prop}, t},
\end{equation*}
with
\begin{eqnarray*}
H_{\text{prop}, t} &=&  \id \otimes \ket{t-1}\bra{t-1} - U_t \otimes \ket{t}\bra{t-1} - U_i^{\cal y} \otimes \ket{t-1}\bra{t} + \id \otimes \ket{t}\bra{t}. 
\end{eqnarray*}
The Hamiltonian is very similar to the one given by Eq. (\ref{HamiltonianQMA}), the only difference being that the time is now encoded in binary form. On the one hand, we do not need a clock term anymore, as any state of the clock register is valid. On the other, we only need $\log(T)$ qubits to store the time. The price for that i an increase in the locality of the Hamiltonians from $5$ body terms to $O(\log(n))$. Note that the number of qubits in the Hamiltonian is only a logarithmic larger than the number of qubits in the original circuit, in contrast to the case where we use the unary representation, in which the size of the Hamiltonian is $N + T$ (see section \ref{lochamprob}). As it will become clear later in the proof, this is the reason why we need to employ $O(\log(n))$-local Hamiltonians. Although Hamiltonian \ref{HamiltonianDQC1} is identical to the one used by Kitaev for proving $\QMA$-completeness of estimating the ground state energy, some of its registers have a different interpretation in the $\DQC$ case: the initial clean qubit takes the role of the initial ancillas in the state $\ket{0}$, while the maximally mixed qubits takes the role of the proof. What makes the difference from the ability to solve any problem in $\QMA$ to \textit{merely} any problem in $\DQC$ is that in the former we must access a single eigenvalue of the Hamiltonian, while in the latter, as we shall see, the solution is encoded in an exponential number of eigenvalues of the Hamiltonian. 

The analysis of the spectral properties is identical to the one carried out in section \ref{lochamprob}. Indeed, we know that the zero eigenspace of $H = J_{\text{prop}} H_{\text{prop}}$ is spanned 
\begin{equation*}
\ket{\eta_i} = \frac{1}{\sqrt{T + 1}} \sum_{t=0}^T U_t ... U_1\ket{i}\otimes \ket{t},
\end{equation*}
for $i \in \{0, ..., 2^{N} - 1\}$
Using the same reasoning of section \ref{lochamprob}, it also follows that $H$ has a spectral gap $\Delta \geq J_{\text{prop}}\Omega(T^{-3})$. By choosing $J_{\text{prop}} = \poly(n)$ sufficiently large, such that $\Delta = \poly(n)$, and following the perturbation theory approach taken in section \ref{lochamprob}, we find that for energies below $\Delta/2$, the spectrum of (\ref{HamiltonianDQC1}) is approximated to $1/\poly(n)$ accuracy by the spectrum of $H_{\eff} = (T + 1)\Pi H_{\text{out}} \Pi + J_{\text{in}} \Pi H_{\text{in}} \Pi$, where $\Pi$ is the projector onto the groundspace of $H$. 

Similarly to the $5$-local case, the ground-space of $\Pi H_{\text{in}} \Pi$ is spanned by states of the form
\begin{equation} \label{etabasis2DQC}
\ket{\eta_i} = \frac{1}{\sqrt{T + 1}} \sum_{t=0}^T U_t ... U_1\ket{0,i}\otimes \ket{t},
\end{equation}
where $i \in \{ 0, ..., 2^{N-1} - 1 \}$. We now apply Theorem \ref{perturbation} to $H = J_{\text{in}} \Pi H_{\text{in}} \Pi$ and $V = (T + 1)\Pi H_{\text{out}} \Pi$. Note that any eigenvector of $H$ not in the groundspace has energy at least $J_{\text{in}}/ (T + 1)$. Choosing $J_{\text{in}} = \poly(n)$ sufficiently large such that $J_{\text{in}}/ (T + 1) = \poly(n)$, it follows that the for energies below $J_{\text{in}}/2(T + 1)$, the spectrum of $H_{\eff}$, and hence also the low-lying spectrum of the original Hamiltonian (\ref{HamiltonianDQC1}), is $1/\poly(n)$ close to the spectrum of $V = (T + 1)\tilde{\Pi} H_{\text{out}}\tilde{\Pi}$, where $\tilde{\Pi}$ is the projection onto the $2^{N-1}$-dimensional space ${\cal S}$ spanned by the vectors given by Eq. (\ref{etabasis2DQC}). 

From Eq. (\ref{etabasis2DQC}), it is clear that all the $2^{N-1}$ eigenvalues of $V$ in the subspace ${\cal S}$ have $O(1)$ energy. Therefore, the average energy of the first $2^{N-1}$ eigenstates of $H$ is $1/\poly(n)$ close to the average energy of the $2^{N-1}$ eigenvalues of $V$ in ${\cal S}$. It is also clear that $V$ can be diagonalized by a set of vectors of the form
\begin{equation*} 
\ket{\tilde{\eta}_i} = \frac{1}{\sqrt{T + 1}} \sum_{t=0}^T U_t ... U_1\ket{0,\psi_i}\otimes \ket{t},
\end{equation*}
where $\{ \ket{\psi_i} \}_{i=1}^{2^{N-1}}$ forms an arbitrary basis for $(\mathbb{C}^2)^{\otimes N - 1}$. The average energy of $V$ is thus given by
\begin{eqnarray*}
\frac{1}{2^{N-1}} \sum_{i=1}^{2^{N-1}} \lambda_i(V) &=& \frac{1}{2^{N-1}}\sum_{i=1}^{2^{N-1}} \bra{\tilde{\eta}_i} V \ket{\tilde{\eta}_i} \\ &=& \frac{1}{2^{N-1}} \sum_{i=1}^{2^{N-1}} \bra{0,\psi_i} U_x^{\cal y} ( \ket{0}\bra{0} \otimes \id^{\otimes N - 1}) U_x \ket{0,\psi_i} \nonumber \\ &=& \frac{1}{2^{N-1}} \tr \left(\left(  U_x \ket{0}\bra{0} \otimes \frac{\id^{\otimes N - 1}}{2^{N-1}} U_x^{\cal y} \right) \ket{0}\bra{0} \otimes \id^{\otimes N - 1} \right) \nonumber \\ &=& 1 - \mu_{\YES}. 
\end{eqnarray*}
Therefore, if we can estimate to polynomial accuracy the average energy of the first $2^{N-1}$ eigenvalues of Hamiltonian (\ref{HamiltonianDQC1}), we can obtain $\mu_{\YES} \pm 1/\poly(n)$ and determine whether $x \in L_{\YES}$ or not.   

In the remainder of the proof we show that the ability to solve both \partition \hspace{0.08 cm} and \spec \hspace{0.08 cm} allow us to compute such average energy. 

Let us start with \spec \hspace{0.08 cm}, for which the reduction is simpler. From the analysis above we find that there is an energy gap in the spectrum of $H_{\DQC}$ from $J_{\text{in}}/2(T + 1)$ to $J_{\text{in}}/(T + 1)$. If we can solve \spec \hspace{0.08 cm}, we can calculate
\begin{equation}
N_{H_{\DQC}}(0, J_{\text{in}}/2(T + 1)) \pm 1/\poly(n) = \frac{1}{2^{N + \log(T)}} \sum_{i=1}^{2^{N-1}} \lambda_i \pm 1/\poly(n).
\end{equation}
As $T = \poly(r)$ and $N = \poly(r)$, the previous equation implies that we can obtain
\begin{equation*}
\frac{1}{2^{N-1}} \sum_{i=1}^{2^{N-1}} \lambda_i \pm 1/\poly(n),
\end{equation*}
which by the discussion of the previous paragraphs is $1/\poly(n)$ close to $1-\mu_{\YES}$. Note that had we used the 5-local construction, we would have the normalization factor of $2^{N + T}$, instead of $T2^{N}$, rendering the estimation of $\mu_{\YES}$ impossible.

For \partition \hspace{0.08 cm}, we first note that each $O(\log(n))$-local term in Eq. (\ref{HamiltonianDQC1}) has a zero minimum eigenvalue. Therefore, $\lambda = \sum_i \lambda_i(H_i) = 0$. If we can solve \partition \hspace{0.08 cm} for $H_{\DQC}$, then we can calculate, to $1/\poly(n)$ accuracy,
\begin{equation*}
\frac{1}{2^{n}} \sum_{i=1}^{2^n} e^{-\beta \lambda_i},
\end{equation*}
for any $1/\poly(n) \leq \beta \leq \poly(n)$. As $n = N + \log(T)$, it follows that we can compute
\begin{equation} \label{twostepsbehind}
\frac{1}{2^{N-1}} \sum_{i=1}^{2^n} e^{-\beta \lambda_i} \pm 1/\poly(n),
\end{equation}
We have seen that the value of the first $2^{N-1}$ eigenvalues of $H_{\DQC}$ is $O(1)$, while any other eigenvalue is larger than $J_{\text{in}}/2(T + 1) = \poly(n)$. Hence, we can choose a $\beta = 1/\poly(n)$ such that $\beta J_{\text{in}}/2(T + 1) = \poly(n)$. Then, from Eq. (\ref{twostepsbehind}) it follows that we can obtain, with precision $1/\poly(n)$,
\begin{equation*}
\frac{1}{2^{N-1}} \sum_{i=1}^{2^{N-1}} (1 - \beta \lambda_i(H_{\DQC}) + O((\beta \lambda_i)^2) + \frac{1}{2^{N-1}} \sum_{i=2^{N-1}+1}^{2^{n}}O(e^{-\poly(n)}).            
\end{equation*}
As
\begin{equation*}
\frac{1}{2^{N-1}} \sum_{i=1}^{2^{N-1}}  O(\beta \lambda_i^2) + \beta \frac{1}{2^{N-1}} \sum_{i=2^{N-1}+1}^{2^{n}}O(e^{-\poly(n)}) \leq 1/\poly(n),
\end{equation*}
we are able to estimate
\begin{equation*}
\frac{1}{2^{N-1}} \sum_{i=1}^{2^{N-1}} \lambda_i \pm 1/\poly(n), 
\end{equation*}
which is equal to $\mu_{\YES} \pm 1/\poly(n)$.

\chapter{The Complexity of Poly-Gapped Quantum Local Hamiltonians} \label{QCMA}

\section{Introduction}

In chapter \ref{complexity} we reviewed Kitaev's theorem on the $\QMA$-completeness of the \locHam \hspace{0.05 cm} problem \cite{KSV02}, together with further developments \cite{KR03, KKR06, OT05} culminating in the $\QMA$-completeness of the problem already for Hamiltonians of particles arranged in a line \cite{AGIK07}. The importance of this result stems not only from the fact that it started the theory of $\QMA$-completeness and remains perhaps as the most prominent $\QMA$-complete problem, but also from the key role in physics of local Hamiltonians and their ground-state energy. Although it has been known already from the work of Barahona \cite{Bar82} that the determination of the ground energy of many-body Hamiltonians, even classical ones, can be a computationally hard problem, to pinpoint the complexity of calculating quantum ground-state energies is both of fundamental interest and of practical importance. 

Indeed, \locHam \hspace{0.05 cm} for one-dimensional classical systems is in $\P$ and whether the same is true for quantum Hamiltonians remained open until the $\QMA$-completeness result of Aharanov, Gottesman, and Kempe \cite{AGIK07}\footnote{Assuming, of course, $\P \neq \QMA$.}. In another example of the usefulness of considering hardness with respect to $\QMA$, Schuch and Verstraete \cite{SV07} proved that determining the ground energy for a system of electrons interacting via the Coulomb interaction and with local magnetic fields is $\QMA$-hard, by showing that \locHam \hspace{0.05 cm} can be reduced to it. A consequence of their result is that the existence of an efficiently computable universal function for density functional theory \cite{DG95}, arguably the most successful approach to the simulation of interacting electrons systems, would imply $\NP = \QMA$ and is therefore very unlikely. One should appreciate that such a conclusion would not have been possible case the problem was shown to be \textit{merely} $\NP$-hard. 

Considering these results, the analysis of the impact of particular properties of many-body Hamiltonians to the complexity of \locHam \hspace{0.05 cm} appears as a very interesting question. A property that seems to have a crucial role in this context is the \textit{spectral gap}, given by the difference of the ground and the first excited energy levels, $\Delta := \lambda_{1}(H) - \lambda_0(H)$. 

The existence of a gap $\Delta = O(1)$ has direct consequences to the properties of local Hamiltonians, including the exponential decay of ground-state correlation functions \cite{HK06}. Moreover, in one dimension, it is conjectured that to determine local properties of gapped Hamiltonians is easy, although no polynomial algorithm is presently known. A groundbreaking result by Hastings shows that ground states of gapped 1-D Hamiltonians have at least an efficient classical description, as a Matrix-Product-State (MPS) of polynomial bond dimension\footnote{A state $\ket{\psi} \in (\mathbb{C}^d)^{\otimes n}$ has a MPS representation with bond dimension $D$ if it can be written as
\begin{equation}
\ket{\psi} = \sum_{i_1,...,i_n=1}^d \tr(A_{i_1}^{[1]}...A_{i_n}^{[n]}) \ket{i_1,...,i_n},
\end{equation}
with $A_i^{[k]}$ $D \times D$ matrices \cite{FNW92, PVWC07}. Note that only $ndD^2$ complex numbers are needed to specify the state.} \cite{Has07}. Since expectation values of local observables of a MPS can be calculated in polynomial time in the number of sites and in its bond dimension (see e.g. \cite{PVWC07}), Hastings result implies that for one dimensional gapped Hamiltonians, \locHam \hspace{0.05 cm} cannot be $\QMA$-complete, unless $\QMA = \NP$. 

An interesting intermediate regime for the spectral gap is the one in which $\Delta = 1/\poly(n)$, where $n$ is the number of sites of the model. Even though poly-gapped Hamiltonians do not have the same distinctive properties as gapped Hamiltonians, such as exponential decaying correlation functions, the existence of this constraint on the gap does seem to have an impact on the complexity of determining low-energy properties of the model.

For instance, in Ref. \cite{SWVC07} Schuch, Wolf, Verstraete, and Cirac considered the class of problems which can be solved in a quantum computer with the help of an oracle which prepares one copy of the ground-state of a local Hamiltonian\footnote{With the promise that it has a unique ground state and that its ground-state energy if zero.}. For general local Hamiltonians, they showed that the class of such problems is precisely $\PP$\footnote{PP is the class of decision problems solvable by a probabilistic Turing machine in polynomial time, with an error probability of less than 1/2 for all instances. Note that such probability of error might be exponential close to 1/2 \cite{AB08}.}, while for poly-gapped Hamiltonians, the class is contained in $\QMA$. Hence, unless $\PP = \QMA$, which is considered unlikely \cite{Vya03}, we find that the promise of a polynomial gap has significant implications in this setting. 

A second example is the determination of the ground-state expectation value of local observables, considering local Hamiltonians that are adiabatically connected to a Hamiltonian with a simple ground state. We say that two Hamiltonians $H$ and $H_0$ are adiabatically connected with a minimum gap $\tilde{\Delta}$ if, for every $s \in [0, 1]$, $\Delta(sH + (1 - s)H_0) \geq \tilde{\Delta}$. Consider a local Hamiltonian $H$ and suppose we want to compute the ground-state expectation value of a local observable, given that $H$ is adiabatically connected to a Hamiltonian $H_0$ whose ground state is a known product state. If we make no assumptions on the minimum gap $\tilde{\Delta}$, then the notion of adiabatic connectivity is vacuous and the problem is again as hard as $\PP$, while it itself belongs to $\PP$. If we require that $\tilde{\Delta} = 1/\poly(n)$, then the complexity of computing such expectation values is precisely $\BQP$, as shown by Aharonov \textit{et al} \cite{AvDK+07}\footnote{This statement is the well-known equivalence of the adiabatic quantum computation model and the circuit model proved in \cite{AvDK+07}. The proof actually uses a variant of Kitaev's encoding of a quantum circuit into a local Hamiltonian discussed in chapter \ref{complexity}.}. Finally, for $\tilde{\Delta} = O(1)$ and Hamiltonians defined on a lattice of bounded dimension, Osborne showed that the expectation value of local observables can be obtained efficiently in a classical computer\footnote{One also needs the promise that the ground-state energy is zero throughout the adiabatic evolution. However, even under this further assumption we still find $\PP$ and $\BQP$ as the complexity of the problem in the unrestricted and $1/\poly(n)$ cases.} \cite{Osb07}.

What are the consequences of a polynomial gap promise to \locHam? Such a question was raised in Refs. \cite{SWVC07, AGIK07} and is the focus of this chapter. For dimensions larger than two, we know that the problem is $\NP$-hard \cite{Bar82}, while the best lower in one dimension was recently presented by Schuch, Verstraete, and Cirac in Ref. \cite{SCV08}. They tackled this question with the further promise that the ground state is described by a MPS of polynomial bond dimension and showed that, in this case, the problem is as hard as $\UNP\cap \UcoNP$. The class $\UNP$ is the subclass of $\NP$ containing the problems for which $\YES$ instances have a unique witness, and $\UcoNP$ is the class of problems for which $\NO$ instances have a unique witness. As the intersection $\UNP\cap \UcoNP$ includes $\FACTORING$, it is believed not be equal to $\P$. Note, however, that it could well be the case that the problem has a quantum algorithm\footnote{At least $\FACTORING$ is not a candidate for placing the class outside $\BQP$!}. In fact, if any poly-gapped Hamiltonian were adiabatically connected, with a known polynomial minimal gap, to a known Hamiltonian with a simple ground-state, the problem would indeed be in $\BQP$\footnote{The $\NP$-hardness of 2-D classical Hamiltonians with a constant gap shows that this cannot be the case for dimensions higher than one, unless $\NP \subseteq \BQP$. For one dimensional Hamiltonians, however, this possibility had not been refuted yet.}. 

In this chapter we present a sharper lower bound on the complexity of this problem and show that such a restriction of \locHam \hspace{0.05 cm} to poly-gapped Hamiltonians is hard for $\QCMA$\footnote{With respect to probabilistic Cook reductions.} (see section \ref{QMA}). As it is unlikely that $\QCMA = \BQP$, we find that the problem appears to be intractable even for quantum computation. Crucial in our approach is the celebrated Valiant-Vazirani Theorem \cite{VV86, AB08}, concerning the hardness of $\NP$ problems with a unique witness. 

The structure of this chapter is the following. In section \ref{defQMAchapter} we define the problem we are interested in, state the main results and outline some of its consequences. In section \ref{quantumVV} we review the Valiant-Vazirani Theorem and prove Theorem  \ref{VVquantum}. Finally, we prove Theorem \ref{gappedlocHam} in section \ref{proofgappedloc}.

\section{Definitions and Main Results} \label{defQMAchapter}

The problem we are interested in can be stated as follows.

\begin{definition}
\gappedlocHam: We are given a $k$-local Hamiltonian on $(\mathbb{C}^d)^{\otimes n}$, $H = \sum_{j=1}^r H_j$, with $r = \poly(n)$ and $d = O(1)$. Each $H_j$ has a bounded operator norm $|| H_j || \leq \poly(n)$ and its entries are specified by $\poly(n)$ bits. We are also given two constants $a$ and $b$ with $b - a \geq 1/\poly(n)$ and and lower bound $\Delta$ on the spectral gap of $H$ such that $\Delta \geq 1/\poly(n)$. In $\YES$ instances, the smallest eigenvalue of $H$ is at most $a$. In $\NO$ instances, it is larger than $b$. We should decide which one is the case.
\end{definition}

In the sequel we prove the following theorem.

\begin{theorem} \label{gappedlocHam}
For $d=k=2$, \gappedlocHam \hspace{0.08 cm} is $\QCMA$-hard under probabilistic Cook reductions. The same is true for $d=12$, $k=2$ and one-dimensional Hamiltonians. Therefore, there is no polynomial quantum algorithm for \gappedlocHam, unless $\QCMA = \BQP$.   
\end{theorem}

Hardness under probabilistic Cook reductions means that with several calls to an oracle for \gappedlocHam \hspace{0.08 cm}, which gives the correct solution to the problem with high probability case the input satisfies the promise and an arbitrary answer case it does not, one can solve any other problem in $\QCMA$ with polynomial quantum computation. As it will be clear later, it is crucial in our approach that we allow calls to the oracle also with inputs which do not satisfy the promise, even though we allow for an arbitrary output in this case.  

We conjecture that in fact the problem is $\QMA$-complete. As it will be clear in section \ref{quantumVV}, the element missing is a quantum version of the Valiant-Vazirani result \cite{VV86}, a possibility first raised in Ref. \cite{AGIK07}. Although it seems plausible that such a version might exist, we leave it as an open problem. 

An interesting question, from the point of view of classical simulability of many-body models, is to determine the minimum bond dimension $D$ needed to approximate by a MPS the ground state of one dimensional local Hamiltonians to $1/\poly(n)$ accuracy. Theorem \ref{gappedlocHam} shows that in general the scaling of $D$ with the spectral gap has to be exponential, unless $\MA = \QCMA$\footnote{Alternatively, Theorem \ref{gappedlocHam} implies that there are one-dimensional poly-gapped Hamiltonians whose ground-state violates an area law \cite{ECP08} for any $\alpha$-R\'enyi entropy with $\alpha < 1$, even beyond a logarithm correction \cite{SWVC08}.}$^{,}$\footnote{We have to use $\MA$ instead of $\NP$ because \gappedlocHam \hspace{0.08 cm} is complete to $\QCMA$ only under probabilistic Cook reductions.}. This should be contrasted to the fact that every $d$-dimensional poly-gapped local Hamiltonian can be approximated by the boundary of a $d+1$-dimensional projected-entangled-pair-state \cite{VC04}, a generalization of MPS to higher dimensions, of constant bond dimension \cite{SWVC07}. 

Theorem \ref{gappedlocHam} has actually stronger consequences for schemes to efficiently store and manipulate ground-states of 1-D local Hamiltonians in a classical computer. Consider any set of states (i) which are described by $\poly(n)$ parameters and (ii) from which one can efficiently compute expectation values of local observables. Matrix-Product-States are an example of such a set and recently several others have been proposed \cite{APD+06, Vid07, HCH+08}. A general set of states satisfying properties (i) and (ii) capable of approximating the ground state of every 1-D local Hamiltonian would place \gappedlocHam \hspace{0.08 cm} in $\NP$. It hence follows from Theorem \ref{gappedlocHam} that no such method can exist, unless $\QCMA = \MA$. As ground-states of 1-D Hamiltonians with a constant gap \textit{can} be efficiently approximated by an MPS \cite{Has07}, we see that it is actually this class of ground-states that appear to be somewhat special, and not ground-states of general 1-D local Hamiltonians. 

Let us consider the following variant of $\QCMA$, which we call Unique-witness Quantum Classical Merlin Arthur. 

\begin{definition} \label{UQMAdef}
\UQCMA: Let $A = (\A_{\YES}, \A_{\NO})$ be a promise problem and let $a, b : \mathbb{N} \rightarrow [0, 1]$ be functions. Then $A \in \UQCMA(a, b)$ if, and only if, for every $n$ bit string $L_n \in A$, there exists a polynomial-time generated quantum circuit $Q_{L_n}$ acting on $N = \poly(n)$ qubits and a function $m = \poly(n)$ such that
\begin{enumerate}
	\item (completeness) If $L_n \in \A_{\YES}$, there exists a computational basis state $\ket{\psi} = \ket{a_1, ..., a_m}$, with $a_i \in \{ 0, 1 \}$, satisfying 
\begin{equation*}	
\tr(Q_{L_n} \left((\ket{0}\bra{0})^{\otimes N - m}\otimes \ket{\psi}\bra{\psi} \right) Q_{L_n}^{\cal y} (\ket{1}\bra{1} \otimes \id^{\otimes N - 1})) \geq a(n),
\end{equation*}
while for any other computational basis state $\ket{\phi} = \ket{b_1, ..., b_m}$,
\begin{equation*}	
\tr(Q_{L_n} \left((\ket{0}\bra{0})^{\otimes N - m}\otimes \ket{\phi}\bra{\phi} \right) Q_{L_n}^{\cal y} (\ket{1}\bra{1} \otimes \id^{\otimes N - 1})) \leq b(n).
\end{equation*}	
	\item (soundness) If $L_n \in \A_{\NO}$, for every computational basis state $\ket{\psi} \in (\mathbb{C}^2)^{\otimes m}$, 
\begin{equation*}	
\tr(Q_{L_n} \left((\ket{0}\bra{0})^{\otimes N - m}\otimes \ket{\psi}\bra{\psi} \right) Q_{L_n}^{\cal y} (\ket{1}\bra{1} \otimes \id^{\otimes N - 1})) \leq b(n).
\end{equation*}	
\end{enumerate}
We define $\UQCMA = \UQCMA(2/3,1/3)$. 
\end{definition}
Note that also for $\UQCMA$ it is possible to amplify the soundness and completeness of the protocol exponentially, i.e. $\UQCMA(a, b) = \UQCMA(1 - 2^{-\poly(n)}, 2^{-\poly(n)})$ for any $a(n) - b(n) \geq 1/\poly(n)$. 

In section \ref{quantumVV} we prove the next theorem, which is the generalization of the Valiant-Varirani Theorem to $\QCMA$.  

\begin{theorem} \label{VVquantum}
$\BQP^{\UQCMA} = \QCMA$.
\end{theorem}

\section{Proof of Theorem \ref{VVquantum}} \label{quantumVV}

The Valiant-Vazirani Theorem says that any problem in $\NP$ can be solved in $\BPP(1 - 2^{-\poly(n)}, 1)$ if we can solve instances of $\SAT$ that are known to contain either zero or only one satisfying assignment \cite{VV86}. A detailed account of this important result, which is also a key ingredient in Toda's Theorem \cite{Tod91}, can be found in several references, including \cite{AB08}. 

The main idea is to construct a randomized polynomial algorithm that on an instance $C(x_1,...,x_n)$ of $\SAT$ with $n$ variables (which, for simplicity, we assume to be given as a circuit $C(x_1,...,x_n)$ with inputs $x_i$, but could equally well be given as a conjunctive normal form), outputs $n + 2$ $\SAT$ instances $C_1, ..., C_{n+2}$, also with $n$ variables each, such that 
\begin{itemize}
	\item If $C$ is unsatisfiable, all $C_i$ are unsatisfiable. 
	\item If $C$ is satisfiable, with probability larger than $\frac{1}{8}$, one of the $C_i$ has exactly one satisfying assignment.   
\end{itemize}
Given an algorithm which solves $\SAT$ with unique solutions we can then solve any problem in $\NP$ with an exponential small probability of error, by using the above reduction \cite{VV86}. To construct the $C_i$, the idea is to add pseudo-random constraints to $C$, derived from a universal hash function. 
%The details of the contruction can be found in Refs. \cite{VV86, AB08}.

\begin{definition}
\cite{AB08} (pairwise independent hash function) A family ${\cal H}$ of functions from $\{0, 1\}^{n}$ to $\{0, 1\}^k$ is pairwise independent if 
\begin{itemize}
	\item for every $x \in \{0, 1\}^n$ and $a \in \{0, 1\}^k$,
	\begin{equation*}
	\Pr_{h \in {\cal H}}(h(x) = a) = \frac{1}{2^k},
	\end{equation*}
	\item for every $x, y \in \{0, 1\}^n$, $x \neq y$, and $a, b \in \{0, 1\}^k$,
	\begin{equation*}
	\Pr_{h \in {\cal H}}(h(x) = a \hspace{0.1 cm}| \hspace{0.1 cm} h(y) = b) = \frac{1}{2^k}.
	\end{equation*}
\end{itemize}
\end{definition}

There are efficient constructions of pairwise independent functions, with circuit complexity polynomial in $n$; an example is the set of affine transformations over integers mod 2 \cite{VV86}. The key observation of Valiant and Vazirani for establishing their result is the following lemma, which we state in a slightly more general form, already preparing for the proof of Theorem \ref{VVquantum}. 

\begin{lemma} \label{hash}
\cite{VV86} Let ${\cal H}$ be a family of pairwise independent hash functions $h : \{0, 1\}^n \rightarrow \{0, 1\}^k $. Let $S \subseteq \{ 0, 1 \}^n$ be such that $|S| \leq 2^{k-1}$ and $V \subseteq S$. Then
\begin{equation}
\Pr_{h \in {\cal H}}(|S \cap h^{-1}(0^k)| = 1 \hspace{0.1 cm} \text{\normalfont and} \hspace{0.1 cm} S \cap h^{-1}(0^k) \in V) \geq \frac{|V|}{2^{k + 1}}.
\end{equation}
\end{lemma}
\begin{proof}
Following \cite{VV86}, with $X := \Pr_{h \in {\cal H}}(|S \cap h^{-1}(0^k)| = 1 \hspace{0.1 cm} \text{\normalfont and} \hspace{0.1 cm} S \cap h^{-1}(0^k) \in V)$ we have
\begin{eqnarray}
X &=& \sum_{x \in V} \Pr_{h \in {\cal H}}(h(x)=0^k \hspace{0.1 cm} \text{\normalfont and} \hspace{0.1 cm} \forall y \in S - \{ x \}, h(y) \neq 0^k) \nonumber \\ &=& 
\sum_{x \in V} \Pr_{h \in {\cal H}}(h(x) = 0^k)\Pr_{h \in {\cal H}}(\forall y \in S - \{ x \}, h(y) \neq 0^k \hspace{0.05 cm} | \hspace{0.05 cm} h(x) = 0^k) \nonumber \\ &=& \sum_{x \in V}\frac{1}{2^k}\left( 1 - \Pr_{h \in {\cal H}}(\exists y \in S - \{ x \}, h(y) = 0^k \hspace{0.05 cm} | \hspace{0.05 cm} h(x) = 0^k)  \right) \nonumber \\ &\geq& \sum_{x \in V}\frac{1}{2^k}\left( 1 - \sum_{y \in S - \{x\}}\Pr_{h \in {\cal H}}(h(y) = 0^k \hspace{0.05 cm} | \hspace{0.05 cm} h(x) = 0^k)  \right) \nonumber \\ &\geq& \sum_{x \in V}\frac{1}{2^k} \left( 1 - \frac{|S|}{2^k} \right) \nonumber \\ &=& \frac{|V|}{2^k}\left(1 - \frac{|S|}{2^k} \right) \geq \frac{|V|}{2^{k + 1}}.
\end{eqnarray}
\end{proof}

The Valiant-Vazirani Theorem can now be easily established as follows. We let $C_i = C \cap M_i$, where $M_i$ is a circuit which computes the hash function $h_i$ from $n$ bits to $i$ bits, chosen at random from the family ${\cal H}$. If $C$ is not satisfiable, it is clear that none of the $C_i$ will be either. If $C$ is satisfiable with $2^{i-2} \leq |S| \leq 2^{i-1}$ different satisfying assignments, then an application of Lemma \ref{hash} with $V = S$ shows that with probability larger than 1/8, $C_i$ will have exactly one satisfying assignment.   

The proof of Theorem \ref{VVquantum} will follow very closely the proof of the Valiant-Vazirani Theorem. One difference is that in $\QCMA$ the witnesses do not have to accept with certainty, so we might have several mediocre witnesses exponentially close to the boundary of acceptance. As we show in the sequel, we can overcome this difficulty by choosing the set $V$ in Lemma \ref{hash} appropriately.

\vspace{1 cm}

\begin{proof} (Theorem \ref{VVquantum}) 

Given a problem $A \in \QCMA$, we show how to solve it with polynomial quantum computation given access to an oracle that solves any problem in $\UQCMA$. Let $L_n \in A$ be a $n$-bit instance of $A$. We know that if $L_n \in A_{\YES}$, then there is a classical state $\ket{\psi} = \ket{a_1, ..., a_m}$, with $a_i \in \{0, 1\}$, of $m = \poly(n)$ bits which makes the verification procedure $Q_{L_n}$ to accept with probability larger than $2/3$, while if $L_n \in A_{\NO}$, no $m$-bit classical state accepts with probability larger than $1/3$. 

%We can consider the following verification ptocedure. Arthur first performs a measurement in the computational basis of the state given to him by Merlin, collapsing the original state into a computational basis state of the form $\ket{\psi} = \ket{a_1, ..., a_m}$. He then copy this state to $\poly(n)$ ancillas and execute in each of them the original verfification procedure. Finally, encode the $\poly(n)$ measurements into a single one, using standard techniques \cite{NC01}, in order to put the proocol in the standard form given by \ref{QCMAdef}. In this way, he can amplify the soundness and completeness to $2^{-\poly(n)}$ and $1 - 2^{-\poly(n)}$. The point here is that as the witness is classical, he can perform error amplification without having to ask for larger witnesses\footnote{In the case of quantum proofs Marriott and Watrous showed that this is also possible \cite{MW05}. A naive application of their contruction to $\QCMA$ would fail, however, as the after the first round of teh verification procedure the witness would cease to be classical.}. This fact will be important in the proof of Theorem \ref{gappedlocHam}.

%We therefore assume that if $L_n \in A_{\YES}$, there is at least one classical witness that accepts with probability larger than $1 - 2^{-r(n)}$, for a $r(n) = \poly(n)$, while if $L_n \in A_{\NO}$, every computational basis state accept with probability smaller than $2^{-r(n)}$. 

To probabilistically transform this problem into one with a unique witness, we consider the following family of protocols $P_{l, k, h}$, indexed by $l, k \in \{1, ..., m \}$ and a hash function $h \in {\cal H}$ from $m$ to $l$ bits:
\begin{enumerate}
	\item Arthur measures in the computational basis Merlin's witness, obtaining the classical witness $\ket{\psi} = \ket{a_1, ..., a_m}$\footnote{Such a measurement is not really needed, as Arthur has the promise that Merlin sends a classical witness. However, for the proof of Theorem \ref{gappedlocHam} it is useful to consider it.}.
	\item He then computes $h(a_1, ..., a_m)$. If the output is $0^{l}$, he accepts and goes to step 3. If not, he rejects. 
	\item He copies the classical witness $\ket{\psi} = \ket{a_1, ..., a_m}$ to $p(n) = \poly(n)$ ancilla registers and executes the original verification procedure, $Q_{L_n}$, $p(n)$ times in parallel, accepting if the number of individual acceptances is larger than $p(n)\left(\frac{1}{3} + \frac{k - 1/2}{3m}  \right)$.    
\end{enumerate}

We claim that 
\begin{itemize}
	\item If $L \in \NO$, all $P_{l, k, h}$ accepts any classical witness with probability smaller than $2^{-\frac{p^2}{10^3 m^2}}$. 
	\item If $L \in \YES$, with probability larger than $\frac{1}{16}$ over the choice of $h \in {\cal H}$, one of the $P_{l, k, h}$ has only one classical witness which accepts with probability larger than $1 - 2^{-\frac{p^2}{10^3 m^2}}$, while any other computational basis state is accepted with probability at most $2^{-\frac{p^2}{10^3 m^2}}$. 
\end{itemize}

It is then clear that given an oracle to problems in $\UQCMA$, we can solve any problem in $\QCMA$ by choosing $\poly(n)$ random instances of $P_{l, k, h}$, solving them with the help of the oracle (which outputs an arbitrary answer case the $P_{l, k, h}$ does not have a unique witness), and accepting if the number of $\YES$ outputs is large enough. 

Let us prove the claim. The case $L \in \NO$ is straightforward, so we focus in the $L \in \YES$ one. We consider as a witness any computational basis state that accepts with probability larger than $1/3$. There might be at most $2^m$ different such witnesses. We divide the interval $[1/3, 2/3]$ into $m$ intervals $r_i$, $i \in \{ 1, ..., m \}$, of equal length $\frac{1}{3m}$, i.e. $r_i := [\frac{1}{3} + \frac{i-1}{3m}, \frac{1}{3} + \frac{i}{3m})$, and set $r_{m+1} = (2/3, 1]$. Let $N_i$ be the number of witnesses with acceptance probability in the range $r_i$.

The key point is to realize that there must be a $l \in \{1, ..., m + 1 \}$ such that 
\begin{equation*}
\sum_{k=l+1}^m N_k > N_l
\end{equation*}
and 
\begin{equation}
\sum_{k=l}^m N_k \leq 2^{m-l+1}.
\end{equation}
We prove it by contradiction. If it were not true, it would follow that 
\begin{equation} \label{auxauxaux}
N_1 \geq \sum_{k=2}^m N_k,
\end{equation}
as there are at most $2^m$ witnesses. By Eq. \ref{auxauxaux}, we must have $\sum_{k=2}^m N_i \leq 2^{m-1}$, which leads, by the assumption, to  
\begin{equation*}
N_2 \geq \sum_{k=3}^m N_k.
\end{equation*}
Carrying out the same analysis as before we eventually find that for every $j$, 
\begin{equation*}
N_j \geq \sum_{k=j+1}^m N_k.
\end{equation*}
As $N_{m+1} \geq 1$ (there is at least a witness which accepts with probability larger than $2/3$), we must have $N_{m-k} \geq 2^{k}$ and hence
\begin{equation*}
\sum_{k=1}^{m+1}N_k \geq 1 + \sum_{k=0}^{m-1}2^k \geq 1 + 2^m, 
\end{equation*}
which is in contradiction with the bound of $2^m$ on the maximum number of witnesses. 

Given such a $l$, we apply Lemma \ref{hash} with $S$ as the set of witnesses which accepts with probability larger than $\frac{1}{3} + \frac{l-1}{3m}$ (which we assume has cardinality in the interval $[2^{k-2}, 2^{k-1}$]), $V$ as the set of witnesses with acceptance probability larger than $\frac{1}{3} + \frac{l}{3m}$, and $n = m$. By the discussion above, we have $|V| > |S|/2 = 2^{k-3}$ and, with probability larger than 
\begin{equation*}
\frac{|V|}{2^{k+1}} \geq \frac{1}{16},
\end{equation*}
there will be only one witness with acceptance probability larger than $\frac{1}{3} + \frac{l}{3m}$ (call it $W_{good}$) and no witness which accepts with probability in the range $r_l$, if we consider the protocol $P_{l, k, h}$ with a random choice of $h \in {\cal H}$ (all the original witnesses with acceptance in the interval $r_i$ will be rejected in step 2).  

Consider a $h$ for which there is a unique witness $W_{\text{good}}$ which accepts with probability larger than $\frac{1}{3} + \frac{l}{3m}$ and no witness with acceptance probability in the range $r_l$ and set $V_i(W_{good})$, $i \in \{1, ..., p(n) \}$, as the outcome of the $i$-th test of $P_{l, k, h}$. We have
\begin{equation*}
\mu_{W_{good}} := \mathbb{E}\left( \frac{1}{p}\sum_{i=1}^p V_i(W_{good}) \right) \geq \frac{1}{3} + \frac{l}{3m},
\end{equation*}
and thus, by the Hoeffding-Chernoff \cite{Dud02} bound, 
\begin{eqnarray*}
\Pr(P_{l, k, h} \hspace{0.2 cm} \text{accepts} \hspace{0.2 cm} W_{good}) &=& \Pr\left( \frac{1}{p}\sum_{i=1}^p V_i(W_{good}) \geq \frac{1}{3} + \frac{l-1/2}{3m} \right) \nonumber \\ &\geq& 1 - \Pr\left( \left| \frac{1}{p}\sum_{i=1}^p V_i - \mu_{W_{good}} \right| \geq \frac{1}{6m} \right) \geq 1 - 2^{-\frac{p}{10^3m^2}},  
\end{eqnarray*}
while for any other witness $W$, 
\begin{equation*}
\mu_{W} := \mathbb{E}\left( \frac{1}{p}\sum_{d=1}^p V_i(W) \right) \leq \frac{1}{3} + \frac{l-1}{3m},
\end{equation*}
and therefore 
\begin{eqnarray*}
\Pr(P_{l, k, h} \hspace{0.2 cm} \text{accepts} \hspace{0.2 cm} W) &=& \Pr\left( \frac{1}{p}\sum_{i=1}^p V_i(W) \geq \frac{1}{3} + \frac{l-1/2}{3m} \right) \\ &\leq& \Pr\left( \left| \frac{1}{p}\sum_{i=1}^p V_i - \mu_{W} \right| \geq \frac{1}{6m} \right) \leq 2^{-\frac{p}{10^3m^2}}.  
\end{eqnarray*}

The argument is completely analogous to $L \in A_{\NO}$, which leads to a soundness of $2^{-\frac{p}{10^3m^2}}$.
\end{proof}

\section{Proof of Theorem \ref{gappedlocHam}} \label{proofgappedloc}

We employ the following Lemma in the proof of Theorem \ref{gappedlocHam}, which appears as Proposition 9 in \cite{ABD+08}.

\begin{lemma} \label{unionboundPOVM}
\cite{ABD+08} Given a $k$-partite state $\rho_{1,...,k}$, suppose there are pure states $\ket{\psi_1}$, ..., $\ket{\psi_k}$ such that $\bra{\psi_i}\rho_i\ket{\psi_i} \geq 1 - \epsilon_i$, for all $i$. Let $\ket{\Psi} := \ket{\psi_1} \otimes ... \otimes \ket{\psi_k}$ and $\epsilon := \epsilon_1 + ... + \epsilon_k$. Then $\bra{\Psi}\rho_{1, ..., k}\ket{\Psi} \geq 1 - \epsilon$.
\end{lemma}

\begin{proof} (Theorem \ref{gappedlocHam})

For simplicity we first prove the theorem for 5-local Hamiltonians and then extend it to one dimensional Hamiltonians. 

The idea is to encode each verification procedure $P_{l, k, h}$, defined in the proof of Theorem \ref{VVquantum}, in a local Hamiltonian, using Kitaev's construction explained in section \ref{lochamprob}. Note that we can easily modify the verification procedure $P_{l, k, h}$ to the standard form used in Definition \ref{QCMAdef}, in which a quantum circuit is applied to the witness and ancilla registers and then a single measurement is performed in the end \cite{NC00}. Let us denote this circuit by $Q_{l, k, h} = U_T...U_2U_1$, which consists of $T$ single and two qubit gates and acts on $N$ qubits, where $N - m$ qubits are initialized in the $\ket{0}$ state and the witness register has $m = \poly(n)$ qubits. Define the following 5-local Hamiltonian, acting $N + T$ qubits:
%Given such a circuit $P_{l, k, h}$ acting on $N-1 = \poly(n)$ qubits, we define the following related circuit $Q_{l, k, h}$ on $N$ qubits: it applies $P_{l, k, h}$ to the first $N-1$ qubits, then copies the answer qubit to the additonal qubit added by a CNOT gate, and then uncompute the the circuit by applying $P_{l, k, h}^{\cal y}$\footnote{The use of this construction will be clear later}.
\begin{equation}
H_{l, k, h} := J_{\text{out}}H_{\text{out}} + J_{\text{in}} H_{\text{in}} + J_{\text{prop}} H_{\text{prop}} + J_{\text{clock}} H_{\text{clock}} + J_{\text{gap}} H_{\text{gap}},
\end{equation}
where $J_{\text{out}}, J_{\text{in}}, J_{\text{clock}}, J_{\text{gap}} = \poly(n)$ are functions to be chosen later and the terms $H_{\text{in}}$, $H_{\text{out}}$, $H_{\text{prop}}$, and $H_{\text{clock}}$ are given by Eqs. (\ref{Hkit1}), (\ref{Hkit2}), and (\ref{Hkit3}).  The new term $H_{\text{gap}}$, which will be used to create a polynomial gap in $\NO$ instances, is given by
\begin{equation}
H_{\text{gap}} := \sum_{i=1}^m \ket{0}\bra{0}_1 \otimes \ket{000}\bra{000}_{1,2,3}^c.
\end{equation}

%The main point is to prove that for a random choice of $h \in {\cal H}$ and $l, k \in \{1, ..., m \}$, with probability larger than $\frac{1}{16m^2}$, one of the $H_{l, k, h}$ will have a poly-gap and the solution of the problem will be encoded in its the ground state energy., i.e. case $L \in \YES$, the ground state energy of $H$ will be below a certain number $a$ and if $L \in \NO$, it will be larger than a $b$ such that $b - a \geq 1/\poly(n)$. 

As shown in the previous section, the verification protocols $P_{l, k, h}$ (and thus also $Q_{l, k, h}$) are such that
\begin{itemize}
	\item case $L \in A_{\YES}$, with probability larger than $\frac{1}{16m^2}$ over $l, k$, and $h$, there is a unique classical witness that makes $P_{l, k, h}$ accept with probability larger than $1 - 2^{-p(n)}$, where $p(n) = \poly(n)$\footnote{Which we can choose it as large as we want by increasing only the number of ancilla qubits in the verification procedure and the size of the circuit}, while any other classical state accepts with probability smaller than $2^{-p(n)}$. 
  \item case $L \in A_{\NO}$, no classical witness makes $P_{l, k, h}$ accept with probability larger than $2^{-p(n)}$, for any choice of $l, k$, and $h$.  
\end{itemize}

In the sequel we prove that such a property implies that
\begin{itemize}
	\item case $L \in A_{\YES}$, with probability larger than $\frac{1}{16m^2}$ over $l, k$, and $h$, $H_{l, k, h}$ has a unique ground-state with a poly-gap\footnote{For which an explicitly lower bound $\Delta = 1/\poly(n)$ calculated from the choices of $J_{\text{out}}, J_{\text{in}}, J_{\text{clock}}, J_{\text{gap}}$ and $p$.} above it and its ground-state energy is smaller than $a$. 
  \item case $L \in A_{\NO}$, $H_{l, k, h}$ has a unique ground-state with a poly-gap$^{13}$ above it and its ground-state energy is larger than $b \geq a + 1/\poly(n)$, for every choice of $l, k$, and $h$.  
\end{itemize}
Given an oracle to \gappedlocHam \hspace{0.08 cm} we can then solve any problem in $\QCMA$ with polynomial quantum computation, just as in the proof of Theorem \ref{VVquantum}, replacing the oracle for general problems in $\UQCMA$ by the oracle for \gappedlocHam. 

Let us analyse the spectral properties of the low energy sector of $H_{l,, k, h}$. As in section \ref{lochamprob}, we apply Theorem \ref{perturbation} to $\tilde{H} = H + V$, with $H = J_{\text{clock}}H_{\text{clock}} + J_{\text{prop}} H_{\text{prop}}$ and $V = J_{\text{in}} H_{\text{in}} +  J_{\text{out}}H_{\text{out}} + J_{\text{gap}}H_{\text{gap}}$. As before, the ground-space of $H$ is spanned by
\begin{equation} \label{etabasispolygap}
\ket{\eta_i} = \frac{1}{\sqrt{T + 1}} \sum_{t=0}^T U_t ... U_1\ket{i}\otimes \ket{1}^{\otimes t} \otimes \ket{0}^{\otimes T - t},
\end{equation}
for $i \in \{ 0, ..., 2^N - 1 \}$. The spectral gap of $H$ ($\Delta(H)$), in turn, is lower bounded by $J_{\text{prop}}\Omega(T^{-3})$. Letting $\Pi$ be the projector onto the ground-space of $H,$ we find that, for energies below $\Delta(H)/2$, the spectrum of $H_{l, k, h}$ is approximated to $1/\poly(n)$ accuracy, for sufficiently large $J_{\text{prop}} = \poly(n) \gg J_{\text{in}}, J_{\text{out}}, J_{\text{gap}}$, by the spectrum of 
\begin{equation}
J_{\text{out}}\Pi H_{\text{out}} \Pi + J_{\text{in}} \Pi H_{\text{in}} \Pi + J_{\text{gap}} \Pi H_{\text{gap}} \Pi
\end{equation}
in the subspace spanned by the states \ref{etabasispolygap}. Still following the discussion of section \ref{lochamprob} (see also section \ref{DQC1hardsection}), another application of Theorem \ref{perturbation}, this time with $H = J_{\text{in}} H_{\text{in}}$ and $V = J_{\text{out}}H_{\text{out}} + J_{\text{gap}}H_{\text{gap}}$, shows that for sufficiently large $J_{\text{in}} = \poly(n) \gg  J_{\text{out}}, J_{\text{gap}}$, for energies below $J_{\text{in}}/(2(T + 1))$, the energy spectrum of $H_{l, k, h}$ is approximated to $1/\poly(n)$ accuracy by the spectrum of
\begin{equation}
H_{\text{eff}} := J_{\text{out}}\tilde{\Pi} H_{\text{out}} \tilde{\Pi} + J_{\text{gap}} \tilde{\Pi} H_{\text{gap}} \tilde{\Pi},
\end{equation}
where, as in section \ref{DQC1hardsection}, $\tilde{\Pi}$ is the projector onto the subspace spanned by 
\begin{equation} \label{etabasispolygap0}
\ket{\eta_i} = \frac{1}{\sqrt{T + 1}} \sum_{t=0}^T U_t ... U_1\ket{i,0}\otimes \ket{1}^{\otimes t} \otimes \ket{0}^{\otimes T - t},
\end{equation}
for $i \in \{ 0, ..., 2^{m} - 1 \}$. 

Let us first analyse the gap and the ground-state energy of a $\YES$ instance. We assume that $l, k$ and $h$ are such that there is a unique classical witness that makes $P_{l, k, h}$ accept with probability larger than $1 - 2^{-p(n)}$, while any other classical state is accepted with probability smaller than $2^{-p(n)}$. As shown before, this happens with probability larger than $\frac{1}{16m^2}$ over a random choice of $l, k$ and $h$. Let $\ket{\psi} = \ket{a_1,..., a_m}$ be such a witness. Then, with 
\begin{equation} \label{etapsi}
\ket{\eta_{\psi}} := \frac{1}{\sqrt{T + 1}} \sum_{t=0}^T U_t ... U_1\ket{\psi,0}\otimes \ket{1}^{\otimes t} \otimes \ket{0}^{\otimes T - t},
\end{equation}
$\bra{\eta_{\psi}}H_{\text{eff}}\ket{\eta_{\psi}} \leq J_{\text{out}} 2^{-p(n)} + m J_{\text{gap}}$ and the ground state energy of $H_{\text{eff}}$ and $H$ must be smaller than $2mJ_{\text{gap}}$ (for a sufficiently large choice of $p(n)$).

As every other classical state $\ket{b_1,..., b_m}$ different from $\ket{\psi}$ is accepted with probability at most $2^{-p(n)}$, a simple calculation\footnote{Indeed, if ${\cal N} = \text{span}(\{ \ket{\phi_k} \}_{k=1}^{2^m-1})$ and $\ket{\phi} = \sum_{k}^{2^m-1}c_k \ket{\phi_k}$,
\begin{eqnarray*}
\Pr(\phi \hspace{0.1 cm} \text{accepts}) &=& \bra{\phi}Q_{l, k, h}^{\cal y}\ket{1}\bra{1}Q_{l, k, h}\ket{\phi} \\ &=& \left| \sum_{k, k'} c_k c_{k'}^* \bra{\phi_{k'}}Q_{l, k, h}^{\cal y}\ket{1}\bra{1}Q_{l, k, h}\ket{\phi_k} \right| \\ &\leq& \sum_{k, k'} |c_k| |c_{k'}^*| |\bra{\phi_{k'}}Q_{l, k, h}^{\cal y}\ket{1}|\bra{1}Q_{l, k, h}\ket{\phi_k}| \\ &\leq& 2^{-p(n)}\left(\sum_{k=1}^{2^m-1}|c_k| \right) \leq 2^{-p(n) + m}.  
\end{eqnarray*}} shows that any state in their span ${\cal N}$ is accepted with probability at most $2^{-p(n) + m}$. Thus $\bra{\eta_{\phi}}H_{\text{eff}}\ket{\eta_{\phi}} \geq J_{\text{out}}/2$ (choosing $p(n) \geq 2m$) for any state $\ket{\phi}$ in ${\cal N}$. Choosing $J_{\text{out}} = \poly(n) \gg m J_{\text{gap}}$ large enough, the ground state of $H_{\text{eff}}$ must be $1/\poly(n)$ close to $\ket{\eta_{\psi}}$ (otherwise the ground-state energy would be larger than $2mJ_{\text{gap}}$), and hence we have a $\poly(n)$ gap (whose exact functional form depends on the choice of $J_{\text{out}}, J_{\text{gap}}$). 

Case $L \in \NO$, every classical state is accepted with probability at most $2^{-p(n)}$ and from the previous discussion an arbitrary state in $(\mathbb{C}^2)^{\otimes m}$ is accepted with probability smaller than $2^{-p(n) + m}$. This implies that the ground state energy of $H_{\text{eff}}$ is larger than $(1 - 2^{-p(n) + m})J_{\text{out}}$, and we find that there is a polynomial large separation in the ground-state energy of $H_{l, k, h}$ for $\YES$ and $\NO$ instances. In the remainder of the proof we show that $H_{l, k, h}$ has a poly-gap (for any choice of $l, k, h$). The idea will be to use the extra term $H_{\text{gap}}$, which penalizes every qubit of the witness register which is not initialized in the state $\ket{1}$. 

If we choose $\ket{\psi} = \ket{1}^{\otimes m}$ as a witness, the associated state $\ket{\eta_{\psi}}$ (see Eq. \ref{etapsi}) has energy smaller or equal to $J_{\text{out}}$.  The witness (possibly quantum) associated to the ground-state of $H_{\text{eff}}$, $\ket{\eta_{\psi_g}}$, has then to satisfy 
\begin{equation}
2^{-p(n) + m} \geq \frac{J_{\text{gap}}}{T + 1}\bra{\psi_g}H_{\text{gap}}\ket{\psi_g} = \frac{J_{\text{gap}}}{T + 1} \sum_{k=1}^m \tr(\ket{0}_k\bra{0} \psi_g^k),
\end{equation}
where $\psi_g^k := \tr_{\backslash k}(\ket{\psi_g}\bra{\psi_g})$, as otherwise $\ket{\psi}$ would have a smaller energy. Choosing $J_{\text{gap}} = m^2(T + 1)$, we find from Lemma \ref{unionboundPOVM} that 
\begin{equation} \label{approxgorundstate}
|\bra{1}^{\otimes m}\ket{\psi_g}| \geq 1 - m^{-1}2^{-p(n) + m}. 
\end{equation}
Suppose we had a state $\ket{\phi}$ orthogonal to $\ket{\psi_g}$ such that $\ket{\eta_{\phi}}$ had energy below $J_{\text{out}} + 1$. Then
\begin{equation}
1 \geq \frac{J_{\text{gap}}}{T + 1}\bra{\phi_g}H_{\text{gap}}\ket{\phi_g} = m^2 \sum_{k=1}^m \tr(\ket{0}_k\bra{0} \phi_g^k),
\end{equation}
and applying Lemma \ref{unionboundPOVM}, $|\bra{1}^{\otimes m}\ket{\phi}| \geq 1 - 1/m$, which contradicts Eq. (\ref{approxgorundstate}), as $\braket{\phi}{\psi_g} = 0$. Therefore, $H_{\eff}$ has a gap of at least $1$, which implies that $H$ has a gap larger than $1 - 1/\poly(n)$.

The construction in the case of one dimensional local Hamiltonians is completely analogous. In $\YES$ instances the existence of a poly-gap comes from the uniqueness of the witness, so the same analysis carried out here to 5-local Hamiltonians is valid for the 1-D construction of \cite{AGIK07}. For $\NO$ instances, we add a term similar to $H_{\text{gap}}$ to the Hamiltonian give by Eq. (\ref{HamiltonianQMA1D}), but acting on the data qubits of the first $m$ sites of the chain and without a clock register. As such a term acts trivially in the control registers, the Clairvoyance Lemma (see section \ref{QMAcomp1D} and Ref. \cite{AGIK07}) can be applied in the same fashion and, once restricted to the space of legal configurations, the analysis is completely analogous to the one we just did.
\end{proof}

\part{Quantum Simulation of Many-Body Physics in Quantum Optical Systems} \label{part3}

\chapter{Quantum Simulations in Arrays of Coupled Cavities} \label{AMO}

\section{Introduction}

Quantum optics and atomic physics are concerned with the interaction of light and matter in the quantum regime, typically on the scale of a few atoms and light quanta \cite{KA83, CDG89, MW95, Lou00}. As atoms do not naturally interact strongly with each other and with light, the physics of such systems can usually be very well understood by considering them on an individual level and treating the interactions perturbatively. This is somehow a much simpler situation than the one encountered in the condensed matter context. There, strong interactions among the basic constituents, such as nuclei and electrons, lead to the appearance of completely new physics when one considers a mesoscopic or macroscopic number of interacting individual particles. Thus, even though the building interactions are usually known, it is challenging to fully describe the properties of such systems. 

The higher level of isolation of quantum optical systems have permitted over the past decades the development of experimental techniques for their manipulation with an unprecedent level of control. Highlights include the invention of trapping and cooling techniques \cite{MS99}, which have found applications e.g. in spectroscopy and precision measurements \cite{AD02} and which culminated in the realization of a Bose-Einstein condensate (BEC) of neutral atoms \cite{Leg01, PS01}; the manipulation of the state of individual atomic and photonic systems at the quantum level \cite{LBM+03,RBH01} and its use in the study of the foundations of quantum theory \cite{RBH01,Zei99}; and more recently the realization of primitives for the processing of quantum information and for the implementation of quantum computation in several different quantum atomic and optical set-ups \cite{MM02}.  

Such a progress raised the possibility of engineering strong interactions in these systems and of realizing such strong correlated models in systems with large number of constituents. This is of interested from a fundamental point of view, as in this way we can form new phases of matter for atoms, photons, or even molecules that do not exist outside the laboratory. Moreover, due to the high level of experimental control over them, it is possible to study many-body physics in a much cleaner manner than possible employing condensed-matter systems, in which the strong interaction with the environment and the very short space- and time-scales involved inhibit a precise control over their static and dynamical properties. As mentioned in chapter \ref{complexity}, the use of a quantum mechanical system to simulate the physics of another is an idea that dates back to Feynman \cite{Fey82} and has attracted a lot of interest recently in quantum information science \cite{MM02}. 

A very successful example in this direction is the study of cold trapped atoms in optical lattices \cite{BDZ08}. By loading a BEC of neutral atoms in a periodic optical lattice, formed by the off-resonant dipole interaction of atoms with overlapping standing waves of counterpropagating lasers, the otherwise weak interactions of the atoms can be increased to such an extend that a strongly correlated many-body model is formed \cite{JBC+98}. This idea led to several important experiments, such as the realization of the Bose-Hubbard model \cite{GME+02} and of a Tonks-Girardeau gas \cite{BWM+04} and is current an active area of theoretical and experimental research \cite{LSA+07,BDZ08}. 

Also of note is the early study of strongly correlated many-body models in Josephson junctions arrays \cite{MSS01}. Although a Josephson junction is a mesoscopic solid state device, it behaves in many respects just like the type of systems studied in quantum optics, as the relevant physics can be understood by looking at a few quantum levels. In this context, arrays of interacting Josephson junctions have been used to reproduce properties of bosonic particles. For example, the Mott insulator phase \cite{ZFE+92} and the transition superconductor-to-insulator \cite{OM96} have been observed in Josephson junctions arrays.    

Finally, theoretical proposals for using trapped ions for realizing many-body models have been put forwarded \cite{PC04a,PC04b,RPTS08} and a first benchmark experiment in this direction has been recently reported \cite{FSG+08}. 

In this chapter and in the next three we analyse the use of the atom-light interaction in coupled microcavity arrays to create strongly correlated many-body models. Such a possibility has attracted a lot of interest recently, see e.g.  \cite{HBP06,ASB07,GTCH06,HP07,HBP07a,BHP07,HBP07b,RF07,NUT+08,PAK07,IOK08,BKKY07,MCT+08,HP08,LG07,KA08,CAB08,RFS08,OIK08, HBP08}. Recent experimental progress in the fabrication of microcavity arrays and the realization of the quantum regime in the interaction of atomic-like structures and quantized electromagnetic modes inside those \cite{AKS+03,SKPV03,KSMV04,ADW+06,BLS+06,ASV07,DPA+08,AAS+03,BHA+05,SNAA05,MSG+07,SPS07} opened up the possibility of using them as quantum simulators of many-body physics. The first motivation for this study is the desire to put photons or combined photonic-atomic excitations in new states of matter, which do not naturally appear in nature. The second is that this set-up offer advantages over other proposals for the realization of strongly interacting many-body models in quantum optical systems. Due to the small separations between neighboring sites, it is experimentally very challenging to access individual sites in Josephson junctions arrays, optical lattices, and to a lesser degree in ion traps\footnote{The problem here is related to scalability. Although for a small number of ions it is possible to address individual sites, this becomes harder when more ions are confined in the trap, as the average spacing gets smaller.}. Neighboring sites in coupled cavity arrays, on the other hand, are usually separated by dozens of micrometers and can therefore be accurately accessed by optical frequencies. This allows e.g. the measurement of local properties as well as the exploration of inhomogeneous systems.  

The organization of this chapter is the following. In section \ref{cQED} we review the basics of cavity quantum electrodynamics (cQED) and overview some particular realizations of cavity QED that will show useful for our aims: toroidal microcavities coupled to neutral ions \cite{AKS+03,SKPV03,KSMV04,ADW+06,BLS+06,ASV07,DPA+08}, in subsection \ref{todoid}, photonic-band-gap defect nanocavities coupled to quantum dots \cite{AAS+03,BHA+05,SNAA05}, in subsection \ref{phddot}, and microwave stripline cavities coupled to Cooper pair boxes \cite{MSG+07,SPS07}, in subsection \ref{Cooperpair}. In section \ref{cca}, in turn, we develop the basic description for an array of interacting microcavities, which will be employed in the next three chapters. 

%The organization of the paper is the following. In section \ref{cQED} we discuss the use of cavities to enhance the interaction of light and matter, develop the basic description for an array of interacting cavities, and overview some particular realizations of cavity quantum electrodynamics (QED) that will show useful for our aims: toroidal microcavities coupled to neutral ions \cite{AKS+03,SKPV03,KSMV04,ADW+06,BLS+06,ASV07,DPA+08}, photonic-band-gap defect nanocavities coupled to quantum dots \cite{AAS+03,BHA+05,SNAA05,HBW+07}, and microwave stripline cavities coupled to Cooper pair boxes \cite{MSG+07,SPS07}. In section \ref{BHmodel} we review proposals \cite{HBP06,ASB07,GTCH06,HP07,RF07,HBP07b} for observing strongly correlated phenomena of atomic-photonic polaritons in coupled cavity arrays, in particular the proposal of Refs. \cite{HBP06,HP07} for creating a Bose-Hubbard model for the polaritons. To do so we briefly discuss the quantum mechanical effect of electromagnetically-induced-transparency \cite{FIM05} and recent approaches for creating large photonic Kerr nonlinearities in cavity QED \cite{SI96,ISWD97,WI99,GAG99,HP07,BHP07}. In section \ref{Heisenberg}, in turn, we survey recent works, in particular Ref. \cite{HBP07a}, on the engineering of spin Hamiltonians in arrays of coupled cavitities. Finally, in section \ref{dis} we outline ideas on the use of such a set-up in the study of disorder and inhomogeneous systems \cite{RF07,HP08}. We then finish with the conclusions in section \ref{conc}.

\section{Cavity QED \label{cQED}}

An optical cavity or optical resonator is a system, usually formed by two or more mirrors, which allows for the formation of standing waves of light of particular frequencies and shape, called resonant modes or cavity modes. All the other modes are suppressed by destructive interference effects. Among the many applications of optical resonators, its use in light amplification by stimulated emission of radiation (laser) is of particular note \cite{Shi86}. 

The interaction of electromagnetic fields with atoms, which is in general rather complex, is substantially simplified inside a cavity, partially due to the few modes of radiation that must be taken into account. In fact the standard textbook example of an electric dipole interacting with a monochromatic electromagnetic field can be experimentally realized in the quantum regime in optical and microwave cavities \cite{MD02}. Cavity quantum electrodynamics (cQED) studies the physics of such systems. A detailed account of cQED can be found in \cite{KA83, WM94}. In the next paragraphs we briefly recall the basic interactions in cQED, mainly in order to set up the notation for the forthcoming discussion.  

Let us consider the simplest situation in which the cavity has a single resonating mode and interacts with a single neutral atom. The physics of this system is well described by the Jaynes-Cummings Hamiltonian \cite{JC63, SK93}, 
\begin{equation} \label{jynescummungsmodel}
H = \hbar \omega_C a^{\cal y}a + \hbar \omega_0 \ket{e}\bra{e} + \hbar g(a^{\cal y}\ket{g}\bra{e} + a\ket{e}\bra{g}),
\end{equation}
where $\omega_C, \omega_0$ are the frequencies of the resonant mode of the cavity and of the atomic transition, respectively, $g$ is Jaynes-Cummings coupling between the cavity mode and the two level system, $a^{\cal y}$ is the creation operator of a photon in the resonant cavity mode, and $\ket{g}, \ket{e}$ are the ground and excited states of the two level system. The vacuum Rabi frequency $g$ is given by 
\begin{equation}
g = \mu f(r_0)\left(\frac{\omega_C}{2  \epsilon_0 V_{\text{mode}}}\right)^{1/2},
\end{equation}
where $V_{\text{mode}}$ is the mode volume of the cavity, $f(r_0)$ is the mode-function at the position $r_0$ of the two level system, $\mu$ is the electric dipole transition moment, and $\epsilon_0$ is the free-space permittivity. 

Two important simplifications must be taken in order to derive the Jaynes-Cummings model (\ref{jynescummungsmodel}). First we consider the dipole approximation, in which only two electronic levels of the atom are considered and their interaction with the light mode is taken to be the one of an electric dipole \cite{KA83, WM94}. This approximation is valid as long as the light intensity is not too high and the wavelength of the light mode is much larger than the atomic dimension, which is usually the case in cavity QED. The second simplification that we used is the rotating wave approximation, where terms that oscillate with a frequency of $\omega_C + \omega_0$ are neglected. Again this is a good approximation if the vacuum Rabi frequency $g$ is much smaller than the resonance frequency $\omega_C$, which is the case for optical and microwave cavities \cite{KA83, WM94}.     

The atomic transition $g-e$ might also be driven by an monochromatic external laser field. The Hamiltonian for this process reads
\begin{equation}
H_L = \Omega(e^{- \delta i t}\ket{e}\bra{g} + e^{\delta i t}\ket{g}\bra{e}),
\end{equation}
where $\Omega$ and $\delta$ are the Rabi frequency and the frequency of the laser. 
%\begin{figure}
%\centering
%\psfrag{g4}{\hspace{-0.04cm}\color{red2}$\gamma$\normalcolor}
%\psfrag{g3}{\hspace{-0.1cm}\color{red2}$\gamma$\normalcolor}
%\psfrag{k}{\hspace{-0.04cm}\color{green4}$\kappa$\normalcolor}
%\psfrag{g13}{\hspace{-0.04cm}\color{green4}$g$\normalcolor}
%\psfrag{g24}{\hspace{-0.1cm}\color{green4}$g$\normalcolor}
%\psfrag{o}{\hspace{0.02cm}\color{nblue}$\Omega$\normalcolor}
%\psfrag{d2}{\hspace{-0.16cm}\raisebox{-0.06cm}{$\Delta$}}
%\psfrag{d1}{\hspace{-0.02cm}\raisebox{-0.06cm}{$\delta$}}
%\psfrag{w1}{\raisebox{-0.04cm}{\hspace{-0.05cm}$\omega$}}
%\psfrag{w2}{\raisebox{-0.04cm}{\hspace{-0.05cm}$\omega$}}
%\psfrag{1}{\hspace{-0.06cm}\raisebox{-0.18cm}{$1$}}
%\psfrag{2}{\hspace{-0.06cm}\raisebox{-0.18cm}{$2$}}
%\psfrag{3}{\hspace{-0.06cm}\raisebox{-0.0cm}{$3$}}
%\psfrag{4}{\hspace{-0.08cm}\raisebox{-0.0cm}{$4$}}
%\includegraphics[width=.7\linewidth]{cavityloss2.eps}
%\caption{\label{cavityqedfig} Cavity QED: A two level atom interacts via a dipole coupling $g$ with the photons in the cavity. Excitations a lost via %spontaneous emission at a rate $\gamma$ and cavity decay at a rate $\kappa$.}
%\end{figure}

On top of the coherent interaction, there are two main loss processes that affect the dynamics of the system: spontaneous emission from level $\ket{e}$ to level $\ket{g}$ at a rate $\gamma$ and leaking out of photons of the cavity mode at rate $\kappa$. The cavity decay rate $\kappa$ is connected to the quality factor of the cavity $Q$ by 
\begin{equation}
Q = \frac{\omega_C}{2\kappa}.
\end{equation}
Although other processes, such as thermal motion of the atom or dephasing due to background electromagnetic fields, also contribute to the losses of the system, their effect are usually much smaller and can in many application be disregarded.

The strong coupling regime of cavity QED is reached when the cooperativity factor, given by the the vacuum Rabi frequency to the square over the product of the spontaneous emission and cavity decay rates, is much larger than one
\begin{equation}
\xi := g^2/ \gamma \kappa \gg 1. 
\end{equation}
In this regime the coherent part of the evolution dominates over the decoherence processes and quantum dynamics of the joint atom cavity mode system can be observed. Such a regime has been first achieved in seminal experiments for microwave \cite{RBH01} and optical \cite{TRK92} frequencies, both using a single atom inside a Fabry-P\'erot cavity formed between two miniature spherical mirrors. In these works, the strong coupling regime has been employed e.g. to create conditional quantum dynamics and entanglement between the photonic and atomic degrees of freedom \cite{RBH01,THL+95}, to perform quantum non-demolition measurements \cite{RBH01}, and to realize the photon blockade effect \cite{BBM05}, which we revisit in section \ref{BHpol}.

In the next chapters we will be interested in the situation in which several cavities operating in the strong coupling regime are coupled to each other. It turns out that the Fabry-P\'erot architecture is not very suitable for the coupling of separate cavities. Nonetheless, several new cavity QED set-ups have recently emerged in which (i) the strong coupling regime has been achieved and (ii) the construction of arrays of coupled cavities has been realized, or at least seem reasonable to be so. These new systems, which we now briefly describe in the next paragraphs, are the strongest candidates for realizing the proposals we analyse in chapters \ref{BH}, \ref{kerrnonlinearity}, and \ref{heisenberg}.

\subsection{Photonic-band-gap-defect Nanocavities Coupled to Quantum Dots} \label{phddot}

The first example consists of quantum dots coupled to photonic-band-gap defect nanocavities. Photonic crystals are structures with periodic dielectric properties which affects the motion of photons, in a similar manner as the peridiocity of a semiconductor affects the motion of electrons \cite{JJ02}. Such structures can have band gaps in frequency space in which no photon of particular frequencies can propagate in the material. A nanocavity can be created in a photonic-band-gap material by producing a localized defect in the structure of the crystal, in such a way that light of a particular frequency cannot propagate outside the defect area. Large arrays of such nanocavities have been produced \cite{AV04} and photon hopping in the microwave and optical domains has been observed \cite{BTO00a,BTO00b}. Quantum dots, semiconductors whose excitations are confined in all three directions, behave in many respects like atoms and can be effectively addressed by lasers \cite{JA98}. Quantum dots can be created inside photonic crystal nanocavities \cite{BHA+05} and made to interact with the cavity mode to form a standard cavity QED system. Such a system has already been put in the strong coupling regime \cite{HBW+07}. Due to the extremely small volume of the nanocavity, the Jaynes-Cummings coupling coefficient can be extremely large \cite{AAS+03,SNAA05}. These cavities however have relatively small quality factors in comparison to other proposals \cite{SKV+05}, due mainly to fact that they are presently fabricated in two dimensions only. 

\subsection{Toroid Microcavities Coupled to Neutral Atoms} \label{todoid}

A second promising cavity QED set-up is composed of neutral atoms interacting with a toroid microcavity. Silica based toroid shaped microcavities can be produced on a chip with very large $Q$ factors \cite{AKS+03} and with resonant frequency in the optical range. They have moreover already been produced in arrays \cite{BLS+06,ASV07}. The coupling of a toroidal microcavity to a taper optical fiber has been realized with very high efficiency \cite{SKPV03} and could be used to provide the coupling between different cavities. Strong coupling of such cavities to neutral atoms has been reported in Ref. \cite{ADW+06}. Theoretical predictions \cite{SKV+05} show that very large cooperativity factors in the range $10^2-10^3$ are expected to be achieved in such systems.

\subsection{Stripline Superconducting Resonator Coupled to Cooper-pair Boxes}  \label{Cooperpair}

A third example is the cavity QED system formed by a Cooper-pair box coupled to a superconducting transmission line resonator. A Cooper-pair box is formed by two superconducting islands separated by a Josephson junction. Although it is a macroscopic device, it behaves in many ways like an atom as the relevant structure can be described by a two level quantum system. A superconducting transmission line cavity, on the other hand, is a quasi-one-dimensional coplanar waveguide resonator, formed on a chip with a resonance frequency in the microwave range. The strong-coupling regime in such a system has been reported in Ref. \cite{WSB+04} and recently two Cooper-pair boxes have been strongly coupled to the cavity mode \cite{MSG+07,SPS07}. The largest cooperativity factor to date has been achieved in this set-up \cite{MSG+07}. Although coupling between different cavities has not been realized, up to ten Cooper-pair boxes interacting in the strong coupling regime with the resonant cavity mode has been predicted to be achievable using the architecture of Ref. \cite{MSG+07}.

\section{Coupled Cavity Arrays} \label{cca}
%
%\begin{figure}
%\centering
%\includegraphics[width=\linewidth]{crystal1a.eps}
%\caption{\label{crystal1a} An array of coupled cavities. Photon hopping occurs due to the overlap (shaded green) of the light modes (green lines) of %adjacent cavities.}
%\end{figure}
%
In this section we present a description for the coupling of cavities in an array. Photon hopping can happen between neighboring cavities due to the overlap of the spacial profile of the cavity modes. In order to model such a process, we follow Refs. \cite{HBP06,SM98,YXLS99,BTO00a} and consider the array of cavities by a periodic dielectric constant, $\epsilon (\vec{r}) = \epsilon (\vec{r} + R\vec{n})$, where $\vec{r}$ is a given three dimensional vector, $R$ is a constant and $\vec{n}$ labels tuples of integers. In the Coulomb gauge the electromagnetic field is represented by a vector potential $\vec{A}$ satisfying $\nabla \cdot (\epsilon(\vec{r}) \, \vec{A}) = 0$. We can expand $\vec{A}$ in Wannier functions, $\omega_{\vec{R}}$, each localized in a single cavity at location $\vec{R} = R \vec{n}$. We describe this single cavity by the dielectric function $\epsilon_{\vec{R}}(\vec{r})$. Then, from Maxwell's Equations we can write the following eigenvalue Equation
\begin{equation} \label{Wanniereval}
\frac{\epsilon_{\vec{R}}(\vec{r}) \, \omega_C^2}{c^2} \vec{w}_{\vec{R}} \, -
\, \nabla \times \left( \nabla \times \vec{w}_{\vec{R}} \right) = 0 \, ,
\end{equation}
where the eigenvalue $\omega_C^2$ is the square the resonance frequency of the cavity, which is independent of $\vec{R}$ due to the periodicity. 
Assuming that the Wannier functions decay strongly enough outside the cavity, only Wannier modes of nearest neighbor cavities have a nonvanishing overlap.

Introducing the creation and annihilation operators of the Wannier modes, $a_{\vec{R}}^{\dagger}$ and $a_{\vec{R}}$, the Hamiltonian of the field can be written as
\begin{eqnarray} \label{arrayham2}
\mathcal{H} &=& \omega_C \sum_{\vec{R}}
\left(  a_{\vec{R}}^{\dagger} a_{\vec{R}} + \frac{1}{2} \right) + 2 \omega_C \alpha \sum_{<\vec{R}, \vec{R}'>}
\left( a_{\vec{R}}^{\dagger} a_{\vec{R}'} + \text{h.c.} \right) \, .
\end{eqnarray}
Here $\sum_{< \vec{R}, \vec{R}' >}$ is the sum of all pairs of cavities which are nearest neighbors
of each other. Since $\alpha \ll 1$, we neglected rotating terms which contain products of two
creation or two annihilation operators of Wannier modes in deriving (\ref{arrayham2}). The coupling parameter $\alpha$ 
is given by \cite{YXLS99,BTO00a},
\begin{equation}
\alpha = \int d^3r \, \left(\epsilon_{\vec{R}}(\vec{r}) \, - \, \epsilon(\vec{r}) \right) \,
\vec{w}_{\vec{R}}^{\star} \vec{w}_{\vec{R}'} \, ; \quad 
|\vec{R} - \vec{R}'| = R \, ,
\end{equation}
and has been obtained numerically for specific models \cite{MV97}.

Let us comment on the accuracy of this model in realistic cavity QED implementations. For cavity arrays in 
photonic crystals it provides an excellent approximation, as shown in the experiments reported in Refs. \cite{BTO00a,BTO00b}, in which the experimentally 
determined coupling coefficient $\alpha$ was found to be in good agreement with the theoretical prediction. 

For toroidal microcavities, 
coupling from the resonant mode to the fundamental mode of an optical fiber and back again with negligible noise has been 
achieved \cite{SKPV03}. A challenge that would have to be overcome for realizing Hamiltonian \ref{arrayham2} would be to prevent that significant photon population is localized in the fiber instead of in the cavities themselves. We point out that even if the fiber contains 
appreciable excitations the proposals we discuss in the sequel could still be realized, as one could considers the 
fiber mode as an extra site for the model (perhaps with different local terms as the other sites).  

In deriving Hamiltonian (\ref{arrayham2}) we have assumed that all the cavities have the same resonant frequency. In practice of course there will always be some detuning $\Delta$ between the resonant frequencies of two neighboring cavities, and this leads to a decrease in their coupling. Nonetheless, as long as $\alpha \gg \Delta$, model (\ref{arrayham2}) still provides a very good approximation. 

Although realizing Hamiltonian \ref{arrayham2} might be interesting for quantum information propagation \cite{HRP06,PHE04}, we do not need a quantum simulator to understand its properties. Indeed, the model is harmonic and can be very easily solved in terms of Bloch waves, which allows for a simple understanding of all its basic properties. The situation changes dramatically if we add an on-site interaction term. The interplay of tunneling and interaction leads to interesting many-body physics, as we show in the next chapter, where we analyse the creation of the Bose-Hubbard model in an array of coupled microcavities.

\chapter{The Bose-Hubbard Model in Coupled Cavity Arrays} \label{BH}

\section{Introduction}

The Bose-Hubbard model describes the physics of interacting bosons in a lattice. This model was first discussed by Fisher \textit{et al} in Refs. \cite{FWG+89, FGG90} as the bosonic counterpart of the Hubbard model \cite{Hub63}, largely used in solid state physics for explaining the transition between conducting and insulating systems, and is described by the following Hamiltonian
\begin{eqnarray} \label{bosehubbard}
H_{BH} = U \sum_{k} b^{\cal y}_k b_k(b^{\cal y}_k b_k - 1) - J \sum_{<k, k'>} (b_{k}^{\cal y}b_{k'} + h.c.) + \mu \sum_{k} b^{\cal y}_k b_k,
\end{eqnarray}
where $b_k^{\cal y}$ creates a boson at the $k$-th site. There are three different processes in the Hamiltonian: Bosons hop to neighboring sites at a rate $J$, two or more bosons in the same site repulse each other with a strength $U$ (or attract each other if $U < 0$), and new bosons are added to the systems with a rate determined by the chemical potential $\mu$.   

%Here $J$ is the hopping rate, $U$ the on-site interaction strength, and $\mu$ the chemical potential.

The Bose-Hubbard model has two different phases at zero temperature and a quantum phase transition when the ration $U/J$ crosses the critical value \cite{FWG+89, FGG90}. When $J \gg U$, the tunneling term dominates and the lowest energy state of the system is a condensate of delocalised bosons: the system is in the superfluid phase. If we start to increase the on-site interaction, it will be energetically favorable for the bosons to localise more and more in the sites, in order to minimise the repulsion. Above a specific value of $U/J$, the system ceases to be a superfluid and becomes a Mott insulator, with a well defined number of particles in each site\footnote{In the case of incommensurate filling factors, there will be some superfluid population even in the Mott insulator phase.}. The superfluid-Mott-insulator phase transition is of a quantum character and is driven by quantum fluctuations, as it occurs even at zero temperature \cite{Sac99}. An indicator of such transition is the variance of the number of bosons in a given site, $\langle (b^{\cal y}_k b_k - \langle b^{\cal y}_k b_k \rangle )^2 \rangle$. In the superfluid phase, the number of bosons in each site fluctuates and thus the variance has a non-zero value. In the Mott insulator phase, in turn, the variance is close to zero, as the bosons tend to be localised in each site\footnote{The interested reader is referred to \cite{LSA+07} for an in-depth study of the model.}. 

There has been a renascence in the study of the Bose-Hubbard model since it was shown to describe the physics of ultra-cold neutral atoms in an optical lattice \cite{JBC+98, GME+02, LSA+07,BDZ08}. As mentioned in the previous chapter, such possibility of realizing strongly interacting quantum many-body systems in a well controlled manner is appealing both from the point of view of condensed matter physics and quantum information science.  

In this chapter we propose to realize the Bose-Hubbard Hamiltonian in arrays of coupled cavities, using combined atom-photon excitations, polaritons, and even photons as the bosonic particles of the model. This approach is interesting for two reasons. First it could allow us to put polaritons and photons in phases which do not occur naturally. Second, as discussed in chapter \ref{AMO}, in coupled cavity arrays one has full accessibility to individual sites, something that has not been achieved in optical lattices to date. Such individual access is important for measuring local properties of the system and for engineering inhomogeneous models.  

The chapter is organized as follows. In section \ref{BHpol} we show that the dynamics of polaritons in coupled cavity arrays filled in with an ensemble of 4 level atoms is described by the Bose-Hubbard model and in subsection \ref{secd} we analyse the robustness to noise and imperfections of the set-up. In section \ref{twocompBH}, in turn, we show that for a particular choice of parameters the configuration of section \ref{BHpol} can be used to create a two component Bose-Hubbard model. In section \ref{photonicBH} we discuss the feasibility of realizing the Bose-Hubbard model for the cavity photons themselves. Finally, in section \ref{sec:measure} we discuss ways to measuring the states of the polaritons.  

\section{Polaritonic Bose-Hubbard Model} \label{BHpol}

As we saw in chapter \ref{AMO}, the coupling between neighboring cavities naturally leads to the tunneling of photons in arrays of coupled-cavities. In order to realize the Bose-Hubbard model we thus need to a find a way to generate the repulsion interaction. Indeed, up to harmonic terms which can easily be compensated for, the on-site repulsion interaction has the form of a self Kerr non-linearity $(b_k^{\cal y})^2 b_k^2$ \cite{Blo65}. Although such an interaction naturally appears in some mediums as a result of a non-zero third order electric susceptibility \cite{Blo65}, its effect on the single quantum level is negligible, which is a reason for the difficulty of realizing nonlinear optics for individual quantum systems. However, using the enhanced light-matter interaction in cavity QED it is possible to engineer much stronger nonlinearities. Intuitively the strong interaction of the light mode with atoms inside the cavity, under particular circumstances, mediates strong nonlinear interaction among the photons of the cavity mode. The strength of the nonlinearity can be increased even further if instead of considering photons as the bosonic particles of the model, we consider polaritons, joint photonic-atomic excitations. In this section analyse the feasibility of realizing the Bose-Hubbard Hamiltonian using such polaritons. 
\begin{figure}
	\centering
	\includegraphics[width=0.8\textwidth]{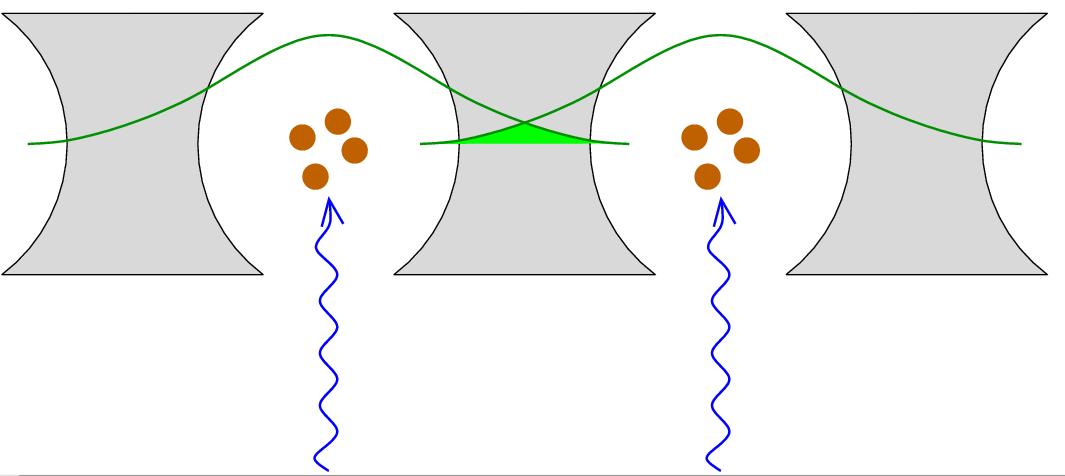}
	\caption{An array of cavities as described by our model. Photon hopping occurs due to the overlap (shaded green) of the light modes (green lines) of adjacent cavities. Atoms in each cavity (brown), which are driven by external lasers (blue) give
	rise to an on site potential.}
	\label{crystal}
\end{figure}

\subsection{Derivation of the Model}

In order to derive the polaritonic Bose-Hubbard model, we consider an array of cavities that are filled with atoms of a particular 4 level structure, which
are driven with an external laser, see figure \ref{crystal}. Thereby the laser drives the atoms in the same manner as in Electromagnetically Induced Transparency (EIT) \cite{FIM05}: The transitions between levels 2 and 3 are coupled to the laser field and the transitions between levels 2-4 and 1-3 couple via dipole moments to the cavity resonance mode, as shown in fig. \ref{level}. We assume furthermore that levels 1 and 2 are metastable and thus ignore spontaneous emission from them.
\begin{figure}
	\centering
	\includegraphics[width=.7\textwidth]{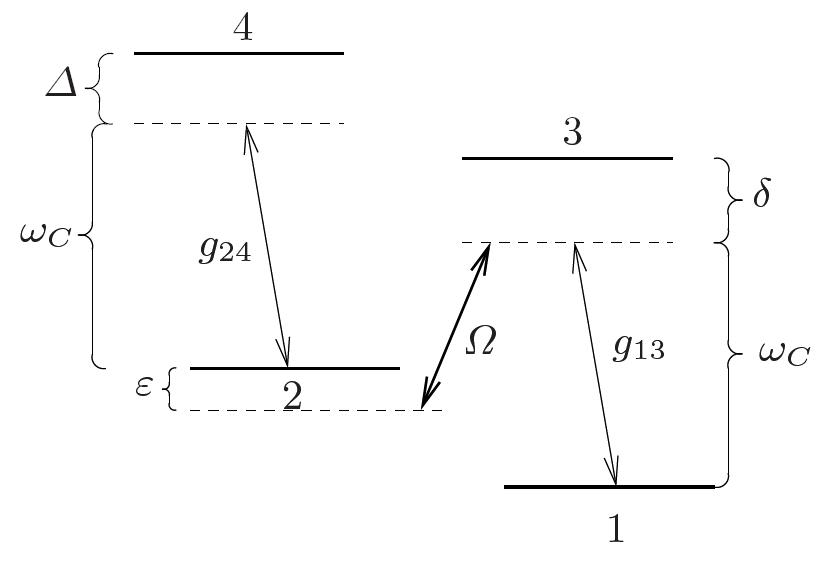}
	\caption{The level structure and the possible transitions of one atom, $\omega_C$ is the frequency of the cavity mode,
	$\Omega$ is the Rabi frequency of the driving by the laser, $g_{13}$ and $g_{24}$
	are the parameters of the respective dipole couplings and $\delta$, $\Delta$ and $\varepsilon$
	are detunings.}
	\label{level}
\end{figure}

It has been shown by Imamoglu and co-workers, that this atom cavity system can show a very large nonlinearity \cite{SI96,ISWD97,WI99,GAG99}, and a similar nonlinearity has recently been observed experimentally \cite{BBM05}, where an instance of the photon blockade effect was realised. In the photon blockade effect, due to the strong interaction of the cavity mode with atoms, a single photon can modify the resonance frequency of the cavity mode in such a way that a second photon can not enter the cavity before the first leaks out. Such a phenomenon has been predicted first with the four level structure of Refs. \cite{SI96,ISWD97}, and then shown experimentally using a two level structure \cite{BBM05}. Recently this idea has been used in Refs. \cite{ASB07,GTCH06}, where it has been proposed that in an array of optical cavities with two level atoms operating in the photon blockade regime, strongly correlated systems, having similarities with the Bose-Hubbard model, can be created. However, the use of two levels has the fundamental drawback that losses due to spontaneous emission cannot be avoided, which renders the feasibility of the schemes \cite{ASB07,GTCH06} dependent on extremely large cooperativity numbers, yet to be achieved in experiments. 

Considering the level structure of figure \ref{level}, in a rotating frame with respect to
\begin{equation*}
H_0 = \omega_C \left( a^{\dagger} a + \frac{1}{2} \right) + \sum_{j=1}^N \left( \omega_C \sigma_{22}^j + \omega_C \sigma_{33}^j + 2 \omega_C \sigma_{44}^j \right),
\end{equation*}
the Hamiltonian of the atoms in the cavity reads,
\begin{eqnarray} \label{H_manyatom}
H^I & = &
\sum_{j=1}^N \left(\varepsilon \sigma_{22}^j + \delta \sigma_{33}^j + (\Delta + \varepsilon)
\sigma_{44}^j \right) \\
& + & \, \sum_{j=1}^N \left( \Omega \, \sigma_{23}^j \, + \,
g_{13} \, \sigma_{13}^j \, a^{\dagger} \, + \,
g_{24} \, \sigma_{24}^j \, a^{\dagger} \, + \, \text{h.c.} \right)  \, , \nonumber
\end{eqnarray}
where $\sigma_{Al}^j = \ket{k_j}\bra{l_j}$ projects level $l$ of atom $j$ to level $k$ of the same atom,
$\omega_C$ is the frequency of the cavity mode,
$\Omega$ is the Rabi frequency of the driving by the laser and $g_{13}$ and $g_{24}$
are the parameters of the dipole coupling of the cavity mode to the respective atomic
transitions.

In a cavity array, $N$ atoms in each cavity couple to the cavity mode via the interaction $H^I$ and photons tunnel between neighboring cavities as described in
Equation (\ref{arrayham2}). Hence the full Hamiltonian describing this system reads
\begin{eqnarray} \label{polariton_raw}
H = \sum_{\vec{R}} H^I_{\vec{R}} + \omega_C \sum_{\vec{R}}
\left(  a_{\vec{R}}^{\dagger} a_{\vec{R}} + \frac{1}{2} \right)
+ 2 \omega_C \alpha \sum_{<\vec{R}, \vec{R}'>}
\left( a_{\vec{R}}^{\dagger} a_{\vec{R}'} + \text{h.c.} \right)
\end{eqnarray}
and can emulate a Bose-Hubbard model for polaritons as we will see.

Assuming that all atoms interact in the same way with the cavity mode\footnote{We can relax this assumption and only require that the atoms do not move very fast. In such a case the Dicke-type operators have to be modified taking into account the phase as a function of the position of the atoms.}, we can restrict ourselves to Dicke type dressed states, in which the atomic excitations are delocalised among all the atoms. In the case where $g_{24} = 0$ and $\varepsilon = 0$ , level 4 of the atoms decouples from the dressed-states excitation manifold \cite{WI99}.
If we furthermore assume that the number of atoms is large, $N \gg 1$, Hamiltonian (\ref{H_manyatom}) can be diagonalised in the subspace spanned by symmetric Dicke dressed states.
Let us therefore define the following creation (and annihilation) operators:
\begin{eqnarray} \label{polariton_operators}
p_0^{\dagger} & = & \frac{1}{B} \, \left(g S_{12}^{\dagger} - \Omega a^{\dagger} \right), \\
p_+^{\dagger} & = & \sqrt{\frac{2}{A (A + \delta)}} \, \left(\Omega S_{12}^{\dagger} + g a^{\dagger} +
\frac{A + \delta}{2} S_{13}^{\dagger} \right) \nonumber \\
p_-^{\dagger} & = & \sqrt{\frac{2}{A (A - \delta)}} \, \left(\Omega S_{12}^{\dagger} + g a^{\dagger} -
\frac{A - \delta}{2} S_{13}^{\dagger} \right) \nonumber \, ,
\end{eqnarray}
where $g = \sqrt{N} g_{13}$, $B = \sqrt{g^2 + \Omega^2}$, $A = \sqrt{4 B^2 + \delta^2}$,
$S_{12}^{\dagger} = \frac{1}{\sqrt{N}} \sum_{j=1}^N \sigma_{21}^j$
and $S_{13}^{\dagger} = \frac{1}{\sqrt{N}} \sum_{j=1}^N \sigma_{31}^j$.

The operators $p_0^{\dagger}$, $p_+^{\dagger}$ and $p_-^{\dagger}$ describe polaritons,
quasi particles formed by combinations of atom and photon excitations.
By looking at their matrix representation, we can see that in the limit of large atom numbers, $N \gg 1$ and in the symmetric subspace, they satisfy bosonic commutation relations,
\begin{equation} \label{polariton_comm}
\left[ p_j, p_l \right] = 0 \: \: \text{and} \: \:
\left[ p_j, p_l^{\dagger} \right] = \delta_{jl} \: \: \text{for} \: \: j,l = 0,+,-.
\end{equation}
$p_0^{\dagger}$, $p_+^{\dagger}$ and $p_-^{\dagger}$ thus describe independent bosonic
particles. In terms of these polaritons, the Hamiltonian (\ref{H_manyatom}) for $g_{24} = 0$ and
$\varepsilon = 0$ reads,
\begin{equation} \label{H0polariton}
\left[H^I\right]_{g_{24} = 0, \varepsilon = 0} = 
\mu_0 \, p_0^{\dagger} p_0 + \mu_+ \, p_+^{\dagger} p_+ + \mu_- \, p_-^{\dagger} p_- \, ,
\end{equation}
where the frequencies are given by $\mu_0 = 0$, $\mu_+ = (\delta - A)/2$ and
$\mu_- = (\delta + A)/2$. 

The polaritons $p_0$ only contain atomic contributions in the two metastable states $1$ and $2$ but not in level $3$, which shows spontaneous emission. Hence, these polaritons do not live in radiating atomic levels and are thus called dark state polaritons \cite{FIM05}.

\subsubsection{Polariton-polariton Interactions}

We now show that the dynamics of $p_0$ can be described by the self Kerr nonlinearity interaction.
To this aim we fist write the full Hamiltonian $H^I$, given by Eq. (\ref{H_manyatom}), in the polariton basis,
expressing the operators $\sum_{j=1}^N \sigma_{22}^j$ and
$a^{\dagger} \, \sum_{j=1}^N \sigma_{24}^j$
in terms of $p_0^{\dagger}$, $p_+^{\dagger}$ and $p_-^{\dagger}$. 
The coupling of the polaritons to the level 4 of the atoms via the dipole moment $g_{24}$ reads,
\begin{equation} \label{coupletolevel4}
g_{24} \, \left( \sum_{j=1}^N \sigma_{42}^j \, a \, + \text{h.c.} \right) \approx
- g_{24} \, \frac{g \Omega}{B^2} \, \left( S_{14}^{\dagger} \, p_0^2 \, + \text{h.c.} \right) \, ,
\end{equation}
where $S_{14}^{\dagger} = \frac{1}{\sqrt{N}} \sum_{j=1}^N \sigma_{41}^j$. 
In deriving (\ref{coupletolevel4}), we made use of the rotating wave approximation: In a
frame rotating with respect to (\ref{H0polariton}), the polaritonic creation operators rotate with the frequencies $\mu_0$, $\mu_+$ and $\mu_-$.
Furthermore, the operator $S_{14}^{\dagger}$ rotates at the frequency $2 \mu_0$, (see Eq. (\ref{H_manyatom})). Hence, provided that 
\begin{equation} \label{rotatingwappr}
|g_{24}| \, , \, |\varepsilon| \, , \, |\Delta| \, \ll \, |\mu_+| \, , \, |\mu_-|
\end{equation}
all terms that rotate at
frequencies $2 \mu_0 - (\mu_+ \pm \mu_-)$ or $\mu_0 \pm \mu_+$ or $\mu_0 \pm \mu_-$
can be neglected, which eliminates all interactions that would couple $p_0^{\dagger}$ and $S_{14}^{\dagger}$
to the remaining polariton species.

For $|g_{24 }g \Omega / B^2| \ll |\Delta|$, the coupling to level 4 can be treated in a perturbatively.
This results in an energy shift of $2 U$ with
\begin{equation} \label{osint}
U = - \frac{g_{24}^2}{\Delta} \,
\frac{N g_{13}^2 \, \Omega^2}{\left(N g_{13}^2 \, + \, \Omega^2 \right)^2}
\end{equation}
and in an occupation probability of the state of one $S_{14}^{\dagger}$ excitation of
$- 2 U / \Delta$, 
which will determine an effective decay rate for the polariton $p_0^{\dagger}$ via
spontaneous emission from level 4. Note that $U > 0$ for $\Delta < 0$ and vice versa.
In a similar way, the two photon detuning $\varepsilon$ leads to and energy shift of
$\varepsilon \, g^2 \, B^{-2}$ for the polariton $p_0^{\dagger}$, which plays the role of a chemical
potential in the effective Hamiltonian.

Hence, provided (\ref{rotatingwappr}) holds, the Hamiltonian for the dark state polariton $p_0^{\dagger}$
can be written as
\begin{equation} \label{Heffect}
H_{\text{eff}} = U \, \left(p_0^{\dagger}\right)^2 \left(p_0\right)^2 \, + \, \varepsilon \, \frac{g^2}{B^2} \, p_0^{\dagger} p_0 \, ,
\end{equation}
in the rotating frame.

\subsubsection{Polariton Tunnelling}

Let us now look at an array of interacting cavities, each under the conditions discussed above. 
The on-site interaction and the chemical potential of the Hamiltonian
(\ref{arrayham2}) have already been incorporated in the polariton analysis for one individual cavity.
The tunneling term transforms into the polariton picture via (\ref{polariton_operators}).
To distinguish between the dark state polaritons in different
cavities, we introduce the notation $p_{\vec{R}}^{\dagger}$ to label the polariton $p_0^{\dagger}$
in the cavity at position $\vec{R}$. We get,
\begin{eqnarray} 
a_{\vec{R}}^{\dagger} a_{\vec{R}'}
& \approx & \frac{\Omega^2}{B^2} \, p_{\vec{R}}^{\dagger} \, p_{\vec{R}'} \\
& + & \text{"terms for other polariton species"} \nonumber \, ,
\end{eqnarray}
where we have applied a rotating wave approximation.
Contributions of different polaritons decouple due to the separation of their
frequencies $\mu_0$, $\mu_+$ and $\mu_-$.
As a consequence the Hamiltonian for the polaritons $p_{\vec{R}}^{\dagger}$ takes on the form
(\ref{bosehubbard}), with 
\begin{equation}
J = \frac{2 \omega_C \Omega^2}{N g_{13}^2 + \Omega^2} \alpha \, ,
\end{equation}
$U$ as given by equation (\ref{osint}) and $\mu = \epsilon g^2 / B^2$. 

\subsection{Spontaneous Emission and Cavity Decay} \label{secd}

Let us now analyse the loss sources in the system. Levels 1 and 2 are metastable and therefore have negligible decay rates on the relevant time scales. The decay mechanism for the dark state polariton originates then mainly from the leaking out of photons from the cavity and from the small, but still non-zero, population of level 4, due to the coupling $\sum_{j=1}^N (\sigma_{42}^ja + h.c.)$, which leads to spontaneous emission. The resulting effective decay rate for the dark state polariton reads
\begin{equation} \label{gamma0}
\Gamma_0 = \frac{\Omega^2}{B^2}\kappa + \Theta(n - 2)\frac{g_{24}^2g^2\Omega^2}{\Delta^2 B^4}\gamma_4,
\end{equation}
where $\kappa$ is the cavity decay rate, $\gamma_4$ the spontaneous emission rate from level 4, $n$ is the average number of photons, and $\Theta$ the Heaveside step function. Assuming that $g_{13} = g_{24}$, the maximal achievable rate $U/ \Gamma_0$ can be readily seen to be $g_{13}/\sqrt{4 \kappa \Theta(n - 2)\gamma_4}$. Therefore, in the strong coupling regime we have $U/ \Gamma_0 \gg 1$, which ensures that the dynamics of the Bose-Hubbard model can be observed under realistic conditions. 

\subsection{The Phase Transition} 

Of great interest is of course whether the quantum phase transition from a Mott insulator to a superfluid state could be observed with the present polariton approach and realisable technology. To study this phase transition, we consider a system with on average one polariton per
cavity. Here, the Mott insulator state is characterised by the fact that the local state in each cavity is a single polariton Fock state with vanishing fluctuations of the polariton number. The superfluid phase in contrast shows polariton number fluctuations.

Figure \ref{vardiff4} shows numerical simulations of the full dynamics of
an array of three cavities as described by equation (\ref{polariton_raw}), including spontaneous emission and cavity decay, and studies the number of polaritons in one cavity/site, 
$n_l = \langle p_l^{\dagger} p_l \rangle$ for site $l$, and the number fluctuations,
$F_l = \langle (p_l^{\dagger} p_l )^2 \rangle - \langle p_l^{\dagger} p_l \rangle^2$.
A comparison between the full model (\ref{polariton_raw}) and the effective model (\ref{bosehubbard}) is done by considering the differences
in the occupation numbers, $\delta n_l = \left[n_l\right]_{\text{cavities}} - \left[n_l\right]_{\text{BH}}$,
and number fluctuations, $\delta F_l = \left[F_l\right]_{\text{cavities}} - \left[F_l\right]_{\text{BH}}$, at each timestep.
Initially there is exactly one polariton in each cavity. We take parameters for toroidal microcavities from \cite{SKV+05}, $g_{24} = g_{13} = 2.5 \times 10^9 s{-1}$, $\gamma_4 = \gamma_3 = 1.6 \times 10^7 s^{-1}$ and $\kappa = 0.4 \times 10^5 s^{-1}$ \cite{HBP06}.
The results are shown in figure \ref{vardiff4}. As the system is driven from a Mott insulator state to a superfluid state by ramping up the driving laser $\Omega$,
the particle number fluctuations increase significantly. 
The numerics show an excellent agreement of the full dynamics as described by equation (\ref{polariton_raw}) and the dynamics of the corresponding Bose-Hubbard model (\ref{bosehubbard}) and thus confirm the possibility of observing the Mott-insulator-to-superfluid transition in such a system.
The oscillatory behavior is related to the fact that the initial state of the system is, due to the
nonzero tunneling rate $J$, not its ground state. This initial state was chosen since it's preparation in an experiment is expected to be easier compared to other states and figure \ref{vardiff4} thus shows the dynamics that is expected to be observed in an experiment.
\begin{figure*}
%\sidecaption
\centering
\includegraphics[width=1\textwidth]{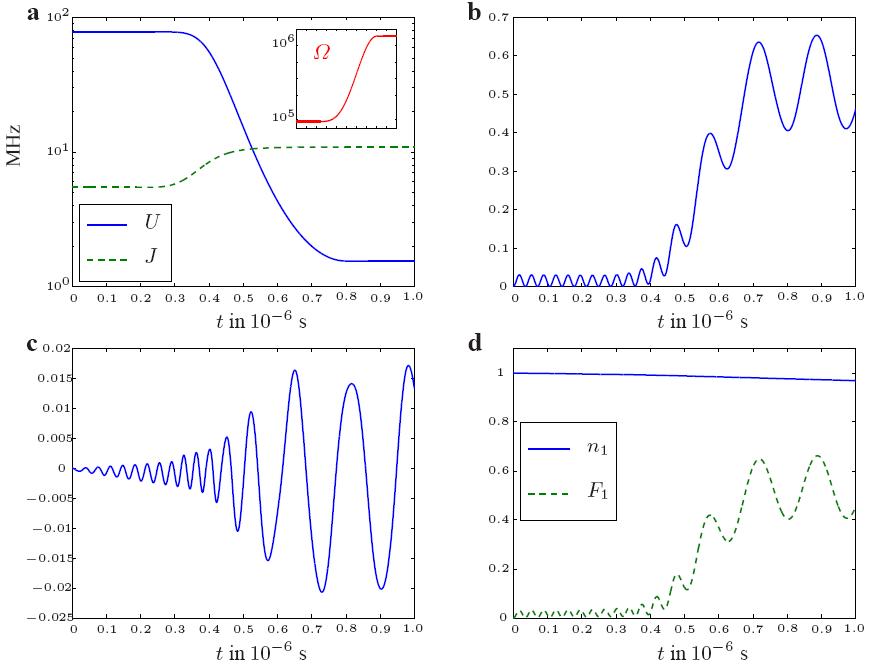}
\caption{\label{vardiff4} The Mott insulator to superfluid transition for 3 polaritons in 3 cavities compared to 3 particles in a 3 site Bose-Hubbard model.
{\bf a:} Log-plot of $U$ and $J$ and linear plot of the time dependent $\Omega$ (inset).
{\bf b:} Number fluctuations for polaritons in cavity 1, $F_1$, for a single quantum jump trajectory.
{\bf c:} Difference between number fluctuations for polaritons in a cavity and number fluctuations in a pure BH model, $\delta F_1$, for a single quantum jump trajectory.
{\bf d:} Expectation value and fluctuations for the number of polaritons in one cavity/site according to
the effective model (\ref{bosehubbard}) with damping (\ref{gamma0}).
$g_{24} = g_{13} = 2.5 \times 10^9 s{-1}$, $\gamma_4 = \gamma_3 = 1.6 \times 10^7 s^{-1}$, $\kappa = 0.4 \times 10^5 s^{-1}$, $N = 1000$,
$\Delta = -2.0 \times 10^{10} \text{s}^{-1}$ and $2 \omega_C \alpha = 1.1 \times 10^{7} \text{s}^{-1}$.
The Rabi frequency of the driving laser is increased from initially $\Omega = 7.9 \times 10^{10} \text{s}^{-1}$ to finally $\Omega = 1.1 \times 10^{12} \text{s}^{-1}$.
Deviations from the pure BH model are about 2\%.}
\end{figure*}

The phase diagram for this model and its dependence on the atom number $N$ have been analysed in \cite{RF07}, where also a glassy phase has been predicted for a system with disorder.

An polariton model which generates an effective Lieb-Liniger Hamiltonian and is
capable of reproducing the Tonks Girardeau regime was proposed in \cite{CGM+07}.
The scheme employs atoms of the same level structure as in figure \ref{level} that
couple to light modes of a tapered optical fiber.

\section{Two Component Bose-Hubbard Model} \label{twocompBH}

In this section we discuss an extension of the proposal just presented for creating a two component polaritonic Bose-Hubbard model. This model describes the dynamics of two independent bosonic particles in each site and reads 
\begin{eqnarray} \label{bosehubbardtwocomp}
H_{\text{eff}} & = &
\sum_{k,j = b,c} \mu_j \, n^{(j)}_{k} \, - \, \sum_{< k, k' >; j,l = b,c} 
J_{j,l} \left( j_{k}^{\dagger} \, l_{k'}\, + \, \text{h.c.} \right) \, \nonumber \\
& + & \sum_{k,j = b,c} U_j \, n^{(j)}_{k} \left(n^{(j)}_{k} - 1\right) \, 
+ \sum_{k} U_{b,c} \, n^{(b)}_{k} n^{(c)}_{k} \, ,
\end{eqnarray}
where $b_{k}^{\dagger}$($c_{k}^{\dagger}$) create polaritons of the type $b$($c$) in the cavity at site $k$,  $n^{(b)}_{k} = b_{k}^{\dagger} b_{k}$ and
$n^{(c)}_{k} = c_{k}^{\dagger} c_{k}$. $\mu_b$ and $\mu_c$ are the polariton energies, $U_b$, $U_c$ and $U_{b,c}$ their on-site interactions and $J_{b,b}$, $J_{c,c}$ and $J_{b,c}$ their tunneling rates. This model exhibits interesting quantum many-body phenomena, which are partially also known for a Luttinger liquid \cite{Lut63}, such as spin density separation \cite{Hal81} and phase separation \cite{CH03}. 

We can realize Hamiltonian \ref{bosehubbardtwocomp} in the same set-up considered above, by considering the interaction in a particular dispersive regime. Suppose we are in the regime in which $\delta \gg \Omega, g$. Then, from Eq. \ref{polariton_operators} we find 
\begin{equation} \label{polariton_operators_2}
\begin{array}{rlrl}
p_0^{\dagger} = & \frac{1}{B} \, \left(g S_{12}^{\dagger} - \Omega a^{\dagger} \right) & \quad \mu_0 = & 0 \nonumber \\
p_-^{\dagger} \approx & \frac{1}{B} \, \left(\Omega S_{12}^{\dagger} + g a^{\dagger} \right) - \frac{B}{\delta} S_{13}^{\dagger} & \quad \mu_- = & - \frac{B^2}{\delta} \nonumber \\
p_+^{\dagger} \approx & S_{13}^{\dagger} + \frac{1}{\delta} \, \left(\Omega S_{12}^{\dagger} + g a^{\dagger} \right) & \quad \mu_+ = & \delta + \frac{B^2}{\delta}
\end{array}\,, 
\end{equation}
in leading order in $\delta^{-1}$. Note that now, not only the polariton $p_0$, but also $p_-^{\dagger}$ do not experience loss from spontaneous emission to leading order in $\delta^{-1}$. We can therefore define two dark state polaritons species
\begin{equation} \label{b_c_def}
b^{\dagger}  = \frac{1}{B} \, \left(g S_{12}^{\dagger} - \Omega a^{\dagger} \right) \, ; \quad
c^{\dagger} = \frac{1}{B} \, \left(\Omega S_{12}^{\dagger} + g a^{\dagger} \right) \, .
\end{equation}
From an analysis similar to the one outlined before for the one component case, it can be shown that the dynamics of these two polaritons species is given by the two component Bose-Hubbard model.
The parameters of the model (\ref{bosehubbardtwocomp}) are given by,
$J_{bb} = \alpha \frac{g^2}{B^2}$,
$J_{cc} = \alpha \frac{\Omega^2}{B^2}$,
$J_{bc} = \alpha \frac{g \Omega}{B^2}$,
$U_b = - \frac{g_{24}^2 g^2 \Omega^2}{B^4 \Delta}$,
$U_c = - \frac{g_{24}^2 g^2 \Omega^2}{B^4 (\Delta + 2 B^2/\delta)}$ and
$U_{bc} = - \frac{g_{24}^2 (g^2 - \Omega^2)^2}{B^4 (\Delta + B^2/\delta)}$.
For further details the reader is referred to Ref. \cite{HBP07b}.

An example how the parameters of the effective Hamiltonian (\ref{bosehubbardtwocomp})
vary as a function of the intensity of the driving laser $\Omega$. is given in figure
\ref{range_HBP07c}, where the parameters of the atom cavity system are chosen to be
$g_{24} = g_{13}$, $N = 1000$, $\Delta = - g_{13} / 20$, $\delta = 2000 \sqrt{N} g_{13}$ and $\alpha = g_{13} / 10$. figure \ref{range_HBP07c} shows the interactions $U_b$, $U_c$ and $U_{bc}$, the tunneling rates $J_{bb}$, $J_{cc}$ and $J_{bc}$ and $|\mu_c - \mu_b|$ as a function of $\Omega / g_{13}$.
\begin{figure}
\centering
\includegraphics[width=1\linewidth]{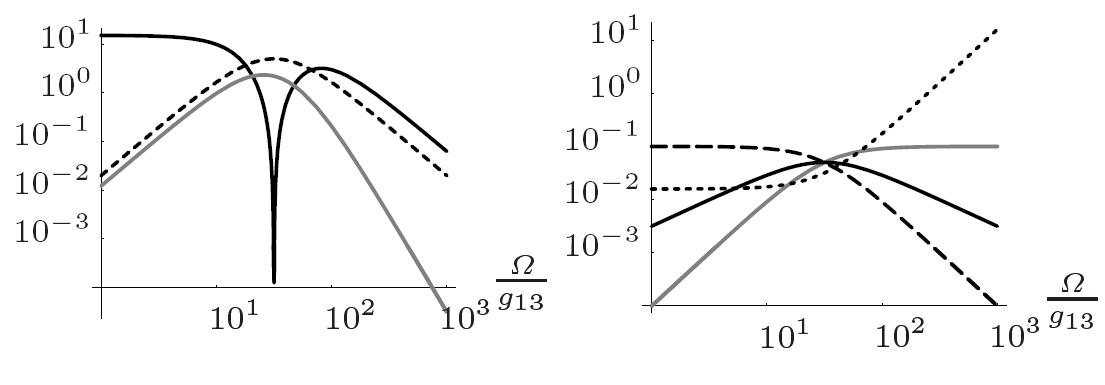}
\caption{Left: The polariton interactions $U_b$ (dashed line), $U_c$ (gray line) and $U_{bc}$
(solid line) in units of $g_{13}$ as a function of $\Omega / g_{13}$. Right: The tunneling rates $|J_{bb}|$ (dashed line), $|J_{cc}|$ (gray line) and $|J_{bc}|$ (solid line) together with $|\mu_c - \mu_b|$ (dotted line) in units of $g_{13}$ as a function of $\Omega / g_{13}$. The parameters of the system are $g_{24} = g_{13}$, $N = 1000$, $\Delta = - g_{13} / 20$, $\delta = 2000 \sqrt{N} g_{13}$ and $\alpha = g_{13} / 10$.}
\label{range_HBP07c}
\end{figure}
For $g \approx \Omega$ one has $|U_{bc}| \ll |U_b|, |U_c|$ and $J_{bb} \approx J_{cc} \approx J_{bc}$.
Whenever $|\mu_c - \mu_b| < |J_{bc}|$, $b^{\dagger}$ polaritons get converted into $c^{\dagger}$ polaritons
and vice versa via the tunneling $J_{bc}$.
With the present choice of $\alpha$ and $\delta$, this happens for
$0.16 g < \Omega <  1.6 g$.

\section{Photonic Bose-Hubbard Model} \label{photonicBH}

In the last section we saw that, as soon as the strong coupling regime is achieved, it is possible to engineer a polaritonic Bose-Hubbard model in coupled cavity arrays. The prospect of realizing even higher cooperativity factors in cavity QED, in the regime $10^2-10^3$ \cite{SKV+05}, opens up the possibility of engineering the Bose-Hubbard model also for the cavity mode photons. Such a setting is actually conceptually simpler than the polaritonic and has already been outlined before. Given that the photonic tunneling term is a consequence of the interaction of neighboring cavities, to realize a photonic Bose-Hubbard model we must find a way to create large photon Kerr nonlinearities using the interaction of the atoms with the cavity mode. 

To create a pure photonic Bose-Hubbard model would be very interesting as in this way new regimes for light, which do not occur naturally, could be engineered. For example, in the photonic Mott insulator exactly one photon exists in each cavity,
provided the whole structure contains on average one photon per cavity. Moreover, the photons are localized in the cavity they are in and are not able to hop between different cavities. In such a situation photons behave as strongly correlated particles that are each "frozen" to their lattice site, a system that would correspond to a crystal formed by light.

The first proposal for generating strong Kerr nonlinearities is the EIT based setting analysed in the last section. As shown by Hartmann and Plenio in Ref. \cite{HP07}, in the regime characterized by $g \ll \Omega$ and $g_{24}g \ll |\Delta \Omega|$, the dark sate polariton $p_0$ approximates well minus the photon annihilation operator, $- a$, and the non-harmonic part of Hamiltonian \ref{Heffect} becomes  
\begin{equation}
-  \frac{g_{24}g}{|\Delta \Omega|}  \frac{g}{\Omega} g_{24} (a^{\cal y})^2 a^2. 
\end{equation}
From the decay rate $\Gamma_0$ for the dark state polariton we can infer the decay rate for the photonic interaction, $\Gamma = \kappa + \Theta(n - 2)g_{24}^2g^2\Delta^{-2}\Omega^{-2}\gamma_4$. As $g_{24}g \ll |\Delta \Omega|$, we then see that cavity decay is the main source of loss. 

Taking, for example, the case of toroidal microcavities, for which the achievable parameters are predicted to be $g_{24} = 2.5 \times 10^9 s^{-1}$, $\gamma_4 = 1.6 \times 10^7 s^{-1}$, and $\kappa = 0.4 \times 10^{5} s^{-1}$ \cite{SKV+05}, and setting the parameters such that $g/\Omega = g_{24}g / |\Delta \Omega| = 0.1$, we obtain $U/\Gamma = 625$, which would be sufficient to observe a Mott insulator of photons. 

A second proposal to realize large Kerr nonlinearities will be discussed in depth in chapter \ref{kerrnonlinearity}. 

\section{Measurements \label{sec:measure}}

In this section we discuss ways how information about the state of the polaritonic
Bose Hubbard models can be obtained in experiments. The possibilities to access the system in measurements are somewhat complementary to those for cold atoms in optical lattices \cite{BDZ08}. Whereas time of flight images give access to global quantities
such as phase coherence in optical lattices, coupled cavity arrays allow to access
the local particle statistics in measurements. Unlike optical lattices they would thus allow to obtain direct experimental evidence for the exact integer particle number
in each lattice site in a Mott Insulator regime. We first discuss the measurement scheme for the one-component model and then turn to its extension for the two-component model.

\subsection{One-component Model}

For the model (\ref{bosehubbard}), the number of polaritons in one individual cavity can be measured via state selective resonance fluorescence.
To that end, the polaritons are made purely atomic excitations by adiabatically switching off the driving
laser. The process is adiabatic if $g B^{-2} \frac{d}{dt} \Omega \ll |\mu_+| , |\mu_-|$ \cite{FIM05}, which means it can be done fast enough to prevent polaritons from hopping between
cavities during the switching. Hence, in each cavity, the final number of atomic excitations in level 2 is equal to the initial number of dark state polaritons, while the other polariton species will not be transferred to the atomic level 2 (c.f. (\ref{polariton_operators})). In this state cavity loss is eliminated and the atomic excitations are very long lived. They can then be measured with high efficiency and precision using state-selective resonance fluorescence \cite{Ima02,JK02}.
% The time available for carrying out this measurements is limited by the
% trapping time of the atoms.

% By measuring the number of polaritons in several individual cavities in possibly repeated experiments one
% can obtain the local number fluctuations which allow to distinguish between the different phases of the system, Mott insulator, superfluid and strong attractive potential.
% 
% Furthermore the fluorescence spectrum contains information about the spectrum of the effective model (\ref{bosehubbard}), where, for a correct interpretation of the obtained data, the Rabi frequency of the probe laser and the linewidth of the involved atomic levels need to be taken into account.

\subsection{Two-component Model}

The number statistics for both polariton species for the model (\ref{bosehubbardtwocomp}), $b^{\dagger}$ and $c^{\dagger}$, in one cavity can again be
measured using state selective resonance fluorescence. In the two-component case the STIRAP can however not be applied as in the single component case because the energies $\mu_b$ and $\mu_c$ are similar and the passage would thus need to be extremely slow to be adiabatic.

For two components, one can do the measurements as follows. First the external driving laser $\Omega$ is switched off. Then the roles of atomic levels 1 and 2 are interchanged in each atom via a Raman transition by applying a $\pi / 2$-pulse. To this end the transitions $1 \leftrightarrow 3$ and $2 \leftrightarrow 3$ are driven with two lasers (both have the same Rabi frequency $\Lambda$) in two-photon resonance for a time
$T = \pi \delta_{\Lambda} / |\Lambda|^2$ ($\delta_{\Lambda}$ is the detuning from atomic level 3). The configuration is shown in figure \ref{flippass}a. This pulse results in the mapping
$\ket{1_j} \leftrightarrow \ket{2_j}$ for all atoms $j$.

Next another laser, $\Theta$, that drives the transition
$1 \leftrightarrow 4$ is switched on, see figure \ref{flippass}b. Together with the coupling $g_{24}$, this configuration can be described in terms of three polaritons, $q_0^{\dagger}$, $q_+^{\dagger}$ and $q_-^{\dagger}$, in an analogous way to
$p_0^{\dagger}$, $p_+^{\dagger}$ and $p_-^{\dagger}$, where now the roles of the atomic
levels 1 and 2 and the levels 3 and 4 are interchanged.
Hence, if one chooses $\Theta = \Omega$ the $\pi / 2$-pulse maps the $b^{\dagger}$ onto the dark state polaritons of the new configuration, $q_0^{\dagger}$, whereas
for $\Theta = - \Omega$ it maps the $c^{\dagger}$ onto $q_0^{\dagger}$. The driving laser is then adiabatically switched off, $\Theta \rightarrow 0$, and the corresponding STIRAP process maps the $q_0^{\dagger}$ completely onto atomic excitations of level 1. This process can be fast since the detuning $\Delta$ is significantly smaller than $\delta$ and hence the energies of all polariton species $q_0^{\dagger}$, $q_+^{\dagger}$ and $q_-^{\dagger}$ are well separated. Another $\pi / 2$-pulse finally maps the excitations of level 1 onto excitations of level 2,
which can be measured by state selective resonance fluorescence in the same way as for the one-component model.

The whole sequence of $\pi / 2$-pulse, STIRAP process and another $\pi / 2$-pulse can be done much faster than the timescale set by the dynamics of the Hamiltonian (\ref{bosehubbardtwocomp}) and $b^{\dagger}$ or $c^{\dagger}$ can be mapped onto atomic excitations in a time in which they are not able to move between sites. 
The procedure thus allows to measure the instantaneous local particle statistics of each species separately.
\begin{figure}
	\centering
	\includegraphics[width=1\linewidth]{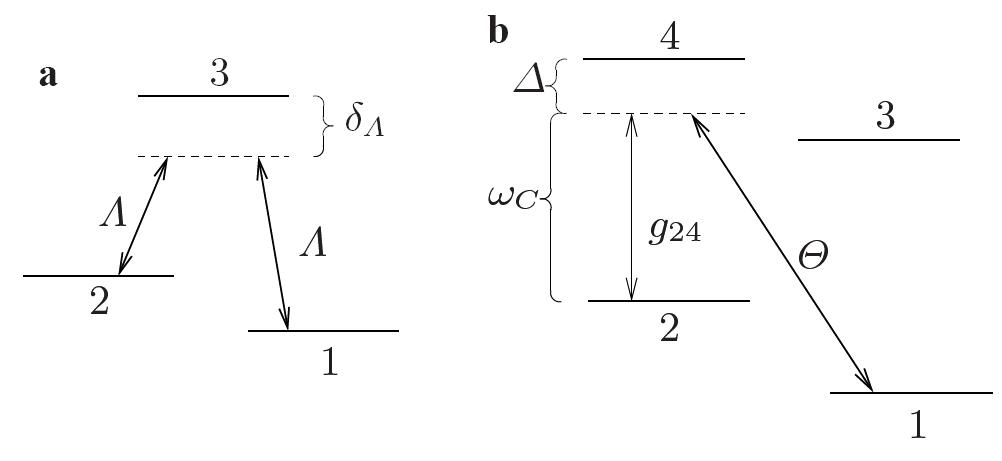}
	\caption{{\bf a}: Configuration of the $\pi / 2$-pulse. Two driving lasers in two-photon transition with identical Rabi frequencies $\Lambda$ couple to the atomic transitions $1 \leftrightarrow 3$ and $2 \leftrightarrow 3$. {\bf b}: Configuration for the STIRAP process. A driving laser couples to the $1 \leftrightarrow 4$ transition with Rabi frequency $\Theta$. The cavity mode couples to transitions $2 \leftrightarrow 4$ and $1 \leftrightarrow 3$, where the coupling to $1 \leftrightarrow 3$ is ineffective and not shown.}
	\label{flippass}
\end{figure}

\chapter{Stark-shift-induced Kerr non-Linearities} \label{kerrnonlinearity}

\section{Introduction}

Quantum properties of light, such as photon anti-bunching 
\cite{WM94} and photonic entanglement \cite{RBH01}, can 
only be produced by nonlinear interactions between photons. 
Strong nonlinearities are also important for quantum information 
processing, with applications ranging from quantum nondemolition 
measurements \cite{IHY85} to quantum memories for light \cite{FL02} 
and optical quantum computing architectures \cite{THL+95}. 
However, photon-photon interactions are usually extremely weak, 
and several orders of magnitude smaller than those needed in 
the above applications. As discussed in chapters \ref{AMO} and \ref{BH}, a possible route towards larger 
nonlinearities is the use of coherent interaction between 
light and matter in high finesse QED cavities. The goal here is to produce large nonlinearities with negligible 
losses, something which cannot be accomplished by merely tuning 
atom-light interactions close to resonance. To the best of the author's 
knowledge, the setting generating the largest nonlinear interactions 
in this context is the EIT-based proposal by Imamoglu and co-workers \cite{SI96,ISWD97,WI99,GAG99, HP07}, discussed in chapter \ref{BH}\footnote{As shown in 
Ref. \cite{HP07} and pointed out in chapter \ref{BH}, the EIT setting can produce nonlinearities 
with a strength of $\alpha \beta g$, where $g$ is the Rabi 
frequency of the atoms-cavity interaction and $\alpha, \beta$ 
are two parameters which must be much smaller than one.}. A distinctive feature of this scheme is that the ratio nonlinearity strength over losses can be kept constant when several atoms interact with the same cavity mode. We note that our figure of merit when comparing one to several atoms schemes is that the total population in the excited state is kept the same. Therefore we say that more atoms produce a larger nonlinearity in comparison when, while maintaining the same decoherence rate due to spontaneous emission, the strength of the effective non-linearity increases with the number of atoms coupled to the cavity mode.

In this chapter we propose a new method for producing Kerr nonlinearities 
in cavity QED which is (i) experimentally less demanding, requiring one atomic level and one 
coupling to the cavity mode less than in the EIT setting, (ii) 
virtually absorption free, and (iii) produces nonlinearities comparable 
or even superior to the state-of-art EIT scheme \cite{SI96,ISWD97,WI99,GAG99, HP07}. By 
applying suitable laser pulses at the beginning and end of the evolution 
of the proposed set-up, we obtain nonlinear interactions whose (iv) 
strength increases with increasing number of atoms interacting with 
the cavity mode, leading to effective nonlinear interactions at least 
two orders of magnitude larger than previously considered possible. This brings closer to 
reality a number of proposals for quantum computation and communication \cite{MNBS05, NM04, DGCZ00} based 
on photonic nonlinearities as well as the observation of strongly correlated photons in coupled cavity arrays, as discussed 
in section \ref{photonicBH}.

The organization of the chapter is the following. In section \ref{derkerrnon} we derive the 
Kerr non-linear effective Hamiltonian from the basic light-mater interaction for 
one atom inside the cavity, while in section \ref{manykerr} we consider the many-atoms
case. In section \ref{errorkerr}, in turn, we analyse the main source of errors in 
the proposed setting. In section \ref{crosskerr} we show how to generate cross Kerr 
nonlinearities. Finally in section \ref{expkerr} we discuss the feasibility of realizing 
the proposal in state-of-the-art cavity QED set-ups.

\section{Derivation of the Model} \label{derkerrnon}

In our approach, the relevant atomic level structure, depicted in fig. \ref{figkerr} (d), 
is a $\Lambda$ system with two metastable states and an excited state. The 
cavity mode couples dispersively only to the $0-2$ transition and the levels 
$0$ and $1$ are coupled via a far-detuned Raman transition. Finally, the 
detuning associated with the lasers and with the cavity mode are assumed 
to be very different from each other. 

\begin{figure} \label{figkerr}
\begin{center}
\includegraphics[scale=0.4]{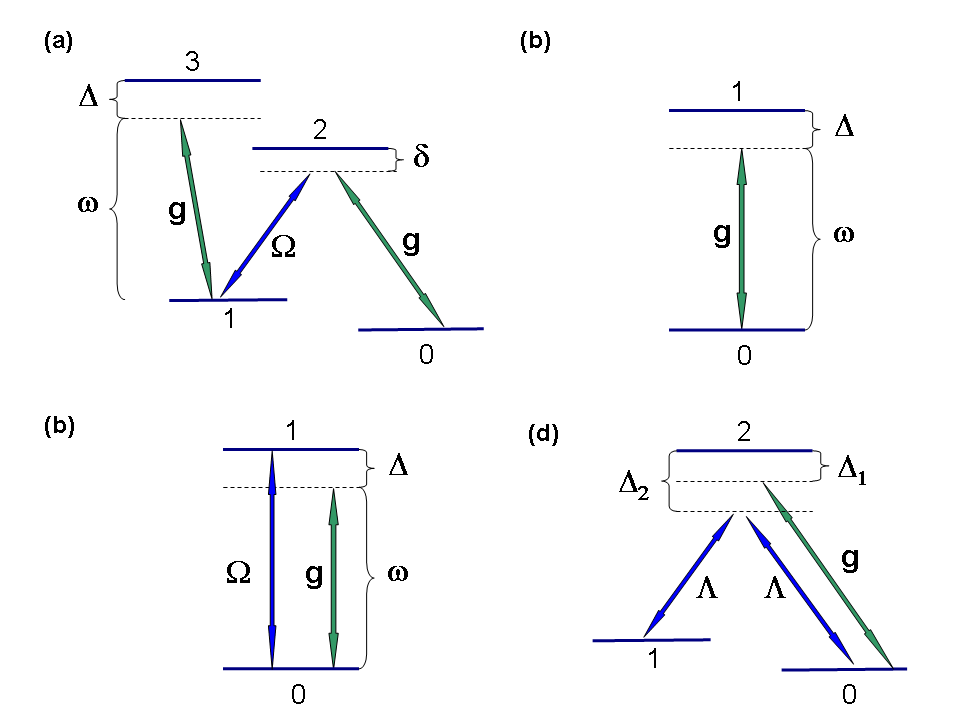}
\caption{(a) EIT scheme considered in Refs. \cite{SI96,ISWD97,WI99,GAG99, HP07}, 
(b) dispersive regime of the Jaynes-Cummings interaction, (c) large but 
noisy nonlinearity scheme, and (d) scheme proposed in this paper. Coupling 
to the cavity mode is shown in green and to classical lasers in blue.}
\end{center}
\end{figure}

In this section we show how the effective photonic non-linearity can be obtained 
from the set-up proposed above. In order to gain intuitive understanding why 
our model works, we first present two simple schemes to engineer non-linearities and 
discuss their drawbacks. We then turn to the new set-up we propose and show how 
it overcomes the limitations of the previous ones. 

\subsection{Dispersive Regime non-Linearity}

To understand intuitively how the scheme works, let us start with the 
simplest system producing a nonlinearity: $N$ two level atoms interacting 
dispersively with a cavity mode $a$ (see fig. \ref{figkerr} (b)). The system can be described by the 
Jaynes-Cummings Hamiltonian, which in the interaction picture with respect
to  $H_0 = \omega a^{\cal y}a + \omega_0 \sum_{k}\ket{1_k}\bra{1_k}$ reads
\begin{equation}
H = g(e^{-i \Delta t}a S_{10} + e^{i \Delta t}a S_{01})
\end{equation}
where $S_{10} := \sum_k \ket{1_k}\bra{0_k}$, $g$ is the Rabi frequency of 
the Jaynes-Cummings interaction, and $\Delta$ the detuning of the cavity mode 
frequency to the atomic transition frequency. If
$\sqrt{N} g/ \Delta \ll 1$, the system is in the so-called dispersive regime and 
we can adiabatically eliminate the upper level, as its population is negligible. Including fourth order terms in the 
perturbation theory and neglecting non-energy preserving terms, we can approximate the Hamiltonian above as 
\begin{equation}
H_{eff}^{(1)} = \frac{N g^2}{\Delta} a^{\cal y}a 
(S_{00} - S_{11}) + \frac{N g^4}{\Delta^3} a^{\cal y}a a^{\cal y}a 
(S_{00} - S_{11}), 
\end{equation} 
where $S_{aa} := \sum_{k=1}^N \ket{a_k}\bra{a_k}$. Hence, if all atoms are prepared in their ground states, we 
obtain a self Kerr nonlinearity $\chi=Ng^4/\Delta^3$. 
This scheme has two major drawbacks. First, the obtained 
nonlinearities are at least one order of magnitude smaller than in 
the EIT setting. Second, because of $\sqrt{N} g/ \Delta \ll 1$, 
characterizing the dispersive regime, 
the nonlinearity decreases as $g/\sqrt{N}$.

\subsection{Extra Driving Laser Field}

A simple solution to both these problems is to add a classical laser 
in resonance to the transition $\ket{0} \rightarrow \ket{1}$ (see 
fig. \ref{figkerr} (c)). Under the conditions $\sqrt{N} g/ \Delta \ll 1$ and 
$\sqrt{N} g^2 / \Delta \ll \Omega$, the dynamics of the system is 
well described by 
\begin{equation}
H_{eff}^{(2)} = \frac{N g^2}{\Delta} a^{\cal y}a S_{3} 
+ \frac{\sqrt{N} g}{\Delta}  \frac{ \sqrt{N}g^2 }{\Delta \Omega} 
g a^{\cal y}a a^{\cal y}a S_{3},
\end{equation}
where $S_{3} := \sum_{k=1}^N \ket{+_k}\bra{+_k} - \ket{-_k}\bra{-_k}$, with 
$\ket{\pm} = (\ket{0} \pm \ket{1})/\sqrt{2}$. Hence, preparing 
the atomic states in the superposition $\ket{-}$, we find a nonlinearity 
as large as in the EIT setting which does not decrease with the number 
of atoms. Unfortunately, this improvement comes at the price of a 
large increase in the losses due to spontaneous emission
which will destroy the atomic superpositions, quickly decreasing 
the effective nonlinearity.  

\subsection{Stark-shift non-Linearity}

In the proposed scheme (see fig. \ref{figkerr} (d)), we have a situation 
similar to the one discussed above. Indeed, the effective 
Hamiltonian will be identical to $H_{eff}^{(2)}$. However, now the 
states $\ket{\pm}$ are metastable, which makes the system almost 
decoherence free. The full Hamiltonian of the system, in the interaction picture with respect to $H_0 = \omega a^{\cal y}a + \omega_1 \sum_{k}\ket{1_k}\bra{1_k}$
reads:
\begin{eqnarray} \label{hfull}
H &=& g \left( e^{- i \Delta_1 t} a S_{20} + e^{ i \Delta_1 t} a^{\cal y} S_{02} \right)  + \sqrt{2} \Lambda \left( e^{- i \Delta_2 t} S_{2+} +  e^{- i \Delta_2 t} S_{+2}\right), 
\end{eqnarray}
where $S_{20} := \sum_k \ket{2_k}\bra{0_k}$ and  $S_{2+} := \sum_k \ket{2_k}\bra{+_k}$. We assume for the moment that all the atoms interact in the same manner with the cavity mode. As shown in fig. \ref{figkerr} (d), ($g$, $\Delta_1$) and $(\Lambda$, $\Delta_2$) are the Rabi frequencies and detunings of the cavity-atom and laser-atom interactions, respectively. To justify the use of the single mode 
paradigm we require that $\Delta_1, \Delta_2 \gg \kappa$, where 
$\kappa$ is the cavity decay rate. We are interested in the dispersive regime, characterized by:
\begin{equation} \label{cond33}
\frac{\sqrt{N} g}{\Delta_1} \ll 1, \hspace{0.1 cm} \frac{ \sqrt{N} \Lambda}{\Delta_2} \ll 1.
\end{equation}
Moreover, we assume that
\begin{equation}
\sqrt{N}g, \sqrt{N} \Lambda \ll |\Delta_2 - \Delta_1|,
\end{equation}
so that we can treat the processes driven by the cavity-atom and laser-atom interactions independently. Under this conditions the dynamics of the system
will 
be described by the one of a self-Kerr photonic non-linearity. In the sequel
we 
outline the main steps taken to derive the effective Hamiltonian for the
model.

Under the above conditions the excited state will hardly be populated, and we can adiabatically eliminate it, finding an effective Hamiltonian for the two metastable states and the cavity mode. In turn, these will experience a.c. Stark shifts due to the interaction with the upper level. The effective Hamiltonian, dropping out terms proportional to the identity, is given by 
\begin{equation}
H_1 = \frac{g^2}{\Delta_1}a^{\cal y}a S_{00} + \frac{\Theta}{2} (S_{10} + S_{01}),
\end{equation}
with $\Theta := 2\Lambda^2/\Delta_2$. If we now go to a second interaction picture with respect to $H_0 = \frac{g^2}{2\Delta_1}a^{\cal y}a$ we find
\begin{equation} \label{h3}
H_1^{int} = \frac{g^2}{2\Delta_1} a^{\cal y}a \left( S_{+-} + S_{-+}\right) + \frac{\Theta}{2} S_3 ,
\end{equation}
where $S_{+-} := \sum_k \ket{+_k}\bra{-_k}$, $S_{-+} = (S_{+-})^{\cal y}$, and, as before, $S_3 := \sum_{k=1}^N \ket{+_k}\bra{+_k} - \ket{-_k}\bra{-_k}$. It follows that in the $\{ \ket{\pm} \}$ basis the system can be viewed as an ensemble of two level atoms driven by a laser with a photon-number-dependent Rabi frequency. If we consider the dispersive regime of this system, i.e.
\begin{equation} \label{cond11}
\frac{\sqrt{N}g^2}{2\Delta_1 \Theta} \ll 1,
\end{equation}
the atoms prepared in the $\ket{-}$ state will experience a Stark shift proportional to $(a^{\cal y}a)^2$, which gives rise to the desired Kerr nonlinearity. The effective Hamiltonian will be given by $H_{eff} = \frac{g^4}{4\Delta_1^2\Theta} a^{\cal y}a a^{\cal y}a S_3$. Therefore, if we prepare all the atomic states in the $\ket{-}$ state, one obtains an essentially absorption free Kerr nonlinearity given by
\begin{equation} \label{h2}
H_{kerr} = \underbrace{\frac{\sqrt{N}g^2}{2\Delta_1 \Theta}}_{\ll 1} \underbrace{\frac{\sqrt{N} g}{2 \Delta_1}}_{\ll 1} g a^{\cal y}a a^{\cal y}a.
\end{equation}
In the regime we consider, characterized by Eqs. (\ref{cond33}, \ref{cond11}), the first two terms $\frac{\sqrt{N}g^2}{2\Delta_1 \Theta}$ and $\frac{\sqrt{N} g}{2 \Delta_1}$ must be much smaller than one. We can also see that the strength of the effective Kerr nonlinearity does not depend on the number of atoms, as we can tune $\Delta_1$ and $\Theta$ independently. As such, the maximum achievable nonlinearity is limited by $g$, which depends on the properties of the cavity and of the atoms employed. 

We note that under the conditions we impose, each atom interacts independently 
with the cavity mode. Therefore, if each atom has a different Rabi 
frequency $g_i$, which is the case when e.g. an atomic cloud is released 
in an optical cavity, then the resulting nonlinear coupling term will 
be identical to the one given by Eq. (\ref{h2}), but with $g^4$ replaced 
by $\sum_i g_i^4$. Moreover, for the same reason the scheme is robust 
if some of the atoms are not in the $\ket{-}$ state, which is clear 
from the form of $H_{eff}$.

\section{Many Atoms Regime} \label{manykerr}

We now proceed to show that, in fact, our set-up can be modified to give nonlinear interactions which increase with $N$. The joint atomic operators $S_{+-}, S_{-+},$ and $S_3$ satisfy $su(2)$ commutation relations: 
\begin{equation}
[S_{+-}, S_{-+}] = S_3, \hspace{0.9 cm} [S_3, S_{\pm \mp}] = \pm S_{\pm \mp}. 
\end{equation}
Defining the canonical transformation 
\begin{equation*}
U = \exp(\mu a^{\cal y }a(S_{+-} - S_{-+})),
\end{equation*}
with $\mu := g^2/(\Delta_1 \Theta) \ll 1$, we can use the Hausdorff expansion ($\exp(x A)B \exp(-x A) = B + x[A, B] + x^2[A, [A, B]]/2 + ...$) to obtain from Eq. (\ref{h3})
\begin{eqnarray} \label{h5}
H_{rot} &:=& U^{\cal y}(H_1^{int})U \approx \frac{\Theta}{2} S_3 -  \left(\frac{g^2}{2\Delta_1}\right)^2 \frac{1}{\Theta} (a^{\cal y}a)^2S_3 \nonumber \\ &+& \left(\frac{g^2}{2\Delta_1}\right)^3 \frac{1}{\Theta^2} (a^{\cal y}a)^3(S_{+-} + S_{-+}). 
\end{eqnarray} 
Suppose we had a way of generating $H_{rot}$. If we prepared all the atoms in the $\ket{-}$ state, the effective photonic Hamiltonian would be given by the second term in Eq. (\ref{h5}) as long as
\begin{equation} \label{cond22}
\left(\frac{g^2}{2\Delta_1}\right)^3 \frac{\sqrt{N}}{\Theta^3} \ll 1.
\end{equation}
The condition above is given by the ratio of the coefficient of the normalized
atomic operator $S_3/\sqrt{N}$, given by $\frac{\sqrt{N}\Theta}{2}$, and
the coefficient of the term $(a^{\cal y}a)^3(S_{+-} + S_{-+}$, equal to $\left(\frac{g^2}{2\Delta_1}\right)^2 \Theta^{-1}$. The former gives the energy spacing between the first
and second collective atomic excitation in the rotated basis, whereas the
later is responsible for Rabi oscillations between these. Condition (\ref{cond22})
then ensures that basically no transition from $\ket{-}$ to $\ket{+}$ happens. Note that it is much less stringent than the one given by Eq. (\ref{cond11}). In particular, setting $\Theta$ such that $(g^2/2\Delta_1)/\Theta = N^{-1/4}$, Eq. (\ref{cond22}) is satisfied for large $N$, and we obtain a nonlinearity of 
\begin{equation}
(\sqrt{N}g/\Delta_1) N^{1/4} g a^{\cal y}a a^{\cal y}a,
\end{equation}
while maintaining the same level of error due to spontaneous emission, dephasing, and cavity decay rate. For instance, with $N = 10^4$ and $(\sqrt{N}g/\Delta_1) = 0.1$ , we obtain a nonlinearity equal to the Rabi frequency $g$, which is at least two orders of magnitude larger than possible in the single atom case.  

\begin{figure} \label{kerrlaserfig}
\begin{center}
\includegraphics[scale=0.6]{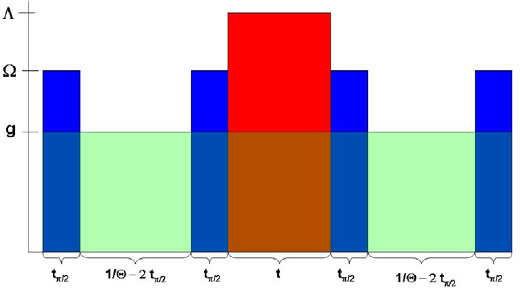}
\caption{Sequence of operations generating $V(t)$. Blue and red areas correspond to the time the lasers with Rabi frequencies $\Omega$ and $\Lambda$ are on, respectively. The green area illustrates that the Jaynes-Cummings interaction between the cavity mode and the atoms is on at all times. The scales do not reproduce reality in general. For example, $t$ will be much larger than $1/\Theta$ in most cases.}
\end{center}
\end{figure}
It is indeed possible to realize the unitary operator $V(t) = \exp(i t H_{rot})$ for every $t$. We have that $V(t) := U^{\cal y}\exp(i t H_1^{int})U$, hence it suffices to show how to create the unitary $U$. Consider the Hamiltonian given by Eq. (\ref{hfull}) when the classical lasers are switched off, i.e. $H = g(e^{-i \Delta_1 t}aS_{20} + e^{i \Delta_1 t}a^{\cal y}S_{02})$. Suppose we apply to the atoms a fast laser pulse, described by the unitary operator $M$ that generates the transformation
\begin{eqnarray}
\ket{0_k} \rightarrow (\ket{0_k} + i\ket{1_k})/\sqrt{2} \\
 \ket{1_k} \rightarrow (\ket{0_k} - i\ket{1_k})/\sqrt{2}, \nonumber
\end{eqnarray}
let Hamiltonian $H$ run for a time $t$ and apply the inverse transformation $M^{\cal y}$. Then, the total evolution operator will be given by
\begin{eqnarray} \label{eeeer}
&&M^{\cal y} \exp(i g \int_0^t (e^{-i \Delta_1 t'}aS_{20} + e^{i \Delta_1 t'}a^{\cal y}S_{02})dt') M  \nonumber \\ &=& \exp\left(i g \int_0^t \sqrt{2}g(e^{-i \Delta_1 t'}a(S_{20} + i S_{21}) + h.c.)dt'\right).
\end{eqnarray}
As $\sqrt{N}g/\Delta_1 \ll 1$, we can once more adiabatically eliminate level 2 and approximate the unitary evolution above by $\exp( t (g^2/\Delta_1)a^{\cal y}a( S_+ - S_-))$, which is $U$ when $t = 1/\Theta$. The sequence of operations executing $V(t)$ is shown in fig. \ref{kerrlaserfig}. 

\section{Error Analysis} \label{errorkerr}

\begin{figure} \label{errorkerr1}
\begin{center}
\includegraphics[scale=0.45]{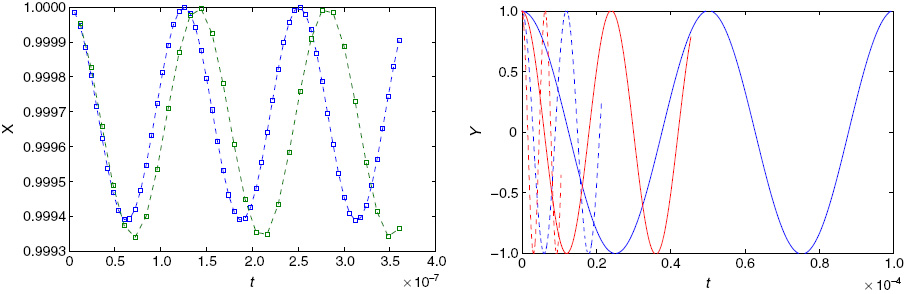}
\caption{(a) $X = |\bra{n} \otimes \bra{-}V(t)\ket{-}\otimes \ket{n}|$ versus $t$ for $(N, n) = (1, 2)$ (blue) and $(2, 2)$ (green), with $g = 10^8 s^(-1)$, $5 \Theta = gN^{1/3}$, $\Omega = 100g$, and $\Delta_1 = 10 \sqrt{N} g$; (b) $Y = Re(\bra{n} \otimes \bra{-}V(t)\ket{-}\otimes \ket{n})$ versus $t$ for $(N, n) = (1, 1)$ (blue solid line), $(1, 2)$ (blue dashed line), $(2, 1)$ (red solid line), and $(2, 2)$ (red dashed line), with $g = 10^8 s^{-1}$, $\Delta_1 = 10g$, $\Theta = g$, and $\Omega = 100g$.}
\end{center}
\end{figure}

In order to check the accuracy of our approach, we simulated it numerically for one and two atoms. We assume that the unitary evolution given by Eq. (\ref{eeeer}) if created by a $\pi/2$ pulse of the Hamiltonian $H_{laser} = \Omega (\ket{0}\bra{1} + \ket{1}\bra{0})$. We choose $\Omega = 100g$, which gives a $t_{\pi/2} = \pi/(2\Omega)$. The other parameters are $g = 10^8 s^{-1}$, $5 \Theta = gN^{1/3}$, and $\Delta_1 = 10 \sqrt{N} g$. With these values, the quality of the approximation, which is determined by the value of the L.H.S. of Eq. (\ref{cond22})  and  $\sqrt{N}g/\Delta_1$, should be constant. A good figure of merit in this respect is given by $X = |\bra{n} \otimes \bra{-}V(t)\ket{-}\otimes \ket{n}|$, where $V(t)$ is the total unitary operator formed by concatenating the steps explained above (see fig. \ref{kerrlaserfig} for a graphic description of the procedure), $\ket{n}$ is the $n$-th Fock state of the photons, and $\ket{-} := \ket{-_1}\otimes ... \otimes \ket{-_N}$. In the case of an ideal Kerr non-linearity $X$ should be equal to
one at all instants of time. As can be seen in fig. \ref{errorkerr1} (a) we get a good agreement
with the ideal case both for 1 and 2 atoms. In particular the quality of
the approximation does not deteriorate with the number of atoms, as expected
from our analytical calculations. In fig. \ref{errorkerr1} (b) in turn we plotted $Re(\bra{n} \otimes \bra{-}V(t)\ket{-}\otimes \ket{n})$ for $(N, n) = (1, 1), (1, 2), (2, 1), (2, 2)$, with the parameters $g = 10^8 s^{-1}$, $\Delta_1 = 10g$, $\Theta = g$, and $\Omega = 100g$. We found a very good agreement with the dynamics of a pure photonic nonlinearity of strength $\kappa = N g^4 / 4\Delta_1^2 \Theta$, for which $Re(\bra{n} \otimes \bra{-}V(t)\ket{-}\otimes \ket{n}) = \cos(\kappa n^2 t)$.   

\begin{figure} \label{noiekerr}
\begin{center}
\includegraphics[scale=0.4]{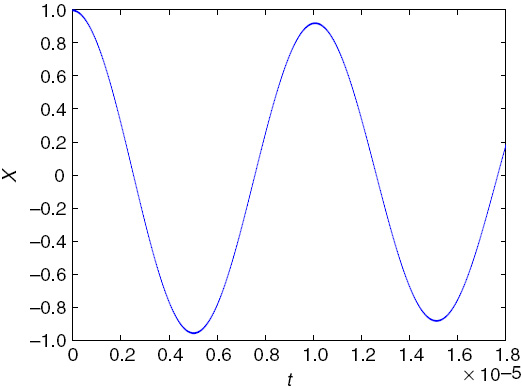}
\caption{$X = Re(\bra{3} \otimes \bra{-}V'(t)\ket{-}\otimes \ket{3})$ versus $t$ for $g = 10^9 s^{-1}$, $\Lambda = 2g$ $\Delta_1 = 10\Lambda$, $\Delta_2 = 5 \Lambda$, $\kappa = 10^6 s^{-1}$, $\gamma = 10^{6} s^{-1}$, and $\theta = 10^{3} s^{-1}$. The damping in the oscillations show the effect of losses
due to spontaneous emission from the excited level, cavity decay rate, and
dephasing from all three levels. With the realistic parameters considered
one achieves a good approximation to the ideal case.}
\end{center}
\end{figure}

The main source of decoherence in the system under analysis are spontaneous emission from the upper level and cavity decay due to the finite quality factor of the cavity used. Dephasing can usually be disregarded as its effects are much smaller. The upper level is hardly populated, so spontaneous emission from it should not play an important role. Indeed, the effective spontaneous emission rate can be estimated by the product of the population of the upper level $\ket{2}$, given by $\frac{g^2}{\Delta_1^2}$, with the spontaneous emission rate. As we work in the regime $g/\Delta_1 \ll 1$, the effective spontaneous emission rate can be shown to be several orders of magnitude smaller than the effective non-linearity (see section 6 for further discussion). The main source of decoherence is therefore cavity decay. First, as mentioned before, in order to justify the use of one single mode describing the cavity field, one need to ensure that $\Delta_1, \Delta_2 \gg \kappa$, where $\kappa$ is the cavity decay rate. The derivation can then be carried through independently of the exact value of $\kappa$. Of course in the end the dynamics will have a contribution from the effective non-linear term and from the leaking out of photons from the cavity. In order to have a good approximation for a pure photonic non-linearity the strength of the effective interaction should therefore be much larger than $\kappa$. As shown in section \ref{expkerr} this condition can indeed be achieved in several cavity QED set-ups. 

We have numerically checked the effect of decoherence by simulating the full evolution in the case of one atom, considering the effect of dephasing, spontaneous emission from the upper level, and cavity decay. The parameters used are $g = 10^9 s^{-1}$, $\Lambda = 2g$ $\Delta_1 = 10\Lambda$, and $\Delta_2 = 5 \Lambda$. Moreover, we consider $\kappa = 10^6 s^{-1}$, $\gamma = 10^{6} s^{-1}$, and $\theta = 10^{3} s^{-1}$ for the cavity decay rate, spontaneous emission rate, and dephasing rate, respectively. As discussed in section \ref{expkerr}, this parameters are within reach in several cavity QEP set-ups. The dynamics of $Re(\bra{3} \otimes \bra{-}V'(t)\ket{-}\otimes \ket{3})$ is plotted in fig. \ref{noiekerr} where now $V'(t)$ is the non-unitary operator describing the full Hamiltonian and the decoherence sources It should be contrasted with the
evolution due to an ideal Kerr non-linearity, given by $\cos(\kappa 9 t)$. 

\section{Cross-Kerr non-Linearities} \label{crosskerr}

\begin{figure} \label{crosskerrfig}
\begin{center}
\includegraphics[scale=0.5]{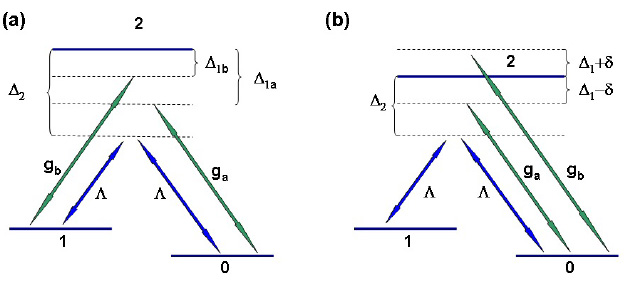}
\caption{Level structures for generating cross-Kerr nonlinearities. Coupling
to the cavity mode is shown in green and to classical lasers in blue.}
\end{center}
\end{figure}

Up to now we have discussed the case of self-Kerr nonlinearities. In a cross-Kerr nonlinearity, one optical field induces a Stark shift in another field proportional to the intensity of the former. In terms of two optical modes $a$ and $b$, the Hamiltonian is given by $H_{cross} = \nu a^{\cal y}a b^{\cal y}b$. There are two possible generalizations of the scheme introduced here to produce large cross-Kerr nonlinearities. 

A shown in fig. \ref{crosskerrfig} (a), we may chose setting as in Eq. (\ref{hfull}), but now, we have a second cavity mode $b$ coupled to the transition $\ket{1} \rightarrow \ket{2}$ with Rabi frequency $g_b$ and detuning $\Delta_{1_b}$. As shown in Ref. \cite{DK03}, this could be achieved using two degenerate cavity modes with orthogonal polarization. Carrying over the same analysis we did for the one mode case, we can find in this case 
\begin{equation} 
H_{eff} = \frac{1}{\Theta}\left( \frac{g_a^2}{2 \Delta_{1_a}}a^{\cal y}a - \frac{g_b^2}{2 \Delta_{1_b}}b^{\cal y}b    \right)^2 S_3.
\end{equation}
Hence, in addition to a self-Kerr nonlinearity for each mode, we find a cross-Kerr nonlinearity between them. 

In the second case, we consider the level structure shown in fig. \ref{crosskerrfig} (b). We have two modes $a$ and $b$ coupled to the transition $\ket{0} \rightarrow \ket{2}$ with detunings $\Delta_1 \pm \delta$, respectively. As discussed in Ref. \cite{ADW+06}, this set-up could be realized in a toroidal microcavity, where $a$ and $b$ are the normal modes of the clock and counter-clockwise propagating modes, and $\delta$ is the rate of tunneling between those two. In this case, the effective Hamiltonian can be found to be
\begin{equation} 
H_{eff} = \frac{1}{\Theta}\left( \frac{g_a^2}{2 (\Delta_1 - \delta)}a^{\cal y}a + \frac{g_b^2}{2 (\Delta_1 + \delta)}b^{\cal y}b \right)^2 S_3
\end{equation}

\section{Experimental Realization} \label{expkerr}

Our scheme can be applied to a variety of cavity QED settings. This includes, as mentioned in section \ref{cQED}, single or an ensemble of atoms trapped in fiber-based cavities \cite{CSD+07} 
or optical microcavities \cite{BBM05, BDR+07, ADW+06, TGD+05, HWS07}, where for the 
single atom case ratios of $\chi/\kappa \approx 625$ and $\chi/\gamma \approx 245$ for the 
nonlinearity strength ($\chi$) over cavity decay rate $\kappa$ and effective spontaneous 
emission $\gamma$, respectively, have been predicted to be feasible \cite{SKV+05}, and quantum dots embedded in photonic band gap structures, where ratios of $\chi/\kappa \approx 40$ and $\chi/\gamma \approx 12$ could be achieved \cite{BHA+05, AAS+03}. Cooper-pair boxes coupled to a 
r.f. transmission line \cite{WSB+04} could also be a suitable set-up in 
the longer run, when the problem related to the absence of two metastable 
levels is overcome. Even considering the decay between levels $\ket{1}$ and $\ket{0}$ 
in such systems, promising ratios of $\chi/\kappa \approx 15$ and $\chi/\gamma \approx 0.5$ could be realized using the values of the recent experiment \cite{SHS+07}. 

In implementations based on atoms one can safely neglect the spontaneous emission from the two metastable levels, 
as their lifetime is several orders of magnitude larger than the time the 
photon stays in the cavity. In settings based on quantum dots, one can 
also find metastable configurations, whose lifetimes, although much 
smaller than in the atomic case, would still be large in comparison to the 
timescales involved in the experiment. In such cases, the effective spontaneous emission rate is given by the product of the population of the upper level $\ket{2}$, given by $\frac{g^2}{\Delta_1^2}$, with the spontaneous emission rate. In the absence of two metastable levels, as it is the present case of rf cavities interacting with Cooper pair-boxes, $\gamma$ is just the spontaneous emission rate from $\ket{1}$ to the ground state $\ket{0}$.

As an example of the applicability of our scheme to engineer nonlinearities which 
grow with the number of atoms, we consider the recent experiment \cite{BDR+07} 
(see also \cite{CSD+07}) in which a Bose-Einstein condensate of $4 \times 10^5$ $^{87}$Rb atoms was strongly 
coupled to a single cavity mode of an ultra-hight finesse optical cavity, 
with cavity decay rate $\kappa$, spontaneous emission rate $\gamma'$ and Rabi frequency $g$ of $8$ MHz, $18$ MHz, and $70$ MHz, 
respectively. We could apply our scheme to this set-up using for instance two 
Zeeman sublevels as the two metastable levels. The estimated nonlinearity 
$\chi$ could then be as large as $g$ itself, which would allow the realization 
of ratios $\chi/\kappa \approx 11$ and $\chi/\gamma \approx 39$ for non-linearity strength ($\chi$) over 
cavity decay rate $\kappa$ and spontaneous emission rate $\gamma$.

\chapter{Spin Hamiltonians in Coupled Cavity Arrays} \label{heisenberg}

\section{Introduction}

Interacting spin systems have a key role in quantum statistical mechanics, condensed-matter physics, and more recently in quantum information science. On the one hand they have a simple microscopic description, which allows the use of techniques from statistical mechanics in their analysis. On the other hand, they have a very rich structure and can hence be employed to model and learn about phenomena of real condensed-matter systems. In quantum information science, spin chains can be harnessed to propagate and manipulate quantum information, as well as to implement quantum computation.

A particular interesting spin Hamiltonian is the Heisenberg anisotropic or XYZ model, given by
\begin{equation} \label{XYZ}
H = \sum_{<j, j'>} J_x \sigma_j^x \sigma_{j'}^x + J_y \sigma_j^y \sigma_{j'}^y + J_z \sigma_j^z \sigma_{j'}^z + \sum_j B \sigma_j^z,
\end{equation}
where the first sum runs over nearest neighbors, $\sigma_j^{x/y/z}$ are the Pauli matrices at site $j$, $J_{x/y/z}$ the coupling coefficients and $B$ the magnetic field strength. This model contains a rich phase diagram, can be used to study quantum magnetism \cite{Aue98} and is a limiting case of the fermionic Hubbard model, which is believed to contains the main features of high $T_c$ superconductors \cite{And87, MU03}. From a quantum information perspective, such a model can be used to create cluster states, highly entangled multipartite quantum states which are universal for quantum computation by single-site measurements \cite{RB01}. 

In the previous chapters we have seen that under appropriate conditions interesting bosonic models can be created and probed in coupled cavity arrays. There the atoms were used for detection and most importantly to boost the interaction of photons in the same cavity, either actively, as in the polaritonic case, or in an undirect manner, as in the generation of large Kerr nonlinearities. In this chapter we focus on a complementary regime, in which the tunneling of photons is used to mediate interactions for atoms in neighboring cavities. We study the realization of effective spin lattices, more concretely of Hamiltonian \ref{XYZ}, with individual atoms in micro-cavities that are coupled to each other via the exchange of virtual photons.

This chapter is organized as follows. In section \ref{derhei} we derive the effective Heisenberg anisotropic Hamiltonian in a coupled cavity array, first focusing on XX and YY interactions in subsection \ref{XXYY} and then on the ZZ interaction in subsection \ref{ZZhei}. In section \ref{exphei} we analyse the main losses mechanisms and find the conditions for a successful experimental implementation of the proposal. Finally, in section \ref{clusterstate} we outline a possible application of generating Hamiltonian \ref{XYZ} in coupled cavity arrays: the creation of cluster states for one-way quantum computing \cite{RB01}.  

\section{Derivation of the Model} \label{derhei}

On general lines, the set-up we employ works as follows. The two spin polarizations $\ket{\uparrow}$ and $\ket{\downarrow}$ of the Hamiltonian will be represented by two
longlived atomic levels of the $\Lambda$ level-structure outlined in fig. \ref{xy_levels} and \ref{zz_levels}).
Together with external lasers, the cavity mode that couples to the atom inside each cavity can induce Raman transitions between these two 
longlived levels. Due to a large detuning between laser and cavity mode, these transitions can only create virtual photons in 
the cavity mode which mediate an interaction with another atom in a neighboring cavity. With appropriately chosen 
detunings, both the excited atomic levels and photon states have vanishing occupation and can be eliminated from the 
description. As a result, the dynamics is confined to only two states per atom, the long-lived levels, and can be 
described by a spin-1/2 Hamiltonian. The small occupation of photon states and excited atomic levels also strongly suppresses spontaneous emission and
cavity decay. As discussed later, this proposal can be realized as soon as the strong coupling regime of cavity QED is achieved. 

We note that similar ideas have been recently proposed to realize spin systems of larger spins \cite{KA08,CAB08}. 

In the next subsections we show how to engineer effective $\sigma^x \sigma^x$, 
$\sigma^y \sigma^y$ and $\sigma^z \sigma^z$ interactions 
as well as the effective magnetic field $B \sigma_z$ and
then explain how to generate the full anisotropic Heisenberg 
model. 

\subsection{XX and YY Interactions} \label{XXYY}

\begin{figure}
\centering
\includegraphics[width=.7\linewidth]{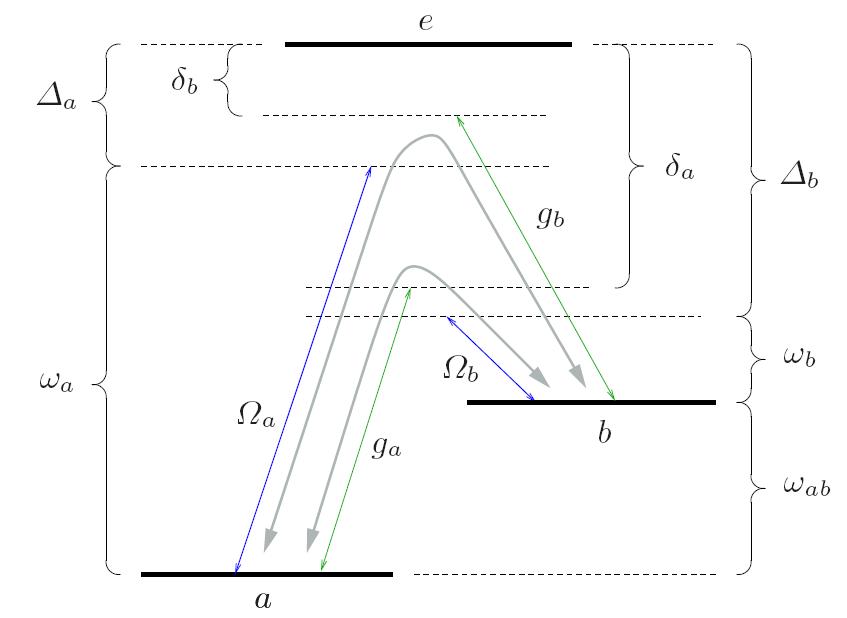}
\caption{\label{xy_levels} Level structure, driving lasers and relevant couplings to the cavity mode to generate effective $\sigma^x \sigma^x$- and $\sigma^y \sigma^y$-couplings for one atom. The cavity mode couples with strengths $g_a$ and $g_b$ to transitions $\ket{a} \leftrightarrow \ket{e}$ and $\ket{b} \leftrightarrow \ket{e}$ respectively. One laser with frequency $\omega_a$ couples to transition $\ket{a} \leftrightarrow \ket{e}$ with Rabi frequency $\Omega_a$ and another laser with frequency $\omega_b$ to $\ket{b} \leftrightarrow \ket{e}$ with $\Omega_b$. The dominant 2-photon processes are indicated in faint gray arrows.}
\end{figure}
We consider an array of coupled cavities with one 3-level atom in each cavity (see figure \ref{xy_levels}).
Two long lived levels, $\ket{a}$ and $\ket{b}$, represent the 
two spin states. The cavity mode couples to the transitions 
$\ket{a} \leftrightarrow \ket{e}$ and $\ket{b} \leftrightarrow \ket{e}$, 
where $\ket{e}$ is the excited state of the atom. Furthermore, 
two driving lasers couple to the transitions 
$\ket{a} \leftrightarrow \ket{e}$ respectively $\ket{b} \leftrightarrow \ket{e}$.
For the sake of simplicity we consider here a one-dimensional array.
The generalization to higher dimensions is straight forward. 
The Hamiltonian of the atoms reads
\begin{equation}
H_A = \sum_{j=1}^N \omega_e \ket{e_j} \bra{e_j} + \omega_{ab} \ket{b_j} \bra{b_j}, 
\end{equation}
where the index $j$ counts the cavities, $\omega_e$ is the 
energy of the excited level and $\omega_{ab}$ the energy of 
level $\ket{b}$. The energy of level $\ket{a}$ is set to zero. 
The Hamiltonian that describes the photons in the cavity modes is 
given by Eq. (\ref{arrayham2}). For simplicity of notation we set
$J_C = 2 \omega_C \alpha$.  For convenience we assume periodic boundary 
conditions, where $H_C$ can be diagonalized via the Fourier transform
\begin{equation*}
a_k = \frac{1}{\sqrt{N}} \sum_{j=1}^N e^{i k j} a_j, \hspace{0.3 cm} k = \frac{2 \pi l}{N}, \hspace{0.1 cm} - \frac{N}{2} \le l \le \frac{N}{2}
\end{equation*}
for $N$ odd to give $H_C = \sum_{k} \omega_k a_k^{\dagger} a_k$ with $\omega_k = \omega_C + 2 J_C \cos(k)$. Finally the interaction between the atoms and the photons 
as well as the driving by the lasers are described by
\begin{equation*}
H_{AC} = \sum_{j=1}^N \left[ \left(\frac{\Omega_a}{2} e^{-i \omega_a t} + g_a a_j \right) \ket{e_j} \bra{a_j} + \left(\frac{\Omega_b}{2} e^{-i \omega_b t} + g_b a_j \right) \ket{e_j} \bra{b_j} + \text{h.c.} \right].
\end{equation*}
Here $g_a$ and $g_b$ are the couplings of the respective transitions to the cavity mode, $\Omega_a$ is the Rabi frequency of one laser with frequency $\omega_a$ and  $\Omega_b$ the Rabi frequency of a second laser with frequency $\omega_b$. The complete Hamiltonian is then given by $H = H_A + H_C + H_{AC}$.

We now switch to an interaction picture with respect to $H_0 = H_A + H_C - \delta_1 \, \sum_{j=1}^N \ket{b_j} \bra{b_j}$, where $\delta_1 = \omega_{ab} - (\omega_a - \omega_b)/2$, and adiabatically eliminate the excited atom levels $\ket{e_j}$ and the photons \cite{BPM07}. We consider terms up to 2nd order in the effective Hamiltonian and drop fast oscillating terms. For this approach the detunings $\Delta_a \equiv \omega_e - \omega_a$, $\Delta_b \equiv \omega_e - \omega_b - (\omega_{ab} - \delta_1)$, $\delta_a^k \equiv \omega_e - \omega_k$ and $\delta_b^k \equiv \omega_e - \omega_k - (\omega_{ab} - \delta_1)$ have to be large compared to the couplings $\Omega_a, \Omega_b, g_a$ and $g_b$. Furthermore, the parameters must be such that the dominant Raman transitions between levels $a$ and $b$ are those that involve one laser photon and one cavity photon each (c.f. figure \ref{xy_levels}). To avoid excitations of real photons via these transitions, we furthermore require 
$\left| \Delta_a - \delta_b^k \right|, \left| \Delta_b - \delta_a^k \right| \gg \left| \frac{\Omega_a g_b}{2 \Delta_a} \right|, \left| \frac{\Omega_b g_a}{2 \Delta_b} \right|$.
 
Hence whenever the atom emits or absorbs a virtual photon 
into or from the cavity mode, it does a transition from 
level $\ket{a}$ to $\ket{b}$ or vice versa. If one atom emits 
a virtual photon in such a process that is absorbed by a 
neighboring atom, which then also does a transition between 
$\ket{a}$ to $\ket{b}$, an effective spin-spin interaction 
has happened. Dropping irrelevant constants, the resulting 
effective Hamiltonian reads
\begin{equation}
H_{\text{xy}} = \sum_{j=1}^N B \sigma_j^z
+ ( J_1 \sigma_j^+ \sigma_{j+1}^- + J_2 \sigma_j^- \sigma_{j+1}^- + \text{h.c.} ),
\end{equation}
where $\sigma_j^z = \ket{b_j} \bra{b_j} - \ket{a_j}\bra{a_j}$, $\sigma_j^+ = \ket{b_j}\bra{a_j}$ and $B$, 
$J_1$ and $J_2$ are simple functions of the various Rabi frequencies and detunings\footnote{To second order they are 
given by 
\begin{equation*}
B = \frac{\delta_1}{2} - \frac{1}{2} \left[\frac{|\Omega_b|^2}{4 \Delta_b^2} 
\left( \Delta_b - \frac{|\Omega_b|^2}{4 \Delta_b} - \frac{|\Omega_a|^2}{4 (\Delta_a - \Delta_b)} - \gamma_b g_b^2 - \gamma_1 g_a^2 + \gamma_1^2 \frac{g_a^4}{\Delta_b} \right) -
(a \leftrightarrow b) \right], 
\end{equation*}
\begin{equation*}
J_1 = \frac{\gamma_2}{4} \left( \frac{|\Omega_a|^2 g_b^2}{\Delta_a^2} + \frac{|\Omega_b|^2 g_a^2}{\Delta_b^2} \right)
\end{equation*}
and
\begin{equation*}
J_2 = \frac{\gamma_2}{2} \frac{\Omega_a^{\star} \Omega_b g_a g_b}{\Delta_a \Delta_b},
\end{equation*}
where
$\gamma_{a,b} = \frac{1}{N} \sum_{k} \frac{1}{\omega_{a,b} -\omega_k}$,
$\gamma_1 = \frac{1}{N} \sum_{k} \frac{1}{(\omega_a + \omega_b)/2 - \omega_k}$ and
$\gamma_2 = \frac{1}{N} \sum_{k} \frac{\exp (i k)}{(\omega_a + \omega_b)/2 - \omega_k}$.
}. If $J_2^{\star} = J_2$, this Hamiltonian reduces to the XY model,
\begin{equation} \label{HXYb}
H_{\text{xy}} = \sum_{j=1}^N B \sigma_j^z + J_x \sigma_j^x 
\sigma_{j+1}^x + J_y \sigma_j^y \sigma_{j+1}^y \, ,
\end{equation}
with $J_x = (J_1 + J_2)/2$ and $J_y = (J_1 - J_2)/2$.

For $\Omega_a = \pm (\Delta_a g_a / \Delta_b g_b) \Omega_b$ 
with $\Omega_a$ and $\Omega_b$ real, the interaction is either 
purely $\sigma^x \sigma^x$ ($+$) or purely $\sigma^y \sigma^y$ 
($-$) and the Hamiltonian (\ref{HXYb}) becomes the Ising model 
in a transverse field, whereas the isotropic XY model 
($J_x = J_y$, i.e. $J_2 = 0$) is obtained for either 
$\Omega_a \rightarrow 0$ or $\Omega_b \rightarrow 0$. The 
effective magnetic field $B$ in turn can, independently of 
$J_x$ and $J_y$, be tuned to assume any value between 
$|B| \gg |J_x|, |J_y|$ and $|B| \ll |J_x|, |J_y|$ by varying 
$\delta_1$. Thus we will be able to drive the system through
a quantum phase transition. Now we proceed to effective ZZ interactions.

\subsection{ZZ Interactions} \label{ZZhei}

To obtain an effective $\sigma^z \sigma^z$ interaction, we again use the same atomic level configuration but now only one laser with frequency $\omega$ mediates atom-atom coupling via virtual photons. A second laser with frequency $\nu$ is used to tune the effective magnetic field via a Stark shift. The atoms together with their couplings to cavity mode and lasers are shown in figure \ref{zz_levels}.
\begin{figure}
\centering
\includegraphics[width=.7\linewidth]{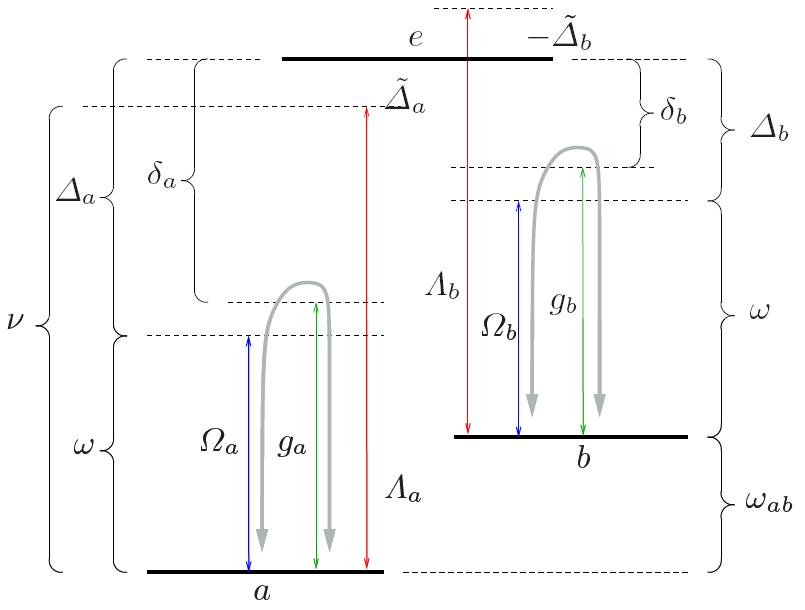}
\caption{\label{zz_levels} Level structure, driving lasers and relevant couplings to the cavity mode to generate effective $\sigma^z \sigma^z$-couplings for one atom. The cavity mode couples with strengths $g_a$
and $g_b$ to transitions $\ket{a} \leftrightarrow \ket{e}$ and $\ket{b} \leftrightarrow \ket{e}$ respectively. Two lasers with frequencies $\omega$ and $\nu$ couple with Rabi frequencies $\Omega_a$ respectively $\Lambda_a$ to transition $\ket{a} \leftrightarrow \ket{e}$ and $\Omega_b$ respectively $\Lambda_b$ to $\ket{b} \leftrightarrow \ket{e}$. The dominant 2-photon processes are indicated in faint gray arrows.}
\end{figure}
Again, we consider the one-dimensional case as an example. The generalization to higher dimensions is straightforward.
The Hamiltonians $H_A$ of the atoms and $H_C$ of the cavity modes thus have the same form as above, whereas $H_{AC}$ now reads:
\begin{align}
H_{AC} = & \sum_{j=1}^N \left[ \left(\frac{\Omega_a}{2} \exp^{-i \omega t} +
\frac{\Lambda_a}{2} \exp^{-i \nu t} + g_a a_j \right) \ket{e_j}\bra{a_j} \right. \\
+ & \left. \left(\frac{\Omega_b}{2} \exp^{-i \omega t} +
\frac{\Lambda_b}{2} \exp^{-i \nu t} + g_b a_j \right) \ket{e_j}\bra{b_j} + \text{h.c.} \right] \nonumber .
\end{align}
Here, $\Omega_a$ and $\Omega_b$ are the Rabi frequencies of the driving laser with frequency $\omega$ on transitions $\ket{a} \rightarrow \ket{e}$ and $\ket{b} \rightarrow \ket{e}$, whereas $\Lambda_a$ and $\Lambda_b$ are the Rabi frequencies of the driving laser with frequency $\nu$ on transitions $\ket{a} \rightarrow \ket{e}$ and $\ket{b} \rightarrow \ket{e}$.

We switch to an interaction picture with respect to $H_0 = H_A + H_C$ and adiabatically eliminate the excited atom levels $\ket{e_j}$ and the photons \cite{BPM07}. Again, the detunings $\Delta_a \equiv \omega_e - \omega$, $\Delta_b \equiv \omega_e - \omega - \omega_{ab}$, $\tilde{\Delta}_a \equiv \omega_e - \nu$, $\tilde{\Delta}_b \equiv \omega_e - \nu - \omega_{ab}$, $\delta_a^k \equiv \omega_e - \omega_k$ and $\delta_b^k \equiv \omega_e - \omega_k - \omega_{ab}$ have to be large compared to the couplings $\Omega_a, \Omega_b, \Lambda_a, \Lambda_b, g_a$ and $g_b$, whereas now Raman transitions between levels $a$ and $b$ should be suppressed. Hence parameters must be such that the dominant 2-photon processes are those that involve one laser photon and one cavity photon each but where the atom does no transition between levels $a$ and $b$ (c.f. figure \ref{zz_levels}). To avoid excitations of real photons in these processes, we thus require 
$\left| \Delta_a - \delta_a^k \right|, \left| \Delta_b - \delta_b^k \right| \gg \left| \frac{\Omega_a g_a}{2 \Delta_a} \right|, \left| \frac{\Omega_b g_b}{2 \Delta_b} \right|$.

Whenever two atoms exchange a virtual photon in this scheme, both experience a Stark shift that depends on the state of the partner atom. This conditional Stark shift plays the role of an effective $\sigma^z \sigma^z$-interaction.
Dropping irrelevant constants, the effective Hamiltonian reads:
\begin{equation} \label{HZZ}
H_{\text{zz}} = \sum_{j=1}^N \left( \tilde{B} \sigma_j^z + J_z \sigma_j^z \sigma_{j+1}^z \right) \, .
\end{equation}
$\tilde{B}$ and $J_z$ are determined by the detunings and Rabi frequencies employed\footnote{To second order they are give by
\begin{eqnarray*}
\tilde{B} = - \frac{1}{2} \left[ \right. && \frac{|\Lambda_b|^2}{16 \tilde{\Delta}_b^2} \left(4 \tilde{\Delta}_b
- \frac{|\Lambda_a|^2}{\tilde{\Delta}_a - \tilde{\Delta}_b} - 
\frac{|\Lambda_b|^2}{\tilde{\Delta}_b} - \sum_{j=a,b} \left(\frac{|\Omega_j|^2}{\Delta_j - 
\tilde{\Delta}_b} + 4 \tilde{\gamma}_{jb} g_j^2 \right) \right) \\ &+& \left. \frac{|\Omega_b|^2}{16 \Delta_b^2}
\left(4 \Delta_b - \frac{|\Omega_a|^2}{\Delta_a - \Delta_b} - \frac{|\Omega_b|^2}{\Delta_b}
- \sum_{j=a,b} \left(\frac{|\Lambda_j|^2}{\tilde{\Delta}_j - \Delta_b} + 4 \gamma_{jb} g_j^2 \right) + 4 \gamma_{bb}^2 \frac{g_b^4}{\Delta_b} \right)
- \left( a \leftrightarrow b \right) \right] 
\end{eqnarray*}
and
\begin{equation*}
J_z = \gamma_2 \left| \frac{\Omega_b^{\star} g_b}{4 \Delta_b} - \frac{\Omega_a^{\star} g_a}{4 \Delta_a} \right|^2
\end{equation*}
with
$\gamma_1 = \frac{1}{N} \sum_{k} \frac{1}{\omega - \omega_k}$,
$\gamma_2 = \frac{1}{N} \sum_{k} \frac{\exp (i k)}{\omega - \omega_k}$,
$\gamma_{aa} = \gamma_{bb} = \frac{1}{N} \sum_{k} \frac{1}{\omega - \omega_k}$,
$\left. \begin{array}{c}
\gamma_{ab}\\
\gamma_{ba}
\end{array} \right\}= \frac{1}{N} \sum_{k} \frac{1}{\omega \pm \omega_{ab} - \omega_k}$,
$\left. \begin{array}{c}
\tilde{\gamma}_{ab}\\
\tilde{\gamma}_{ba}
\end{array} \right\}= \frac{1}{N} \sum_{k} \frac{1}{\nu \pm \omega_{ab} - \omega_k}$,
$\tilde{\gamma}_{aa} = \tilde{\gamma}_{bb} = \frac{1}{N} \sum_{k} 
\frac{1}{\nu - \omega_k}$.
}. Here again, the interaction $J_z$ and the field $\tilde{B}$ can be tuned independently, either by varying
$\Omega_a$ and $\Omega_b$ for $J_z$ or by varying $\Lambda_a$ and $\Lambda_b$ for $\tilde{B}$. In particular,
$|\Lambda_a|^2$ and $|\Lambda_b|^2$ can for all values of $\Omega_a$ and $\Omega_b$ be chosen such that either $J_z \ll \tilde{B}$ or $J_z \gg \tilde{B}$. 

\subsection{The Complete Effective Model} 

Making use of the Suzuki-Trotter formula, the two 
Hamiltonians (\ref{HXYb}) and (\ref{HZZ}) can now be combined 
to one effective Hamiltonian. To this end, the lasers that 
generate the Hamiltonian (\ref{HXYb}) are turned on for a 
short time interval $dt$ ($||H_{\text{xy}}|| \cdot dt \ll 1$) 
followed by another time interval $dt$ ($||H_{\text{zz}}|| \cdot dt \ll 1$) with the lasers that generate the Hamiltonian (\ref{HZZ}) 
turned on. This sequence is repeated until the total time 
range to be simulated is covered. 

The effective Hamiltonian 
simulated by this procedure is $H_{\text{spin}} = H_{\text{xy}} + H_{\text{zz}}$, which is precisely the Heisenberg anisotropic 
model of Eq. (\ref{XYZ}) with $B_{\text{tot}} = B + \tilde{B}$. The time interval 
$dt$ should thereby be chosen such that 
$\Omega^{-1}, g^{-1} \ll  dt_1 , dt_2 \ll J_x^{-1} , J_y^{-1} , J_z^{-1} , B^{-1}$ and $\tilde{B}^{-1}$, so that the Trotter sequence 
concatenates the effective Hamiltonians $H_{XY}$ and $H_{ZZ}$. 
The procedure can be generalized to higher order Trotter formulae
or by turning on the sets of lasers for time intervals of different 
length.

The validity of all above approximations is shown in figure \ref{run2},
where numerical simulations of the dynamics generated by the full Hamiltonian $H$ 
are compared it to the dynamics generated by the effective model (\ref{HXYb}).

The present example considers two atoms in two cavities, initially 
in the state $\frac{1}{\sqrt{2}} (\ket{a_1} + \ket{b_1}) \otimes \ket{a_2}$, and calculates
the occupation probability $p(a_1)$ of the state 
$\ket{a_1}$ which corresponds to the probability of spin 1 to 
point down, $p( \downarrow_1 )$. Figure \ref{run2}{\bf a} shows 
$p(a_1)$ and $p( \downarrow_1 )$ for an effective Hamiltonian 
(\ref{XYZ}) with $B_{\text{tot}} = 0.135$MHz, $J_x = 0.065$MHz, 
$J_y = 0.007$MHz and $J_z = 0.004$MHz and hence 
$|B_{\text{tot}}| > |J_x|$, whereas figure \ref{run2}{\bf b} 
shows $p(a_1)$ and $p( \downarrow_1 )$ for an effective Hamiltonian 
(\ref{XYZ}) with the same $J_x$, $J_y$ and $J_z$ but $B_{\text{tot}} = -0.025$MHz
and hence  $|B_{\text{tot}}| < |J_x|$ \cite{HRP06}.
\begin{figure}
\centering
\includegraphics[width=.8\linewidth]{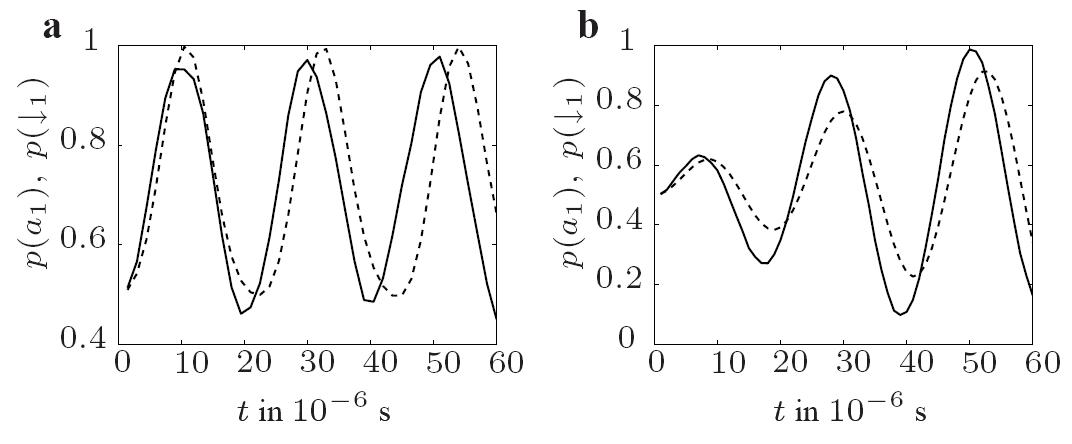}
\caption{The occupation probability $p(a_1)$ of state $\ket{a_1}$ (solid line) and the
probability $p( \downarrow_1 )$ of spin 1 to point down (dashed line) for the parameters
$\omega_e = 10^6$GHz, $\omega_{ab} = 30$GHz, 
$\Delta_a = 30$GHz, $\Delta_b = 60$GHz, $\omega_C = \omega_e - \Delta_b + 2$GHz, $\tilde{\Delta}_a = 15$GHz, $\Omega_a = \Omega_b = 2$GHz, $\Lambda_a = \Lambda_b = 0.71$GHz, $g_a = g_b = 1$GHz, $J_C = 0.2$GHz and
$\delta_1 = -0.0165$GHz (plot {\bf a}) respectively $\delta_1 = -0.0168$GHz (plot {\bf b}). Both, the occupation of the excited atomic states $\langle \ket{e_j} \bra{e_j} \rangle$ and the photon number $\langle a^{\dagger} a \rangle$ are always smaller than 0.03.}
\label{run2}
\end{figure}
Discrepancies between numerical results for the full and 
the effective model are due to higher order terms for the 
parameters $B$, $\tilde{B}$, $J_x$, $J_y$ and $J_z$,
which lead to relative corrections of up to 10\% in the
considered cases. Despite this lack
of accuracy of the second order approximations, the effective model 
is indeed a spin-1/2 Hamiltonian as occupations of excited 
atomic and photon states are negligible.

\section{Cluster State Generation} \label{clusterstate}

Arrays of coupled cavities can be used for the generation of cluster states \cite{RB01},
which form, together with the local addressability, a platform for one-way quantum computation.
One way to generate these states is via the effective Hamiltonian (\ref{HZZ}).
To this end, all atoms are initialized in the states $(\ket{a_j} + \ket{b_j})/\sqrt{2}$,
which can be done via a STIRAP process \cite{FIM05}, and then evolved 
under the Hamiltonian (\ref{HZZ}) for $t = \pi / 4 J_z$.

Figure \ref{cluster} shows the von Neumann entropy of the 
reduced density matrix of one effective spin and the purity of the reduced density matrix of the effective 
spin chain $P_{\text{s}}$ for a full three cavity model. Since 
$E_{\text{vN}} = S \approx log_2 2$ for $t \approx 50 \mu$s while 
the state of the effective spin model remains highly pure 
($P_{\text{s}} = tr[\rho^2] > 0.995$) the degree of entanglement
will be very close to maximal. Thus the levels  $\ket{a_j}$ and $\ket{b_j}$
have been driven into a state which is, up to local unitary rotations, 
very close to a three-qubit cluster state.
\begin{figure}
\centering
\includegraphics[width=.8\linewidth]{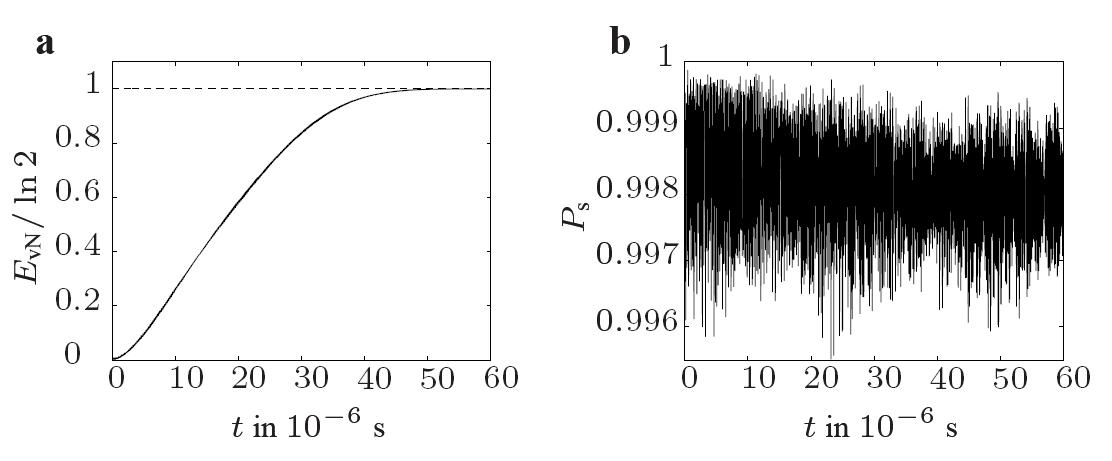}
\caption{{\bf a}: The von Neumann entropy $E_{\text{vN}}$ of the reduced density matrix of 1 effective spin in multiples of $\ln 2$ and {\bf b}: the purity of the reduced state of the effective spin model for 3 cavities where $J_z = 0.021$MHz. The plots assume that no spontaneous emission took place.
For a spontaneous emission rate of $0.1$MHz ($g = 1$GHz), the probability for a decay event in the total time range is 1.5\%. Hence, cluster state generation fails with probability $0.005 \times n$ for $n \ll 1/0.005 = 200$ cavities, irrespectively of the lattice dimension.}
\label{cluster}
\end{figure}

\section{Experimental Implementation} \label{exphei}

For an experimental implementation, the parameters of the 
effective Hamiltonian, $J_x$, $J_y$, $J_z$, $B$ and $\tilde{B}$ 
have to be much larger than rates for decay mechanisms via 
the photons or the excited states $\ket{e_j}$. 
With the definitions $\Omega = \text{max}(\Omega_a, \Omega_b)$, $g = \text{max}(g_a, g_b)$, $\Delta = \text{min}(\Delta_a, \Delta_b)$, the occupation of the excited levels $\ket{e_j}$ and the photon number $n_p$ can be estimated to be $\langle \ket{e_j}\bra{e_j} \rangle \approx |\Omega / 2 \Delta |^2$ and $n_p \approx |(\Omega g / 2 \Delta) \gamma_1 |^2$, whereas the couplings $J_x$, $J_y$ and $J_z$ are approximately $|(\Omega g / 2 \Delta)|^2 \gamma_2$. 

Spontaneous emission from the levels $\ket{e_j}$ at a rate $\Gamma_E$ and cavity decay of photons at a rate $\Gamma_C$ thus lead to effective decay rates 
\begin{equation*}
\Gamma_1 = |\Omega / 2 \Delta |^2 \Gamma_E \hspace{0.4 cm} \text{and} \hspace{0.4 cm} \Gamma_2 = |(\Omega g / 2 \Delta) \gamma_1 |^2 \Gamma_C.
\end{equation*}
Hence, we require 
\begin{equation*}
\Gamma_1 \ll |(\Omega g / 2 \Delta)|^2 \gamma_2  \hspace{0.4 cm} \text{and} \hspace{0.4 cm} \Gamma_2 \ll |(\Omega g / 2 \Delta)|^2 \gamma_2
\end{equation*}
 which implies $\Gamma_E \ll J_C \, g^2 / \delta^2$ and $\Gamma_C \ll J_C$
($J_C < \delta/2$), where, $\delta =  |(\omega_a + \omega_b)/2 - \omega_C|$ for the XX and YY interactions and $\delta =  |\omega - \omega_C|$ for the ZZ interactions and we have approximated $|\gamma_1| \approx \delta^{-1}$ and $|\gamma_2| \approx J_C \delta^{-2}$.
Since photons should be more likely to tunnel to the next cavity than decay into free space,
$\Gamma_C \ll J_C$ should hold in most cases. For $\Gamma_E \ll J_C g^2 / \delta^2$,
to hold, cavities with a high ratio $g / \Gamma_E$ are favorable. Since $\delta > 2 J_C$, the two requirements together imply that the cavities should have a high cooperativity factor. Promising cavity QED systems satisfying these requirements have been discussed in sections \ref{cQED} and \ref{expkerr}.

%
%\include{c-4}
%\include{c-5}
%\part{Entanglement Distillation  with Linear Optics}\label{part3}
%\include{c-6}
%\include{c-7}
%\include{c-8}
%\input{c-9}
%\include{conc}

%\include{c-3}

%% Make the following chapters appear as appendices.
%\appendix
%\include{a-1}
%\part{}
%\include{backmatter}
%\include{acknow}

\appendix

\chapter{Notation} \label{not}

\newcommand{\tabstart}[1]{\noindent \begin{tabular}{p{2.95cm}p{8.7cm}}
    \multicolumn{2}{l}{{\bf #1}} \\ \hline \\[-2.5ex] } 

\newcommand{\tabstop}{\\ \hline \end{tabular}}

\newcommand{\tabinter}{\vspace{2ex}}

\tabstart{General}
  $\log$ & binary logarithm  \\
  $\ln$ & natural logarithm \\
  $e$ & Euler's constant \\
  $\lim$ & limit \\
  $\min$ & minimum \\
  $\max$ & maximum \\
  $\sup$ & supremum \\ 
  $\inf$ & infimum  \\
  $\limsup$ & supremum limit \\ 
  $\liminf$ & infimum limit  \\
  $\poly(n)$ & a polynomial in $n$ (determined from the context) \\
  $\Pr(X \geq \epsilon)$& probability that $X$ is larger or equal to $\epsilon$ \\
  $\mathbb{E}(X)$ & expectation value of random variable $X$ \\
  $\nabla \cdot$ & divergence \\
  $\nabla \times$ & rotational \\
\tabstop

\tabinter

\tabstart{Norms and Distance Measures}

$|| A ||_1$ & trace norm of $A$ \\
$|| A ||$ & operator norm of $A$ \\
$|| f ||_L$ & Lipschitz seminorm of the function $f$\\
$\beta(\mu, \nu)$ & L\'evy-Prohorov metric of the measures $\mu$ and $\nu$ \\
$F(A, B)$ & fidelity of $A$ and $B$ \\
$|| A ||_S$ & separable trace norm of $A$ \\
\tabstop

\tabinter

\tabstart{Sets}

$\mathbb{N}$ & set of natural numbers \\
$\mathbb{R}$ & set of real numbers \\
$\mathbb{C}$ & set of complex numbers \\
${\cal S}_n$ & symmetric group over $\{1, ..., n \}$ \\
$\overline{{\cal M}}$ & closure of ${\cal M}$ \\
$\text{span}({\cal M})$ & linear span of ${\cal M}$ \\
$\text{cone}({\cal M})$ & cone of ${\cal M}$ : $\{ \lambda x : \lambda \in \mathbb{R}_+, x \in {\cal M} \}$ \\
$\text{co}({\cal M})$ & convex hull of ${\cal M}$ \\
${\cal M}^*$ & dual cone of ${\cal M}$ \\
$\text{Sym}({\cal H}^{\otimes k})$ &  symmetric subspace of ${\cal H}^{\otimes k}$ \\
${\cal B}({\cal H})$ & set of bounded operators over ${\cal H}$ \\
${\cal D}({\cal H})$ & state space on ${\cal H}$ \\
${\cal S}({\cal H})$ & set of separable states over ${\cal H}$ \\
$\text{cone}({\cal S})$ & set of unnormalized separable states \\
${\cal W}({\cal H})$ & set of entanglement witnesses over ${\cal H}$ \\
${\cal S}_k({\cal H}^{\otimes k})$ & set of permutation-symmetric states over ${\cal H}^{\otimes k}$  \\
$\ket{\theta}^{[\otimes, n, r]}$ & set of almost power states in $\ket{\theta}$ \\
$B_{\epsilon}(\rho)$ & $\epsilon$ ball around $\rho$ : $\{ \tilde{\rho} : || \rho - \tilde{\rho} ||_1 \leq \epsilon \}$

\tabstop

\tabinter

\tabstart{Abbreviations and Acronyms}  
  
  cQED & Cavity Quantum Electrodynamics \\ 
  CP & Completely Positive \\
  i.i.d. & independent and identically distributed \\
  LOCC & Local Operations and Classical Communication \\
  MPS & Matrix-Product-State \\
  NPPT & Non Positive Partial Transpose \\
  PPT & Positive Partial Transpose \\
  POVM & Positive Operator Valued Measurement \\
  SEPP & Non-Entangling (Separability Preserving) Maps \\
  SEPP$(\epsilon)$ & $\epsilon$-Non-Entangling Maps \\
  SLOCC & Stochastic Local Operations and Classical Communication \\
  EIT & Electromagnetic-Induced-Transparency \\  
  QED & Quantum Electrodynamics \\
  STIRAP & Stimulated Raman adiabatic passage \\  
  TPCP & Trace Preserving Completely Positive \\ 
 % $\P$ & Polynomial Time \\
 % $\NP$ & Non-Deterministic Polynomial Time \\
 % $\UNP$ & Unique Witness $\NP$ \\
 % $\MA$ & Merlin Arthur  \\
 % $\BQP$ & Bounded Error Quantm Polynomial \\
 % $\DQC1$ & Bounded Error One-Clean-Qubit Quantum Polynomial \\
 % $\QMA$ & Quantum Merlin Arthur \\
 % $\QCMA$ & Quantum Classical Merlin Arthur

\tabstop

\tabinter

\tabstart{Bachmann-Landau notations}  
 
$g(n) = O(f(n))$ & $\exists k > 0, n_0 : \forall n > n_0, \hspace{0.1 cm} g(n) \leq k f(n)$ \\ 
$g(n) = \Omega(f(n))$ & $\exists k > 0, n_0 : \forall n > n_0, \hspace{0.1 cm} g(n) \geq k f(n)$ \\ 
$g(n) = o(f(n))$ & $\forall k > 0, \exists n_0 : \forall n > n_0, \hspace{0.1 cm} g(n) \leq k f(n)$ \\ 
$g(n) = \Theta(f(n))$ & $g(n) = O(f(n))$ and $f(n) = \Omega(g(n))$ 
 
\tabstop

\tabinter

\tabstart{Operators and Operations}

$\id$ & identity \\
$\Phi(K)$ & $K \times K$ maximally entangled state \\
$\phi_2$ & $2 \times 2$ maximally entangled state \\
$\tr_{k}$ & partial trace over the $k$-th Hilbert space \\
$\tr_{\backslash k}$ & partial trace over all except the $k$-th Hilbert space \\
$\hat{S}_k$ & symmetrization operator over ${\cal H}^{\otimes k}$ \\
$\dim({\cal V})$ & dimension of the vector space ${\cal V}$ \\
$\tr(X)$ & trace of the hermitian operator $X$ \\
$\rank(X)$ & rank of the hermitian operator $X$ \\
$\lambda_{\max}(X)$ & maximum eigenvalue of $X$ \\
$\lambda_{\min}(X)$ & minimum eigenvalue of $X$ \\
$X^T$ & transpose of the hermitian operator $X$ \\
$X^{\cal y}$ & conjugate transpose of Hermitian operator $X$ \\
$X^{\Gamma_A}$ & partial transpose of the bipartite hermitian operator $X$ with respect to $A$ \\
$X^{\Gamma_B}$ & partial transpose of the bipartite hermitian operator $X$ with respect to $B$ \\

\tabstop

\tabinter

\tabstart{Quantum Optics and Atomic Physics}

$\hbar$ & reduced Plank constant \\
$c$ & speed of light \\
$\epsilon_0$ & free space permittivity \\ 
%$\omega_C$ & cavity mode resonance frequency \\
%$\omega_0$ & atomic transition frequency \\
$V_{\text{mode}}$ & mode volume of the cavity \\
$Q$ & quality factor \\
$\kappa$ & cavity decay rate \\
$\gamma$ & spontaneous emission rate \\
$\xi$ & cooperativity parameter \\

\tabstop

% \tabinter

% To indicate that an operator acts on a product space, we add
% subscripts which correspond to each of the subspaces, e.g., $\rho_{A
%   B} \in \NN(\cH_A \otimes \cH_B)$.  To denote the operator obtained
% by taking the partial trace over some of the subsystems, we omit the
% corresponding subscripts, e.g., $\rho_A = \tr_B(\rho_{A B})$.

\tabinter

\tabstart{Various Quantities}

$h(x)$ & binary Shannon entropy of $(x, 1 - x)$ \\
$S(\rho)$ & von Neumann entropy of $\rho$ \\
$S(\rho | \sigma)$ & relative entropy of $\rho$ and $\sigma$ \\
$S(A|B)_{\rho}$ & conditional entropy of $\rho_{AB}$ \\
$I_c(A\rangle B)_{\rho}$ & coherent information of $\rho_{AB}$ \\
$F_K(\rho)$ & singlet-fraction of $\rho$ to a $K \times K$ maximally entangled state \\
$E_D(\rho)$ & distillable entanglement of $\rho$ \\
$E_C(\rho)$ & entanglement cost of $\rho$ \\
$E_F(\rho)$ & entanglement of formation of $\rho$ \\
$E(\psi)$ & entropy of entanglement of $\ket{\psi}$ \\
$E_{R}(\rho)$ & relative entropy of entanglement of $\rho$ \\
$E_R^{\infty}(\rho)$ & regularized relative entropy of entanglement of $\rho$ \\
$R(\rho)$ & robustness of entanglement of $\rho$ \\
$R_G(\rho)$ & global robustness of entanglement of $\rho$ \\
$LR(\rho)$ & log-robustness of entanglement of $\rho$ \\
$LR_G(\rho)$ & log-global robustness of entanglement of $\rho$ \\
$E_{sq}(\rho)$ & squashed entanglement of $\rho$ \\
$C^{\leftarrow}(\rho)$ & Henderson-Vedral measure of classical correlations of $\rho_{AB}$ \\
$G^{\leftarrow}(\rho)$ & mixed convex-roof of $C^{\leftarrow}(\rho)$ \\
$E_{\cal M}(\rho)$ & ${\cal M}$-relative entropy of $\rho$ \\
$LR_{\cal M}(\rho)$ & log ${\cal M}$-robustness of $\rho$ \\
$S_{\max}(\rho || \sigma)$ & max-relative entropy of $\rho$ and $\sigma$ \\
$LR_{\cal M}^{\epsilon}(\rho)$ & smooth log ${\cal M}$-robustness of $\rho$ \\
$E_D^{ane}(\rho)$ & distillable entanglement of $\rho$ under asymptotically non-entangling operations \\
$E_D^{ne}(\rho)$ &  distillable entanglement of $\rho$ under non-entangling operations \\
$E_C^{ane}(\rho)$ & entanglement cost of $\rho$ under asymptotically non-entangling maps \\
$E_C^{ane}(\rho)$ & entanglement cost of $\rho$ under non-entangling maps \\
$F_{sep}(\rho ; K)$ & deterministic singlet-fraction of $\rho$ with a $K \times K$ maximally entangled states under non-entangling maps \\
$\Sigma_-(z)$ & self-energy operator \\
$Z(H, \beta)$ & partition function of $H$ at inverse temperature $\beta$ \\
$N_{H}(a,b)$ & eigenvalue counting function of $H$ in the interval $[a, b]$ \\
$\Delta(H)$ & spectral gap of $H$ \\
$\Theta$ & Headvise step function \\
\tabstop

\end{document}